\documentclass{article}
\pdfoutput=1

\usepackage[utf8]{inputenc}
\usepackage[english]{babel}
\usepackage[T1]{fontenc}
\usepackage{amsmath}
\usepackage{amsfonts}
\usepackage{amssymb}
\usepackage{amsthm}
\usepackage{latexsym}
\usepackage{stmaryrd}
\usepackage{tikz}
\usepackage{enumitem}
\usetikzlibrary{calc, intersections, fpu}
\usetikzlibrary{arrows.meta} 
\usetikzlibrary{decorations.text}
\usepackage{tikz-cd}
\usepackage[colorlinks=true, linkcolor=blue, hyperindex=true, citecolor = green!60!purple]{hyperref}
\usepackage{comment}
\usepackage{geometry}
\usepackage{dsfont}

\newcommand{\R}{\mathbb{R}}

\newcommand{\Sp}{\mathbb{S}}
\newcommand{\T}{\mathcal{T}}
\newcommand{\N}{\mathds{N}}
\newcommand\inter[2]{\llbracket #1,#2 \rrbracket}

\newcommand{\control}[2]{ ($#2+#1$) .. #1}

\newcommand\modulo[3]{\pgfmathparse{mod (#1,#2)}\pgfmathtruncatemacro{#3}{\pgfmathresult}}

\newcommand{\arrowIn}{\tikz \draw[arrows = {-Stealth[scale width=1]}] (-1pt,0) -- (1pt,0);}
\newcommand{\arrowOut}{\tikz \draw[arrows = {-Stealth[scale width=1]}] (1pt,0) -- (-1pt,0);}
\newcommand{\arrowIns}[1]{\tikz \draw[arrows = {-Stealth[scale width=#1]}] (-1pt,0) -- (1pt,0);}
\newcommand{\arrowOuts}[1]{\tikz \draw[arrows = {-Stealth[scale width=#1]}] (1pt,0) -- (-1pt,0);}

\newcommand{\emphdef}[1]{\textbf{#1}}
\newcommand{\wid}{\text{w}}
\newcommand{\cwid}{\text{cw}}
\newcommand{\bwid}{\text{bw}}
\newcommand{\mcut}{\text{m-cut}}
\newcommand{\cut}{\delta}
\newcommand{\dirmed}{\text{DM}}

\newtheorem{theorem}{Theorem}[section]
\newtheorem{definition}[theorem]{Definition}
\newtheorem{proposition}[theorem]{Proposition}

\newtheorem{lemma}[theorem]{Lemma}

\newcommand{\etc}{\textit{etc}}
\newcommand{\ie}{\textit{i.e.}}
\newcommand{\eg}{\textit{e.g.}}

\title{Well-quasi-orders on embedded planar graphs}  

\author{Corentin Lunel\thanks{Charles University, Prague, Czechia, corentin.lunel@kam.mff.cuni.cz} \and Cl\'ement Maria\thanks{INRIA UniCA, clement.maria@inria.fr}}

\date{}

\begin{document}

\maketitle

\begin{abstract}
The central theorem of structural graph theory states that the graph minor relation is a well-quasi-order on graphs. It has far-reaching consequences, in particular in the study of graph structures and the design of (parameterized) algorithms. In this article, we study two embedded versions of classical minor relations from structural graph theory and prove that they are also well-quasi-orders on general or restricted classes of embedded planar graphs. These embedded minor relations appear naturally for intrinsically embedded objects, such as knot diagrams and surfaces in $\mathbb{R}^3$. 

Handling the extra topological constraints of the embeddings requires careful analysis and extensions of classical methods for the more constrained embedded minor relations. We prove that the embedded version of immersion induces a well-quasi-order on bounded carving-width plane graphs by exhibiting particularly well-structured tree-decompositions and leveraging a classical argument for well-quasi-orders on forests. We deduce that the embedded graph minor relation defines a well-quasi-order on plane graphs via their directed medial graphs, when their branch-width is bounded. We conclude that the embedded graph minor relation is a well-quasi-order on all plane graphs, using classical grids theorems in the unbounded branch-width case.
\end{abstract}

\tableofcontents

\section{Introduction.}

An abstract graph $H$ is a \emphdef{minor} of an abstract graph $G$ if a graph isomorphic to $H$ can be obtained from $G$ by a sequence of edge contractions and edge or vertex deletions. In a series of twenty papers, written over the course of twenty years, Robertson and Seymour proved the \textit{graph minor theorem}~\cite{Graph_Minors_XX}, stating that the minor relation is a \emphdef{well-quasi-order} on the set of graphs, which roughly means that any infinite sequence of graphs contains two comparable elements with respect to the minor relation. Notably, it implies that any class of graph that is closed under taking minors can be characterised by a finite set of excluded minors. This exceptional achievement has had vast consequences in mathematics and computer science, and in particular for our understanding of graph structures and the design of algorithms~\cite{Dvorak_minor_book, Lovasz_minor_survey} (all definitions discussed in this introduction will be properly defined in the subsequent sections). 

The theory has close ties with topology. In particular, for a graph to be embeddable on a surface of genus $g$ is a property closed under taking minors, which implies that the family of genus $g$ graphs is characterised by a finite set of excluded minors. For example, planar graphs are exactly those graphs which exclude the clique $K_5$ and the complete bipartite graph $K_{3,3}$ as minors. In a somehow opposite direction, a keystone of Robertson-Seymour proof is the \textit{graph structure theorem}~\cite{Graph_Minors_XVII}, which broadly speaking states that any graph excluding a fixed minor $H$ can be reconstructed by combining pieces nearly embeddable on a surface where $H$ cannot be embedded.

Despite its deep connection with topology, their theory only applies to \textit{abstract} graphs but not \textit{embedded} graphs. For example, the characterisation of \emphdef{planar graphs} with excluded minors concerns graphs that \textit{can be embedded} in the plane, as opposed to \emphdef{plane graphs}, which are planar graphs together with one (of possibly many distinct) fixed embedding in the sphere $\Sp^2$ (or similarly in the plane $\R^2$). Plane graphs are considered up to equivalence, where two graphs are \emphdef{equivalent} if there is an ambient isotopy of the sphere transforming one of the embedded graphs into the other. This is a stronger notion of equivalence than graph isomorphism: see Figure~\ref{pic_ex_non_equi} where $G$ and $G'$ are distinct plane graphs, although they both stem from the same abstract graph. 

The notion of a minor naturally adapts to the embedded context. More precisely, edge and vertex deletions can be defined on embedded graphs, and an \emphdef{embedded contraction} of an embedded edge $e$ consists of taking a closed disc $D$ in the plane containing $e$ and not intersecting the graph otherwise and contracting $D$ into a point (precise definitions are given in Section~\ref{sec_minor_def}). In consequence, a plane graph $H$ is an \emphdef{embedded minor} of a plane graph $G$ if a plane graph equivalent to $H$ can be obtained from $G$ by a sequence of embedded edge contractions and edge and vertex deletions. This notion is not equivalent to the (abstract) minor relation, as illustrated in Figure~\ref{pic_ex_non_equi}.

\begin{figure}[ht]
\begin{center}
\begin{tikzpicture}
\def\e{0.05}

\coordinate (c0) at (90:1);
\coordinate (c1) at (-150:1);
\coordinate (c2) at (-30:1);
\coordinate (c3) at (0,0);
\coordinate (c4) at ($(c0)+(-10:1)$);

\draw (c0) -- (c1) -- (c2);
\draw (c0) -- (c3);
\draw (c0) -- (c4);
\draw [black!70] (c0) -- (c2);
\node [black!70] at ($(c0)!0.5!(c2)+(50:0.2)$) {$e$};

\foreach \i in {0,...,4}{\fill (c\i) circle (\e cm);};
\node at (-90:1) {$G$};

\begin{scope}[xshift = 3.5cm]
\coordinate (c0) at (90:1);
\coordinate (c1) at (-150:1);
\coordinate (c2) at (-30:1);
\coordinate (c3) at (0,0);
\coordinate (c4) at ($(c0)+(-10:1)$);

\draw (c0) -- (c1) .. controls +(0:1.5) and \control{(c0)}{(-60:1.5)};
\draw (c0) -- (c3);
\draw (c0) -- (c4);

\foreach \i in {0,1,3,4}{\fill (c\i) circle (\e cm);};
\node at (-90:1) {$H$};
\end{scope}

\begin{scope}[xshift = 7cm]
\coordinate (c0) at (90:1);
\coordinate (c1) at (-150:1);
\coordinate (c2) at (-30:1);
\coordinate (c3) at ($(c0)+(-105:1)$);
\coordinate (c4) at ($(c0)+(-75:1)$);

\draw (c0) -- (c1) -- (c2) --cycle;
\draw (c0) -- (c3);
\draw (c0) -- (c4);

\foreach \i in {0,...,4}{\fill (c\i) circle (\e cm);};
\node at (-90:1) {$G'$};
\end{scope}
\end{tikzpicture}
\caption{Two plane graphs $G$ and $G'$ that are isomorphic but not equivalent as plane graphs. As abstract graphs, the graph $H$ is a minor of both $G$ and $G'$. However, it is an embedded minor of $G$, obtained by contracting the embedding of $e$, but it is not an embedded minor of $G'$. Indeed, no embedded edge contraction or deletion in $G'$ results in a plane graph with $2$ faces, both containing a dangling edge in their inside.}
\label{pic_ex_non_equi}
\end{center}
\end{figure}

The \textit{embedded minor} relation appears as a natural relation for low dimensional topologists for whom the precise embedding of a graph matters. In particular, the question of whether plane graphs are well-quasi-ordered by embedded minor is still open, since one cannot apply already known results on abstract graphs \textit{embeddable} on a given surface; notably, it has been raised several times on these highly voted Mathoverflow questions~\cite{Question_mathover_1,Question_mathover_2}. 

\subsection{Our results.}
\label{subsec_intro_result}

The goal of this paper is to prove that \textit{embedded} counterparts of the minor relation, and related concepts such as \textit{immersion}, do define well-quasi-orders on graphs \textit{embedded} in the plane. We strive to offer a self-contained and accessible proof for the bounded width case, which is the main contribution of this work. Our main result is the following general theorem:

\begin{theorem}\label{th_embedded_gm}
Graphs embedded in $\R^2$ are well-quasi-ordered by embedded minors.
\end{theorem}

On a high level, the proof amounts to finding a proper pair of graphs in which one is an embedded minor of the other in any given infinite sequence of embedded graphs. Let us then notice that if one tries to naively apply the graph minor theorem on the abstract graphs associated to the plane graphs of our sequence, one will find a pair of graphs $G_1,G_2$ such that $G_1$ is a minor of $G_2$. However, there is no reason why the embedding that was fixed for $G_1$ would be compatible with the embedding of $G_2$ (see for example, $G_1 = H$ and $G_2 = G'$ of Figure~\ref{pic_ex_non_equi}). 

Hence, we cannot apply the graph minor theorem as a black box, and in fact this is also true for most of the results established for abstract minors. We nonetheless follow the classical approach for this kind of proof; it is common to study separately two cases which depend on whether the \textit{width} of graphs in the sequence is bounded. We rely on different notions of width in this paper, but in our case, they are all equivalent to the famous \emphdef{treewidth}, and they all measure, in their own way, how close a graph is to a tree. They are all defined as the minimum, among all specific tree-like decompositions, of the maximum size of a separation of the graph considered which is induced by the decomposition (the separation is made of vertices in the case of branch-width and treewidth, and of edges in the case of carving-width). 

If the width is unbounded, then most of the usual techniques and results of the graph minor theory can be readily applied to the embedded framework. Indeed, it is known that every planar graph is a minor of a large enough grid. In addition, for a sequence of planar graphs with unbounded width, it is always possible to find graphs admitting (arbitrarily) large grids as minor. Hence, it is possible to construct a pair of comparable graphs by interleaving a large enough grid between the first graph of the sequence and any further graph of the sequence with large enough width. Luckily for our purpose, embeddings of grids are very constrained so that we need only check that the existing constructions are compatible with the embedded framework.

However, this is not true in the bounded width case, where one has no control over the possible embeddings and cannot exhibit substructures as powerful as grids. In this case, graphs are close to trees: one can see them as thickened trees. As such, the classical approach is to generalise the original proof made by Kruskal that trees are well-quasi-ordered~\cite{Kruskal_tree_minor} but with the thickened trees instead. We do so, while considering embedded minors instead of abstract minors. This requires to locally keep track of the embeddings of the subgraphs, considered during the proof, on discs such that the boundary of the subgraph (the separation that it defines) lies on the boundary of the disc. 

On a high level, if one tries to prove the classical results in the embedded context with this approach, one runs into two problems. The first one, for embedded minors, is that the natural tree-like decompositions to consider rely on separations made of vertices, \ie, sets of vertices whose removal leaves the graph disconnected. It becomes then hard to always maintain local embeddings on nice surfaces like discs: one needs to consider pinched surfaces as depicted in Figure~\ref{pic_plane_pinch} if a separation consisting of a single vertex exists on the boundary. 

\begin{figure}[ht]
\begin{center}
\begin{tikzpicture}[scale = 1.5]
\def\e{0.05} \def\op{0.1}
\coordinate (i0) at (45:0.6);
\coordinate (i1) at (135:0.6);
\coordinate (i2) at (-135:0.6);
\coordinate (i3) at (-45:0.6);
\coordinate (i4) at (0,0);
\coordinate (c) at (1.1,0);
\coordinate (c1) at ($(c)+(15:0.6)$);
\coordinate (c2) at ($(c)+(-45:0.6)$);

\begin{scope}[]
\draw (i0) -- (i1) -- (i2) -- (i3) -- cycle;
\draw (i4) -- (i0);
\draw (i4) -- (i1);
\draw (i4) -- (i2);
\draw (i4) -- (i3);
\draw (i0) -- (c);
\draw (i3) -- (c);
\draw (c) -- (c1) -- (c2) -- cycle;
\end{scope}
\begin{scope}[gray!80!white]
\draw (i1) -- +(135:0.75);
\draw (i2) -- +(-150:0.75);
\draw (i2) -- +(-120:0.75);
\draw (c) -- +(80:0.75);
\draw (c) -- +(-85:0.75);
\draw (c1) -- +(30:0.75);
\draw (c1) -- +(70:0.75);
\end{scope}

\filldraw [red!50!purple, smooth cycle, tension = 0.9, fill opacity = \op] plot coordinates {($(i0)+(90:0.15)$) (i1) (i2) ($(i3)+(-90:0.5)$) (c)};
\filldraw [red!50!purple, smooth cycle, tension = 0.9, fill opacity = \op] plot coordinates {($(c2)+(-45:0.15)$) (c1) (c)};
\foreach \i in {i0,i1,i2,i3,i4,c,c1,c2}{\fill (\i) circle (\e cm);};

\node [red!70!black] at (-70:0.8) {$D$};
\node at (0.75,0) {$G_1$};

\begin{scope}[xshift = 5cm]
\coordinate (i0) at (45:0.6);
\coordinate (i1) at (135:0.6);
\coordinate (i2) at (-135:0.6);
\coordinate (i3) at (-45:0.6);
\coordinate (i4) at (0,0);
\coordinate (c) at (1.1,0);
\coordinate (c1) at ($(c)+(15:0.6)$);
\coordinate (c2) at ($(c)+(-45:0.6)$);

\begin{scope}
\draw (i0) -- (i1) -- (i2) -- (i3) -- cycle;
\draw (i4) -- (i0);
\draw (i4) -- (i1);
\draw (i4) -- (i2);
\draw (i4) -- (i3);
\draw (i0) -- (c);
\draw (i3) -- (c);
\draw (c) -- (c1) -- (c2) -- cycle;
\end{scope}
\begin{scope}[gray!80!white]
\draw (i1) -- +(135:0.75) coordinate [midway] (m1);
\draw (i2) -- +(-150:0.75) coordinate [midway] (m2);
\draw (i2) -- +(-120:0.75) coordinate [midway] (m3);
\draw (c) -- +(80:0.75) coordinate [midway] (m4);
\draw (c) -- +(-85:0.9) coordinate [pos = 0.75] (m9);
\draw (c1) -- +(30:0.75) coordinate [midway] (m5);
\draw (c1) -- +(70:0.75) coordinate [midway] (m6);
\path (i3) -- +(-70:0.75) coordinate [midway] (m7);
\path (c2) -- +(-45:0.75) coordinate [midway] (m8);
\end{scope}

\filldraw [red!50!purple, smooth cycle, tension = 0.9, fill opacity = \op] plot coordinates {(m1) (m3) (m7) (m9) (m8) (m5) (m4)};
\foreach \i in {i0,i1,i2,i3,i4,c,c1,c2}{\fill (\i) circle (\e cm);};

\node [red!70!black] at (-45:0.85) {$D$};
\node at (0.75,0) {$G_1$};
\end{scope}
\end{tikzpicture}
\end{center}
\caption{Left: Local embedding of $G_1$ on $D$ with vertices on the boundary, $D$ needs to be pinched. Right: Local embedding of $G_1$ on $D$ with edges on the boundary, $D$ can always be a disc.}
\label{pic_plane_pinch}
\end{figure}

To tackle this issue, we exploit the embeddings of our graph, which allow us to define associated graphs called \textit{medial graphs}. Then, the notion of embedded minors readily translates into another renown graph relation called \textit{immersion}, that we study in the embedded setting and on the medial graphs. An \emphdef{immersion} of a graph $H$ in a graph $G$ is a map that sends vertices of $H$ to vertices of $G$ injectively and edges of $H$ to edge-disjoint paths in $G$ connecting the images of their endpoints. In that framework, the natural width parameter is the \emphdef{carving-width}, which, again, intuitively measures the size of optimal edge-cuts of tree-like decompositions, called \emphdef{carving-decompositions} \cite{Seymour_Ratcatcher}, of the graph (see Section~\ref{subsec_backg_graph} for precise definitions). The separations of the carving-decompositions are, in this case, made of edges so that we can always locally keep track of embeddings on a disc (see right of Figure~\ref{pic_plane_pinch}). In the bounded width case, proving Theorem~\ref{th_embedded_gm} boils down to the following theorem, that we prove:

\begin{theorem}\label{th_emb_immersion_bounded}
Graphs embedded in $\R^2$ with bounded carving-width are well-quasi-ordered by embedded immersion.
\end{theorem}

The second problem encountered, if one attempts to prove either Theorem~\ref{th_embedded_gm} or Theorem~\ref{th_emb_immersion_bounded}, is that one ends up in a situation where a high number of paths leaving a subgraph $G_1$ of a graph $G$ need to be found. However, with plane graphs, these paths can leave $G_1$ in the ``wrong'' direction, see the left of Figure~\ref{pic_plane_pb}. This last problem is, however, avoided if $G \smallsetminus G_1$ lies on the same side of $G_1$, or more formally if $G \smallsetminus G_1$ is connected (see right of Figure~\ref{pic_plane_pb}). Hence, we intuitively want some decompositions of our graphs such that each separation leaves the graph connected on both sides; these are called \emphdef{bond} decompositions \cite{Seymour_Ratcatcher}. For graphs that are connected enough, the bond property allows us to circumvent the difficulties faced when trying to respect the embeddings.

\begin{figure}[ht]
\begin{center}
\begin{tikzpicture}[scale=1.25]
\filldraw [red!5!purple, even odd rule, fill opacity = 0.2] (0,0) circle (1.5 and 0.5) (0,0) circle (0.75 and 0.25);

\filldraw [blue!80!black, even odd rule, fill opacity = 0.2] (0,0) circle (2.4 and 1.25) (0,0) circle (1.75 and 0.75);

\filldraw [blue!80!black, fill opacity = 0.2] (0,0) circle (0.5 and 0.15);

\path (90:0.25) arc (90:75:0.75 and 0.25) coordinate (c1);
\path (90:0.25) arc (90:30:0.75 and 0.25) coordinate (c2);

\draw [black!50] (c1) -- +(-160:0.25) coordinate (c1'');
\draw [black!50] (c1) -- +(-120:0.25) coordinate (c2'');
\draw [black!50] (c2) -- +(-160:0.25) coordinate (pb);

\path (90:0.5) arc (90:60:1.5 and 0.5) coordinate (o1);
\path (90:0.5) arc (90:45:1.5 and 0.5) coordinate (o2);
\path (90:0.5) arc (90:20:1.5 and 0.5) coordinate (o3);

\draw [black!50] (o1) -- +(30:0.15) ;
\draw [black!50] (o2) -- +(30:0.15);
\draw [black!50] (o3) -- +(30:0.15);

\node [red!50!purple] at (-1.125,0) {$G_1$};
\node [blue!80!black] at (-2.05,0) {$G_2$};

\path (90:0.75) arc (90:-90:1.75 and 0.75) coordinate (p1);

\draw [thick, dashed] (pb) -- (p1);

\begin{scope}[xshift= 6cm]
\filldraw [red!5!purple, fill opacity = 0.2] (0,0) circle (1.5 and 0.5);

\filldraw [blue!80!black, even odd rule, fill opacity = 0.2] (0,0) circle (2.4 and 1.25) (0,0) circle (1.75 and 0.75);

\path (90:0.5) arc (90:75:1.5 and 0.5) coordinate (c1);
\path (90:0.5) arc (90:-45:1.5 and 0.5) coordinate (c2);

\draw [black!50] (c1) -- +(80:0.15) coordinate (c1'');
\draw [black!50] (c1) -- +(120:0.15) coordinate (c2'');
\draw [black!50] (c2) -- +(-90:0.15) coordinate (pb);

\path (90:0.5) arc (90:60:1.5 and 0.5) coordinate (o1);
\path (90:0.5) arc (90:45:1.5 and 0.5) coordinate (o2);
\path (90:0.5) arc (90:20:1.5 and 0.5) coordinate (o3);

\draw [black!50] (o1) -- +(30:0.15) ;
\draw [black!50] (o2) -- +(30:0.15);
\draw [black!50] (o3) -- +(30:0.15);

\node [red!50!purple] at (0,0) {$G'_1$};
\node [blue!80!black] at (-2.05,0) {$G'_2$};

\end{scope}
\end{tikzpicture}
\end{center}
\caption{Left: Paths can leave $G_1$ on the ``wrong'' side in a general decomposition. Right: All paths leave $G'_1$ properly in a bond decomposition.}
\label{pic_plane_pb}
\end{figure}

However, to perform the generalisation of the original proof made by Kruskal as mentioned earlier, we rely on a classical method of Nash-Williams \cite{nash-williams_kruskal} which requires another strong combinatorial property of decompositions which is called \textit{linked}. The definition of linked is quite technical, we refer to Section~\ref{sec_bond_linked} for its precise definition; we simply mention here that this property is purely combinatorial and does not depend on embeddings. The existence of bond carving-decompositions of connected enough graphs is already known, as well as the existence of linked carving-decompositions. We prove here that we can enforce both properties at the same time for $2$-vertex-connected abstract graphs:

\begin{theorem}\label{th_bond_link_decomp}
Let $G$ be a $2$-vertex-connected graph. Then, there exists a bond-linked carving-decomposition of $G$ of width $\cwid(G)$.
\end{theorem}

To get rid of the connectivity condition in the last theorem, we relax the bond property from separations leaving two connected subgraphs to separations leaving two subgraphs lying on disjoint discs. Note that this property, called the \emphdef{disc property}, is dependent on the specific embeddings, but that by \textit{cut-cycle duality}, the disc property is satisfied by any bond decomposition of a plane graph. Furthermore, the disc property is exactly what is needed to circumvent the issues presented in Figure~\ref{pic_plane_pb}. We call decompositions which are both linked, and satisfy the disc property, \emphdef{disc-linked decompositions}. Exploiting bond-linked carving-decompositions, we prove their existence for general plane graphs:  

\begin{proposition}\label{prop_disc_decomp}
Let $G$ be a graph embedded in $\Sp^2$. Then, there exists a disc-linked carving-decomposition of $G$ of width $\cwid (G)$.
\end{proposition}

Thanks to decompositions of Proposition~\ref{prop_disc_decomp}, we can leverage the classical argument by Nash-Williams~\cite{nash-williams_kruskal} to prove Theorem~\ref{th_emb_immersion_bounded}. Then, Theorem~\ref{th_emb_immersion_bounded} applied to medial graphs proves Theorem~\ref{th_embedded_gm} in the bounded width case, which is the case that remained. We also present two applications of Theorem~\ref{th_embedded_gm} to knot theory, which will be presented in Section~\ref{subsec_related_app}.

\subsection{Related work and applications.}
\label{subsec_related_app}

\subparagraph*{Abstract, embeddable, and embedded graphs.}
We first emphasise again the differences between abstract graphs, embeddable graphs, and embedded graphs. Graphs, called abstract graphs in this paper, consist of the data of vertices and edges. Embeddable graphs in some surface $\Sigma$ are abstract graphs for which it is known that at least \textit{an embedding in $\Sigma$ exists}. Embedded graphs in $\Sigma$, are embeddable graphs together with a \textit{fixed embedding}, it is the data of a set of vertices, a set of edges, and a given embedding. Embedded graphs are the focus of this paper and the challenge that we address is to find a well-quasi-order on embedded graphs compatible with their fixed embeddings.

It is well-known that all abstract graphs are embeddable graphs in some surface $\Sigma$ provided that the genus of $\Sigma$ is high enough. Indeed, for an abstract graph $G$ one can build such a surface by drawing $G$ in the plane and replacing each crossing between two edges with a handle bridging over the crossing. This construction does not yield, in general, a surface of minimal genus containing $G$. In fact, finding such a surface is \textbf{NP}-hard \cite{Thomassen_min_gen_hard}, and finding methods to find them is still subject to recent developments~\cite{Conder_min_gen}. Explicit embeddings of embeddable graphs can be very useful from a complexity point of view. On planar graphs, for example, an embedding can be computed in linear time~\cite{Tarjan_planar}, and leveraged to compute balanced separators of planar graphs of size $\sqrt{V(G)}$ via their dual graph. These separators can then be leveraged in efficient divide-and-conquer algorithms, for instance, many \textbf{NP}-hard problems on graphs admit sub-exponential time algorithm when restricted to planar graphs~ \cite{Bodlaender_planar_efficient,Tarjan_separator}.

From the graph minor perspective, embeddings have been instrumental to the theory since the graph minor theorem exploits the structure theorem~\cite[(1.3)]{Graph_Minor_XVI}, which shows that graphs are ``almost'' embeddable on some surface. However, even though the graph minor theory exploits embeddings, the graph minor relation does not take into account the fixed embeddings as explained earlier in this introduction: existing results do not apply directly to our framework. However, let us detail 2 cases where existing results are readily enough to prove Theorem~\ref{th_embedded_gm}. The first one is when the graphs considered are restricted to $3$-vertex-connected graphs. By Whitney's theorem~\cite{Whitney_unique}, they only admit one embedding, so that if an abstract graph $H$ is an abstract minor of an abstract $G$, the unique embedding of $H$ will be an embedded minor of the unique embedding of $G$. Note that this fact does not provide a simple proof for our theorem by simply seeing each embedded graph $G$ as a subgraph of a 3-connected graph $G'$ which combinatorially encodes the embedding. First it is not clear how to deal with non-simple graphs, and more importantly, one cannot deduce that $H$ is an embedded minor of $G$ if $H'$ is a minor of $G'$: added edges of $G'$ may be necessary for $H$ to be a minor of $G'$. The second case is when graphs are restricted to plane trees, i.e., embedded trees in the sphere. They are handled by the original proof of Kruskal~\cite{Kruskal_tree_minor}, which proves that embedded trees (with labels) are well-quasi-ordered. Nevertheless, the trees that we leverage in our proofs are those of decompositions, and we do not exploit their embeddings in our proofs.

\subparagraph*{Graph Immersion.}
The Nash-Williams' conjecture is the counterpart of the graph minor theorem for immersions (on abstract graphs): abstract graphs are well-quasi-ordered by immersion. It was established in the last paper of the graph minor series~\cite{Graph_Minors_XXIII} using the tools developed in it. In this article, they mention that a stronger notion of immersion is left unproven: \textit{strong immersion}, where vertices in the image of edges should be distinct from the images of vertices. Since then, establishing structural results about strengthening of immersion or generalisation of immersion, to directed graphs for instance, has been the focus of an increasing body of research~\cite{Devos_min_degree_immersion, Dvorak_structure_th, Liu_wqo_strong_digraph, Diestel_Walls_strong, Kawarabayashi_immer_embedable, Cavallaro_wqo_bd_cw}. In that regard, Theorem~\ref{th_emb_immersion_bounded}   is a parallel to the classical result of well-quasi-ordering by immersion in the planar embeddings context with the additional bounded carving-width constraint.

We now emphasise that we exploit immersions in our proof via the correspondence between embedded minors in a plane graph $G$, and embedded (directed) immersions in its (directed) \textit{medial graph} $DM(G)$. The idea of exploiting the planar graph-directed medial graph correspondence in this very context was used in~\cite{Medina_wqolink} in order to attempt a proof of Theorem~\ref{th_embedded_gm} to get a result about knot theory. Unfortunately, their proof contains a major gap: their references 10, 11, and 20 (respectively \cite{Devos_min_degree_immersion, DeVos_forbid_clique, Hannie_PhD} in this paper) that they use as black boxes to prove their Fact 2.3 and Proposition 2.4 are not expressed in the embedded context. Hence, as explained several times in this introduction, they cannot be used as black boxes in their context of ``good digraphs'', which are embedded directed graphs (the restriction to good digraphs is not enough as can be seen in Figure~\ref{pic_embeddable_vs_embedded}).

Crucially, we boil down the proof of Theorem~\ref{th_embedded_gm} to proving that directed medial graphs embedded in the sphere of bounded carving-width are well-quasi-ordered by embedded directed immersion (Proposition~\ref{prop_imm_emb_dirmed}). These directed medial graphs naturally have an Eulerian embedded structure (same number of incoming and leaving edges at each vertex, which alternate on the embedding) so that our result is really close to the main results of \cite{Kawarabayashi_immer_embedable} stating that Eulerian directed graphs that are Eulerian embeddable into a surface $\Sigma$ are well-quasi-ordered by strong immersion. The objective of the research by Cavallaro, Kawarabayashi, and Kreutzer in \cite{Kawarabayashi_immer_embedable, Cavallaro_wqo_bd_cw}\footnote{The paper \cite{Kawarabayashi_immer_embedable}, a preprint as of the writing of this paper, has been split. The first half, dedicated to proving that directed graphs of bounded carving width are well-quasi-ordered by strong immersion, has been published in \cite{Cavallaro_wqo_bd_cw}. The results there do not concern embedded graphs and their associated relations.} is to develop an analogue of the well-quasi-ordering by immersion in the context of directed graphs (digraphs) and strong immersion. Their main theorem does not directly imply Proposition~\ref{prop_imm_emb_dirmed} as pictured in Figure~\ref{pic_embeddable_vs_embedded}, since they consider embeddable graphs as opposed to embedded ones; and to the best of our understanding, their methods do not respect embeddings so that our result is not a direct consequence of theirs. However, as testimony to the deep connections between topology and structural theory, they do exploit embeddings, and in particular their \textit{representativity}, to exhibit useful substructures in their digraphs.

\begin{figure}[ht]
\begin{center}
\begin{tikzpicture}[scale=0.675]
\def\e{0.065} \def\ll{1} \def\nl{0.75}
\coordinate (c0) at (0,0);
\coordinate (c1) at (0,1);
\coordinate (c2) at (1.5,2.5);
\coordinate (c3) at (2.5,2.5);
\coordinate (c4) at (4,1);
\coordinate (c5) at (4,0);
\coordinate (c6) at (-0.75,1);
\coordinate (c7) at (-0.75,0);
\coordinate (c8) at (0.75,1);
\coordinate (c9) at (0.75,0);
\coordinate (c10) at (3.25,1);
\coordinate (c11) at (3.25,0);
\coordinate (c12) at (4.75,1);
\coordinate (c13) at (4.75,0);
\coordinate (c14) at (1.5,3.5);
\coordinate (c15) at (2.5,3.5);
\coordinate (c16) at (1.5,4.5);
\coordinate (c17) at (2.5,4.5);
\coordinate (c2) at (1.5,2.5);
\coordinate (c3) at (2.5,2.5);

\begin{scope}[gray]
\draw [-Stealth] (c0) -- (c1);
\draw [Stealth-] (c1) .. controls +(90:\ll) and \control{(c2)}{(180:\ll)};
\draw [-Stealth] (c2) -- (c3);
\draw [Stealth-] (c3) .. controls +(0:\ll) and \control{(c4)}{(90:\ll)};
\draw [-Stealth] (c4) -- (c5);
\draw [Stealth-] (c5) .. controls +(-90:2*\ll) and \control{(c0)}
{(-90:2*\ll)};
\draw [-Stealth] (c3) -- (c15);
\draw [-Stealth] (c17) -- (c15);
\draw [-Stealth] (c17) .. controls +(90:\nl) and \control{(c16)}{(90:\nl)};
\draw [-Stealth] (c14) -- (c16);
\draw [-Stealth] (c14) -- (c2);
\draw [-Stealth] (c3) .. controls +(-90:\nl) and \control{(c2)}{(-90:\nl)};
\draw [-Stealth] (c15) -- (c14);
\draw [-Stealth] (c16) .. controls +(180:\nl) and \control{(c14)}{(180:\nl)};
\draw [-Stealth] (c16) -- (c17);
\draw [-Stealth] (c15) .. controls +(0:\nl) and \control{(c17)}{(0:\nl)};
\draw [-Stealth] (c1) -- (c8) .. controls +(0:\nl) and \control{(c9)}{(0:\nl)} -- (c0);
\draw [-Stealth] (c1) -- (c6); 
\draw [-Stealth] (c7) -- (c0); 
\draw [-Stealth] (c7) .. controls +(180:\nl) and \control{(c6)}{(180:\nl)};
\draw [-Stealth] (c6) -- (c7);
\draw [-Stealth] (c6) .. controls +(90:\nl) and \control{($(c6)+(-1,0)$)}{(90:\nl)} -- ($(c7)+(-1,0)$) .. controls +(-90:\nl) and \control{(c7)}{(-90:\nl)};
\draw [-Stealth] (c5) -- (c13) .. controls +(0:\nl) and \control{(c12)}{(0:\nl)} -- (c4);
\draw [-Stealth] (c5) -- (c11) .. controls +(180:\nl) and \control{(c10)}{(180:\nl)} -- (c4);
\end{scope}

\foreach \i in {0,1,2,3,4,5,6,7,14,15,16,17}{\fill (c\i) circle (\e cm);};

\node at (2,-2) {$G_1$};

\begin{scope}[xshift = 8cm]
\coordinate (c0) at (0,0);
\coordinate (c1) at (0,1);
\coordinate (c2) at (1.5,2.5);
\coordinate (c3) at (2.5,2.5);
\coordinate (c4) at (4,1);
\coordinate (c5) at (4,0);
\coordinate (c6) at (-0.75,1);
\coordinate (c7) at (-0.75,0);
\coordinate (c8) at (0.75,1);
\coordinate (c9) at (0.75,0);
\coordinate (c10) at (3.25,1);
\coordinate (c11) at (3.25,0);
\coordinate (c12) at (4.75,1);
\coordinate (c13) at (4.75,0);
\coordinate (c14) at (1.5,3.5);
\coordinate (c15) at (2.5,3.5);
\coordinate (c16) at (1.5,4.5);
\coordinate (c17) at (2.5,4.5);
\coordinate (c2) at (1.5,2.5);
\coordinate (c3) at (2.5,2.5);

\begin{scope}[gray]
\draw [-Stealth] (c0) -- (c1);
\draw [Stealth-] (c1) .. controls +(90:\ll) and \control{(c2)}{(180:\ll)};
\draw [-Stealth] (c2) -- (c3);
\draw [Stealth-] (c3) .. controls +(0:\ll) and \control{(c4)}{(90:\ll)};
\draw [-Stealth] (c4) -- (c5);
\draw [Stealth-] (c5) .. controls +(-90:2*\ll) and \control{(c0)}
{(-90:2*\ll)};
\draw [-Stealth] (c3) -- (c15);
\draw [-Stealth] (c17) -- (c15);
\draw [-Stealth] (c17) .. controls +(90:\nl) and \control{(c16)}{(90:\nl)};
\draw [-Stealth] (c14) -- (c16);
\draw [-Stealth] (c14) -- (c2);
\draw [-Stealth] (c3) .. controls +(-90:\nl) and \control{(c2)}{(-90:\nl)};
\draw [-Stealth] (c15) -- (c14);
\draw [-Stealth] (c16) .. controls +(180:\nl) and \control{(c14)}{(180:\nl)};
\draw [-Stealth] (c16) -- (c17);
\draw [-Stealth] (c15) .. controls +(0:\nl) and \control{(c17)}{(0:\nl)};
\draw [Stealth-] (c8) .. controls +(0:\nl) and \control{(c9)}{(0:\nl)};
\draw [-Stealth] (c1) -- (c8); 
\draw [-Stealth] (c9) -- (c0); 
\draw [Stealth-] (c0) -- (c7) .. controls +(180:\nl) and \control{(c6)}{(180:\nl)} -- (c1);
\draw [-Stealth] (c8) -- (c9);
\draw [-Stealth] (c8) .. controls +(90:\nl) and \control{($(c8)+(1,0)$)}{(90:\nl)} -- ($(c9)+(1,0)$) .. controls +(-90:\nl) and \control{(c9)}{(-90:\nl)};
\draw [-Stealth] (c5) -- (c13) .. controls +(0:\nl) and \control{(c12)}{(0:\nl)} -- (c4);
\draw [-Stealth] (c5) -- (c11);
\draw [-Stealth] (c10) .. controls +(180:\nl) and \control{(c11)}{(180:\nl)};
\draw [-Stealth] (c10) -- (c4);
\draw [-Stealth] (c11) .. controls +(-90:\nl) and \control{($(c11)+(-1,0)$)}{(-90:\nl)} -- ($(c10)+(-1,0)$) .. controls +(90:\nl) and \control{(c10)}{(90:\nl)};
\draw [-Stealth] (c11) -- (c10);
\end{scope}

\foreach \i in {0,1,2,3,4,5,8,9,10,11,14,15,16,17}{\fill (c\i) circle (\e cm);};

\node at (2,-2) {$G_2$};
\end{scope}
\end{tikzpicture}
\caption{$G_1$ is a directed immersion of $G_2$ where $G_1$ and $G_2$ are both Eulerian embeddable directed graphs, but $G_1$ is not an embedded immersion of $G_2$ for the given embeddings.}
\label{pic_embeddable_vs_embedded}
\end{center}
\end{figure}

\subparagraph*{Contributions.}
Since our approach is classical, many approaches and ideas are standards, although implementing them in the embedded case often requires extra care. In that regard, we detail at the beginning of each section the references that we exploit when it is relevant. For the sake of completeness, and because we felt that it could be enlightening for our results, we also included some classical proofs up to reformulation (in such cases they are clearly mentioned as such). As explained in the next part which present some applications of Theorem~\ref{th_embedded_gm} to knot theory, the extra care needed to handle the embeddings is trickier that one might think, since at least two attempts to prove Theorem~\ref{th_embedded_gm} contain gaps.

The main technical contribution of this paper is Theorem~\ref{th_bond_link_decomp} stating the existence of carving-decompositions both bond and linked, which is established in Section~\ref{sec_bond_linked}. To prove it, our method amounts to combine two existing approaches: the ones used to prove the existence of bond carving-decompositions \cite{Seymour_Ratcatcher} and linked carving-decompositions \cite{Geelen_wqo_branchwidth}. It is highly non-trivial that these two approaches can be combined: they both amount to minimising a relevant potential defined on carving-decompositions. However, it turns out that the minimals for each of these two different potentials do not always overlap. Nonetheless, on the one hand, the bond property is only satisfied at the minimum of its potential, while, on the other hand, the linked property can be satisfied even if its potential is non-minimal. In the proof of Theorem~\ref{th_bond_link_decomp}, we show that a linked carving-decomposition of minimal width and minimal potential for the bond property is also bond. As otherwise, we could decrease the bond potential while preserving the linkedness and general width (it relies on Lemma~\ref{lem_bl} and Lemma~\ref{lem_1nl} whose proofs and applications are the crux of this paper).

These bond-linked decompositions cannot exist on graphs admitting a \textit{cut-vertex}, hence, to establish Theorem~\ref{th_emb_immersion_bounded}, we prove in Section~\ref{sec_disc_linked_carving_dec} Proposition~\ref{prop_disc_decomp} stating the existence of decompositions which are both linked and such that each separation is realised by an embedded Jordan curve (this is the disc property). In \cite{Bodlaender_planar_efficient} Dorn, Penninkx, Bodlaendern, and Fomin established the existence of branch-decompositions satisfying the disc property (sphere-cut decompositions), but it is unclear to us how one can adapt them to be additionally linked. The work of Section~\ref{sec_disc_linked_carving_dec} is new and technical, but more tedious than hard.

Once Proposition~\ref{prop_disc_decomp} is established, proving Theorem~\ref{th_emb_immersion_bounded} is done following the classical argument of Nash-Williams \cite{nash-williams_kruskal}. This requires some extra work compared to the usual setting; this is done in Section~\ref{sec_nash_williams_argument}. As mentioned earlier, the hard case of Theorem~\ref{th_embedded_gm} is the bounded width case, and is handled using Theorem~\ref{th_emb_immersion_bounded}. The connection between embedded immersions and embedded minors comes from the directed medial graphs (following ideas of \cite{Medina_wqolink}), we repeat ideas from Section~\ref{sec_nash_williams_argument} to adapt them to the directed setting in Section~\ref{subsec_im_gm}. We additionally check that existing constructions already handle the unbounded width case in Section~\ref{subsec_high_bw}.

\subparagraph*{Application to knot theory.}
The \textit{embedded} version of the graph minor theorem has already found applications for intrinsically embedded objects found in knot theory~\cite{baader_divide,Medina_wqolink}. Indeed, knots, which are embeddings of the circle $\Sp^1$ into $\Sp^3$ considered up to \textit {ambient isotopy} of $\Sp^3$, are intrinsically defined by an embedding. We refer to \cite{Rolfsen_Knots} for an introduction to knot theory but want to emphasise that \emph{this discussion is an explanation for some of our motivation, and not a prerequisite to understanding the remainder of this paper}. When knots admit a non-degenerate embedding, polygonal, for example, they also admit a \textit{regular projection} in the plane, which is called a \emphdef{knot diagram}. Formally, these knot diagrams are $4$-regular embedded graphs where each vertex, which represents a crossing in the projection of the knot, is labelled to indicate which strand lies above the other one in space. Knots are intrinsically defined as embeddings constraining which knot diagrams, i.e., which plane graphs, can be obtained from them. The Reidemeister theorem \cite{Reidemeister_reide_th}, is a fundamental result of knot theory stating that two knots are equivalent if and only if their diagrams are connected by local modifications called Reidemeister moves. Thus, there is a strong impetus to leverage embedded graphs to study knots.

Approaching knots by their diagrams has been a very successful in knot theory, both from the theoretical and computational point of view, as witnessed by the numerous knot invariant defined at the level of diagrams \cite{Kauffman_Jones_poly, Manolescu_floer_diag} and the many FPT algorithms relying on knot diagrams of low treewidth \cite{maria_FPT_quantum, Burton_2018,Malowsky_Jones}. This approach, when based on a structural graph theory, has been already been fruitful \cite{baader_minor_quasipositive, baader_positive, Lunel_spherewidth}, and would greatly benefit from Theorem~\ref{th_embedded_gm} as shown by \cite{baader_divide, Medina_wqolink}. 

In \cite{baader_divide}, Baader and Dehornoy study knots and surfaces arising from \textit{divides} \cite[Definition 1.1]{baader_divide}, which are graphs embedded on a disc, they determine minors characterising maximal \textit{signature} (this is an invariant of surfaces with boundaries). On a high level, divides allow them to define embedded surfaces with boundaries, and the embedded minor relation readily translates to a containment relation on surfaces preserving the property of having maximal signature. Then, by exploiting their Lemma 2.3 (unproven at the time, as acknowledged by the authors on arXiv), they prove a variation of Theorem~\ref{th_embedded_gm} (their Proposition B) allowing them to prove that their surface containment relation is a well-quasi-order on their surfaces. Finally, they exhibit the set of forbidden minors characterising maximal signature. The boundaries of these surfaces are knots or links (union of knots), whose signature is an invariant.

In \cite{Medina_wqolink}, Medina, Mohar, and Salazar claim to prove Theorem~\ref{th_embedded_gm} (their Theorem 1.1), and use it to study a diagrammatic relation on knot. As explained above, their proof of contains a gap, however, the application to knot theory is very interesting, we broadly present it here. They say that $ D \leadsto L$ where $L$ is a link and $D$ a diagram, if there exists a sequence of crossing exchanges (switching between crossings \tikz{\draw (0,0) -- (0.2,0.2); \draw (0.2,0) -- ++(135:0.08); \draw (0,0.2) -- ++(-45:0.08);} and \tikz{\draw (0,0.2) -- (0.2,0); \draw (0,0) -- ++(45:0.08); \draw (0.2,0.2) -- ++(-135:0.08);}) and smoothings (replacing a crossing \tikz{\draw (0,0) -- (0.2,0.2); \draw (0.2,0) -- ++(135:0.08); \draw (0,0.2) -- ++(-45:0.08);} by one of its two resolutions \tikz{\draw (0,0) .. controls (0.09,0.1) and (0.09,0.1) .. (0,0.2); \draw (0.2,0) .. controls (0.11,0.1) and (0.11,0.1) .. (0.2,0.2);} and \tikz{\draw (0,0) .. controls (0.1,0.09) and (0.1,0.09) .. (0.2,0); \draw (0,0.2) .. controls (0.1,0.11) and (0.1,0.11) .. (0.2,0.2);}) transforming $D$ into a diagram of $L$. Then, Theorem~\ref{th_embedded_gm} implies that for each link $L$, there exists a finite number of diagrams ${C_1, ..., C_m}$ such that $D \leadsto L$ if and only if one of the $C_i$'s is comparable to $D$ (embedded immersion of $D$ where vertex labels are ignored). By describing an algorithm checking the embedded relations (embedded immersion via embedded minor thanks to directed medial graph, see Lemma~\ref{lem_equiv_emb_minors}), they provide a polynomial-time algorithm solving $\cdot \leadsto L$ for each link $L$.

By providing a proof for Theorem~\ref{th_embedded_gm}, we fix former contributions pertaining to knot theory~\cite{Medina_wqolink, baader_divide} (although \cite{baader_divide} may require some adaptations). Furthermore, we expect this result to have further implications in knot theory, if only for diagrammatic notions similar to the one introduced in \cite{Medina_wqolink} or surface constructions present in \cite{Lunel_defect,baader_minor_quasipositive, baader_positive} which can now exploit all plane graphs as a basis. 
On a broader scale, there is a growing interested in low dimensional topology for structural approaches~\cite{Huszar_JSJ, Lunel_defect,Huszar_treewidth,mariapurcell,Lunel_spherewidth}, our result naturally fits in the development of bridges between these two fields.

\subsection{Organisation of this paper.} 

We first provide a wide standard background on graphs, topology, and orders in Section~\ref{sec_prelim}. We then provide formal definitions and characterisations for the embedded minors relation in Section~\ref{sec_minor_def}. Then, we establish the existence of bond-linked carving-decompositions, Theorem~\ref{th_bond_link_decomp}, in Section~\ref{sec_bond_linked}, and disc-linked carving-decomposition, Theorem~\ref{prop_disc_decomp}, in Section~\ref{sec_disc_linked_carving_dec}. We then use these decompositions to prove Theorem~\ref{th_emb_immersion_bounded} in the bounded carving-width case in Section~\ref{sec_bounded_cw}, by building on an argument of Nash-Williams for well-quasi-orders on trees (Section~\ref{subsec_order_edge} and~\ref{sec_nash_williams_argument}). We then prove that plane graphs are well-quasi-ordered by embedded minor in Section~\ref{sec_em_gb} by first applying our methods to medial graphs when initial graphs have bounded branch-width, in Section~\ref{subsec_im_gm}. The case of unbounded branch-width is treated by studying embedded versions of grid theorems in Section~\ref{subsec_high_bw}.

\section{Preliminaries.}
\label{sec_prelim}

To ensure accessibility to both the graph theory and the topology communities, we introduce basic and more advanced notions from both fields. We invite the reader to skip any non necessary background section. We refer the reader to~\cite{Diestel_Graph_Theory} for concepts in graph theory that we do not reintroduce in the following.

\subsection{Background on graphs.}
\label{subsec_backg_graph}

\subparagraph*{Graphs.}
Let $G$ be a graph, whose set of vertices is denoted by $V(G)$, and multi-set of edges denoted by $E(G)$. Graphs may have \emphdef{self-loops} (edges with identical endpoints) and \emphdef{multi-edges} (a set with several identical edges, we do not give here a formal framework for this as they will be introduced in the relevant sections).

The \emphdef{neighbours} $N_G(v)$ of a vertex $v$ of $G$ is the subset of vertices of $G$ sharing an edge with $v$, \ie, $u \in N_G(v) \Leftrightarrow (u,v) \in E(G)$. 
Let $A,B \subseteq V(G)$ be two disjoint subsets of vertices of $G$; the \emphdef{cut} between $A$ and $B$ in $G$, denoted by $\cut_G(A,B)$, is the set of edges with exactly one endpoint in $A$ and one endpoint in $B$, \ie, $\cut_G(A,B) \colon= \{ (u,v) \in E(G) | u \in A,~  v \in B\}$. When $B = V(G) \smallsetminus A$, we simply write $\cut_G(A)=\cut_G(A,B)$ for the cut induced by $A$. 

The \emphdef{subgraph induced} by $A \subseteq V(G)$, denoted by $G[A]$, is the subgraph of $G$ on vertices $A$ where only edges with both endpoints in $A$ remain: $G[A] = (A, \{ e \in E(G) \cap A \times A \})$. The cut between two vertex-disjoint subgraphs is defined to be the cut between their vertex sets. A graph is said to be \emphdef{$k$-vertex-connected} if removing any $k-1$ vertices leaves the graph connected. A connected graph which is not $2$-vertex-connected admits at least a \emphdef{cut-vertex}: a vertex whose removal renders $G$ disconnected. The \emphdef{degree} of a vertex is the number of edges incident to it, counted with multiplicity: self-loops count as two-edges. A graph is said to be \emphdef{$k$-regular} if all vertices have degree $k$.

When there is no ambiguity, we drop $G$ from the notations, and write $V$ and $E$ for the vertices and edges, $\cut(A,B)$ for a cut between $A$ and $B$, \etc.

A \emphdef{path} $P$ of size $k \leq 1$ in $G$, between vertices $u,v \in V(G)$, is a sequence $v_0,e_0,v_1,e_1, \ldots, e_{k-1},v_k$, where $u=v_0$,  $v_k = v$, and for each $i \in \inter{0}{k-1}$, $v_i$ and $v_{i+1}$ are endpoints of $e_i$ and the $v_i$'s are distinct. A \emphdef{cycle} of $G$ is a path of $G$, except on its endpoints which are the same.  A \emphdef{tree} is a connected acyclic graph (no cycles in it) and a \emphdef{forest} is a union of trees, i.e., an acyclic graph. A graph is said to be \emphdef{Hamiltonian} if there exists a single cycle containing all its vertices.

A \emphdef{digraph} or directed graph is a graph where edges have a direction, \ie, a head and a tail. A path in a digraph has the additional restriction that the head of $e_i$ must be tail of $e_{i+1}$.

An \emphdef{edge subdivision} in $G$ is the operation of removing edge $(u,v)$ from $G$ and replacing it by $(u,w)$ and $(w,v)$ where $w$ is a new vertex. If the edge is embedded, it can be realised by choosing a point in the relative interior of the edge to be the image of the new vertex $w$.

\subparagraph*{Width invariants}

Introduced by Seymour and Thomas in 1994~\cite{Seymour_Ratcatcher}, the carving-width is another graph parameter aiming at quantifying how close a graph is to a tree. 

A tree is \emphdef{binary} if its vertices have either degree $3$, called \emphdef{inner vertices}, or degree $1$, called \emphdef{leaves}. Denote by $L(T)$ the set of leaves of a tree $T$, i.e., its vertices of degree $1$. Let $G$ be a graph. A \emphdef{carving-decomposition} of $G$ is a pair $\mathcal (T,\phi)$ where $T$ is a binary tree and $\phi$ is an injective map from the vertices of $G$ to the leaves of $T$: $\phi : V(G) \rightarrow L(T)$. Note that in this definition, some leaves of $T$ may be unlabelled. However it is always possible to remove these edges by recursively removing unlabelled leaves, and merging edges incident to a vertex of degree $2$. On the contrary, one can add an unlabelled leaf by subdividing any edge into two, and adding an edge incident to an unlabelled leaf at the vertex of degree $2$.

Let $e \in E(T)$ be an edge in the tree $T$ of a tree decomposition of $G$. Removing $e$ from $T$ yields two subtrees $T_1^e,T_2^e$ of $T$. In this paper, to ease notations, we will indicate substructures obtained from an edge $e$ of a tree $T$ by using the notation $\cdot^e$, while sets of vertices and edges related to $e$ will be denoted using $\cdot(e)$. 

The sets of vertices \emphdef{displayed} by $e$ are the vertices falling on either side of $e$ in $T$, \ie, $V_1(e) = \phi^{-1} (L(T_1^e))$ and $V_2(e)= \phi^{-1} (L(T_2^e))$, and the graph cut associated to $e$ is defined by $T(e) = \cut(V_1 (e), V_2 (e))$. We also call the subgraphs $G[V_1(e)]$ and $G[V_2(e)]$ of $G$ the subgraphs \emphdef{displayed} by $e$, and write for short $G_1^e$ and $G_2^e$. 
The \emphdef{width} $\wid (e)$ of $e$ is the size of the associated cut $|T(e)|$, and the width of a carving-decomposition $D = (T,\phi)$ is the maximal width of an edge of $T$: $\wid(D) = \max_{e \in E(T)} \wid (e)$. Finally, the \emphdef{carving-width} $\cwid (G)$ of $G$ is the minimal width of all of its carving-decompositions: 
\[\displaystyle\cwid (G) = \min_{\substack{(T,\phi)\\ \text{ carving-dec.}}} \max_{e \in E(T)} |\cut_G(V_1 (e), V_2 (e))|.\]

A \emphdef{branch-decomposition} a similar concept, where edges of a graph $G$ are injectively sent to the leaves of a binary tree. In this context, any edge of the tree partition $E(G) = E_1 \sqcup E_2$ into two parts; the width of such an edge of the branch-decomposition is the number of vertices appearing as endpoints of both edges in $E_1$ and $E_2$. The \emphdef{branch-width} is the minimum width of such a decomposition where the width of an edge $e$ is the number of vertices of $G$ which are incident to edges in each of the two subgraphs displayed by $e$.

\subparagraph*{Minor relations}
To avoid redundancy with Section~\ref{sec_minor_def}, where we introduce formally the embedded minor relation and the embedded immersion relation we study in this article, we only give informal notions of the more classical minor and immersion relations. We refer to~\cite[Chapter 1.7, 12]{Diestel_Graph_Theory} for a thorough introduction to these relations. An abstract graph $H$ is a \emphdef{minor} of an abstract graph $G$, if $H$ can be obtained from $G$ by a sequence of edge contractions (merging endpoints of an edge), and vertex and edge deletions. Equivalently, $H$ is a minor of $G$ if the vertices of $H$ can be sent to disjoint subgraphs of $G$, and the edges of $H$ can be sent to vertex-disjoint paths connecting the images of their endpoints. An abstract graph $H$ is an \emphdef{immersion minor} of an abstract graph $G$, if $H$ can be obtained from $G$ by a sequence of \emphdef{lifts} (also called splitting off), and vertex and edge deletions. A lift in a graph $G$ consists of replacing a pair of incident edges $\{(u,v),(v,w)\}$ by the single edge $\{(u,w)\}$. Equivalently, $H$ is an immersion minor of $G$ if the vertices of $H$ can be sent to disjoint vertices of $G$, and the edges of $H$ can be sent to edge-disjoint paths connecting the images of their endpoints.

\begin{figure}[ht]
\begin{center}
\label{diff_rel}
\begin{tikzpicture}
\def\e{0.055} \def\a{25} \def\al{35}
\coordinate (c0) at (90:1.5);
\coordinate (c1) at (90+\al:1.5);
\coordinate (c2) at (210-\a:1.5);
\coordinate (c3) at (210:1.5);
\coordinate (c4) at (210+\a:1.5);
\coordinate (c5) at (-30-\a:1.5);
\coordinate (c6) at (-30:1.5);
\coordinate (c7) at (-30+\a:1.5);
\coordinate (c8) at (90 -\al:1.5);
\coordinate (c) at (0,0);

\begin{scope}[gray]
\draw (c0) -- (c1) -- (c2) -- (c3) -- (c4) -- (c5) -- (c6) -- (c7) -- (c8) -- cycle;
\draw (c) -- (c0);
\draw (c) -- (c2);
\draw (c) -- (c4);
\draw (c) -- (c5);
\draw (c) -- (c7);
\end{scope}

\foreach \i in {0,...,8}{\fill (c\i) circle (\e cm);};
\fill (c) circle (\e cm);
\node at (0,-1.75) {$G$};

\begin{scope}[xshift = 4cm]
\coordinate (c0) at (90:1.5);
\coordinate (c1) at (90+\al:1.5);
\coordinate (c2) at (210-\a:1.5);
\coordinate (c3) at (210:1.5);
\coordinate (c4) at (210+\a:1.5);
\coordinate (c5) at (-30-\a:1.5);
\coordinate (c6) at (-30:1.5);
\coordinate (c7) at (-30+\a:1.5);
\coordinate (c8) at (90 -\al:1.5);
\coordinate (c) at (c0);

\begin{scope}[gray]
\draw (c0) -- (c1) -- (c2) -- (c3) -- (c4) -- (c5) -- (c6) -- (c7) -- (c8) -- cycle;
\draw (c) -- (c0);
\draw (c) -- (c2);
\draw (c) -- (c4);
\draw (c) -- (c5);
\draw (c) -- (c7);
\end{scope}

\foreach \i in {0,...,8}{\fill (c\i) circle (\e cm);};
\fill (c) circle (\e cm);

\node at (0,-1.75) {$H_m$};
\end{scope}

\begin{scope}[xshift = 8cm]
\coordinate (c0) at (90:1.5);
\coordinate (c1) at (90+\al:1.5);
\coordinate (c2) at (210-\a:1.5);
\coordinate (c3) at (210:1.5);
\coordinate (c4) at (210+\a:1.5);
\coordinate (c5) at (-30-\a:1.5);
\coordinate (c6) at (-30:1.5);
\coordinate (c7) at (-30+\a:1.5);
\coordinate (c8) at (90 -\al:1.5);
\coordinate (c) at (0,0);

\begin{scope}[gray]
\draw (c0) -- (c1) -- (c2) -- (c3) -- (c4) -- (c5) -- (c6) -- (c7) -- (c8) -- cycle;
\draw (c) -- (c0);
\draw (c) -- (c2);
\draw (c) -- (c4);
\draw (c5) .. controls (c) and (c) .. (c7);
\end{scope}

\foreach \i in {0,...,8}{\fill (c\i) circle (\e cm);};
\fill (c) circle (\e cm);

\node at (0,-1.75) {$H_{i}$};
\end{scope}
\end{tikzpicture}
\caption{Three abstract graphs: $G$, $H_m$ is a minor of $G$, $H_{i}$ is an immersion of $G$.}
\label{pic_diff_rel}
\end{center}
\end{figure}

If $H$ is an immersion of $G$, unlabelling the vertices of $V(G) \smallsetminus V(H)$ from a carving decomposition of $G$ yields a carving-decomposition of $H$ with a smaller carving-width since all elementary operations only reduce the sizes of cuts. Hence, carving-width is non-increasing with respect to immersion minor: if $H$ is an immersion minor of $G$, then $\cwid (H) \leq \cwid (G)$. Similarly, and with the same arguments, branch-width is non-increasing with respect to minor: if $H$ is a minor of $G$, then $\bwid (H) \leq \bwid (G)$.

\subsection{Background on embedded graphs.}

We refer to~\cite[Chapter 4]{Diestel_Graph_Theory} for a broad presentation of planar graphs and related concepts.

Let $\Sigma$ be a surface, a \emphdef{Jordan curve} of $\Sigma$ is the image $\gamma$ of a continuous injective map from $\Sp^1 \rightarrow \Sigma$. By the Jordan–Schoenflies theorem~\cite{Jordan_lecture}, if $\Sigma = \Sp^2$ such a curve splits $\Sp^2$ into two connected components which are homeomorphic to open discs.

Let $G$ be a graph. An embedding of $G$ into $\Sigma$, if it exists, is a drawing of $G$ on $\Sigma$ without crossings. More precisely, for each $e \in E(G)$ define $I_e = [0,1]$, then an embedding of $G$ on $\Sigma$ is a map $f : V(G) \cup (\bigcup_{e \in E(G)} I_e) \rightarrow \Sigma$ such that: 
\begin{itemize}
\item $f$ restricted to $V(G) \cup (\bigcup_{e \in E(G)} \overset{\circ}{I_e} )$ is injective.
\item $f$ restricted to each $I_e$ is continuous
\item For each $e = (u,v)$, on $I_e$, $f \big|_{I_e} (0) = f(u)$, and $f \big|_{I_e} (1) = f(v)$.
\end{itemize} 

Two embedded graphs $G_1$ and $G_2$ are \emphdef{equivalent} if there exists an ambient isotopy, \ie, a continuous family of homeomorphism $f_t \colon \Sigma \to \Sigma$, taking the embedding of $G_1$ to the embedding of $G_2$.  We consider embedded graphs up to this equivalence.

This formally indicates that one can draw $G$ on $\Sigma$ such that each edge is drawn continuously and no two drawn edges cross in their interior (they can meet at the end if they share the same endpoint). In the following sections, we will consider embedded graphs. However, to ease notations, we will denote by $G$ both the embedded graph and the abstract graph (where we forget about the embedding). A \emphdef{plane graph} is a graph together with an embedding in $\Sp^2$.

For $G$ an embedded graph, the set of \emphdef{faces} $F(G)$ are the connected components of $\Sigma \setminus G$. The \emphdef{dual} $G^*$ of $G$ is the abstract graph with set of vertices $F(G)$, and two vertices $F,F'$ of $G*$ are adjacent if they share an edge of $G$. If $G$ is planar, there is a unique embedding of $G^*$ in $\Sp^2$, however this embedding is not unique in $\R^2$. The \emphdef{cut-cycle duality} is a bijection between the minimal cuts of $G$ and cycles of $G^*$ which can then be realised as Jordan curves.

\subsection{Background on order theory.}
A \emphdef{quasi-order} on a set $X$ is a binary relation $R$ satisfying reflexivity ($\forall x \in X, x R x$) and transitivity ($\forall x,y,z \in X^3, x R y$ and $y R z \Rightarrow x R z$). The relation $R$ is an \emphdef{order} if it additionally satisfies antisymmetry ($\forall x,y \in X^2, x R y$ and $y R x \Rightarrow x = y$). A \emphdef{strict quasi-order} satisfies irreflexivity ($\forall x \in X, \neg x R x$), asymmetry ($\forall x,y \in X^2, x R y \Rightarrow \neg y R x$) and transitivity. An order $R$ is said to be \emphdef{linear} or \emphdef{total} if all pairs of elements are comparable ($\forall x,y \in X^2, x R y$ or $y R x$). A \emphdef{minimal element} for a quasi-order or strict quasi-order is an element $m$ with no other element of $X$ in relation with it: $\forall x \in X, \neg x R m$. The \emphdef{lexicographic} order is a strict total order put on $n$-tuples (for $n > 1$) defined by $a <_{\text{lec}} b$ if $a_i < b_i$ for the smallest $1 \leq i \leq n$ such that $a_i$ and $b_i$; this is the analogue of the order of dictionaries on $n$-tuples. The \emphdef{colexicographic} order is the lexicographic order obtained after reversing the order of the $n$-tuples, i.e., $a <_{\text{colec}} b$ if $a_i < b_i$ for the highest $1 \leq i \leq n$ such that $a_i$ and $b_i$ differ.

A \emphdef{circular order} $C$ on $X$ is a ternary relation satisfying cyclicity ($\forall x,y,z \in X^3, C(x,y,z) \Rightarrow C(z,x,y)$), asymmetry ($\forall x,y,z \in X^3, C(x,y,z) \Rightarrow \neg C(x,z,y)$), transitivity ($\forall w,x,y,z \in X^3, C(x,y,z)$ and $C(x,z,w) \Rightarrow C(x,y,w)$), and connectedness ($\forall x,y,z \in X^3, x \neq y \neq z \Rightarrow C(x,y,z) \text{ or } C(x,z,y)$). One can think of $C(x,y,z)$ as $x$ meets $y$ before $z$. Any linear order $\leq$ induces a circular one defined as $C_\leq (x,y,z)$ if one of $x \leq y \leq z$, $z \leq x \leq y$, $y \leq z \leq x$ is satisfied. Similarly, choosing any $x \in X$ which is circularly ordered by $C$ defines a linear order $\leq_x$ defined as $y \leq_x z$ if $C(x,y,z)$, and $x \leq_x y$ and $y \leq_x y$ for all $y \in X$. Notice that the circular order induced by $\leq_x$ is then $C$.

A \emphdef{well-quasi-order} $R$ or \emphdef{w.q.o.} on $X$ is a quasi-order for which every infinite sequence of elements admits an increasing pair, \ie, for all $(x_n)_{n \in \N} \in X^\N$, there exist $i < j$ such that $x_i R x_j$. We refer to~\cite[Chapter 12]{Diestel_Graph_Theory} for an introduction to well-quasi-orders and to~\cite{Davey_book_order} for more on order theory.

\section{Embedded relations and combinatorial characterisation.}
\label{sec_minor_def}
 
In this section, we give a formal topological definition of the embedded minor relations, and we provide a more combinatorial characterisation for the embedded immersion, that will be useful in the following sections. Embedded graph minors have been defined in~\cite{Medina_wqolink}, and to the best of our knowledge, we are not aware of a precise definition of embedded immersion or of our characterisation in the literature. 

\subsection{Embedded graph minors.}

Extending the abstract notion of graph minor to the embedded framework is quite natural: the abstract operations are turned into their embedded equivalents.

Let $H$ and $G$ be two plane graphs. The graph $H$ is said to be a \emphdef{plane minor} of $G$ if $H$ can be obtained from $G$ by a sequence of \textit{embedded vertex deletions}, \textit{embedded edge deletions}, and \textit{embedded edge contractions}. Let $f : V(G) \cup (\bigcup_{e \in E(G)} I_e) \rightarrow \Sigma$ be the embedding of $G$ in the surface $\Sigma$. An \emphdef{embedded edge deletion} of $e \in E(G)$ is done by removing the interior of $f(e)$ from the drawing; formally, it consists of restricting the embedding $f$ to $V(G) \bigcup_{e \in E(G) \smallsetminus \{ e \}} I_e$. An \emphdef{embedded vertex deletion} of $v \in V(G)$ consists of removing $f(v)$ and all edges incident to it; formally, the embedding $f$ is restricted to $(V(G) \smallsetminus \{ v \}) \bigcup_{e \in E(G) \smallsetminus \{ v \times V(G) \}} I_e$.

Abstractly, contracting the edge $e$ of $G$ consists of removing an edge and identifying its endpoints, which yields the graph $G / e$. To perform an \emphdef{embedded edge contraction} on $e$ which is not a self-loop, choose a disc $D_e$ of $\Sp^2$ intersecting $G$ exactly on the two distinct endpoints of $e$, and contract this disc to a point by a homotopy $\psi : \Sp^2 \rightarrow \Sp^2$. It follows that $\psi \circ f|$ is a continuous map from $V(G)$ to $\Sp^2$ which sends $I_e$ to a single point as well as its endpoints and is injective otherwise. Hence, $\psi \circ f = f'$ naturally translates to an embedding of $G / e$. These embedded operations are depicted on Figure~\ref{pic_def_emb_op}.

\begin{figure}[ht]
\begin{center}
\begin{tikzpicture}[scale = 0.85]
\def\e{0.05}
\coordinate (c1) at (0,0);
\coordinate (c2) at (1.5,0);
\coordinate (c3) at (0.75, -1);
\coordinate (c4) at (-1,0.5);
\coordinate (c5) at (-1,-0.75);
\coordinate (c6) at (2.5,0.7);

\draw [red!50!black, thick] (c1) -- (c2);
\draw (c1) -- (c3) -- (c2);
\draw (c1) -- (c4);
\draw (c1) -- (c5);
\draw (c2) -- (c6);
\draw (c1) .. controls +(90:0.75) and \control{(c2)}{(90:0.75)};
\draw (c2) .. controls +(10:0.5) and \control{(c2)}{(-40:0.5)};

\foreach \i in {1,...,6}{\fill (c\i) circle (\e cm);};
\node at (0.75,0) [below, red!50!black] {$e$};

\draw [-Stealth, black!50] (2.35,-0.25) -- ++(1,0);

\begin{scope}[xshift = 4.25cm]
\coordinate (c1) at (0,0);
\coordinate (c2) at (1.5,0);
\coordinate (c3) at (0.75, -1);
\coordinate (c4) at (-1,0.5);
\coordinate (c5) at (-1,-0.75);
\coordinate (c6) at (2.5,0.7);

\draw (c1) -- (c3) -- (c2);
\draw (c1) -- (c4);
\draw (c1) -- (c5);
\draw (c2) -- (c6);
\draw (c1) .. controls +(90:0.75) and \control{(c2)}{(90:0.75)};
\draw (c2) .. controls +(10:0.5) and \control{(c2)}{(-40:0.5)};

\foreach \i in {1,...,6}{\fill (c\i) circle (\e cm);};
\end{scope}

\begin{scope}[xshift = 9.5cm]
\coordinate (c1) at (0,0);
\coordinate (c2) at (1.5,0);
\coordinate (c3) at (0.75, -1);
\coordinate (c4) at (-1,0.5);
\coordinate (c5) at (-1,-0.75);
\coordinate (c6) at (2.5,0.7);

\draw (c1) -- (c2);
\draw (c1) -- (c3) -- (c2);
\draw (c1) -- (c4);
\draw (c1) -- (c5);
\draw (c2) -- (c6);
\draw (c1) .. controls +(90:0.75) and \control{(c2)}{(90:0.75)};
\draw (c2) .. controls +(10:0.5) and \control{(c2)}{(-40:0.5)};

\foreach \i in {1,...,6}{\fill (c\i) circle (\e cm);};
\fill [red!50!black] (c2) circle (\e*1.25 cm);
\node [red!50!black, below] at (c2) {$v$};

\draw [-Stealth, black!50] (2.35,-0.25) -- ++(1,0);

\begin{scope}[xshift = 4.25cm]
\coordinate (c1) at (0,0);
\coordinate (c2) at (1.5,0);
\coordinate (c3) at (0.75, -1);
\coordinate (c4) at (-1,0.5);
\coordinate (c5) at (-1,-0.75);
\coordinate (c6) at (2.5,0.7);

\draw (c1) -- (c3);
\draw (c1) -- (c4);
\draw (c1) -- (c5);

\fill (c1) circle (\e cm);
\foreach \i in {3,...,6}{\fill (c\i) circle (\e cm);};
\end{scope}
\end{scope}

\begin{scope}[yshift = -2.75cm, xshift = 2.75cm]
\coordinate (c1) at (0,0);
\coordinate (c2) at (1.5,0);
\coordinate (c3) at (0.75, -1);
\coordinate (c4) at (-1,0.5);
\coordinate (c5) at (-1,-0.75);
\coordinate (c6) at (2.5,0.7);

\draw [red!50!black, thick] (c1) -- (c2);
\draw (c1) -- (c3) -- (c2);
\draw (c1) -- (c4);
\draw (c1) -- (c5);
\draw (c2) -- (c6);
\draw (c1) .. controls +(90:0.75) and \control{(c2)}{(90:0.75)};
\draw (c2) .. controls +(10:0.5) and \control{(c2)}{(-40:0.5)};

\foreach \i in {1,...,6}{\fill (c\i) circle (\e cm);};
\node at (0.75,0) [below, red!50!black] {$e$};

\draw [-Stealth, black!50] (2.35,-0.25) -- ++(1,0);

\begin{scope}[xshift = 4.5cm]
\coordinate (c1) at (0,0);
\coordinate (c2) at (1.5,0);
\coordinate (c3) at (0.75, -1);
\coordinate (c4) at (-1,0.5);
\coordinate (c5) at (-1,-0.75);
\coordinate (c6) at (2.5,0.7);

\draw [red!50!black, thick] (c1) -- (c2);
\draw (c1) -- (c3) -- (c2);
\draw (c1) -- (c4);
\draw (c1) -- (c5);
\draw (c2) -- (c6);
\draw (c1) .. controls +(90:0.75) and \control{(c2)}{(90:0.75)};
\draw (c2) .. controls +(10:0.5) and \control{(c2)}{(-40:0.5)};

\foreach \i in {1,...,6}{\fill (c\i) circle (\e cm);};
\filldraw [fill opacity = 0.1, red!75!black] (c1) .. controls +(30:0.5) and \control{(c2)}{(150:0.5)} .. controls + (-150:0.5) and \control{(c1)}{(-30:0.5)};
\node at (0.75,-0.15) [below, red!75!black] {\footnotesize $D_e$};

\draw [-Stealth, black!50] (2.35,-0.25) -- ++(1,0);
\end{scope}

\begin{scope}[xshift = 9cm]
\coordinate (c1) at (0.75,0);
\coordinate (c2) at (c1);
\coordinate (c3) at (0.75, -1);
\coordinate (c4) at (-1,0.5);
\coordinate (c5) at (-1,-0.75);
\coordinate (c6) at (2.5,0.7);

\draw [red!50!black, thick] (c1) -- (c2);
\draw (c1) .. controls +(-135:0.25) and \control{(c3)}{(135:0.25)};
\draw (c1) .. controls +(-45:0.25) and \control{(c3)}{(45:0.25)};
\draw (c1) -- (c4);
\draw (c1) -- (c5);
\draw (c2) -- (c6);
\draw (c1) .. controls +(135:0.75) and \control{(c2)}{(45:0.75)};
\draw (c2) .. controls +(10:0.5) and \control{(c2)}{(-40:0.5)};

\foreach \i in {1,...,6}{\fill (c\i) circle (\e cm);};
\fill [red!50!black] (c2) circle (\e*1.25 cm);
\end{scope}
\end{scope}
\end{tikzpicture}
\caption{Top left: Embedded edge deletion of $e$. Top right: Embedded vertex deletion of $v$. Bottom: Embedded edge contraction of $e$.}
\label{pic_def_emb_op}
\end{center}
\end{figure}

\subsection{Embedded immersions.}

Let $G$ be a plane graph. An \emphdef{embedded lift} of the distinct embedded edges $e = (u,v)$ and $ e' = (v,w)$ of $G$ at $v$ in $\Sigma$, is the operation consisting of choosing a closed disc $D$ of $\Sigma$ such that $D \cap V(G) = \{ v \}$ and whose boundary $\partial D$ intersects once each edge incident to $v$. Then, replacing ${D} \cap (e \cup e')$ by one of the two arcs of $\partial D$ between $e \cap \partial D$ and $e' \cap \partial D$ such that the obtained graph is still embedded (see Figure~\ref{pic_def_emb_lift} for example). If $v$ is of degree $2$, the remaining vertex can be on either of the incident faces of the newly created edge. Otherwise, there is at least one additional incident edge to $v$ which will intersect $\partial D$ and determine the proper arc for the operation. An \emphdef{embedded immersion} is an immersion that can be realised by a sequence of vertex deletions, edge deletions, and embedded lifts. 

\begin{figure}[ht]
\begin{center}
\begin{tikzpicture}[scale = 0.8]
\begin{scope}[xshift = -0.75cm]
\draw [red!50!purple] (60: 1) -- (0,0);
\draw [red!50!purple] (-60: 1) -- (0,0);
\fill (0,0) circle (0.1 cm);
\filldraw [fill opacity = 0.1, blue!80!black] (0,0) circle (0.4);
\node at (-160:0.65) [blue!80!black] {$D$};
\node at (130:0.225) {$v$};
\node [red!50!purple, left] at (60: 1) {$e$};
\node [red!50!purple, left] at (-60: 1) {$e'$};
\end{scope}

\draw [-Stealth, gray] (0,0) -- (1.5,0) .. controls +(1.5,0) and \control{(3.5,1.25)}{(180:1)};
\draw [-Stealth, gray] (0,0) -- (1.5,0) .. controls +(1.5,0) and \control{(3.5,-1.25)}{(180:1)};
\node at (0.85,0) [align = center] {embedded \\ lift};

\begin{scope}[xshift = 4cm, yshift = 1.25cm]
\draw [red!50!purple] (60: 1) -- (60:0.4) arc (60:-60:0.4) -- (-60:1);
\fill (0,0) circle (0.1 cm);
\node [red!50!purple, right] at (-60: 1) {$e''$};
\end{scope}

\begin{scope}[xshift = 4cm, yshift = -1.25cm]
\draw [red!50!purple] (60: 1) -- (60:0.4) arc (60:300:0.4) -- (-60:1);
\fill (0,0) circle (0.1 cm);
\node [red!50!purple, right] at (60: 1) {$e''$};
\end{scope}

\begin{scope}[xshift = 7.5cm]
\draw (130:1) -- (0,0) -- (190:1);
\draw (-150:1) -- (0,0);
\draw [red!50!purple] (60: 1) -- (0,0) -- (-60:1);
\fill (0,0) circle (0.1 cm);

\draw [-Stealth, gray] (0.75,0) -- +(0.8,0);

\begin{scope}[xshift = 3cm]
\draw (130:1) -- (0,0) -- (190:1);
\draw (-150:1) -- (0,0);
\draw [red!50!purple] (60: 1) -- (60:0.4) arc (60:-60:0.4) -- (-60:1);
\fill (0,0) circle (0.1 cm);

\draw [-Stealth, gray] (0.75,0) -- +(0.8,0);
\node at (1.05,0.2) [gray] {$\sim$};
\end{scope}

\begin{scope}[xshift = 6 cm]
\draw (130:1) -- (0,0) -- (190:1);
\draw (-150:1) -- (0,0);
\draw [red!50!purple] (60: 1) ..controls +(-120:0.75) and \control{(-60:1)}{(120:0.75)};
\fill (0,0) circle (0.1 cm);
\end{scope}
\end{scope}
\end{tikzpicture}
\caption{Left: Embedded lift on two edges $e,e'$ incident to a degree $2$ vertex $v$, yielding a merged edge $e''$ with two possible embeddings (up to equivalence). Right: Embedded lift where $v$ has degree $>2$; the other edges incident to $v$ determine a unique embedding for the merged edge $e''$.}
\label{pic_def_emb_lift}
\end{center}
\end{figure}

In the following, we give a combinatorial characterisation of embedded immersions as an \textit{ordered immersion} with \textit{tangent paths}. We will establish the equivalence of these notions in Lemma~\ref{lem_equiv_plan_imm} and Lemma~\ref{lem_equiv_plan_imm_dir} (for directed graphs) which will be used in the proofs of our main theorems in Sections~\ref{sec_bounded_cw} and~\ref{sec_em_gb}.

An \emphdef{ordered graph} $G$ is the data of a set of vertices $V(G)$, a set of edges $E(G)$, and for each vertex $v$, a circular order $C_v$ on the set of edges incident to $v$. Note that an embedding of a graph $G$ in an oriented surface induces a natural structure of ordered graph on $G$, where the order of edges incident to a vertex is compatible with the orientation of the surface and the embedding. An ordered immersion is then an immersion respecting the order around vertices:

\begin{definition}[Ordered immersion]\label{def_order}
Let $H$ and $G$ be two ordered graphs. An \emphdef{ordered immersion} of $H$ into $G$ is a map $\phi$ defined on $V(H)$ and $E(H)$, mapping vertices of $H$ to vertices of $G$ and edges of $H$ to paths of $G$, such that: 
\begin{enumerate}
\item $\phi |_{V(H)}$ is injective
\item for each pair of distinct edges $e,e'$, the paths $\phi(e)$ and $\phi(e')$ are edge-disjoint.
\item for each vertex $v \in V(H)$: for each edge $e \in E(H)$ incident to $v$ denote by $e'$ the edge of $\phi(e)$ incident to $\phi(v)$, then for all $e_1,e_2,e_3$ incident to $v$, $C_v(e_1,e_2,e_3) \Rightarrow C_{\phi(v)} (e'_1, e'_2, e'_3)$.
\end{enumerate}
In this case, we say that $H$ is an ordered immersion of $G$.
\end{definition}

Applying properly~\cite[Theorem 1.6]{Graph_Minors_XXIII} allows us to prove that ordered immersions form a well-quasi-order on ordered graphs of bounded degree. However, for embedded graphs, ordered immersion minor is a strictly weaker notion than embedded immersion, where the order of edges is induced by the embedding. Indeed, Figure~\ref{pic_ord_vs_emb} depicts two embedded graphs $G_1$, $G_2$ where $G_1$ is an ordered minor of $G_2$ (orders around vertices are induced by the planar embedding and a fixed orientation of the plane) but is not an embedded minor of $G_2$. Indeed, one can see $G_2$ as a drawing of $G_1$ in the plane except that each edge of two pairs of edges intersects. The ordered immersion is the natural identity on the whole graph, except on the two aforementioned pairs of parallel edges, which are mapped to $4$ transverse paths which cross at the intersections. This immersion pattern is unique, and it enforces transversal paths that cannot be realised by embedded lifts in an embedded immersion.

\begin{figure}[ht]
\begin{center}
\begin{tikzpicture}[scale = 0.65]
\clip (-2.2,-5) rectangle (20.1,4);
\def\e{0.5} \def\se{0.35} \def\me{0.45} \def\lme{1}  \def\le{1.25}
\def\r{0.065}

\coordinate (a0) at (0,0);
\coordinate (a1) at ($(a0)+(135:\e)$);
\coordinate (a2) at ($(a1)+(135:\e)$);
\coordinate (a3) at ($(a0)+(-135:\e)$);
\coordinate (a4) at ($(a3)+(-135:\e)$);
\coordinate (a5) at (-0.75,0);
\coordinate (a6) at ($(a5)+(135:\e)$);
\coordinate (a7) at ($(a6)+(135:\e)$);
\coordinate (a7') at ($(-1.125,0)+(135:2.45*\e)$);
\coordinate (a8) at ($(a5)+(-135:\e)$);
\coordinate (a9) at ($(a8)+(-135:\e)$);
\coordinate (a10) at (-1.5,0);
\coordinate (a11) at ($(a10)+(135:\e)$);
\coordinate (a12) at ($(a10)+(-135:\e)$);
\foreach \i in {0,...,12}{\fill (a\i) circle (\r cm);};
\fill (a7') circle (\r cm);

\coordinate (b0) at (2,0);
\coordinate (b1) at (3.25,0.75);
\coordinate (b2) at (3.25,1.25);
\coordinate (b3) at (3.25,1.75);
\coordinate (b4) at (3.75,0.75);
\coordinate (b5) at (3.75,1.25);
\coordinate (b6) at (3.75,1.75);
\coordinate (b6') at (3.5,0.25);
\coordinate (b7) at (3.25,-0.75);
\coordinate (b8) at (3.25,-1.25);
\coordinate (b9) at (3.25,-1.75);
\coordinate (b10) at (3.75,-0.75);
\coordinate (b11) at (3.75,-1.25);
\coordinate (b12) at (3.75,-1.75);
\coordinate (b13) at (5,0);
\foreach \i in {0,...,13}{\fill (b\i) circle (\r cm);};
\fill (b6') circle (\r cm);

\coordinate (c0) at (7,0);
\coordinate (c1) at ($(c0)+(30:0.6)$);
\coordinate (c1') at ($(c1)+(\e*0.5,0.45)$);
\coordinate (c2) at ($(c1)+(0:\e)$);
\coordinate (c3) at ($(c2)+(0:\e)$);
\coordinate (c4) at ($(c3)+(0:\e)$);
\coordinate (c5) at ($(c4)+(0:\e)$);
\coordinate (c6) at ($(c5)+(0:\e)$);
\coordinate (c7) at ($(c0)+(-30:0.6)$);
\coordinate (c8) at ($(c7)+(0:\e)$);
\coordinate (c9) at ($(c8)+(0:\e)$);
\coordinate (c10) at ($(c9)+(0:\e)$);
\coordinate (c11) at ($(c10)+(0:\e)$);
\coordinate (c12) at ($(c11)+(0:\e)$);
\foreach \i in {0,...,12}{\fill (c\i) circle (\r cm);};
\fill (c1') circle (\r cm);

\node at (-1,1.5) {$G_1$};

\begin{scope}[opacity=0.65]
\draw (a0) -- (a1) -- (a2) .. controls +(135:\se) and \control{(a7)}{(45:\se)} -- (a11) .. controls +(-135:\se) and \control{(a12)}{(135:\se)} -- (a9) .. controls +(-45:\se) and \control{(a4)}{(-135:\se)} -- (a3) -- cycle; 
\draw (a5) -- (a6) -- (a7) .. controls +(135:0.2) and \control{(a7')}{(0:\se)} .. controls +(180:\me) and \control{(a7')}{(90:\me)} .. controls +(-90:\se) and \control{(a11)}{(135:0.2)} -- (a10) -- (a8) -- (a4) .. controls +(-45:\me) and \control{(a3)}{(-45:\me)} -- cycle;
\draw (a5) -- (a8) -- (a9) .. controls +(-135:\me) and \control{(a12)}{(-135:\me)} -- (a10) -- (a6) -- (a2) .. controls +(45:\me) and \control{(a1)}{(45:\me)} -- cycle;

\draw (a0) .. controls +(35:0.75) and \control{(b0)}{(145:0.75)};
\draw (a0) .. controls +(-35:0.75) and \control{(b0)}{(-145:0.75)};

\draw (b0) .. controls +(80:\le) and \control{(b3)}{(180:\me)} -- (b6) .. controls +(0:\me) and \control{(b13)}{(100:\le)};
\draw (b0) .. controls +(-80:\le) and \control{(b9)}{(180:\me)} -- (b12) .. controls +(0:\me) and \control{(b13)}{(-100:\le)};
\draw (b2) -- (b5) .. controls +(0:\lme) and \control{(b11)}{(0:\lme)} -- (b8) .. controls +(180:\lme) and \control{(b2)}{(180:\lme)};
\draw (b1) -- (b4) .. controls +(0:0.65) and \control{(b10)}{(0:0.65)} -- (b7) .. controls +(180:0.65) and \control{(b1)}{(180:0.65)};
\draw (b1) -- (b2) -- (b3) .. controls +(90:\me) and \control{(b6)}{(90:\me)} -- (b5) -- (b4) .. controls +(-90:\me) and \control{(b6')}{(45:0.2)} .. controls +(-135:\e) and \control{(b6')}{(-45:\e)} .. controls +(135:0.2) and \control{(b1)}{(-90:\me)};
\draw (b7) -- (b8) -- (b9) .. controls +(-90:\me) and \control{(b12)}{(-90:\me)} -- (b11) -- (b10) .. controls +(90:\me) and \control{(b7)}{(90:\me)};

\draw (b13) .. controls +(35:0.75) and \control{(c0)}{(145:0.75)};
\draw (b13) .. controls +(-35:0.75) and \control{(c0)}{(-145:0.75)};

\draw (c0) .. controls +(30:0.2) and \control{(c1)}{(180:\se)} -- (c1)  -- (c2)  -- (c3)  -- (c4)  -- (c5)  -- (c6) .. controls +(0:\me) and \control{(c12)}{(0:\me)}  -- (c11)  -- (c10)  -- (c9)  -- (c8)  -- (c7) .. controls +(180:\se) and \control{(c0)}{(-30:0.2)};
\draw (c1) .. controls +(90:0.2) and \control{(c1')}{(-135:\se)} .. controls +(45:\me) and \control{(c1')}{(135:\me)} .. controls +(-45:\se) and \control{(c2)}{(90:0.2)} -- (c8) .. controls +(-90:\me) and \control{(c7)}{(-90:\me)} -- (c1);
\draw (c3) .. controls +(90:\me) and \control{(c4)}{(90:\me)} -- (c10) .. controls +(-90:\me) and \control{(c9)}{(-90:\me)} -- (c3);
\draw (c5) .. controls +(90:\me) and \control{(c6)}{(90:\me)} -- (c12) .. controls +(-90:\me) and \control{(c11)}{(-90:\me)} -- (c5);
\end{scope}

\begin{scope}[xshift=14cm]
\coordinate (a0) at (0,0);
\coordinate (a1) at ($(a0)+(135:\e)$);
\coordinate (a2) at ($(a1)+(135:\e)$);
\coordinate (a3) at ($(a0)+(-135:\e)$);
\coordinate (a4) at ($(a3)+(-135:\e)$);
\coordinate (a5) at (-0.75,0);
\coordinate (a6) at ($(a5)+(135:\e)$);
\coordinate (a7) at ($(a6)+(135:\e)$);
\coordinate (a7') at ($(-1.125,0)+(135:2.45*\e)$);
\coordinate (a8) at ($(a5)+(-135:\e)$);
\coordinate (a9) at ($(a8)+(-135:\e)$);
\coordinate (a10) at (-1.5,0);
\coordinate (a11) at ($(a10)+(135:\e)$);
\coordinate (a12) at ($(a10)+(-135:\e)$);
\foreach \i in {0,...,12}{\fill (a\i) circle (\r cm);};
\fill (a7') circle (\r cm);

\coordinate (b0) at (2,0);
\coordinate (b1) at (3.25,0.75);
\coordinate (b2) at (3.25,1.25);
\coordinate (b3) at (3.25,1.75);
\coordinate (b4) at (3.75,0.75);
\coordinate (b5) at (3.75,1.25);
\coordinate (b6) at (3.75,1.75);
\coordinate (b6') at (3.5,0.25);
\coordinate (b7) at (3.25,-0.75);
\coordinate (b8) at (3.25,-1.25);
\coordinate (b9) at (3.25,-1.75);
\coordinate (b10) at (3.75,-0.75);
\coordinate (b11) at (3.75,-1.25);
\coordinate (b12) at (3.75,-1.75);
\coordinate (b13) at (5,0);
\foreach \i in {0,...,13}{\fill (b\i) circle (\r cm);};
\fill (b6') circle (\r cm);

\coordinate (c0) at (1,-1.5);
\coordinate (c1) at ($(c0)+(-60:0.6)$);
\coordinate (c1') at ($(c1)+(0.45,-\e*0.5)$);
\coordinate (c2) at ($(c1)+(-90:\e)$);
\coordinate (c3) at ($(c2)+(-90:\e)$);
\coordinate (c4) at ($(c3)+(-90:\e)$);
\coordinate (c5) at ($(c4)+(-90:\e)$);
\coordinate (c6) at ($(c5)+(-90:\e)$);
\coordinate (c7) at ($(c0)+(-120:0.6)$);
\coordinate (c8) at ($(c7)+(-90:\e)$);
\coordinate (c9) at ($(c8)+(-90:\e)$);
\coordinate (c10) at ($(c9)+(-90:\e)$);
\coordinate (c11) at ($(c10)+(-90:\e)$);
\coordinate (c12) at ($(c11)+(-90:\e)$);
\foreach \i in {0,...,12}{\fill (c\i) circle (\r cm);};
\fill (c1') circle (\r cm);

\coordinate (i1) at ($(a0)+(30:0.75)$);
\coordinate (i2) at ($(b0)+(150:0.75)$);
\coordinate (i3) at ($(a0)+(-30:0.75)$);
\coordinate (i4) at ($(b0)+(-150:0.75)$);
\foreach \i in {1,...,4}{\fill (i\i) circle (\r cm);};

\node at (-1,1.5) {$G_2$};

\begin{scope}[opacity=0.65]
\draw (a0) -- (a1) -- (a2) .. controls +(135:\se) and \control{(a7)}{(45:\se)} -- (a11) .. controls +(-135:\se) and \control{(a12)}{(135:\se)} -- (a9) .. controls +(-45:\se) and \control{(a4)}{(-135:\se)} -- (a3) -- cycle; 
\draw (a5) -- (a6) -- (a7) .. controls +(135:0.2) and \control{(a7')}{(0:\se)} .. controls +(180:\me) and \control{(a7')}{(90:\me)} .. controls +(-90:\se) and \control{(a11)}{(135:0.2)} -- (a10) -- (a8) -- (a4) .. controls +(-45:\me) and \control{(a3)}{(-45:\me)} -- cycle;
\draw (a5) -- (a8) -- (a9) .. controls +(-135:\me) and \control{(a12)}{(-135:\me)} -- (a10) -- (a6) -- (a2) .. controls +(45:\me) and \control{(a1)}{(45:\me)} -- cycle;

\draw (a0) .. controls +(30:0.2) and \control{(i1)}{(180:\se)} -- (i2) .. controls +(0:\se) and \control{(b0)}{(150:0.2)};
\draw (a0) .. controls +(-30:0.2) and \control{(i3)}{(180:\se)} -- (i4) .. controls +(0:\se) and \control{(b0)}{(-150:0.2)};

\draw (b0) .. controls +(80:\le) and \control{(b3)}{(180:\me)} -- (b6) .. controls +(0:\me) and \control{(b13)}{(100:\le)};
\draw (b0) .. controls +(-80:\le) and \control{(b9)}{(180:\me)} -- (b12) .. controls +(0:\me) and \control{(b13)}{(-100:\le)};
\draw (b2) -- (b5) .. controls +(0:\lme) and \control{(b11)}{(0:\lme)} -- (b8) .. controls +(180:\lme) and \control{(b2)}{(180:\lme)};
\draw (b1) -- (b4) .. controls +(0:0.65) and \control{(b10)}{(0:0.65)} -- (b7) .. controls +(180:0.65) and \control{(b1)}{(180:0.65)};
\draw (b1) -- (b2) -- (b3) .. controls +(90:\me) and \control{(b6)}{(90:\me)} -- (b5) -- (b4) .. controls +(-90:\me) and \control{(b6')}{(45:0.2)} .. controls +(-135:\e) and \control{(b6')}{(-45:\e)} .. controls +(135:0.2) and \control{(b1)}{(-90:\me)};
\draw (b7) -- (b8) -- (b9) .. controls +(-90:\me) and \control{(b12)}{(-90:\me)} -- (b11) -- (b10) .. controls +(90:\me) and \control{(b7)}{(90:\me)};

\draw (b13) .. controls +(-20:4) and \control{(i1)}{(90:8.5)} -- (i3) .. controls +(-90:\me) and \control{(c0)}{(120:\me)};
\draw (b13) .. controls +(30:2) and \control{(i2)}{(90:5.5)} -- (i4) .. controls +(-90:\me) and \control{(c0)}{(60:\me)};

\draw (c0) .. controls +(-60:0.2) and \control{(c1)}{(90:\se)} -- (c1)  -- (c2)  -- (c3)  -- (c4)  -- (c5)  -- (c6) .. controls +(-90:\me) and \control{(c12)}{(-90:\me)}  -- (c11)  -- (c10)  -- (c9)  -- (c8)  -- (c7) .. controls +(90:\se) and \control{(c0)}{(-120:0.2)};
\draw (c1) .. controls +(0:0.2) and \control{(c1')}{(135:\se)} .. controls +(-45:\me) and \control{(c1')}{(45:\me)} .. controls +(-135:\se) and \control{(c2)}{(0:0.2)} -- (c8) .. controls +(180:\me) and \control{(c7)}{(180:\me)} -- (c1);
\draw (c3) .. controls +(0:\me) and \control{(c4)}{(0:\me)} -- (c10) .. controls +(180:\me) and \control{(c9)}{(180:\me)} -- (c3);
\draw (c5) .. controls +(0:\me) and \control{(c6)}{(0:\me)} -- (c12) .. controls +(180:\me) and \control{(c11)}{(180:\me)} -- (c5);
\end{scope}
\end{scope}
\end{tikzpicture}
\caption{Two embedded graphs, the order around each vertex is induced by the embeddings. $G_1$ is an ordered minor of $G_2$ but not an embedded minor of $G_2$.}
\label{pic_ord_vs_emb}
\end{center}
\end{figure}

This phenomenon motivates the following definition: a family of paths $(c_i)_{i \in I}$ in an ordered graph $G$ is said to be \emphdef{transverse} if there exists a vertex $v \in V(G)$ such that the incident edges to $v$ are \emphdef{interlaced}, \ie, if around $v$ edges $e,e'$ of a path $p_i$ and edges $f,f'$ of path $p_j$ are met in the order $efe'f'$ for $i \neq j$. In that case $v$ is called a \emphdef{transverse vertex}, and is called a \emphdef{tangent vertex} otherwise. A family of paths is \emphdef{tangent} if all of its vertices are tangent vertices. For such a family, the embedded immersion can be realised by embedded lifts as illustrated by Figure~\ref{pic_ex_transverse}. Note that forbidding inner vertices from a path to be the starting point of another one would correspond to the hypothesis of strong immersion as defined on abstract graphs.

\begin{figure}[ht]
\begin{center}
\begin{tikzpicture}[scale=0.85]
\def\r{0.065}

\draw [blue!80!black] (0,0) -- +(-110:1);
\draw [blue!80!black] (0,0) -- +(120:1);
\draw [purple!50!red] (0,0) -- +(10:1);
\draw [green!70!purple] (0,0) -- +(-70:1);
\fill (0,0) circle (\r cm);

\draw [-Stealth] (0,-1.25) -- +(0,-0.75); 

\begin{scope}[yshift = -3cm]
\draw [blue!80!black] ($(0,0)+(-110:1)$) .. controls +(70:0.75) and \control{(120:1)}{(-60:0.75)};
\draw [purple!50!red] (0,0) -- +(10:1);
\draw [green!70!purple] (0,0) -- +(-70:1);
\fill (0,0) circle (\r cm);
\end{scope}

\begin{scope}[xshift = 2.75cm]
\draw [blue!80!black] (0,0) -- +(135:1);
\draw [blue!80!black] (0,0) -- +(-135:1);
\draw [purple!50!red] (0,0) -- +(-45:1);
\draw [purple!50!red] (0,0) -- +(45:1);
\draw [green!70!purple] (0,0) -- +(65:1);
\draw [green!70!purple] (0,0) -- +(-65:1);
\draw [opacity = 0.6] (0,0) -- +(10:1);
\fill (0,0) circle (\r cm);

\draw [-Stealth] (0,-1.25) -- +(0,-0.75); 

\begin{scope}[yshift = -3cm]
\draw [blue!80!black] ($(0,0)+(-135:1)$) .. controls +(45:0.75) and \control{(135:1)}{(-45:0.75)};
\draw [purple!50!red] ($(0,0)+(-45:1)$) .. controls +(135:0.75) and \control{(45:1)}{(-135:0.75)};
\draw [green!70!purple] ($(0,0)+(-60:1)$) .. controls +(120:1) and \control{(60:1)}{(-120:1)};

\draw [opacity = 0.6] (10:1) -- ++(-170:0.3) coordinate (c);
\fill (c) circle (\r cm);
\end{scope}
\end{scope}

\begin{scope}[xshift = 7.5cm]
\draw [blue!80!black] (0,0) -- +(-110:1);
\draw [blue!80!black] (0,0) -- +(120:1);
\draw [purple!50!red] (0,0) -- +(10:1);
\draw [green!70!purple] (0,0) -- +(-170:1);
\fill (0,0) circle (\r cm);

\draw [-Stealth] (0,-1.25) -- +(0,-0.75); 

\begin{scope}[yshift = -3cm]
\draw [blue!80!black] ($(0,0)+(-110:1)$) .. controls +(70:0.75) and \control{(120:1)}{(-60:0.75)};
\draw [purple!50!red] (0,0) -- +(10:1);
\draw [green!70!purple] (0,0) -- +(-170:1);
\fill (0,0) circle (\r cm);
\end{scope}
\end{scope}

\begin{scope}[xshift = 10.25cm]
\draw [purple!50!red] (0,0) -- +(-45:1);
\draw [purple!50!red] (0,0) -- +(45:1);
\draw [green!70!purple] (0,0) -- +(65:1);
\draw [green!70!purple] (0,0) -- +(10:1);
\fill (0,0) circle (\r cm);

\draw [-Stealth] (0,-1.25) -- +(0,-0.75); 

\begin{scope}[yshift = -3cm]
\draw [purple!50!red] ($(0,0)+(-45:1)$) .. controls +(135:0.75) and \control{(45:1)}{(-135:0.75)};
\draw [green!70!purple] ($(0,0)+(10:1)$) .. controls +(-170:0.75) and \control{(65:1)}{(-115:0.75)};
\end{scope}
\end{scope}

\end{tikzpicture}
\caption{Diverse paths of a family of paths (red, blue, and green paths) around a vertex. Left: tangent paths around a vertex and the associated embedded lift. Right: transverse paths around a vertex and the associated lift which is not an embedded lift.}
\label{pic_ex_transverse}
\end{center}
\end{figure}

In the planar framework, embedded immersions are exactly captured by ordered immersion for which images of edges are tangent paths, and where the order around vertices is induced by the embedding.

\begin{lemma}\label{lem_equiv_plan_imm}
Let $H$ and $G$ be two plane graphs. Then $H$ is an embedded immersion of $G$ if and only if $H$ is an ordered immersion of $G$ for the circular order induced by the embedding on the sphere, with the additional property that the family of paths defined by images of edges is tangent.
\end{lemma}

\begin{proof}
Suppose that there exists an embedded immersion $\phi$ of $H$ in $G$. Without loss of generality, we first perform all edge deletions. Afterwards, each embedded lift of two edges $e$ and $e'$ at a common vertex $v$ yields a new embedded edge $e''$, where the embedded edges $\phi |_{I_e}$ and $\phi |_{I_{e}'}$ have been replaced by their merger $\phi |_{I_{e}''}$, \ie, their concatenation with a small perturbation near $v$; see Figure~\ref{pic_def_emb_lift}. The set $P$ of edges of $G$ which takes part in embedded lifts is a set of edge-disjoint paths and cycles (the latter for lifts creating a self-loop). By definition of embedded lift, $P$ is a family of tangent paths. In addition, the edges at the beginning and end of these paths are unchanged at their starts and ends, respectively, so that the ordering condition around vertices is satisfied and we have an ordered immersion from $H$ in $G$. 

Reciprocally, if $H$ is an ordered immersion of $G$ where the family of paths $P$, which are images of edges, is a tangent family of paths, we first delete all edges of $H$ not part of $P$ and follow by deleting all isolated vertices. Let $v$ be an inner vertex of a path of $P$. At this point, all edges around $v$ are edges of $P$. Since $P$ is a family of tangent paths, there exist two edges $e,e'$ of the same path incident to $v$ and on the boundary of a common face $F$ (hence they are consecutive around $v$). An embedded lift can be performed at $v$ on $e,e'$ by choosing the arc on $F$ so that $e,e'$ are replaced by their merger in $P$. We perform this lift and repeat as long as there remain inner vertices in $P$. We suppress the remaining isolated vertices and, by construction, the resulting graph is $H$ and was obtained by the sequence described above.
\end{proof}

In Section~\ref{subsec_im_gm}, we consider embedded directed regular graphs. A plane digraph $H$ is a \emphdef{embedded immersion} of a plane digraph $G$ if it is an embedded immersion with the additional constraint that paths that are images of edges are directed paths. In term of embedded lift, it means that we can only merge edges that follow themselves: the head of the first one is the tail of the second one as illustrated in Figure~\ref{pic_dir_lifts}.

\begin{figure}[ht]
\begin{center}
\begin{tikzpicture}
\def\e{0.045} \def\eb{0.075}

\begin{scope}
\draw (0.5,0) -- ($(0.5,0)+(135:0.75)$) node [midway, sloped] {$\arrowIns{1.5}$};
\draw (0.5,0) -- ($(0.5,0)+(-135:0.75)$) node [midway, sloped] {$\arrowOuts{1.5}$};
\draw (0.5,0) -- ($(0.5,0)+(45:0.75)$) node [midway, sloped] {$\arrowIns{1.5}$};
\draw (0.5,0) -- ($(0.5,0)+(-45:0.75)$) node [midway, sloped] {$\arrowOuts{1.5}$};
\draw [dotted] ($(0.5,0)+(135:0.75)$) -- ++(135:0.25);
\draw [dotted] ($(0.5,0)+(-135:0.75)$) -- ++(-135:0.25);
\draw [dotted] ($(0.5,0)+(45:0.75)$) -- ++(45:0.25);
\draw [dotted] ($(0.5,0)+(-45:0.75)$) -- ++(-45:0.25);

\fill (0.5,0) circle (\eb cm);
\end{scope}

\draw [-Stealth] (1.5,0) -- ++(0.75,0);
\draw [-Stealth] (-0.5,0) -- ++(-0.75,0);

\begin{scope}[xshift = 3cm]
\draw ($(0.5,0)+(135:0.75)$) .. controls +(-45:0.65) and \control{($(0.5,0)+(-135:0.75)$)}{(45:0.65)} node [pos = 0.5, sloped] {$\arrowIns{1.25}$};
\draw [dotted] ($(0.5,0)+(135:0.75)$) -- ++(135:0.25);
\draw [dotted] ($(0.5,0)+(-135:0.75)$) -- ++(-135:0.25);
\draw [dotted] ($(0.5,0)+(45:0.75)$) -- ++(45:0.25);
\draw [dotted] ($(0.5,0)+(-45:0.75)$) -- ++(-45:0.25);
\draw ($(0.5,0)+(45:0.75)$) .. controls +(-135:0.65) and \control{($(0.5,0)+(-45:0.75)$)}{(135:0.65)} node [pos = 0.5, sloped] {$\arrowOuts{1.25}$};
\end{scope}

\begin{scope}[xshift = -3cm]
\draw ($(0.5,0)+(135:0.75)$) .. controls +(-45:0.65) and \control{($(0.5,0)+(45:0.75)$)}{(-135:0.65)} node [pos = 0.5, sloped] {$\arrowIns{1.25}$};
\draw [dotted] ($(0.5,0)+(135:0.75)$) -- ++(135:0.25);
\draw [dotted] ($(0.5,0)+(-135:0.75)$) -- ++(-135:0.25);
\draw [dotted] ($(0.5,0)+(45:0.75)$) -- ++(45:0.25);
\draw [dotted] ($(0.5,0)+(-45:0.75)$) -- ++(-45:0.25);
\draw ($(0.5,0)+(-45:0.75)$) .. controls +(135:0.65) and \control{($(0.5,0)+(-135:0.75)$)}{(45:0.65)} node [pos = 0.5, sloped] {$\arrowOuts{1.25}$};
\end{scope}
\end{tikzpicture}
\caption{Embedded directed lifts.}
\label{pic_dir_lifts}
\end{center}
\end{figure} 

The same proof as Lemma~\ref{lem_equiv_plan_imm} with directed paths allows us to prove its equivalent in the directed setting: 

\begin{lemma}\label{lem_equiv_plan_imm_dir}
Let $H$ and $G$ be two plane digraphs. Then $H$ is an embedded directed immersion of $G$ if and only if $H$ is an ordered immersion of $G$ for the circular order induced by the embedding on the sphere, with the additional property that the family of directed paths defined by images of edges is tangent.
\end{lemma}

\section{Bond-linked decompositions.}
\label{sec_bond_linked}

The aim of this section is to establish the existence of well-structured decompositions on $2$-vertex-connected abstract graphs. These decompositions will be used as primary blocks for the decompositions used in Section~\ref{sec_bounded_cw}: The bond property is very useful for embeddings, while the linked property is a cornerstone of Nash-Williams' argument. The existence of decompositions satisfying the linked property has been established in the graph minor series from Robertson and Seymour~\cite{Graph_Minors_IV} and improvements to the methods used there were later developed by Geelen, Gerards, and Whittle~\cite{Geelen_wqo_branchwidth}. The existence of bond carving-decomposition was established in~\cite{Seymour_Ratcatcher}. The novelty of this section is to enforce both properties at the same time. To do so we first consider a linked carving-decomposition, and then refine it into being bond by following a constructive version of the existence proof of a bond decomposition from~\cite{Seymour_Ratcatcher} and checking that the linked property of the decomposition is preserved. The existence of such decompositions might be of independent interest.

\subsection{Linked decompositions.}
\label{subsec_linked}

Let $G$ be an abstract graph. Let $A,B \subseteq V(G)$ be two disjoint sets of vertices of $G$, and denote by $\emphdef{\mcut (A,B)} = \min_{A \subset X \subset B^c} |\cut(X)|$ as the minimum size of a cut separating $A$ and $B$. Note that for $A,B,C \subset V(G)$ with $A$ and $C$ disjoint, and $B \subset C$; we have that $\mcut(A,B) \leq \mcut(A,C)$, because a cut separating $A$ and $C$ is also a cut separating $A$ and $B$. Let $(T, \phi)$ be a carving-decomposition of $G$. Let $a,b$ be two edges of $T$, and let $P$ be the minimal path in $T$ containing both $a$ and $b$. Denote by ${T}^a_i$ the subtree of $T \smallsetminus a$ that does not contain $b$, and denote by $A \subseteq V(G)$ the vertices of $G$ displayed by $a$ and associated to ${T}^a_i$. Denote similarly the set of vertices $B \subseteq V(G)$ displayed by $b$ and associated to the tree of $T \smallsetminus b$ not containing $a$; see Figure~\ref{pic_def_FG_linked}. The edges $a$ and $b$ are said to be \emphdef{linked} if the minimal width of an edge of $P$ is equal to the minimal cut separating $A$ and $B$, \ie, if $\mcut (A,B) = \min_{e \in P} \wid (e)$. A carving-decomposition $(T,\phi)$ is said to be \emphdef{linked} if all pairs of edges in the tree $T$ are linked. Note that any edge in $P$ displays a set of vertices separating $A$ and $B$, so that by definition $\mcut (A,B) > \min_{e \in P} \wid (e)$ is impossible.

\begin{figure}[ht]
\begin{center}
\begin{tikzpicture}
\def\e{0.065} \def\op{0.15}

\path [rotate around={-40:(0,0)}] (1.5,0) arc (0:50:1.5 and 1) coordinate (c1);
\path [rotate around={-40:(0,0)}] (1.5,0) arc (0:75:1.5 and 1) coordinate (c2);
\path [rotate around={-40:(0,0)}] (1.5,0) arc (0:110:1.5 and 1) coordinate (c3);
\coordinate (c4) at (0.25,-0.5);
\coordinate (c5) at (-0.5,0.25);
\coordinate (c6) at (5,0.15);
\path [rotate around={40:(4.5,0)}] (6,0) arc (0:130:1.5 and 1) coordinate (c7);
\path [rotate around={40:(4.5,0)}] (6,0) arc (0:90:1.5 and 1) coordinate (c8);

\draw [thick] (c4) -- (c1);
\draw [thick] (c4) -- (c2);
\draw [thick] (c5) -- (c3);
\draw [thick] (c3) .. controls +(-135:0.5) and \control{(c2)}{(-135:0.5)};
\draw [thick] (c6) -- (c7);
\draw [thick] (c6) -- (c8);
\draw [thick] (c3) -- ++(50:0.3);
\draw [dotted] ($(c3)+(50:0.3)$) -- ++(50:0.3);
\draw [thick] (c3) -- ++(20:0.3);
\draw [dotted] ($(c3)+(20:0.3)$) -- ++(20:0.3);
\draw [thick] (c2) -- ++(60:0.3);
\draw [dotted] ($(c2)+(60:0.3)$) -- ++(60:0.3);
\draw [thick] (c2) -- ++(20:0.3);
\draw [dotted] ($(c2)+(20:0.3)$) -- ++(20:0.3);
\draw [thick] (c1) -- ++(10:0.3);
\draw [dotted] ($(c1)+(10:0.3)$) -- ++(10:0.3);
\draw [thick] (c7) -- ++(170:0.3);
\draw [dotted] ($(c7)+(170:0.3)$) -- ++(170:0.3);
\draw [thick] (c8) -- ++(160:0.3);
\draw [dotted] ($(c8)+(160:0.3)$) -- ++(160:0.3);
\draw [thick] (c8) -- ++(200:0.5);
\draw [dotted] ($(c8)+(200:0.5)$) -- ++(200:0.3);

\begin{scope}[black, opacity = 0.5]
\draw (c1) -- (0.65,-0.75);
\draw (0.65,-0.75) -- (0,-1.75);
\draw (c4) -- (0.65,-0.75);
\draw (c2) .. controls +(210:0.75) and \control{(0,-1.75)}{(110:0.75)};
\draw (c3) -- (-1.4, -1.5);
\draw (c5) .. controls +(160:0.5) and \control{(-1.4, -1.5)}{(110:0.75)};
\draw (-1.4, -1.5) -- (-0.2,-2);
\draw (0,-1.75) -- (-0.2,-2);
\draw (c8) .. controls +(10:2.25) and \control{(4.5,-1.75)}{(45:0.75)};
\draw  (4.25,-1) -- (c6);
\draw  (4.25,-1) -- (c7);
\draw  (4.5,-1.75) -- (4.25,-1);
\draw (0.8,-2.4) -- ++(60:0.3);
\draw [dotted, thick] ($(0.8,-2.4)+(60:0.3)$) -- ++(60:0.3);
\draw (3.5,-2.4) -- ++(120:0.3);
\draw [dotted, thick] ($(3.5,-2.4)+(120:0.3)$) -- ++(120:0.3);
\end{scope}
\draw [red!50!purple] (-0.2,-2) .. controls +(-90:0.5) and \control{(0.8,-2.4)}{(180:0.5)} coordinate [pos = 0.75] (m1);
\draw [blue!80!black] (4.5,-1.75) .. controls +(-90:0.5) and \control{(3.5,-2.4)}{(0:0.5)} coordinate [pos = 0.75] (m2);

\draw [dashed] (0.8,-2.4) -- ++(-60:0.3) -- ($(3.5,-2.4)+(-120:0.3)$) -- (3.5,-2.4);

\filldraw [rotate around={-40:(0,0)}, purple!50!red, fill opacity = \op] (0,0) circle (1.65 and 1.15);
\filldraw [rotate around={40:(4.5,0)}, blue!80!black, fill opacity = \op] (4.5,0) circle (1.65 and 1.15);

\foreach \i in {1,...,8}{\fill (c\i) circle (\e cm);};
\node at (-0.5, 0.9) [color={rgb: red!50!purple,2; black,1}] {$A$}; 
\node at (m1) [below, color={rgb: red!50!purple,2; black,1}] {$a$}; 
\node at (m2) [below, color={rgb: blue,2; black,1}] {$b$}; 
\node at (5, 0.9) [color={rgb: blue,2; black,1}] {$B$}; 
\node at (2.25, 0.9) {\Large $G$};
\node at (2.25,-1.8) [opacity = 0.5] {\Large $T$}; 
\draw [thick, green!30!black] ($(m1)+(-135:0.5)$) .. controls +(-90:0.75) and \control{($(m1)!0.5!(m2)+(0,-0.75)$)} {(90:0.5)};
\draw [thick, green!30!black] ($(m2)+(-45:0.5)$) .. controls +(-90:0.75) and \control{($(m1)!0.5!(m2)+(0,-0.75)$)} {(90:0.5)};
\node [below, green!30!black] at ($(m1)!0.5!(m2)+(0,-0.75)$) {$P$}; 
\end{tikzpicture}
\caption{Definition of the sets $A$ and $B$ from the edges $a,b$, and the path $P$.}
\label{pic_def_FG_linked}
\end{center}
\end{figure}

The existence of linked carving-decomposition is non-trivial; in fact, this was a critical point of~\cite{Graph_Minors_IV} and the following Theorem~\ref{th_link_decomp} is mentioned as one of the main contributions of~\cite{Geelen_wqo_branchwidth} (described as linked branch decomposition there, as they work in the broader framework of matroids and submodular functions). 

\begin{theorem}{\cite[Theorem 2.1]{Geelen_wqo_branchwidth}}
\label{th_link_decomp}
Let $G$ be an abstract graph of carving-width $k \in \N$. Then, there exists a linked carving-decomposition of $G$ of width $\leq k$.
\end{theorem}

For the sake of completeness, and because elements of these proofs will be used in the proof of Theorem~\ref{th_bond_link_decomp}, we include a slight adaptation of Geelen, Gerards and Whittle's proof~\cite{Geelen_wqo_branchwidth} in the remainder of this Section~\ref{subsec_linked}. 

Intuitively, a linked carving-decomposition is a carving-decomposition where an edge should have a width as low as possible, since for each path joining two edges, an edge should achieve the minimal possible cut. Hence, to find such decomposition, it is natural to try to minimise the width of each edge of a decomposition.

Let $(T,\phi)$ be a carving-decomposition of $G$. We denote by $T_k$ the subforest of $T$ induced by the union of edges of width $\geq k$, and write $C(T_k)$ for the set of connected components of $T_k$. Let us consider a strict partial order $<_w$ on carving-decompositions of $G$, such that $(T, \phi) <_w (T',\phi')$ if there exists $k \in \N$ such that: 
\[
  \begin{array}{l}
    \text{at weight} \ k: \left\{\begin{array}{ll}
          \text{either} & \ |E(T_k)| < |E(T'_k)|\\
          \text{or} & \ |E(T_k)| = |E(T'_k)| \ \text{and} \ |C(T_k)| > |C(T_k')|
          \end{array}\right.\\[10pt]
    \text{and for all }\ell > k:  |E(T_\ell)| = |E(T'_\ell)| \ \text{and} \ |C(T_\ell)| = |C(T'_\ell)|.\\
  \end{array}
\]

When there is no ambiguity, we simply write $T <_w T'$. On a high level, $<_w$ is a refinement of a colexicographic order on the tuples $(|w^{-1}(k)|)_{0 \leq k \leq |E(G)|}$. Note that it follows from the definition that if $T <_w T'$, then $\wid(T) \leq \wid(T')$.
By definition of $M$, and since $f$ and $x_i$ are linked in $T$:
\begin{lemma}\label{lem_linked_min}
Let $G$ be an abstract graph and $(T,\phi)$ be a carving-decomposition of $G$. If $(T,\phi)$ is not linked, then there exists a carving-decomposition $(\hat{T}, \hat{\phi})$ of $G$ such that $\hat{T} <_w T$.
\end{lemma}

\begin{proof}
Let $(T,\phi)$ be a carving-decomposition of $G$ that is not linked, and let $a$ and $b$ be two edges of $T$ which are not linked. Let $A$ be the set of vertices displayed by the tree of $T \smallsetminus a$ not containing $b$, let $B$ be the set of vertices displayed by the tree of $T \smallsetminus b$ not containing $a$, and let $P$ be the minimal path of $T$ containing both $a$ and $b$. Up to re-indexing, we have $A = V_1(a) = \phi^{-1} (L(T_1 (a)))$ and $B = V_1 (b)$.

For two sets of vertices $X,Y \subseteq V(G)$, we say that $X$ \emphdef{splits} $Y$ if $Y \smallsetminus X \neq \varnothing$ and $X \cap Y \neq \varnothing$ (this relation is not reflexive). Pick a set of vertices $X \subseteq V(G)$ minimising the cut between $A$ and $B$, \ie, such that $A \subset X \subset B^c$ and $|\cut(X)| = \mcut (A,B)$. Additionally, pick such a set $X$ such that it splits a minimum number of sets of vertices displayed by edges of $T$, \ie, $X$ splits a minimum number of sets $V_i(c)$ with $c \in E(T)$ and $i \in \{ 1,2 \}$. 

We now define a new carving-decomposition $(\hat{T}, \hat{\phi})$ of $G$. The subtrees $T_2(b)$ and $T_2(a)$ both have exactly $1$ vertex of degree $2$, the ones incident to $b$ and $a$ in $T$ respectively; see Figure~\ref{pic_proof_linked_decomp}. We define $\hat{T}$ by taking a copy of $T_2(b)$ and a copy of $T_2(a)$, and connecting their vertices of degree $2$ by a new edge. Then, we (partially) define the function $\hat{\phi}$ on $L(\hat{T})$ by $\hat{\phi}(v) = \phi(v)$ if $v$ belongs to the copy of $T_2(b)$ and $\phi(v) \in X$, or if $v$ belongs to the copy of $T_2(a)$ and $\phi(v) \in X^c$, and $\hat{\phi}(v)$ is undefined otherwise (see Figure~\ref{pic_proof_linked_decomp}).

\begin{figure}[ht]
\begin{center}
\begin{tikzpicture}[scale = 1.20]
\def\e{0.05} \def\el{0.2} \def\exl{0.3} \def\se{0.5} \def\ses{0.33}
\def\op{1} \def\ei{0.03} 

\clip (-2.15,-1.7) rectangle (10.05,1.72);

\coordinate (d0) at (0,0);
\coordinate (d1) at (30:\se);
\coordinate (d2) at (150:\se);
\coordinate (d3) at (-90:\ses);
\coordinate (d4) at ($(d1)+(90:\se)$);
\coordinate (d5) at ($(d1)+(-30:\se)$);
\coordinate (d6) at ($(d2)+(90:\se)$);
\coordinate (d7) at ($(d2)+(-150:\se)$);
\coordinate (d8) at ($(d7)+(-90:\se)$);
\coordinate (d9) at ($(d7)+(150:\se)$);
\coordinate (d10) at ($(d4)+(30:\se)$);
\coordinate (d11) at ($(d4)+(150:\ses)$);
\coordinate (d12) at ($(d8)+(-30:\ses)$);
\coordinate (d13) at ($(d8)+(-150:\ses)$);
\coordinate (d14) at ($(d10)+(90:\ses)$);
\coordinate (d15) at ($(d10)+(-30:\ses)$);
\coordinate (d16) at ($(d5)+(30:\se)$);
\coordinate (d17) at ($(d5)+(-90:\se)$);
\coordinate (d18) at ($(d17)+(-150:\ses)$);
\coordinate (d19) at ($(d17)+(-30:\ses)$);
\coordinate (d20) at ($(d18)+(150:\ses)$);
\coordinate (d21) at ($(d18)+(-90:\ses)$);
\coordinate (d22) at ($(d16)+(90:\ses)$);
\coordinate (d23) at ($(d16)+(-30:\se)$);
\coordinate (d24) at ($(d23)+(30:\ses)$);
\coordinate (d25) at ($(d23)+(-90:\ses)$);

\draw (d0) -- (d1); 
\draw (d0) -- (d2); 
\draw (d0) -- (d3); 
\draw (d1) -- (d4); 
\draw (d1) -- (d5); 
\draw (d2) -- (d6); 
\draw [blue!80!black, thick] (d2) -- (d7); 
\draw (d7) -- (d8); 
\draw (d7) -- (d9); 
\draw (d4) -- (d10); 
\draw (d4) -- (d11); 
\draw (d8) -- (d12); 
\draw (d8) -- (d13); 
\draw (d10) -- (d14);
\draw (d10) -- (d15);
\draw (d5) -- (d16);
\draw (d5) -- (d17);
\draw (d17) -- (d18);
\draw (d17) -- (d19);
\draw (d18) -- (d20);
\draw (d18) -- (d21);
\draw (d16) -- (d22);
\draw [red!70!black, thick] (d16) -- (d23);
\draw (d23) -- (d24);
\draw (d23) -- (d25);

\foreach \i in {0,1,2,4,5,7,8,10,16,17,18,23}{\fill (d\i) circle (\ei cm);};
\foreach \i in {3,6,9,11,12,13,14,15,19,20,21,22,24,25}{\fill (d\i) circle (\e cm);};
\foreach \i in {3,6,9,12,13,19,22}{\fill (d\i) [green!70!purple] circle (\e*0.75 cm);};

\draw [blue!80!black, smooth cycle, tension = 0.6, opacity = \op] plot coordinates {($(d9)+(120:\el)$) ($(d13)+(-150:\el*0.5)$) ($(d12)+(-30:\el)$) ($(d7)$)};
\draw [blue!80!black, smooth cycle, tension = 0.8, opacity = \op] plot coordinates {($(d6)+(90:\exl)$) ($(d14)+(90:\exl)$) ($(d24)+(0:\exl)$) (2.4,-0.75) ($(d21)+(-60:\exl)$) ($(d3)+(-90:\el)$) ($(d2)$)};
\draw [red!70!black, smooth cycle, tension = 0.8, opacity = \op] plot coordinates {($(d23)$) ($(d25)+(-105:\el)$) ($(d24)+(45:\el)$)};
\draw [red!70!black, smooth cycle, tension = 0.5, opacity = \op] plot coordinates {($(d16)$) ($(d22)+(30:\el)$) ($(d14)+(90:\el*0.5)$) ($(d9)+(150:\exl)$) (-2.1,-0.25) ($(d13)+(-150:\exl)$) ($(d21)+(-90:\el*0.5)$) ($(d19)+(-30:\el)$)};

\node at ($(d2)+(-150:\se*0.5)+(120:0.15)$) [blue!50!black] {$a$};
\node at (-1.75,-0.25) [blue!80!black] {\footnotesize $T_1(a)$};
\node at (2.15,1) [blue!80!black] {\footnotesize $T_2(a)$};
\node at ($(d16)+(-30:\se*0.5)+(60:0.2)$) [red!70!black] {$b$};
\node at (2.25,-0.3) [red!70!black] {\footnotesize $T_1(b)$};
\node at (-1.1,0.95) [red!70!black] {\footnotesize $T_2(b)$};
\node at (0,-1.5) {$T$};
 
\draw [-Stealth] (2.5,0.5) -- +(1,0);

\begin{scope}[xshift = 5cm]
\coordinate (d0) at (0,0);
\coordinate (d1) at (30:\se);
\coordinate (d2) at (150:\se);
\coordinate (d3) at (-90:\ses);
\coordinate (d4) at ($(d1)+(90:\se)$);
\coordinate (d5) at ($(d1)+(-30:\se)$);
\coordinate (d6) at ($(d2)+(90:\se)$);
\coordinate (d7) at ($(d2)+(-150:\se)$);
\coordinate (d8) at ($(d7)+(-90:\se)$);
\coordinate (d9) at ($(d7)+(150:\se)$);
\coordinate (d10) at ($(d4)+(30:\se)$);
\coordinate (d11) at ($(d4)+(150:\ses)$);
\coordinate (d12) at ($(d8)+(-30:\ses)$);
\coordinate (d13) at ($(d8)+(-150:\ses)$);
\coordinate (d14) at ($(d10)+(90:\ses)$);
\coordinate (d15) at ($(d10)+(-30:\ses)$);
\coordinate (d16) at ($(d5)+(30:\se)$);
\coordinate (d17) at ($(d5)+(-90:\se)$);
\coordinate (d18) at ($(d17)+(-150:\ses)$);
\coordinate (d19) at ($(d17)+(-30:\ses)$);
\coordinate (d20) at ($(d18)+(150:\ses)$);
\coordinate (d21) at ($(d18)+(-90:\ses)$);
\coordinate (d22) at ($(d16)+(90:\ses)$);

\draw (d0) -- (d1); 
\draw (d0) -- (d2); 
\draw (d0) -- (d3); 
\draw (d1) -- (d4); 
\draw (d1) -- (d5); 
\draw (d2) -- (d6); 
\draw [blue!80!black, thick] (d2) -- (d7); 
\draw (d7) -- (d8); 
\draw (d7) -- (d9); 
\draw (d4) -- (d10); 
\draw (d4) -- (d11); 
\draw (d8) -- (d12); 
\draw (d8) -- (d13); 
\draw (d10) -- (d14);
\draw (d10) -- (d15);
\draw (d5) -- (d16);
\draw (d5) -- (d17);
\draw (d17) -- (d18);
\draw (d17) -- (d19);
\draw (d18) -- (d20);
\draw (d18) -- (d21);
\draw (d16) -- (d22);

\foreach \i in {0,1,2,4,5,7,8,10,16,17,18,23}{\fill (d\i) circle (\ei cm);};
\foreach \i in {3,6,9,12,13,19,22}{\fill (d\i) circle (\e cm);};
\foreach \i in {3,6,9,12,13,19,22}{\fill (d\i) [green!70!purple] circle (\e*0.75 cm);};

\draw [red!70!black, smooth cycle, tension = 0.5, opacity = \op] plot coordinates {($(d16)$) ($(d22)+(30:\el)$) ($(d14)+(90:\el*0.5)$) ($(d9)+(150:\el)$) ($(d13)+(-150:\el)$) ($(d21)+(-90:\el*0.5)$) ($(d19)+(-30:\el)$)};

\node at ($(d2)+(-150:\se*0.5)+(120:0.15)$) [blue!50!black] {$a$};
\node at (-1.1,0.95) [red!70!black] {\footnotesize $T_2(b)$};
\draw ($(2.4,0)+(150:\se)$) -- (d16);
\node at (1.75,-1.5) {$\hat{T}$};

\begin{scope}[xshift = 2.4cm]
\coordinate (d0) at (0,0);
\coordinate (d1) at (30:\se);
\coordinate (d2) at (150:\se);
\coordinate (d3) at (-90:\ses);
\coordinate (d4) at ($(d1)+(90:\se)$);
\coordinate (d5) at ($(d1)+(-30:\se)$);
\coordinate (d6) at ($(d2)+(90:\se)$);
\coordinate (d10) at ($(d4)+(30:\se)$);
\coordinate (d11) at ($(d4)+(150:\ses)$);
\coordinate (d14) at ($(d10)+(90:\ses)$);
\coordinate (d15) at ($(d10)+(-30:\ses)$);
\coordinate (d16) at ($(d5)+(30:\se)$);
\coordinate (d17) at ($(d5)+(-90:\se)$);
\coordinate (d18) at ($(d17)+(-150:\ses)$);
\coordinate (d19) at ($(d17)+(-30:\ses)$);
\coordinate (d20) at ($(d18)+(150:\ses)$);
\coordinate (d21) at ($(d18)+(-90:\ses)$);
\coordinate (d22) at ($(d16)+(90:\ses)$);
\coordinate (d23) at ($(d16)+(-30:\se)$);
\coordinate (d24) at ($(d23)+(30:\ses)$);
\coordinate (d25) at ($(d23)+(-90:\ses)$);

\foreach \i in {0,1,2,4,5,10,16,17,18,23}{\fill (d\i) circle (\ei cm);};
\foreach \i in {11,14,15,20,21,24,25}{\fill (d\i) circle (\e cm);};

\draw (d0) -- (d1); 
\draw (d0) -- (d2); 
\draw (d0) -- (d3); 
\draw (d1) -- (d4); 
\draw (d1) -- (d5); 
\draw (d2) -- (d6);
\draw (d4) -- (d10); 
\draw (d4) -- (d11); 
\draw (d10) -- (d14);
\draw (d10) -- (d15);
\draw (d5) -- (d16);
\draw (d5) -- (d17);
\draw (d17) -- (d18);
\draw (d17) -- (d19);
\draw (d18) -- (d20);
\draw (d18) -- (d21);
\draw (d16) -- (d22);
\draw [red!70!black, thick] (d16) -- (d23);
\draw (d23) -- (d24);
\draw (d23) -- (d25);

\node at ($(d16)+(-30:\se*0.5)+(60:0.2)$) [red!70!black] {$b$};
\node at (2.15,1) [blue!80!black] {\footnotesize $T_2(a)$};
\draw [blue!80!black, smooth cycle, tension = 0.8, opacity = \op] plot coordinates {($(d6)+(90:\exl)$) ($(d14)+(90:\exl)$) ($(d24)+(0:\exl)$) (2.4,-0.75) ($(d21)+(-60:\exl)$) ($(d3)+(-90:\el)$) ($(d2)$)};
\end{scope}
\end{scope}
\end{tikzpicture}
\caption{Construction of $\hat{T}$ from $T$. Green leaves are labelled by elements of $X$ while black ones belong to $X^c$. Leaves with no points are unlabelled.}
\label{pic_proof_linked_decomp}
\end{center}
\end{figure}

Let $\hat{e}$ be an edge of $\hat{T}$ that is the copy of an edge $e$ in $T$ (note that there might be a second copy of $e$ in $\hat{T}$). If $\hat{e} \in T_2(b)$, let us assume by symmetry that $W = V_1(e)$ is the set of vertices of $G$ displayed by the tree of $T \smallsetminus e$ not containing $b$. By definition, $\wid (e) = |\cut(W)|$ and $\wid (\hat{e}) = | \cut( X \cap W ) |$. Note that for any set of vertices $X, Y$ of $V(G)$, the submodularity of $|\cut(\cdot)|$ implies that $|\cut(X \cap Y)| + |\cut(X \cup Y)| \leq |\cut(X)| + |\cut(Y)|$. Then, $\wid(\hat{e}) + |\cut(X \cup W)| = |\cut(X \cap W)| + |\cut(X \cup W)| \leq |\cut(X)| + |\cut(W)| = \mcut (A,B) + \wid(e)$. Since $W \subset B^c$ by definition of $W$, $A \subset X \cup W \subset B^c$ so that $\mcut (A,B) \leq |\cut(X \cup W)|$. Hence, $\wid(\hat{e}) \leq \wid(e)$ with equality if and only if $\mcut (A,B) = |\cut(X \cup W)|$.

Now let us assume that $\wid (\hat{e}) = \wid (e)$ and that there is a second copy $e^*$ of $e$ in $T_2(a)$. We saw that $\mcut (A,B) = |\cut(X \cup W)|$ so that by definition of $X$, $W \cup X$ splits at least as many displayed sets as $X$. Let us assume by contradiction that $X$ splits $W$. Let $Y$ be a set of vertices of $G$ displayed by some edge of $T$. If $W \cup X$ splits $Y$, then $Y \smallsetminus (W \cup X) \neq \varnothing$, which implies that $Y \smallsetminus X \neq \varnothing$. It also follows that $Y$ is not contained in $X$. Since displayed sets are either disjoint or satisfy a containment relation, either $W \subset Y$, which involves that $\varnothing \varsubsetneq W \cap X \subset Y$, or $W$ and $Y$ are disjoint, which involves that $Y \cap (W \cup X) = Y \cap X \neq \varnothing$. In both cases, $X$ splits $Y$ and $X$ splits all displayed sets split by $W$, but this is absurd since $X$ splits in addition $W$. Hence $X$ does not split $W$: either $W \smallsetminus X = \varnothing$ or $W \cap X = \varnothing$, the former is true since $\wid (\hat{e}) = |\cut(W \cap X)| = \wid (e)$. If $e$ is in $P$, then $\wid (e^*) = |\cut(W \cup X)| = \mcut (A,B)$, otherwise $\wid (e^*) = |\cut(W \smallsetminus X)| = |\cut(\varnothing)| = 0$. Hence, at most one copy of $e$ has width greater than $\mcut (A,B)$ if $\wid (\hat{e}) = \wid (e)$.

If $\hat{e} \in T_2(a)$, the arguments above apply to conclude that $\wid(\hat{e}) \leq \wid(e)$ if and only if $\mcut (A,B) = |\cut(X^c \cup W)|$ where $W = V_2(e)$. Then, noticing that $Y \smallsetminus X \neq \varnothing \Leftrightarrow Y \cap X^c \neq \varnothing$ and $Y \cap X \neq \varnothing \Leftrightarrow Y \smallsetminus X^c \neq \varnothing$, $X^c$ splits as many sets as $X$, and we reach the same conclusion: $\wid (\hat{e}) \leq \wid (e)$ and at most one copy of $e$ has width greater than $\mcut (A,B)$ if $\wid (\hat{e}) = \wid (e)$. 

The results above imply that for $i > \mcut(A,B)$, $|E(\hat{T}_i)| < |E(T_i)|$. By construction of $\hat{T}$, each connected component of $T_i$ is copied as at least a component of $\hat{T}$, which can be separated into several ones in $\hat{T}_i$ since the weight of these edges cannot increase (or disappear entirely if all edges have small enough weight). At this point either there exists $i > \mcut (A,B)$ allowing us to conclude that $\hat{T} <_w T$ or for all $i > \mcut (A,B)$, $|E(\hat{T}_i)| = |E(T_i)|$ and $|C(\hat{T}_i)| = |C(T_i)|$. Let us see that this latter option leads to a contradiction. Set $j = \mcut(A,B) +1$, we know that $|C(\hat{T}_j)| = |C(T_j)|$ and $|E(\hat{T}_j)| = |E(T_j)|$. Hence, every component $C_j \in C(T_j)$ has a copy $\hat{C}_j$ in $\hat{T}$, where each edge of $C_j \in C(T_j)$ is present in $\hat{C_j}$. In particular the component $C$ of $T_j$ containing $P$ is copied as $\hat{C}$ in $\hat{T}_j$, which contains the copy of $P$ and the copies of $a$ and $b$. This is absurd: $\wid (n) = j-1$, the copies of $a$ and $b$ cannot lie in the same component of $\hat{T}_j$. Hence, $\hat{T} <_w T$.
\end{proof}

\begin{proof}[Proof of Theorem~\ref{th_link_decomp}]
Recall that, by definition of $<_w$, if $T' <_w T$, then $\wid (T') \leq \wid (T)$. As mentioned in Section~\ref{subsec_backg_graph}, one can remove unlabelled edges of a carving-decomposition $T$. Doing so yields a carving-decomposition $T'$ with less edges than $T$ without increasing the weight of each edge so that $T' <_w T$ if any edge is removed. Furthermore, all carving-decompositions with no unlabelled edges have the same number of edges, and consequently there is a finite number of such decompositions.
Consequently, applying iteratively Lemma~\ref{lem_linked_min} and suppressing unlabelled edges to a carving-decomposition of $G$ of minimal width $\cwid(G)$ eventually yields a minimal element of $<_w$. This minimal element has width $\cwid(G)$, and is linked by Lemma~\ref{lem_linked_min}.
\end{proof}

\subsection{Bond-linked decompositions.}
\label{subsec_bond}

Let $G$ be a connected abstract graph. A carving-decomposition $(T, \phi)$ of $G$ is \emphdef{bond} if for all edges $s \in E(T)$, the subgraphs of $G$ displayed by $s$, $G_1^s$ and $G_2^s$, are connected. An edge not satisfying this condition in a carving-decomposition is said to \emphdef{produce a non-bond cut}. Note that if a connected graph $G$ admits a cut-vertex $v \in V(G)$, then it admits no bond carving-decomposition. Indeed, any carving-decomposition $(T, \phi)$ of $G$ admits an edge $s \in E(T)$ that is incident to the leaf $\phi(v)$, which necessarily produces a non-bond cut. In consequence, Theorem~\ref{th_bond_link_decomp}, the main theorem of this section, focuses on 2-vertex-connected graphs. 

We will present in the following some notions stemming from the work of Seymour and Thomas~\cite{Seymour_Ratcatcher} which lead to the proof of existence of bond carving-decomposition of minimum carving-width. In particular, Lemma~\ref{lem_branching} is a reformulation of (\cite[(5.3)]{Seymour_Ratcatcher}) expressed with rebranching, a local modification of an edge on a carving-decomposition. While rebranchings are not explicitly mentioned in \cite{Seymour_Ratcatcher}, they appear as the key algorithmic step of a constructive version of their proof of existence of bond carving-decomposition: perform rebranchings on a carving-decomposition to decrease a well-defined potential whose vanishing certifies bondness. The technical novelty of this section, and the crux of this paper, are encapsulated in Lemma~\ref{lem_bl} and its application in the proof of Theorem~\ref{th_bond_link_decomp}. Lemma~\ref{lem_bl} states that rebranchings can preserve the linked property if some geometric constraint is satisfied. The final proof of Theorem~\ref{th_bond_link_decomp} starts with a linked carving-decomposition and is transformed into a bond one by performing rebranchings while preserving the linked property because the initial carving-decomposition enforces the relevant geometric property.

In the following, we consider only carving-decompositions with no unlabelled leaf; this is always possible, as explained in Section~\ref{subsec_backg_graph}. Notice that, for two disjoint sets of vertices $A,B \subseteq V(G)$ inducing connected subgraphs $G[A]$ and $G[B]$, the subgraph $G[A\cup B]$ induced by $A \cup B$ is connected if and only if the cut between $A$ and $B$ is non-empty, \ie, $\cut(A,B) \neq \varnothing$. Rebranchings will only be performed on edges whose endpoints are both inner vertices of a carving-decomposition, such an edge will be called an \emphdef{inner edge} in this section.

Let $v$ be an inner vertex of $T$, and denote by $(T^v_i)_{i \in \{1,2,3\}}$ the subtrees of $T$ rooted at the three neighbours of $v$. These subtrees partition vertices of $G$ into three sets $V^v_i = \phi^{-1}(L(T_i))$, and splits $G$ into three subgraphs $G_i^v = G[V^v_i]$ for $i \in \{1,2,3\}$; see Figure~\ref{pic_cut_scheme}. Note that none of the subgraphs $G_i^v$ is empty, because every tree $T^v_i$ contains a labelled leaf.

\begin{figure}[ht]
\begin{center}
\begin{tikzpicture}[scale = 0.75]
\filldraw [color={rgb: green!50!blue,2; black,1}, fill opacity = 0.1] (0,1.9) circle (2 and 1);
\node [color={rgb: green!50!blue,2; black,1}] at (0,2) {\footnotesize $T^v_2 \mapsto V_2^v$};
\draw (0,0) -- (0,1) -- ++(30:1);
\draw (0,1) -- ++(150:1);
\begin{scope}[rotate = 120]
\filldraw [purple!50!red, fill opacity = 0.1] (0,1.9) circle (2 and 1);
\node [purple!50!red] at (0,2.1) {\footnotesize $T^v_1 \mapsto V_1^v$};
\draw (0,0) -- (0,1) -- ++(30:1);
\draw (0,1) -- ++(150:1);
\end{scope}
\begin{scope}[rotate = -120]
\filldraw [blue, fill opacity = 0.1] (0,1.9) circle (2 and 1);
\node [blue] at (0,2.1) {\footnotesize $T^v_3 \mapsto V_3^v$};
\draw (0,0) -- (0,1) -- ++(30:1);
\draw (0,1) -- ++(150:1);
\end{scope}
\fill (0,0) circle (2pt);
\node [above left] at (0,0) {v};
\node at (-30: 0.75) [above] {$T$};

\begin{scope}[xshift = 7cm]
\draw (120:1) -- (0,0) -- (1,0) node [midway, above] {$s$} -- ++(60:1);
\draw (-120:1) -- (0,0);
\draw (1,0) -- ++(-60:1);
\filldraw [blue, fill opacity = 0.1] (-0.9, 0) circle (1 and 2);
\node [blue] at (-1,0) {\footnotesize $T^s_1 \mapsto V_1^s$};
\node [red!50!purple] at (2,0) {\footnotesize $T^s_2 \mapsto V_2^s$};
\filldraw [purple!50!red, fill opacity = 0.1] (1.9, 0) circle (1 and 2);
\node at ($(1,0)+(-60:1)$) [right] {$T$};
\end{scope}
\end{tikzpicture}
\caption{Subgraphs of $G$ induced by edges and vertices of $T$.}
\label{pic_cut_scheme}
\end{center}
\end{figure}

Because $G$ is connected, we have the property that, for all inner vertices $v \in V(T)$, at most one of the cuts $\cut(G_1^v,G_2^v)$, $\cut(G_2^v,G_3^v)$, and $\cut(G_1^v,G_3^v)$ is empty. Otherwise, one of the $G_i^v$'s would satisfy $\cut(G_i) = \varnothing$ and thus be a disconnected component of $G$. This implies that the following quantity is well defined for all inner vertex $v$ of $T$ where $\{X,Y,Z\} = \{ V^v_1, V^v_2, V^v_3 \}$:

\begin{equation*}
\mu (v) = \left\lbrace 
\begin{array}{cl}
|X| - 1 & \text{ if } \cut(Y,Z) = \varnothing\\
0 & \text{ otherwise.}
\end{array}
\right.
\end{equation*}

The map $\mu$ defines a potential on the whole carving-decomposition by summing over all inner vertices of $T$: $\mu(T) = \sum_{v \in V(T) \smallsetminus L(T)} \mu(v)$. A carving-decomposition with a zero $\mu$-potential is bond:

\begin{lemma}\label{lem_potentialmuzero}
Let $G$ be a 2-vertex-connected graph and $(T,\phi)$ a carving-decomposition of $G$. Then, $\mu(T)=0$ if and only if the carving-decomposition $(T, \phi)$ is bond.
\end{lemma}

\begin{proof} 
Because $G$ is 2-vertex-connected, an empty cut $\cut(G_1^{v}, G_2^{v}) = \varnothing$ implies that $|V( G_3^{v} )| > 1$, whence $\mu (v) > 0$. It follows that if $(T, \phi)$ is not a bond carving-decomposition, there exists an edge $e \in E(T)$ such that $G^{e}_1$ is not connected. Let us choose such an edge $e$ such that $G^{e}_1$ is minimal by inclusion. Let $v$ be the endpoint of $e$ such that $V(G^{e}_1) = V(G^v_1) \cup V(G^v_2)$. By minimality of $G^{e}_1$, both $V(G^v_1)$ and $V(G^v_2)$ are connected but $e$ induces a non-bond  cut. Thus, $\cut(G^v_1, G^v_2) = \varnothing$ and $0 < \mu (v) \leq \mu (T)$.

Conversely, if $\mu (T) > 0$, then there exists $v \in V(T)$ such that $\mu(v) > 0$ and one of the edges of $T$ incident to $v$ induces a non-bond cut.
\end{proof}

First, let us show that, for a connected graph $G$ and an arbitrary carving-decomposition $(T,\phi)$ in which there is an inner edge $s$ producing a non-bond cut, it is possible to decrease the potential $\mu$ of the carving-decomposition by locally modifying the branching of $T$ at $s$. We now describe the \emphdef{rebranching} operation in a tree $T$ (see Figure~\ref{pic_rebranching} for a visual description) on the inner edge $s = (u,v)$ connecting two inner vertices $u$ and $v$.

\begin{figure}[ht]
\begin{center}
\begin{tikzpicture}
\def\e{0.05}
\draw (120:1) -- (0,0) -- (1,0) node [midway, above] {$s$} -- ++(60:1);
\draw (-120:1) -- (0,0);
\draw (1,0) -- ++(-60:1);
\fill (0,0) circle (\e cm);
\fill (1,0) circle (\e cm);
\fill (120:1) circle (\e cm);
\fill (-120:1) circle (\e cm);
\fill ($(1,0)+(60:1)$) circle (\e cm);
\fill ($(1,0)+(-60:1)$) circle (\e cm);

\filldraw [rotate around={30:(120:1.5)}, purple!50!red, fill opacity = 0.15] (120:1.5) circle (1 and 0.75);
\filldraw [rotate around={-30:(-120:1.5)}, orange!70!brown, fill opacity = 0.15] (-120:1.5) circle (1 and 0.75);
\filldraw [rotate around={-30:($(1,0)+(60:1.5)$)}, blue!80!black, fill opacity = 0.15] ($(1,0)+(60:1.5)$) circle (1 and 0.75);
\filldraw [rotate around={30:($(1,0)+(-60:1.5)$)}, green!50!blue, fill opacity = 0.15] ($(1,0)+(-60:1.5)$) circle (1 and 0.75);

\path [rotate around={30:(120:1.5)}] ($(120:1.5)+(1,0)$) arc (0:-120: 1 and 0.75) coordinate (A1);
\path [rotate around={-30:(-120:1.5)}] ($(-120:1.5)+(1,0)$) arc (0:120: 1 and 0.75) coordinate (A2);
\draw [purple!50!orange, densely dashed] (A1.south) -- (A2.north) node [midway, black] {$\varnothing$};

\node at (0,0) [left] {$u$};
\node at (1,0) [right] {$v$};
\node at (120:1) [left] {$u_1$};
\node at (-120:1) [left] {$u_2$};
\node at ($(60:1)+(1,0)$) [right] {$v_1$};
\node at ($(-60:1)+(1,0)$) [right] {$v_2$};
\node [color={rgb: red!50!purple,2; black,1}] at (120:1.5) {\footnotesize $T^{u_1} \mapsto A_1$};
\node [color={rgb: orange!70!brown,2; black,1.5}] at (-120:1.5) {\footnotesize $T^{u_2} \mapsto A_2$};
\node [color={rgb: blue!80!black,2; black,1}] at ($(1,0)+(60:1.5)$) {\footnotesize $T^{v_1} \mapsto B_1$};
\node [color={rgb: green!50!blue,2; black,1}] at ($(1,0)+(-60:1.5)$) {\footnotesize $T^{v_2} \mapsto B_2$};
\node at (0.5,-0.6) {\Large $T$};
\draw [-Stealth] (3,0) -- ++(1,0);

\begin{scope}[xshift = 6cm]
\coordinate (a1) at (120:1);
\coordinate (a2) at (-120:1);
\coordinate (b1) at ($(1,0)+(60:1)$);
\coordinate (b2) at ($(1,0)+(-60:1)$);
\coordinate (u) at (0.5,0.5);
\coordinate (v) at (0.5,-0.5);

\draw (a1) -- (u) -- (v) node [midway, right] {$s'$} -- (a2);
\draw (b1) -- (u);
\draw (b2) -- (v);
\fill (u) circle (\e cm);
\fill (v) circle (\e cm);
\fill (120:1) circle (\e cm);
\fill (-120:1) circle (\e cm);
\fill ($(1,0)+(60:1)$) circle (\e cm);
\fill ($(1,0)+(-60:1)$) circle (\e cm);

\filldraw [rotate around={30:(120:1.5)}, purple!50!red, fill opacity = 0.15] (120:1.5) circle (1 and 0.75);
\filldraw [rotate around={-30:(-120:1.5)}, orange!70!brown, fill opacity = 0.15] (-120:1.5) circle (1 and 0.75);
\filldraw [rotate around={-30:($(1,0)+(60:1.5)$)}, blue!80!black, fill opacity = 0.15] ($(1,0)+(60:1.5)$) circle (1 and 0.75);
\filldraw [rotate around={30:($(1,0)+(-60:1.5)$)}, green!50!blue, fill opacity = 0.15] ($(1,0)+(-60:1.5)$) circle (1 and 0.75);

\node at (u) [above] {$u'$};
\node at (v) [below] {$v'$};
\node at (120:1) [left] {$u_1$};
\node at (-120:1) [left] {$u_2$};
\node at ($(60:1)+(1,0)$) [right] {$v_1$};
\node at ($(-60:1)+(1,0)$) [right] {$v_2$};
\node [color={rgb: red!50!purple,2; black,1}] at (120:1.5) {\footnotesize $T^{u_1} \mapsto A_1$};
\node [color={rgb: orange!70!brown,2; black,1.5}] at (-120:1.5) {\footnotesize $T^{u_2} \mapsto A_2$};
\node [color={rgb: blue!80!black,2; black,1}] at ($(1,0)+(60:1.5)$) {\footnotesize $T^{v_1} \mapsto B_1$};
\node [color={rgb: green!50!blue,2; black,1}] at ($(1,0)+(-60:1.5)$) {\footnotesize $T^{v_2} \mapsto B_2$};
\node at (-0.6,0) {\Large $T'$};
\end{scope}
\end{tikzpicture}
\end{center}
\caption{Rebranching on an inner edge $s$ where $\cut(A_1, A_2) = \varnothing$.}
\label{pic_rebranching}
\end{figure}

Aside from $s$, the vertex $u$ is incident to two edges $(u,u_1),(u,u_2)$, and the vertex $v$ is incident to two edges $(v,v_1),(v,v_2)$. Let $A_1,A_2,B_1,B_2 \subseteq V(G)$ be the vertices associated to the subtrees $T^{u_1},T^{u_2},T^{v_1},T^{v_2}$ rooted at $u_1,u_2,v_1,v_2$ respectively, and not containing the edge $s$. A rebranching of $T$ at $s$, as pictured in Figure~\ref{pic_rebranching}, consists of connecting trees $T^{u_1}$ and $T^{v_1}$ with a new vertex $u'$ adjacent to $u_1$ and $v_1$, connecting trees $T^{u_2}$ and $T^{v_2}$ with a new vertex $v'$ adjacent to $u_2$ and $v_2$, and finally connecting $u'$ and $v'$ with an edge $s'$. If $(T,\phi)$ is a carving-decomposition of a graph $G$, and $T'$ is a tree obtained by rebranching an edge of $T$, then $(T',\phi)$ is a carving-decomposition of $G$ (using the same $\phi$ and the direct bijection between the leaves of $T$ and the leaves of $T'$). We now use rebranching to decrease the potential $\mu$ of a tree decomposition.

\begin{lemma}{\cite[(5.3)]{Seymour_Ratcatcher}}\label{lem_branching}
Let $G$ be a connected graph, $(T,\phi)$ be a carving-decomposition of $G$, and $s \in E(T)$ be an inner edge partitioning $V(G)$ into $A_1,A_2,B_1,B_2$, as pictured in Figure~\ref{pic_rebranching}, such that: $\cut(A_1,A_2) = \varnothing$, and both $\cut(A_1,B_1)$ and $\cut(A_2,B_2)$ are not empty. Then $\mu(T') < \mu (T)$ where $T'$ is a rebranching of $T$ at $s$.
\end{lemma}

This lemma is visually summarised by:

\begin{figure}[ht]
\begin{center}
\begin{tikzpicture}
\coordinate (a1) at (120:1);
\coordinate (a2) at (-120:1);
\coordinate (b1) at ($(1,0)+(60:1)$);
\coordinate (b2) at ($(1,0)+(-60:1)$);
\draw (120:1) -- (0,0) -- (1,0) node [midway, above] {$T$} -- ++(60:1);
\draw (-120:1) -- (0,0);
\draw (1,0) -- ++(-60:1);
\node at (0,0) [left] {$u$};
\fill (0,0) circle (2pt);
\node at (1,0) [right] {$v$};
\fill (1,0) circle (2pt);
\begin{scope}[purple!20!red]
\node (A1) at (120:1) [circle, left] {$A_1$};
\node (A2) at (-120:1) [circle, left] {$A_2$};
\node (B1) at ($(1,0)+(60:1)$) [circle, right] {$B_1$};
\node (B2) at ($(1,0)+(-60:1)$) [circle, right] {$B_2$};
\end{scope}
\begin{scope}[purple]
\draw ($(a1)+(-0.1,0.1)$) -- ($(b1)+(0.1,0.1)$);
\draw ($(a2)+(-0.1,-0.1)$) -- ($(b2)+(0.1,-0.1)$);
\draw [densely dashed] (A1.south) -- (A2.north) node [midway, black] {$\varnothing$};
\end{scope}
\draw [->] (2,0) -- (2.8,0);

\begin{scope}[xshift = 4cm]
\coordinate (a1) at (120:1);
\coordinate (a2) at (-120:1);
\coordinate (b1) at ($(1,0)+(60:1)$);
\coordinate (b2) at ($(1,0)+(-60:1)$);
\coordinate (u) at (0.5,0.5);
\coordinate (v) at (0.5,-0.5);
\draw (a1) -- (u) -- (v) node [midway, right] {$T'$} -- (a2);
\draw (b1) -- (u);
\draw (b2) -- (v);
\node at (u) [above] {$u'$};
\fill (u) circle (2pt);
\node at (v) [below] {$v'$};
\fill (v) circle (2pt);
\begin{scope}[purple!20!red]
\node (A1) at (a1) [circle, left] {$A_1$};
\node (A2) at (a2) [circle, left] {$A_2$};
\node at (b1) [right] {$B_1$};
\node at (b2) [right] {$B_2$};
\end{scope}
\node at (4,0) {$\Rightarrow \mu (T') < \mu (T)$};
\end{scope}
\end{tikzpicture}
\caption{A  {\color{purple} \textemdash} link between $X$ and $Y$ means $\cut(X,Y) \neq \varnothing$.}\label{pic_visual_lemma}
\end{center}
\end{figure}

\begin{proof}
We use the notations of the lemma, which are summarised in Figure~\ref{pic_rebranching}. Let $u,v$ be inner vertices of $T$ which are the endpoints of the edge $s$. Recall that, by hypothesis, all leaves of $(T,\phi)$ are labelled by vertices of $G$, and consequently none of the set $A_1,A_2,B_1,B_2$ is empty. Let us assume by contradiction that $\mu(T') \geq \mu (T)$ after rebranching at $s$. The potential $\mu$ being of integral value, this is equivalent to $\mu(T') + 1 - \mu(T) > 0$. Rebranching modifies trees locally, therefore there is a natural map between vertices of $T$ and $T'$ that are away from the rebranching, \ie, $V(T)\smallsetminus \{u,v\}  = V(T')\smallsetminus \{u',v'\}$. Additionally, note that for a vertex $x \in V(T)\smallsetminus \{u,v\} = V(T')\smallsetminus \{u',v'\}$, the partitions of vertices of $G$ induced by the subtrees $(T_i^x)_{i \in \{ 1,2,3 \}}$ in $T$, or induced by the subtrees $({T'}_i^x)_{i \in \{1,2,3\}}$ in $T'$, are the same. In consequence, by definition the potential value $\mu(x)$ in $T$ or $T'$ is unchanged for $x \not \in \{ u,v \}$. It follows that:

\begin{equation*}
\mu(T') + 1 - \mu(T) = \mu(u') + \mu(v') - \mu(u) - \mu(v) +1 > 0,
\end{equation*} 
where $\mu(u),\mu(v)$ are defined in $T$, and $\mu(u'),\mu(v')$ are defined in $T'$.

Now, with the hypothesis $\cut(A_1,A_2) = \varnothing$, we have that $\mu(u) = |B_1 \cup B_2| - 1 = |B_1| + |B_2| - 1$, and therefore, 

\begin{equation}\label{eq_proof_lem_rebranch_1}
(\mu(v') - |B_2| +1) + (\mu(u') - |B_1| +1) - \mu(v) > 0
\end{equation}

This inequality implies that at least one of $\mu(v') - |B_2| +1$ and $\mu(u') - |B_1| +1$ is positive. Without loss of generality, suppose that:

\begin{equation}\label{eq_proof_lem_rebranch_2}
\mu(v') - |B_2| +1 > 0, 
\end{equation} 
the other case being similar by symmetry. 

The inner vertex $v'\in V(T')$ induces the partition $\{A_2,B_2,A_1 \cup B_1\}$ of $V(G)$. From the hypothesis $\cut(A_2, B_2) \neq \varnothing$, we deduce that $\mu(v') \in \{|A_2|-1, |B_2|-1, 0\}$. With the inequality (\ref{eq_proof_lem_rebranch_2}) and $B_2 \neq \varnothing$, we conclude that $\mu(v') = |A_2| - 1$, with $|A_2| > |B_1| > 0$, which implies, by definition of $\mu$, that $\cut(A_1 \cup B_1, B_2) = \varnothing$. 

In consequence, we also have that $\cut(B_1, B_2) =  \varnothing$, because $\cut(B_1, B_2) \subset \cut(A_1 \cup B_1, B_2)$. This in turn involves that $\mu(v) = |A_1| + |A_2| - 1$. Transposing these values of $\mu(v)$ and $\mu(v')$ into Equation~(\ref{eq_proof_lem_rebranch_1}), we get: 

\begin{equation}
\label{eq_proof_lem_rebranch_3}
\mu(u') +1-|B_2|-|B_1| -|A_1| +1 > 0.
\end{equation}

The sets $A_1,B_1,B_2$ are non empty, consequently, Equation~(\ref{eq_proof_lem_rebranch_3}) implies that $\mu(u') > 0$. We deduce from the hypothesis $\cut(A_1, B_1) \neq \varnothing$, that $\mu(u') \in \{|B_1| -1, |A_1| -1\}$. However, we reach a contradiction since Inequality~(\ref{eq_proof_lem_rebranch_3}) implies that $\mu(u') -(|B_1| -1) > 0$ and $\mu(u') -(|A_1| -1) > 0$. Therefore, we prove $\mu (T) > \mu(T')$.
\end{proof}

We now prove that local rebranchings can preserve the linked property of a linked carving-decomposition. The following lemmas and their applications in the proof of Theorem~\ref{th_bond_link_decomp} are the technical crux of our approach and our main technical contribution.

\begin{lemma}\label{lem_bl}
Let $G$ be a connected graph, and $(T,\phi)$ be a linked carving-decomposition. Consider an inner edge $s=(u,v)$, partitioning $V(G)$ into $A_1,A_2,B_1,B_2$, as pictured in Figure~\ref{pic_notations}, such that $\cut(A_1,A_2) = \varnothing$. Assume that $s$ is transformed into an edge $s'=(u',v')$ by a rebranching of $T$ at $s$ yielding $T'$: if $T'$ is not linked, then the pair $a_i,s'$, where $\wid(a_i) > \wid(a_j)$ ($i \neq j, \{i,j\} = \{1,2\}$), is not linked and $\wid(s') < \wid(s)$.
\end{lemma}

\begin{proof}
First, we prove that if any pair of edges is not linked, $s'$ must be involved in the pair. Indeed, we prove by a case study that any pair of edges $f,h \in (E(T')\smallsetminus \{s'\})^2$ is linked in $(T',\phi)$. Note that each path $P'$ in $T'$ is obtained from a path $P$ in $T$ by either replacing $s$ by $s'$ (\eg, from $A_1$ to $B_2$), suppressing $s$ (\eg, from $A_1$ to $B_1$), adding $s'$ (\eg, from $A_1$ to $A_2$), or doing nothing at all (\eg, a path $P'$ using neither $s'$ nor edges incident to it). 

We name edges incident to $s$ in $T$ by the lower-case version of the set that they induce, as summarised in Figure~\ref{pic_notations}:

\begin{figure}[ht]
\begin{center}
\begin{tikzpicture}
\coordinate (a1) at (120:1);
\coordinate (a2) at (-120:1);
\coordinate (b1) at ($(1,0)+(60:1)$);
\coordinate (b2) at ($(1,0)+(-60:1)$);
\draw (0,0)  -- (1,0) node [midway, above] {$s$};
\draw (120:1)  -- (0,0) node [midway, right, gray] {$a_1$};
\draw (-120:1) -- (0,0)  node [midway, right, gray] {$a_2$};
\draw (1,0) -- ++(60:1) node [midway, left, gray] {$b_1$};
\draw (1,0) -- ++(-60:1) node [midway, left, gray] {$b_2$};
\node at (0,0) [left] {$u$};
\fill (0,0) circle (2pt);
\node at (1,0) [right] {$v$};
\fill (1,0) circle (2pt);
\begin{scope}[purple!50!red]
\node (A1) at (120:1) [circle, left] {$A_1$};
\node (A2) at (-120:1) [circle, left] {$A_2$};
\node (B1) at ($(1,0)+(60:1)$) [circle, right] {$B_1$};
\node (B2) at ($(1,0)+(-60:1)$) [circle, right] {$B_2$};
\end{scope}
\node at (-1.1,0) {$T$};

\draw [-Stealth] (2.25,0) -- +(1,0);

\begin{scope}[xshift = 4.5cm]
\coordinate (a1) at (120:1);
\coordinate (a2) at (-120:1);
\coordinate (b1) at ($(1,0)+(60:1)$);
\coordinate (b2) at ($(1,0)+(-60:1)$);
\draw (0.5,0.5)  -- (0.5,-0.5) node [midway, right] {$s'$};
\draw (a1)  -- (0.5,0.5) node [midway, below, gray] {$a_1$};
\draw (a2) -- (0.5,-0.5)  node [midway, above, gray] {$a_2$};
\draw (0.5,0.5) -- (b1) node [midway, below, gray] {$b_1$};
\draw (0.5,-0.5) -- (b2) node [midway, above, gray] {$b_2$};
\node at (0.5,0.5) [above] {$u'$};
\fill (0.5,0.5) circle (2pt);
\node at (0.5,-0.5) [below] {$v'$};
\fill (0.5,-0.5) circle (2pt);
\begin{scope}[purple!50!red]
\node (A1) at (120:1) [circle, left] {$A_1$};
\node (A2) at (-120:1) [circle, left] {$A_2$};
\node (B1) at ($(1,0)+(60:1)$) [circle, right] {$B_1$};
\node (B2) at ($(1,0)+(-60:1)$) [circle, right] {$B_2$};
\end{scope}
\node at (2.1,0) {$T'$};
\end{scope}
\end{tikzpicture}
\caption{Notations for a rebranching at $s$.}
\label{pic_notations}
\end{center}
\end{figure}

Let $f,g \in (E(T')\smallsetminus\{s'\})^2$ be two distinct edges of $T'$ and let $P'$ be the minimal path of $T'$ containing both edges. Cases 1 and 2 below handle $f,g \neq s'$.

\medskip
\noindent
{\bf Case 1:} $f$ and $g$ are also edges of $T$ and are connected by a path $P$ in $T$ using neither $s$ nor any of the edges incident to it.

Then, the path $P'$ connecting $f$ to $g$ in $T'$ is identical to $P$, it follows that $f$ and $g$ are linked.

\medskip
\noindent
{\bf Case 2.} $f$ and $g$ are also edges of $T$ and are connected by a path $P$ in $T$ using $s$ and at least one of the edges $a_1,a_2,b_1,b_2$, which are the edges incident to $s$.

Recall the notations $F = V(G^f_i)$ and $H = V(G^h_j)$ where $i$ is such that ${T}^f_i$ is the unique tree of $T \smallsetminus {f}$ not containing $h$, and ${T}^h_j$ is the unique tree of $T \smallsetminus {h}$ not containing $f$. We prove that there is always an edge $y$ of $P \setminus \{s\}$ whose weight is $\mcut (F,G)$, and that this edge $y$ also appears in $P'$.

\begin{enumerate}[label=(\roman*)]
\item If $P'$ was obtained from $P$ by suppressing $s$, then we can assume by symmetry that $F \subset A_i$ and $H \subset B_i$ for some $i \in \{ 1, 2\}$, and it follows that $a_i$ is in $P'$. Since $A_i \subset A_1 \cup A_2$ and $\cut(A_1, A_2) = \varnothing$, then $|\cut(A_i)| = \wid (a_i) \leq \wid (s) = |\cut(A_1 \cup A_2)|$ so that there exists $y \in P \smallsetminus \{s\}$ such that $\mcut(F,H) = y$. Thus, $\mcut(F,H) = \min_{p \in P'} \wid (p) = \wid (y)$.
\item If $s$ in $P$ was replaced by $s'$ in order to obtain $P'$, then we can assume by symmetry that $F \subset A_i$ and $H \subset B_j$ where $i \neq j$. Because $(T,\phi)$ is a linked carving-decomposition and $\wid (s) \geq \wid (a_i)$, there exists $y \in P \smallsetminus \{s\}$ such that $\mcut(F,H) = y$. It follows that $\wid (s') \geq \mcut(F,H) = \wid (y)$ where $y$ is also in $P'$.
\item If $P'$ was obtained from $P$ by adding $s'$, then we can assume by symmetry that $F \subset X_1$ and $H \subset X_2$, for some $X \in \{ A,B \}$. Since $(T,\phi)$ is a linked carving-decomposition, there exists $y \in P$ such that $\mcut(F,H) = y$ and $y$ also belongs to $P'$. Thus, $\mcut(F,H) = \min_{p \in P'} \wid (p) = \wid (y)$.
\end{enumerate}

Hence, if a pair of edge $(f,g)$ is not linked, one of them must be $s'$. Whence, there remains to deal with paths having $s'$ as an extremal edge, \ie, $f=s'$ or $g=s'$. Without loss of generality, we can assume that $g=s'$.

{\bf Case 3.}  $g = s'$.
Let us assume, by contradiction, that $s'$ is not linked with the edge $f$ of T', and let us call $P'$ the minimal path of $T'$ containing $s'$ and $f$. 

\begin{figure}[ht]
\begin{center}
\begin{tikzpicture}[scale = 0.85]
\def\ls{2} \def\lm{1} \def\li{1} \def\op{0.1}
\coordinate (a1t) at (0,1);
\coordinate (a1b) at ($(a1t)+(0,-2)$);
\coordinate (a2t) at ($(a1t)+(5,0)$);
\coordinate (a2b) at ($(a1b)+(5,0)$);
\coordinate (m1t) at ($(a1t)+(2.5,1)$);
\coordinate (m2b) at ($(a1b)+(2.5,-1)$);
\coordinate (fl) at ($(a1t)+(-60:1.35)$);
\coordinate (bl) at ($(a1t)+(-70:1)$);
\coordinate (fr) at ($(a2t)+(-120:1.35)$);
\coordinate (br) at ($(a2t)+(-130:1)$);

\filldraw [red!50!purple, fill opacity = \op, draw opacity = 0.5] (a1t) .. controls + (180:\ls) and \control{(a1b)}{(180:\ls)} -- (fl) -- (a1t);
\fill [red!50!purple, fill opacity = \op] (a1t) .. controls + (180:\ls) and \control{(a1b)}{(180:\ls)} -- (bl) -- (a1t);
\draw [red!50!purple, dotted, opacity =0.5] (a1t) -- (bl) -- (a1b);

\filldraw [orange!70!brown, fill opacity = \op, draw opacity = 0.5] (a2t) .. controls + (0:\ls) and \control{(a2b)}{(0:\ls)} -- (fr) -- (a2t);
\fill [orange!70!brown, fill opacity = \op] (a2t) .. controls + (0:\ls) and \control{(a2b)}{(0:\ls)} -- (br) -- (a2t);
\draw [orange!70!brown, dotted, opacity =0.5] (a2t) -- (br) -- (a2b);

\filldraw [blue!80!black, fill opacity = \op, draw opacity = 0.5] (a2t) .. controls + (180:\li) and \control{(m1t)}{(0:\lm)} .. controls + (180:\lm) and \control{(a1t)}{(0:\li)} -- (fl) .. controls +(-20:1.5) and \control{(fr)}{(-160:1.5)} -- (a2t);
\fill [blue!80!black, fill opacity = \op] (a2t) .. controls + (180:\li) and \control{(m1t)}{(0:\lm)} .. controls + (180:\lm) and \control{(a1t)}{(0:\li)} -- (bl) .. controls +(20:1.5) and \control{(br)}{(160:1.5)} -- (a2t);
\draw [blue!80!black, dotted, opacity =0.5] (a1t) -- (bl) .. controls +(20:1.5) and \control{(br)}{(160:1.5)} -- (a2t);

\filldraw [green!50!blue, fill opacity = \op, draw opacity = 0.5] (a2b) .. controls + (180:\li) and \control{(m2b)}{(0:\lm)} .. controls + (180:\lm) and \control{(a1b)}{(0:\li)} -- (fl) .. controls +(-20:1.5) and \control{(fr)}{(-160:1.5)} -- (a2b);
\fill [green!50!blue, fill opacity = \op] (a2b) .. controls + (180:\li) and \control{(m2b)}{(0:\lm)} .. controls + (180:\lm) and \control{(a1b)}{(0:\li)} -- (bl) .. controls +(20:1.5) and \control{(br)}{(160:1.5)} -- (a2b);
\draw [green!50!blue, dotted, opacity =0.5] (a1b) -- (bl) .. controls +(20:1.5) and \control{(br)}{(160:1.5)} -- (a2b);

\node [red!50!purple] at (-0.5,0) {$A_1$};
\node [orange!70!brown] at (5.5,0) {$A_2$};
\node [blue!80!black] at (2.5,1) {$B_1$};
\node [green!50!blue] at (2.5,-1) {$B_2$};

\begin{scope}[xshift = 9cm]
\coordinate (a1t) at (0,1);
\coordinate (a1b) at ($(a1t)+(0,-2)$);
\coordinate (a2t) at ($(a1t)+(5,0)$);
\coordinate (a2b) at ($(a1b)+(5,0)$);
\coordinate (m1t) at ($(a1t)+(2.5,1)$);
\coordinate (m2b) at ($(a1b)+(2.5,-1)$);
\coordinate (fl) at ($(a1t)+(-60:1.35)$);
\coordinate (bl) at ($(a1t)+(-70:1)$);
\coordinate (fr) at ($(a2t)+(-120:1.35)$);
\coordinate (br) at ($(a2t)+(-130:1)$);

\filldraw [red!50!purple, fill opacity = \op, draw opacity = 0.2] (a1t) .. controls + (180:\ls) and \control{(a1b)}{(180:\ls)} -- (fl) -- (a1t);
\fill [red!50!purple, fill opacity = \op] (a1t) .. controls + (180:\ls) and \control{(a1b)}{(180:\ls)} -- (bl) -- (a1t);
\draw [red!50!purple, dotted, opacity =0.5] (a1t) -- (bl) -- (a1b);

\filldraw [orange!70!brown, fill opacity = \op, draw opacity = 0.2] (a2t) .. controls + (0:\ls) and \control{(a2b)}{(0:\ls)} -- (fr) -- (a2t);
\fill [orange!70!brown, fill opacity = \op] (a2t) .. controls + (0:\ls) and \control{(a2b)}{(0:\ls)} -- (br) -- (a2t);
\draw [orange!70!brown, dotted, opacity =0.5] (a2t) -- (br) -- (a2b);

\filldraw [blue!80!black, fill opacity = \op, draw opacity = 0.2] (a2t) .. controls + (180:\li) and \control{(m1t)}{(0:\lm)} .. controls + (180:\lm) and \control{(a1t)}{(0:\li)} -- (fl) .. controls +(-20:1.5) and \control{(fr)}{(-160:1.5)} -- (a2t);
\fill [blue!80!black, fill opacity = \op] (a2t) .. controls + (180:\li) and \control{(m1t)}{(0:\lm)} .. controls + (180:\lm) and \control{(a1t)}{(0:\li)} -- (bl) .. controls +(20:1.5) and \control{(br)}{(160:1.5)} -- (a2t);
\draw [blue!80!black, dotted, opacity =0.5] (a1t) -- (bl) .. controls +(20:1.5) and \control{(br)}{(160:1.5)} -- (a2t);

\filldraw [green!50!blue, fill opacity = \op, draw opacity = 0.2] (a2b) .. controls + (180:\li) and \control{(m2b)}{(0:\lm)} .. controls + (180:\lm) and \control{(a1b)}{(0:\li)} -- (fl) .. controls +(-20:1.5) and \control{(fr)}{(-160:1.5)} -- (a2b);
\fill [green!50!blue, fill opacity = \op] (a2b) .. controls + (180:\li) and \control{(m2b)}{(0:\lm)} .. controls + (180:\lm) and \control{(a1b)}{(0:\li)} -- (bl) .. controls +(20:1.5) and \control{(br)}{(160:1.5)} -- (a2b);
\draw [green!50!blue, dotted, opacity =0.5] (a1b) -- (bl) .. controls +(20:1.5) and \control{(br)}{(160:1.5)} -- (a2b);

\path (a2b) .. controls + (180:\li) and \control{(m2b)}{(0:\lm)} coordinate [pos = 0.5] (cb);
\path (a2t) .. controls + (0:\ls) and \control{(a2b)}{(0:\ls)} coordinate [pos = 0.2] (ct);
\draw (cb) .. controls +(35:1) and \control{(ct)}{(-100:0.75)};
\draw [dotted] (cb) .. controls +(60:1) and \control{(ct)}{(-160:0.75)};

\begin{scope}[very thick]
\begin{scope}
\clip (cb) -- ++ (3,0) -- ++(0,3) -- (ct) .. controls +(-100:0.75) and \control{(cb)}{(35:1)};
\draw [blue!40!black] (a2b) -- (fr) -- (a2t);
\end{scope}
\begin{scope}
\clip (cb) -- ++ (3,0) -- ++(0,3) -- (ct) .. controls +(-160:0.75) and \control{(cb)}{(60:1)};
\draw [blue!40!black,dotted] (a2b) -- (br) -- (a2t);
\end{scope}
\node [blue!40!black, below] at (a2b) {$d$};

\begin{scope}
\clip (cb) -- ++ (0,3) -- ++(3,0) -- (ct) .. controls +(-100:0.75) and \control{(cb)}{(35:1)};
\draw [green!40!black] (a2b) -- (fr) -- (a2t);
\end{scope}
\begin{scope}
\clip (cb) -- ++ (0,3) -- ++(3,0) -- (ct) .. controls +(-160:0.75) and \control{(cb)}{(60:1)};
\draw [green!40!black,dotted] (a2b) -- (br) -- (a2t);
\end{scope}
\node [green!40!black, above] at (a2t) {$c$};

\begin{scope}
\clip (a2t) .. controls + (0:\ls) and \control{(a2b)}{(0:\ls)} -- (fr) -- (a2t);
\draw [purple] (cb) .. controls +(35:1) and \control{(ct)}{(-100:0.75)};
\end{scope}
\begin{scope}
\clip (a2t) .. controls + (0:\ls) and \control{(a2b)}{(0:\ls)} -- (br) -- (a2t);
\draw [dotted, purple] (cb) .. controls +(60:1) and \control{(ct)}{(-160:0.75)};
\end{scope}
\node [purple, above] at (ct) {$m_c$};

\begin{scope}
\clip (a2b) .. controls + (180:\li) and \control{(m2b)}{(0:\lm)} .. controls + (180:\lm) and \control{(a1b)}{(0:\li)} -- (fl) .. controls +(-20:1.5) and \control{(fr)}{(-160:1.5)} -- (a2b);
\draw [brown] (cb) .. controls +(35:1) and \control{(ct)}{(-100:0.75)};
\end{scope}
\begin{scope}
\clip (a2b) .. controls + (180:\li) and \control{(m2b)}{(0:\lm)} .. controls + (180:\lm) and \control{(a1b)}{(0:\li)} -- (bl) .. controls +(20:1.5) and \control{(br)}{(160:1.5)} -- (a2b);
\draw [dotted, brown] (cb) .. controls +(60:1) and \control{(ct)}{(-160:0.75)};
\end{scope}
\node [brown, below right] at (cb) {$m_d$};
\end{scope}

\begin{scope}
\clip (cb) .. controls +(60:1) and \control{(ct)}{(-160:0.75)} -- ++(1,0) -- ++(0,-2) --cycle;
\filldraw [green!40!orange, opacity=0.15] (a2b) .. controls + (180:\li) and \control{(m2b)}{(0:\lm)} .. controls + (180:\lm) and \control{(a1b)}{(0:\li)} -- (fl) .. controls +(-20:1.5) and \control{(fr)}{(-160:1.5)} -- (a2t) .. controls + (0:\ls) and \control{(a2b)}{(0:\ls)};
\end{scope}
\begin{scope}
\clip (cb) .. controls +(35:1) and \control{(ct)}{(-100:0.75)}  -- ++(1,0) -- ++(0,-2) --cycle;
\filldraw [green!40!orange, opacity=0.15] (a2b) .. controls + (180:\li) and \control{(m2b)}{(0:\lm)} .. controls + (180:\lm) and \control{(a1b)}{(0:\li)} -- (fl) .. controls +(-20:1.5) and \control{(fr)}{(-160:1.5)} -- (a2t) .. controls + (0:\ls) and \control{(a2b)}{(0:\ls)};
\end{scope}

\node at ($(fr)+(1.5,0)$) [green!40!orange!30!black] {$M$};
\end{scope}
\end{tikzpicture}
\caption{Visualization of $M$ and the associated cuts: $c = \cut(A_2 \smallsetminus M, B_1 \cup (B_2 \smallsetminus M))$, $d = \cut(A_2 \cap M, B_2)$, $m_c = \cut(M, A_2 \smallsetminus M)$, $m_d = \cut(M, A_1 \cup (B_2 \smallsetminus M))$. If $m_c < c$ or $m_d < d$, then one can find a cut separating $B_1$ and one of $A_2$ and $B_2$ shorter than the cuts induced by these sets.}
\label{pic_proof_lem_bl}
\end{center}
\end{figure}

\begin{enumerate}[label=(\roman*)]
\item If $f = a_i$, we make the additional hypothesis for this case that $\wid(s') \geq \wid(s)$. By symmetry of indices, we can assume that $i = 2$ so that $F = A_2$. Then, there exists a set $M \subset V(G)$ such that $A_2 \varsubsetneq M \varsubsetneq A_2 \cup B_2 = (A_1 \cup B_1)^c$, and $\mcut(A_1 \cup B_1,A_2) = |\cut(M)| < \min_{y \in P'} \wid (y) = \min \{\wid (s'), \wid (a_2)\}$. Note that the inclusions $A_2 \varsubsetneq M \varsubsetneq A_2 \cup B_2$ are strict because $\mcut(A_1 \cup B_1,A_2)$ is assumed to not be reached on $a_2$ and $s'$.

Intuitively, the cut induced by the set $M$ will allow us to exhibit a cut separating $A_2$ and $B_1$ which is smaller than the cut induced by $A_2$. This is a shortening argument which corresponds to $m_d < d$ on Figure~\ref{pic_proof_lem_bl}. This fact will lead to a contradiction on $a_2$ and $b_1$ not being linked in $T$.

Formally, we prove that the existence of such $M$ implies that the edges $a_2,b_1 \in E(T)^2$ are not linked in $T$. First, let us detail $\wid (s')$ and $\wid (s)$:

\begin{equation*}
\left\lbrace
\begin{array}{rl}
\wid (s') & = |\cut (A_1 \cup B_1)| = |\cut(B_1, B_2)| + |\cut (A_1, B_2)| +|\cut (B_1, A_2)|\\
\wid (s) & = |\cut (A_1)| + |\cut (A_2)| = |\cut (A_1, B_1)| + |\cut (A_1, B_2)| + |\cut (A_2, B_1)| +|\cut (A_2, B_2)|, 
\end{array}
\right.
\end{equation*}

By hypothesis, $\wid (s') \geq \wid (s)$, hence:

\begin{equation}\label{eq_lem_bl_1}
|\cut(B_1, B_2)|  \geq |\cut (A_1, B_1)| +|\cut (A_2, B_2)|
\end{equation}

Furthermore, we have that:

\begin{equation*}
\left\lbrace
\begin{array}{rl}
|\cut (B_1)| & = |\cut (A_2, B_1)| + |\cut(B_1, B_2)| + |\cut (B_1,A_1)| \\
|\cut(A_2)| & = |\cut(A_2,B_1)| + |\cut(A_2,B_2)|. 
\end{array}
\right.
\end{equation*}

With Inequality~(\ref{eq_lem_bl_1}), we deduce that:
$$|\cut (B_1)| - |\cut (A_2)| = |\cut (B_1, A_1)| + |\cut(B_1, B_2)| - |\cut(A_2, B_2)| \geq 2|\cut(B_1, A_1)| > 0.$$

Hence, $\wid ( a_2 ) < \wid (b_1)$. Furthermore, $\wid (s) = \wid(a_1) + \wid (a_2) > \wid(a_2)$. Thus, we have that $\mcut(B_1, A_2) \leq |\cut(M)| < \wid(a_2) = \min \{ \wid (a_2), \wid (b_1), \wid( s )\}$ which is absurd since $a_2$ and $b_1$ are linked in $T$.

\item If $f = b_i$, by symmetry of indices, we can assume that $i = 2$ so that $F = B_2$. Then, there exists a set $M \subset V(G)$ such that $B_2 \varsubsetneq M \varsubsetneq  A_2 \cup B_2 = (A_1 \cup B_1)^c$ and $\mcut(A_1 \cup B_1,B_2) = |\cut(M)| < \min_{y \in P'} \wid (y) = \min \{\wid (s'), \wid (b_2)\}$. Again, the previous inclusions are strict because $\mcut(A_1 \cup B_1, A_2)$ is assumed to not be reached on $b_2$ and $s'$.

Intuitively, the cut induced by $M$ separates $B_2$ and $B_1$ and is smaller than the cuts induced by $B_1$ and $B_2$: we can shorten them by cutting through $A_2$.  This corresponds to $m_c < c$ on Figure~\ref{pic_proof_lem_bl}. This fact will lead to a contradiction on $b_1$ and $b_2$ not being linked in $T$.

Formally, we prove that the existence of such $M$ implies that the edges $b_1,b_2 \in E(T)^2$ are not linked in $T$. We have that:

\begin{equation*}
\left\lbrace
\begin{array}{rl}
|\cut (A_1 \cup B_1)| &= |\cut(B_1, B_2)| + |\cut(A_1, B_2)| +|\cut(B_1, A_2)|\\
|\cut(M)| & = |\cut(B_1,B_2)| + |\cut(A_1, B_2)| + |\cut(M, A_2 \smallsetminus M)| + |\cut(M \cap A_2, B_1)|
\end{array}
\right.
\end{equation*}

Since $|\cut(M)| < |\cut(A_1 \cup B_1)|$, it follows that $|\cut(M, A_2 \smallsetminus M)| + |\cut(M \cap A_2, B_1)|< |\cut(B_1, A_2)|$. Then:

\begin{equation}\label{eq_lem_bl_3}
\begin{array}{rl}
|\cut(B_1 \cup A_2 \smallsetminus M)| & = |\cut(M, A_2 \smallsetminus M)| +|\cut(M \cap A_2, B_1)| + |\cut(B_1, A_1 \cup B_2)| \\ 
&< |\cut(B_1, A_2)| + |\cut(B_1, A_1 \cup B_2 )| = |\cut(B_1)| = \wid (b_1).
\end{array}
\end{equation}

Using Inequality~(\ref{eq_lem_bl_3}), and the fact that $\cut( B_1 \cup (A_2 \smallsetminus M) )$ separates $B_1$ and $B_2$, we have that $\mcut (B_1, B_2) \leq \cut(B_1 \cup (A_2 \smallsetminus M) < \wid(b_1)$. And by definition of $M$, $\mcut (B_1, B_2) < \wid(b_2)$. Thus, we get that $b_1$ and $b_2$ are not linked in $T$, a contradiction.

\item If $f$ is non-incident to $s'$, denote $X_i$ with $X \in \{A,B\}$ such that $F \subset X_i$ for some $i \in \{1,2\}$. Then, there exists a set $M$ such that $F \varsubsetneq M \varsubsetneq A_i \cup B_i$, and $|\cut(M)| = \mcut(F, (A_i \cup B_i)^c ) < \min_{y \in P'} \wid (y)$ (see $M$ on Figure~\ref{pic_proof_lem_bl}). Again, the previous inclusions are strict because $\mcut (F,(A_i \cup B_i)^c)$ is assumed to not be reached on $f$ and $s'$.

Note that since $f$ and $x_i$ are linked in $T$, $M \subset X_i$ is impossible. Indeed, in this case, $\cut(M)$ would separate $F$ and $X_i^c$ so that an edge $y$ of width $|\cut(M)|$ would exist in the path $P''$ of $T$ between $f$ and $x_i$. The path $P''$ is a subpath of both $P$ and $P'$. This is absurd: such an edge $y$ cannot be in $P'$. Thus, $M \not \subset X_i$.

Intuitively, in Figure~\ref{pic_proof_lem_bl}, we have that $m_c + m_d < c +d$ so that at least one of $m_c<c$ and $m_d < d$ is true. Then, $|\cut(M \cup X_i)| < |\cut(X_i)|$ and it boils down to the points above with the same shortening arguments.

Formally, we have that: 
\begin{equation*}
\left\lbrace
\begin{array}{l}
|\cut(X_i)| = |\cut(M \smallsetminus X_i, M \cap X_i)| + |\cut(X_i \smallsetminus M, X_i^c)| + |\cut(M \cap X_i , (M \cup X_i)^c)|\\
|\cut(M)| = |\cut(X_i \smallsetminus M, M \cap X_i)| + |\cut(M \smallsetminus X_i, M^c)| + |\cut(M \cap X_i , (M \cup X_i)^c)|\\
|\cut(X_i \cap M)| = |\cut(X_i \smallsetminus M, M \cap X_i)| + |\cut(M \smallsetminus X_i, M \cap X_i)| + |\cut(M \cap X_i , (M \cup X_i)^c)|\\
|\cut(X_i \cup M)| = |\cut(X_i \smallsetminus M, X_i^c)| + |\cut(M \smallsetminus X_i, M^c)| + |\cut(M \cap X_i , (M \cup X_i)^c)|
\end{array}
\right.
\end{equation*}

Following the equalities above, we always have:

\begin{equation}\label{eq_lem_bl_4}
|\cut(X_i \cap M)| > |\cut(M)| \Leftrightarrow |\cut(M \smallsetminus X_i, M \cap X_i)| > |\cut(M \smallsetminus X_i, M^c)|.
\end{equation} 

In addition, $f$ and $x_i$ are linked in $T$ but $f$ and $s'$ are not linked in $T'$, so that, necessarily $\mcut(F,(A_i \cup B_i)^c) < \mcut(F,(X_i)^c)$. Since $F \subset (X_i \cap M) \varsubsetneq X_i$, we get that:

\begin{equation*}
|\cut(M)| = \mcut(F,(A_i \cup B_i)^c) < \mcut(F,(X_i)^c) \leq |\cut(X_i \cap M)|
\end{equation*}

By adding $|\cut(X_i \smallsetminus M, X_i^c)| + |\cut(M \cap X_i , (M \cup X_i)^c)|$ to both sides of Inequality~(\ref{eq_lem_bl_4}), we thus obtain that $|\cut(X_i \cup M)| < |\cut(X_i)|$.  A contradiction is obtained from this last property in previous points i) and ii), hence, we have reduced this case to solved ones.
\end{enumerate}

Thus, if $T'$ is not linked, $\wid(s) > \wid(s')$. Furthermore, if a pair $s',f$ is not linked, we know that one of the sets $A_1,A_2$ can be shortened. Since $a_1$ and $a_2$ are linked in $T$, it cannot be the set $A_i$ with $\cut(A_i) \leq \cut(A_j)$ thus proving our lemma.
\end{proof}

\begin{lemma}\label{lem_1nl}
Let $G$ be a connected graph, and $(T,\phi)$ be a linked carving-decomposition. Consider an inner edge $s=(u,v)$, partitioning $V(G)$ into $A_1,A_2,B_1,B_2$, as pictured in Figure~\ref{pic_notations}, such that $\cut(A_1,A_2) = \varnothing$. Assume that $s$ is transformed into an edge $s'=(u',v')$ by a rebranching of $T$ at $s$ yielding $T'$ and into $s''$ by the other rebranching at $s$ (where $A_1$ is paired with $B_2$) yielding $T''$. Then, at most one of $T'$ and $T''$ is not linked. Furthermore, if one of them is not linked, then the other has a width of at most $\wid(T)$.
\end{lemma}

\begin{proof}
Without loss of generality, we assume that $\wid(s') \leq \wid(s'')$ and that $\wid(a_1) \leq \wid(a_2)$. By Lemma~\ref{lem_bl}, we know that if $T'$ is not linked, then the pair $a_2,s'$ is not linked.

Since $a_2,b_1$ is linked in $T$ but $a_2,s'$ is not in $T'$, we know that there exists $M$ such that $A_2 \varsubsetneq M \varsubsetneq A_2 \cup B_2 \varsubsetneq B_1^c$, $\cut(M) < \wid(a_2)$, and $\cut(M) < \wid(s')$. Hence, $\mcut(A_2,B_1)$ is necessarily reached on $b_1$ since $\wid(a_2) < \wid(s) = \wid(a_1) + \wid(a_2)$. It follows that $\wid(b_1) < \wid(a_2)$. Furthermore, since $M$ separates $B_1$ and $A_2$, we necessarily have that $\wid (b_1) < \wid (s')$.

By hypothesis:
\begin{equation*}
\left\lbrace
\begin{array}{l}
|\cut(B_1)| = |\cut(B_1, A_1)| + |\cut(B_1, B_2)| + |\cut(B_1, A_2)|\\
|\wid (s')| = |\cut(B_2, A_1)| + |\cut(B_1, B_2)| + |\cut(B_1, A_2)|
\end{array}
\right.
\end{equation*}

Hence, $\wid (b_1) < \wid (s')$ implies that $\cut (B_2, A_1) > \cut (B_1, A_1)$. This last inequality involves that $\cut (A_1 \cup B_2) = \cut (B_2) + \cut (B_1, A_1) - \cut (B_2, A_1) < \wid (b_2)$. However, $b_2$ and $a_2$ are linked in $T$ and $\wid(s) > \wid(a_2)$ so that $\mcut(A_2,B_2)$ is reached on $a_2$ and $\wid(a_2) < \wid(b_2)$.

By Lemma~\ref{lem_bl}, if $T''$ is not linked, then $s'',a_2$ is not linked and there exists $M$ such that $A_2 \varsubsetneq M \varsubsetneq A_2 \cup B_1 \varsubsetneq B_2^c$ and $\cut(M) < \wid(a_2) < \wid(b_2)$. This is a contradiction with $a_2,b_2$ being linked in $T$. Hence, if $T'$ is not linked, then $T''$ is linked.

In the case where $T'$ is not linked, we established above that $\wid (b_2) \geq \cut (A_1 \cup B_2) = \wid (s'')$. Hence, $\wid (T'') \leq \wid (T)$ since the width of edges between $T$ and $T''$ only differs on $s$ and $s''$.
\end{proof}

We now prove  Theorem~\ref{th_bond_link_decomp} which states that there exists a bond-linked carving-decomposition of any $2$-vertex-connected graph $G$ of width $\cwid(G)$. Note that we need to restrict ourselves to $2$-vertex-connected graphs since a vertex whose deletion would render the graph disconnected necessarily creates a non-bond cut on the edge incident to the leaf labelling it. We proceed by contradiction by assuming that a linked carving-decomposition of minimal carving-width of minimal $\mu$ is not bond and use Lemma~\ref{lem_branching}, Lemma~\ref{lem_bl}, and Lemma~\ref{lem_1nl} to show that we can decrease $\mu$ in this case, which is absurd.

\begin{proof}[Proof of Theorem~\ref{th_bond_link_decomp}]
Let $G$ be a $2$-vertex-connected graph. Among all carving-decompositions of $G$, we consider the ones of carving-width $\cwid(G) = k$ which are linked. This set is not empty by Theorem~\ref{th_link_decomp}. Among these, we choose a carving-decomposition $(T,\phi)$ such that $\mu (T)$ is minimal. If $\mu(T) = 0$, $T$ is bond thanks to Lemma~\ref{lem_potentialmuzero} and we are done. Otherwise, we obtain a contradiction by performing a rebranching at $s$ which decreases $\mu$ while preserving the linked property and the carving-width.

Hence, let us assume that $\mu(T) > 0$. By Lemma~\ref{lem_potentialmuzero}, there exists an inner edge $s = (u,v)$ such that $G_1^s$ is disconnected. Let $u$ be an inner vertex of $T$ incident to the inner edge $s = (u,v)$ such that $G_1^s = G^3_u$ is disconnected an such that $\mu(u)$ is minimum and non-zero. We locally are in the situation of the left picture of Figure~\ref{pic_rebranching} and will use the associated notations. As $u$ is an inner vertex of $G$, we set $G_1^u = A_1$, $G_2^u = A_2$, and $G_3^u = B$, where $\cut(A_1, A_2) = \varnothing$. Since $s$ is not incident to a leaf, we can further refine the splitting around $v$: we set $G_1^v = B_1$, $G_2^v = B_2$, and $G_3^v = A_1 \cup A_2$ (see Figure~\ref{pic_proof_th_bl_1}).

\begin{figure}[ht]
\begin{center}
\begin{tikzpicture}
\coordinate (a1) at (120:1);
\coordinate (a2) at (-120:1);
\coordinate (b1) at ($(1.5,0)+(60:1)$);
\coordinate (b2) at ($(1.5,0)+(-60:1)$);
\draw (120:1) -- (0,0) node [midway, above right, gray] {\large $a_1$} -- (1.5,0) node [midway, above={0.05cm}, gray] {\large $s$} -- ++(60:1) node [midway, above left, gray] {\large $b_1$};
\draw (-120:1) -- (0,0) node [midway, below right, gray] {\large $a_2$};
\draw (1.5,0) -- ++(-60:1) node [midway, below left, gray] {\large $b_2$};
\node at (0,0) [left] {$u$};
\fill (0,0) circle (2pt);
\node at (1.5,0) [right] {$v$};
\fill (1.5,0) circle (2pt);
\begin{scope}[purple!50!red]
\node (A1) at (120:1) [circle, left] {$A_1$};
\node (A2) at (-120:1) [circle, left] {$A_2$};
\node (B1) at ($(1.5,0)+(60:1)$) [circle, right] {$B_1$};
\node (B2) at ($(1.5,0)+(-60:1)$) [circle, right] {$B_2$};
\end{scope}
\draw [purple!50!red, densely dashed] (A1.south) -- (A2.north) node [midway, black] {$\varnothing$};
\end{tikzpicture}
\caption{Sets displayed by $T$ around $s$ (red), and edges name (grey).}
\label{pic_proof_th_bl_1}
\end{center}
\end{figure} 

Then, the assumption of minimality of  $\mu(u) = |B|-1$ implies that both $\cut(B_1, A_1 \cup A_2)$ and $\cut(B_2, A_1 \cup A_2)$ are non-empty. Indeed, if $\cut(B_1, A_1 \cup A_2) = \varnothing$, then $\mu(u) = |B|-1 = |B_1|+|B_2|-1 > |B_2|-1 = \mu(v) > 0$, which is a contradiction. Hence, up to switching $B_1$ and $B_2$, we can assume that $\cut(A_1,B_1) \neq \varnothing$ and $\cut(A_2,B_2) \neq \varnothing$. Additionally, recall that the case $\cut(A_i,B) = \varnothing$ is not possible because $G$ is connected. 

\smallskip

Let us first check that there exists a rebranching of $T$ of width less or equal to $k$. We only need to check the width of $s'$ in $T'$ or $s''$ in $T''$ since the widths of all other tree-edges remain the same. We consider two cases:

\begin{itemize}
\item If $\cut(A_1, B_2) = \varnothing$ or $\cut(A_2, B_1) = \varnothing$: let say, by symmetry, that $\cut(A_1, B_2) = \varnothing$; then $\cut(A_1 \cup B_1) \subset \cut(B_1)$ so that $\wid (s') = |\cut(A_1 \cup B_1)| < |\cut(B_1)| \leq k$.

\item If $\cut(A_1, B_2) \neq \varnothing$ and $\cut(A_2, B_1) \neq \varnothing$. We will show that up to switching $B_1$ and $B_2$, the modification of branching will yield $|\wid (s')| \leq k$ (notice that in that case, even with switching $B_1$ and $B_2$ we can still apply Lemmas~\ref{lem_branching} and~\ref{lem_bl} as explained above).

As $T$ has width at most $k$, each edge incident to $s$ has width at most $k$: 

\begin{equation*}
\left\lbrace 
\begin{array}{l}
|\cut(A_1,B_1)| + |\cut(A_2,B_1)| + |\cut(B_1,B_2)| \leq k\\
|\cut(A_1,B_2)| + |\cut(A_2,B_2)| + |\cut(B_1,B_2)| \leq k
\end{array}
\right.
\end{equation*}

Summing these inequalities yields:

\begin{equation*}
\underbrace{(|\cut(A_1, B_2)| + |\cut(B_1,B_2)| + |\cut(A_2, B_1)|)}_{|\cut(A_1 \cup B_1)|} + \underbrace{(|\cut(A_1, B_1)| + |\cut(B_1,B_2)|) + |\cut(A_2, B_2)| }_{|\cut(A_1 \cup B_2)|} \leq 2k.
\end{equation*}
Hence, at least one of $\cut(A_1 \cup B_1)$ and $\cut(A_1 \cup B_2)$ has size less or equal to $k$: we chose the minimal one and swap $B_1$ and $B_2$ if necessary.
\end{itemize}
Hence, there exist a rebranching non-increasing the width of $T$ at $s$.

\medskip

Hence we perform the rebranching of $T$ at $s$ to get $T'$. If $T'$ is linked we are done, since we are in the condition of Lemma~\ref{lem_branching}. Otherwise, $T'$ is not linked so that we know by Lemma~\ref{lem_bl} that the pair $s',a_i$ is not linked where $\wid (a_j) < \wid(a_i)$ and $i \neq j$. Without loss of generality, we can assume that $i=2$. Then, there exists $M$ such that $A_2 \varsubsetneq M \varsubsetneq B_2 \cup A_2$, $\cut(M) < \wid(a_2)$, and $\cut(M) < \wid(s')$.

If we assume that $\cut(A_1,B_2) = \varnothing$, then $\wid(s') = \wid (b_1)$, $M$ separates $A_2$ and $B_1$ and this contradicts the fact that $a_2$ and $B_1$ are linked in $T$. It follows that $\cut(A_1, B_2) \neq \varnothing$. Furthermore, by minimality of $\mu(T)$, $\cut(A_2, B_1) \neq \varnothing$. Hence, by Lemma~\ref{lem_branching} and Lemma~\ref{lem_1nl}, $T''$ is a linked carving-decomposition of $G$ of carving-width $\cwid(G)$ and such that $\mu (T'') < \mu (T)$. This contradicts the definition of $T$, so that we know that $\mu(T) = 0$ from the start and this concludes our proof. 
\end{proof}

\section{Disc-linked carving-decompositions.}
\label{sec_disc_linked_carving_dec}

As stated before, on the one hand, the linked property of carving-decompositions is essential to Nash-Williams' arguments that we exploit in Section~\ref{sec_bounded_cw}. On the other hand, to handle properly embedded graphs, we require an extra properties from their decompositions being \textit{almost bond}. Indeed removing any edge of a bond decomposition separates the graph into two subgraphs embedded on two disjoint discs by cut-cycle duality, which is a well-suited property for our purpose. However by definition of bond, only graphs that are $2$-vertex-connected can satisfy it. Hence, the aim of this section is to define and establish the existence of decompositions as thin as possible which also satisfy the property of being both linked and ``nicely embeddable on discs''.

All proofs in this section are tedious but not very hard; they rely on carefully merging pre-established disc-linked carving-decompositions of subgraphs satisfying some connectivity conditions in order to yield a disc-linked carving-decomposition for an entire graph which satisfies weaker connectivity conditions.

\subsection{Handling cut-vertices.}

Let us consider $G$ a plane graph and $T$ a carving-decomposition of $G$. As seen earlier, if $G$ has a cut-vertex, it cannot have a bond carving-decomposition. However, as sketched in the introduction, it turns out that we only need a weaker version of the bond property which can be enforced even on graphs admitting cut-vertices. A carving-decomposition $T$ has the \emphdef{disc property} if for all $e \in E(T)$, there exists a Jordan curve $\gamma_e$ such that $G^e_1$ and $G^e_2$ lie on different sides of $\gamma_e$. Equivalently, there exists a disc $D_e$ such that $D_e$ is disjoint from $G^e_2$ and contains $G^e_1$. We call a carving-decomposition which is both linked and has the disc property: a \emphdef{disc-linked carving-decomposition}. 

Note that contrary to the properties, for a carving-decomposition, to be bond or linked, the disc property depends on the embedding. Furthermore, by cut-cycle duality, every bond decomposition of a planar connected graph satisfies the disc property for any embedding of the graph. Hence, we already established with Theorem~\ref{th_bond_link_decomp} the existence of disc-linked carving-decomposition for $2$-vertex-connected graphs. To establish the existence of disc-linked carving-decompositions in general, we will first prove the following technical but simple lemmas, allowing us to merge disc-linked carving-decompositions and providing disc-linked carving-decompositions for components of $G$ that are not $2$-vertex-connected.

\begin{lemma}\label{lem_disc_decomp_merge}
Let $G$ be a connected plane graph and $v$ be a cut-vertex of $G$ such that $G = G_1 \cup G_2$ where $V(G_1) \cap V(G_2) = \{ v \}$. If there exist two disc-linked carving-decompositions $(T_1,\phi_1),(T_2,\phi_2)$ of width $k_1,k_2$ of $G_1,G_2$ respectively and a Jordan curve $\gamma$ intersecting only $v$ such that vertices of $G_1 \smallsetminus v$ and $G_2 \smallsetminus v$ lie on different sides of $\gamma$, then there exists a disc-linked carving-decomposition of $G$ of width $\max \{k_1, k_2, |\cut_G(\{ v \}|) \}$.
\end{lemma}

\begin{proof}
We define a carving-decomposition of $G$ by merging $(T_1,\phi_1)$ and $(T_2,\phi_2)$ at their leaves $\ell_1,\ell_2$ labelled by $v$ (see Figure~\ref{pic_disc_merge}). To do so, we identify these leaves to create a new tree $T$ in which these merged leaves become an inner vertex $\ell$. We then add a new edge incident to $\ell$ and connect it to a newly created leaf $\ell'$ of $T$. We define $\phi$ to be equal to $\phi_1$ on $L(T_1) \smallsetminus \{ \ell_1 \}$, equal to $\phi_2$ on $L(T_2) \smallsetminus \{ \ell_2 \}$, and equal to $v$ on $\ell$.

\begin{figure}[ht]
\begin{center}
\begin{tikzpicture}
\def\e{0.05} \def\se{0.5} \def\ses{0.33}
\coordinate (d0) at (0,0);
\coordinate (d1) at (30:\se);
\coordinate (d2) at (150:\se);
\coordinate (d3) at (-90:\se);
\coordinate (d4) at ($(d1)+(90:\se)$);
\coordinate (d5) at ($(d1)+(-30:\se)$);
\coordinate (d6) at ($(d2)+(90:\se)$);
\coordinate (d7) at ($(d2)+(-150:\se)$);
\coordinate (d8) at ($(d7)+(-90:\se)$);
\coordinate (d9) at ($(d7)+(150:\se)$);
\coordinate (d10) at ($(d4)+(30:\ses)$);
\coordinate (d11) at ($(d4)+(150:\ses)$);
\coordinate (d12) at ($(d8)+(-30:\ses)$);
\coordinate (d13) at ($(d8)+(-150:\ses)$);

\coordinate (a0) at (2,0);
\coordinate (a1) at ($(a0)+(00:\se)$);
\coordinate (a2) at ($(a1)+(-60:\se)$);
\coordinate (a3) at ($(a1)+(60:\se)$);
\coordinate (a4) at ($(a3)+(0:\se)$);
\coordinate (a5) at ($(a3)+(120:\se)$);
\coordinate (a6) at ($(a4)+(60:\se)$);
\coordinate (a7) at ($(a4)+(-60:\se)$);
\coordinate (a8) at ($(a7)+(0:\ses)$);
\coordinate (a9) at ($(a7)+(-120:\ses)$);

\draw (d0) -- (d1); 
\draw (d0) -- (d2); 
\draw (d0) -- (d3); 
\draw (d1) -- (d4); 
\draw (d1) -- (d5); 
\draw (d2) -- (d6); 
\draw (d2) -- (d7); 
\draw (d7) -- (d8); 
\draw (d7) -- (d9); 
\draw (d4) -- (d10); 
\draw (d4) -- (d11); 
\draw (d8) -- (d12); 
\draw (d8) -- (d13); 

\begin{scope}[red!70!black]
\draw (a0) -- (a1); 
\draw (a1) -- (a2); 
\draw (a1) -- (a3); 
\draw (a3) -- (a4); 
\draw (a3) -- (a5); 
\draw (a4) -- (a6); 
\draw (a4) -- (a7); 
\draw (a7) -- (a8); 
\draw (a7) -- (a9);
\end{scope}

\node at (d5) [right] {$v$};
\node at (a0) [left] {$v$};
\node at (0,1.45) {$T_1$};
\node at (3,1.45) [red!70!black] {$T_2$};

\foreach \i in {0,...,13}{\fill (d\i) circle (\e cm);};
\foreach \i in {0,...,9}{\fill [red!70!black] (a\i) circle (\e cm);};

\draw [-Stealth] (4.25,0.25) -- +(1,0);

\begin{scope}[xshift=7cm]
\coordinate (d0) at (0,0);
\coordinate (d1) at (30:\se);
\coordinate (d2) at (150:\se);
\coordinate (d3) at (-90:\se);
\coordinate (d4) at ($(d1)+(90:\se)$);
\coordinate (d5) at ($(d1)+(-30:\se)$);
\coordinate (d6) at ($(d2)+(90:\se)$);
\coordinate (d7) at ($(d2)+(-150:\se)$);
\coordinate (d8) at ($(d7)+(-90:\se)$);
\coordinate (d9) at ($(d7)+(150:\se)$);
\coordinate (d10) at ($(d4)+(30:\ses)$);
\coordinate (d11) at ($(d4)+(150:\ses)$);
\coordinate (d12) at ($(d8)+(-30:\ses)$);
\coordinate (d13) at ($(d8)+(-150:\ses)$);

\coordinate (a0) at (d5);
\coordinate (a1) at ($(a0)+(00:\se)$);
\coordinate (a2) at ($(a1)+(-60:\se)$);
\coordinate (a3) at ($(a1)+(60:\se)$);
\coordinate (a4) at ($(a3)+(0:\se)$);
\coordinate (a5) at ($(a3)+(120:\se)$);
\coordinate (a6) at ($(a4)+(60:\se)$);
\coordinate (a7) at ($(a4)+(-60:\se)$);
\coordinate (a8) at ($(a7)+(0:\ses)$);
\coordinate (a9) at ($(a7)+(-120:\ses)$);
\coordinate (a10) at ($(a0)+(-105:\se)$);

\draw (d0) -- (d1); 
\draw (d0) -- (d2); 
\draw (d0) -- (d3); 
\draw (d1) -- (d4); 
\draw (d1) -- (d5); 
\draw (d2) -- (d6); 
\draw (d2) -- (d7); 
\draw (d7) -- (d8); 
\draw (d7) -- (d9); 
\draw (d4) -- (d10); 
\draw (d4) -- (d11); 
\draw (d8) -- (d12); 
\draw (d8) -- (d13); 

\begin{scope}[red!70!black]
\draw (a0) -- (a1); 
\draw (a1) -- (a2); 
\draw (a1) -- (a3); 
\draw (a3) -- (a4); 
\draw (a3) -- (a5); 
\draw (a4) -- (a6); 
\draw (a4) -- (a7); 
\draw (a7) -- (a8); 
\draw (a7) -- (a9);
\end{scope}

\node [red!40!black] at (a10) [below] {$v$};
\node [red!40!black] at ($(a0)+(-145:\se*0.6)$) {$e'$};
\node at (1.25,1.45) {$T$};

\foreach \i in {0,...,4}{\fill (d\i) circle (\e cm);};
\foreach \i in {6,...,13}{\fill (d\i) circle (\e cm);};
\foreach \i in {1,...,9}{\fill [red!70!black] (a\i) circle (\e cm);};
\fill [red!40!black] (a0) circle (\e cm);
\fill [red!40!black] (a10) circle (\e cm);
\draw [red!40!black] (a0) -- (a10);
\end{scope}
\end{tikzpicture}
\caption{Merging two disc-linked carving-decompositions.}
\label{pic_disc_merge}
\end{center}
\end{figure}

Let $e'$ be the edge incident to $\ell'$, we associate to it a Jordan curve $\gamma_{e'}$, which is the boundary of a disc around $v$. We also associate to each edge incident to $e'$ a small perturbation of $\gamma$ so that $\gamma$ does not intersect $v$, but intersects $G$ exactly once on each edge of either $\cut_G(G_1)$ or $\cut_G(G_2)$. 

Let $e$ be any other edge of $T$ and $\gamma_e$ be the associated Jordan curve for the disc property in the stemming $T_i$. By symmetry, let us assume that $e \in E(T_1)$. The curve $\gamma_e$ separates the vertices of $G$ into two sets, $A$ and $B$, that can be different from the sets of vertices that $T(e)$ separates. Indeed, this curve was defined only with respect to one $G_i$; it can contain vertices of the other subgraph. In this case, we isotope $\gamma_e$ by pushing it on the side of $\gamma$ containing $G_1$ so that $\gamma_e$ and $\gamma$ are disjoint, without modifying its intersection with $G_1$ (see Figure~\ref{pic_curve_push}). This is possible since, by the disc property, no parts of $G_1$ one the side of $\gamma$ where we need to push. The result is a curve satisfying the disc property for $e$. By repeating these modifications, we prove that $(T,\phi)$ satisfies the disc property.

\begin{figure}[ht]
\begin{center}
\begin{tikzpicture}
\def\e{0.05}
\foreach \i in {0, ...,7}{\coordinate (a\i) at (45*\i:1.75);};
\coordinate (a5) at (-160:1.75);
\coordinate (a8) at (-40:1.1);
\coordinate (a9) at (140:1.3);

\coordinate (i0) at (0:0.6);
\coordinate (i1) at (90:0.6);
\coordinate (i2) at (180:0.6);
\coordinate (i3) at (270:0.6);
\coordinate (i4) at (1:0);

\begin{scope}[red!70!black]
\draw (i0) -- (i1) -- (i2) -- (i3) -- cycle;
\draw (i4) -- (i0);
\draw (i4) -- (i1);
\draw (i4) -- (i3);
\draw (i0) -- (a8);
\draw (i3) -- (a8);
\end{scope}
\draw (a7) -- (a8);
\draw (a6) -- (a8);
\draw (a0) -- (a1) -- (a2) -- (a3) -- (a5) -- (a6) -- (a7) -- cycle;

\draw [blue!70!black] (a8) .. controls +(45:1.5) and \control{(a9)}{(45:1.75)} .. controls +(-135:1) and \control{(a8)}{(-135:1.5)};

\foreach \i in {0,...,3}{\fill (a\i) circle (\e cm);};
\foreach \i in {5,...,7}{\fill (a\i) circle (\e cm);};
\fill [red!40!black] (a8) circle (\e cm);
\foreach \i in {0, ...,4}{\fill [red!70!black] (i\i) circle (\e cm);};

\node at (a8) [red!40!black, right] {$v$};
\node at (130:0.8) [red!70!black] {$G_2$};
\node at (-30:2) {$G_1$};
\node at (45:1) [blue!80!black] {$\gamma$};

\begin{scope}[xshift = 5cm]
\foreach \i in {0, ...,7}{\coordinate (a\i) at (45*\i:1.75);};
\coordinate (a5) at (-160:1.75);
\coordinate (a8) at (-40:1.1);
\coordinate (a9) at (140:1.3);

\coordinate (i0) at (0:0.6);
\coordinate (i1) at (90:0.6);
\coordinate (i2) at (180:0.6);
\coordinate (i3) at (270:0.6);
\coordinate (i4) at (1:0);

\begin{scope}[opacity = 0.3]
\begin{scope}[red!70!black]
\draw (i0) -- (i1) -- (i2) -- (i3) -- cycle;
\draw (i4) -- (i0);
\draw (i4) -- (i1);
\draw (i4) -- (i3);
\draw (i0) -- (a8);
\draw (i3) -- (a8);
\end{scope}
\draw (a7) -- (a8);
\draw (a6) -- (a8);
\draw (a0) -- (a1) -- (a2) -- (a3) -- (a5) -- (a6) -- (a7) -- cycle;
\end{scope}

\draw [blue!70!black] (a8) .. controls +(45:1.5) and \control{(a9)}{(45:1.75)} .. controls +(-135:1) and \control{(a8)}{(-135:1.5)}; 

\draw [smooth cycle, tension = 0.75, green!50!black] plot coordinates {($(i4)+(-135:0.2)$) ($(a0)+(0:0.2)$) ($(i1)+(110:0.2)$)};

\foreach \i in {0,...,3}{\fill (a\i) circle (\e cm);};
\foreach \i in {5,...,7}{\fill (a\i) circle (\e cm);};
\fill [red!40!black] (a8) circle (\e cm);
\foreach \i in {0, ...,4}{\fill [red!70!black] (i\i) circle (\e cm);};

\node at (35:0.8) [green!50!black] {$\gamma_e$};
\draw [-Stealth] (1.90,0.35) -- +(1,0);
\end{scope}

\begin{scope}[xshift = 9.5cm]
\foreach \i in {0, ...,7}{\coordinate (a\i) at (45*\i:1.75);};
\coordinate (a5) at (-160:1.75);
\coordinate (a8) at (-40:1.1);
\coordinate (a9) at (140:1.3);
\coordinate (i0) at (0:0.6);
\coordinate (i1) at (90:0.6);
\coordinate (i2) at (180:0.6);
\coordinate (i3) at (270:0.6);
\coordinate (i4) at (1:0);

\begin{scope}[opacity = 0.3]
\begin{scope}[red!70!black]
\draw (i0) -- (i1) -- (i2) -- (i3) -- cycle;
\draw (i4) -- (i0);
\draw (i4) -- (i1);
\draw (i4) -- (i3);
\draw (i0) -- (a8);
\draw (i3) -- (a8);
\end{scope}
\draw (a7) -- (a8);
\draw (a6) -- (a8);
\draw (a0) -- (a1) -- (a2) -- (a3) -- (a5) -- (a6) -- (a7) -- cycle;
\end{scope}

\draw [blue!70!black] (a8) .. controls +(45:1.5) and \control{(a9)}{(45:1.75)} .. controls +(-135:1) and \control{(a8)}{(-135:1.5)}; 

\draw [smooth cycle, tension = 0.75, green!50!black] plot coordinates {($(i4)+(-135:0.2)$) ($(a0)+(180:0.7)$) ($(i1)+(110:0.2)$)};
\node at (35:1) [green!50!black] {$\gamma_e$};

\foreach \i in {0,...,3}{\fill (a\i) circle (\e cm);};
\foreach \i in {5,...,7}{\fill (a\i) circle (\e cm);};
\fill [red!40!black] (a8) circle (\e cm);
\foreach \i in {0, ...,4}{\fill [red!70!black] (i\i) circle (\e cm);};
\end{scope}
\end{tikzpicture}
\caption{Modification of $\gamma_e$.}
\label{pic_curve_push}
\end{center}
\end{figure}

We now want to prove that the carving-decomposition is linked. In this proof, we will be considering cuts and their sizes in $G, G_1,$ and $G_2$. For clarity's sake, for $H \in \{G, G_1, G_2 \}$, we will denote $\mcut(A,B)$ with respect to the graph $H$ by $\mcut(A,B;H)$.

-Now, let $a,b$ be two edges of $T$ that were both in $T_1$ or $T_2$, say $T_1$ by symmetry, and $A,B$ their associated set of vertices displayed in $T_1$. First notice that, for any set of vertices $X$, $\cut_{G_1} (X \cap V(G_1)) \subset \cut_{G} (X)$. Indeed, $e \in \cut_{G_1} (X \cap V(G_1))$ means that $e$ has an endpoint in $X \cap V(G_1)$ and another in $V(G_1) \smallsetminus X$, thus $e \in \cut_{G} (X)$. Furthermore, each edge $x$ of $T$ distinct from $e'$ has an equivalent $x_1$ in either $T_1$ or $T_2$, and by construction of the curves, $\gamma_x \cap G = \gamma_{x_1} \cap G_1$ so that $|T_1 (x_1)| = |T(x)|$.

\begin{itemize}
\item If $v$ is in neither $A$ nor $B$ and $X$ is a set of vertices of $G_1$ such that $A \subset X \subset V(G_1) \smallsetminus B$, then either $|\cut_{G}(X)| = |\cut_{G_1}(X)|$ or $|\cut_{G}(X \cup V(G_2))| = |\cut_{G}(X)|$ depending on whether or not $v$ is in $X$. Hence, $\min_{A \subset X \subset V(G) \smallsetminus B} |\cut_{G} (X)| \leq \min_{A \subset X \subset V(G_1) \smallsetminus B} |\cut_{G_1} (X)|$. Now let $X \in \mathcal{P}(V(G))$ be a set of vertices of $G$ containing $A$ but disjoint from $B$: $A \subset X \subset V(G) \smallsetminus B$. Set $X' = X \cap V(G_1)$, then $A \subset X' \subset V(G_1) \smallsetminus B$, hence by the remark above, $|\cut_{G_1} (X')| \leq |\cut_G (X)|$ so that $\mcut (A,B;G_1) = \mcut(A,B;G)$.

\item Otherwise, let assume by symmetry that $v \in A$. Then, the associated set of vertices to $a$ in $T$ is $A' = A \cup V(G_2)$ by construction of $T$. It follows that sets of vertices $X$ of $G$ satisfying $A' \subset X \subset V(G) \smallsetminus B$ are in bijection by $X \mapsto X \cap V(G_1)$ with sets of vertices $X'$ of $G_1$ satisfying $A \subset X' \subset V(G_1) \smallsetminus B$. It follows that $\mcut (A,B;G_1) = \mcut(A,B;G)$.
\end{itemize}

Since $(T_1,\phi_1)$ is linked, there is one edge $p_1$ in $P_1$, the minimal path of $T_1$ containing both $a$ and $b$, such that $T_1 (p_1) = \mcut (A,B;G_1)$. The path $P_1$ has a natural equivalent $P$ in $T$, and $p_1$ has an equivalent $p$ in $P$. Hence, $\mcut(A,B;G) = |T(p)| = \min_{x \in P} |T(x)|$ (no edge has a smaller associated cut on the path $P$ according to the remark preceding this point).

-Let $a$ be an edge of either $T_1$ or $T_2$, say $T_1$ by symmetry, and set $b = e'$. Let $b'$ be the edge incident to $e'$ in the subtree of $T$ stemming from $T_1$. Notice that the edges $e'$ and $b'$ are linked since $\mcut(A,B;G) = |\cut_G (G_2)| = |T(b')| = \min_{x \in P=\{e', b'\}} |T(x)|$ as every cut separating $v$ and $V(G_1) \smallsetminus \{ v \}$ contains $T(b')$. If $a \neq b'$, then, since $(T_1,\phi_1)$ is linked, there is one edge $p_1$ in $P_1$, the minimal path of $T_1$ containing both $a$ and $b'$, such that $T_1 (p) = \mcut (A,V(G_2);G_1)$. The path $P_1$ has a natural equivalent $P$ in $T$, and $p_1$ has an equivalent $p$ in $P$. 

Let $X \in \mathcal{P} (V(G))$ be a set of vertices containing $A$ but not $v$. Let us assume that there exists a vertex $y \in X \cap V(G_2)$. Then, there is a path $Y$ from $y$ to $v$ in $G_2$ and $Y \cap \cut_G(X) \neq \varnothing$. Hence, $\cut_G(X) > \cut_{G_1} (X \cap V(G_1))$. Thus, $\mcut(A,V(G_2);G) = \mcut(A,\{ v \};G) = \min_{x \in (P \cup {e'})} |T(x)|$ and this is achieved on the edge $p$. 

-Let $a$ be an edge of $T_1$, $b$ be an edge of $T_2$, and $A,B$ be their associated vertex sets with respect to $T$. The minimal cut separating $A$ and $B$ contains only edges of either $G_1$ or $G_2$. Indeed, $v$ belongs to neither $A$ or $B$ since $e'$ is, by construction of $T$ on the tree of $T \smallsetminus a$ containing $b$, and on the tree of $T \smallsetminus b$ containing $a$. If $v$ is not separated by a cut $C$ from both $A$ and $B$, then $A$ and $B$ are not disconnected by $C$. Hence, $v$ needs to be separated from at least one of them, and separating $v$ from either $A$ or $B$ is enough to separate $A$ and $B$ since $v$ is a cut-vertex. Let $C$ be a cut separating $A$ and $B$; then one of $C \cap E(G_1)$ and $C \cap E(G_2)$ separates $A$ and $B$ and is at most as big as $C$. Hence $\mcut(A,B;G) = \min ( \mcut(A,\{ v \};G_1), \mcut(B, \{ v \}; G_2))$. For reasons mentioned above, one edge $p$ of either the path from $e'$ to $a$ or the path from $e'$ to $b$ will realise this minimum. And this edge is not $e'$ since the following edge in each path has a smaller width. Hence the edge $p$ belongs to the path from $a$ to $b$ in $T$, and $a$ and $b$ are linked, and it follows that $(T,\phi)$ is a disc-linked carving-decomposition of $G$.
 
Recall that the width of edges in $T$ is the same as their corresponding edge in $T_2$ or $T_1$. Since the width of $e'$ is $|\cut_G(\{v\})|$, it follows that the width of $(T,\phi)$ is $\max \{ \cwid (G_1), \cwid (G_2), \cut_G(\{ v \}) \}$.
\end{proof}

\begin{lemma}\label{lem_tree_linked}
Let $G$ be a $1$-vertex-connected plane graph without non-trivial $2$-vertex-connected component with at least one edge, then $G$ admits a disc-linked carving-decomposition of width $\cwid (G) = \max_{v \in V(G)} |\cut(\{ v \})|$.
\end{lemma}

\begin{proof}
If $G$ is a non-trivial $1$-vertex-connected component but its only $2$-vertex-connected components are singletons of $V(G)$, then $G$ is a tree with self-loops. Let us prove the result by induction on the size of the tree.

\textbullet If $G$ is reduced to a single edge, then $T$, also a tree with a single edge, is a disc-linked carving-decomposition of $G$ of width $\cwid (G) = 1$.

\textbullet Suppose by induction that our property is verified on any non-trivial embedded tree with less than $n > 2$ vertices. Let $G$ be an embedded tree with $n$ vertices and $\ell$ be a leaf of $G$ when self-loops are removed. $G' = G \smallsetminus \ell$ admits a disc-linked carving-decomposition $(T',\phi')$ of size $\cwid(G')$. Let $v$ be the leaf of $T$ labelling the neighbouring vertex $f$ of $\ell$ in $G$, and $e' = (u,v)$ the edge of $T$ incident to $v$. To obtain a new tree $T$, let us subdivide $e'$ into two edges $(u,u')$ and $(u',v)$ and add the new edge $(u',l)$ to $T'$. Let us set $\phi (x) = \phi' (x)$ if $x \in L(T')$ and $\phi(l) = \ell$ (see Figure~\ref{pic_proof_tree_link} for a visual summary). Then, $(T,\phi)$ is a carving-decomposition of $G$.

\begin{figure}[ht]
\begin{center}
\begin{tikzpicture}[scale = 1.25]
\def\e{0.05} 
\draw [dotted] (120:1)-- +(180:0.5);
\draw [dotted] (120:1)-- +(60:0.5);
\draw (0,0) -- (120:1);
\fill (0,0) circle (\e cm);
\fill (120:1) circle (\e cm);

\node at (0,0) [right] {$v \overset{\phi'}{\mapsto} f$};
\node at (120:1) [below left] {$u$};
\node at (0.25,1.25) {$T'$};

\begin{scope}[xshift = 3cm]
\draw [dotted] (120:1)-- +(180:0.5);
\draw [dotted] (120:1)-- +(60:0.5);
\draw (0,0) -- (120:1);
\draw (120:0.5) -- +(30:0.5) coordinate (l);
\fill (120:1) circle (\e cm);
\fill (0,0) circle (\e cm);
\fill (120:0.5) circle (\e cm);
\fill (l) circle (\e cm);

\node at (0,0) [right] {$v \overset{\phi}{\mapsto} f$};
\node at (l) [right] {$l \overset{\phi}{\mapsto} \ell$};
\node at (120:1) [below left] {$u$};
\node at (120:0.5) [left] {$u'$};
\node at (0.25,1.25) {$T$};
\end{scope}
\end{tikzpicture}
\caption{Subdivision of $(u,v)$.}
\label{pic_proof_tree_link}
\end{center}
\end{figure}

It is possible to isotope each curve $\gamma_e$ for $e \in E(T')$ so that $\ell$ lies on the proper side without modifying the intersections with $G'$ by pushing along the edge incident to it in $G$ if needed. In that case either $f$ and $v$ lie on different sides, or the edge is $(w,v)$ and we can shrink $\gamma_e$ around $f$. Then, we define $\gamma_\ell$ as a small circle around $\ell$ intersecting its incident edge once. Hence, $(T,\phi)$ satisfies the disc property. 

For $a,b$, two edges of $T$, if none of them is $(w,l)$, it is clear that they are linked since they are linked in $T'$. Otherwise, let us assume that $a = (w,l)$. Then the minimal cut separating $\ell$ from any subset of $V(G')$ is the edge incident to $\ell$, which is reached at $a$. Hence $(T, \phi)$ is linked.

The carving-width of $G$ is lower bounded by the maximal degree among its vertices. Let $e$ be an edge of $T' \smallsetminus (u,v)$, by construction of $T$, $T(e) = T'(e)$. In addition, $\wid ((u',l)) = 1 \leq \wid ((u,u')) = |\cut_G(\{ v \})| -1 < \wid ((u',v)) = |\cut_G(\{ v \})|$. If $v$ is a vertex of maximal degree of $G$, then $\cwid (G) = |\cut_G(\{ v \})|$. Otherwise, $\cwid (G) = \max_{x \in V(G')} |\cut_{G'}(\{ x \})| = \max_{x \in V(G)} |\cut_{G}(\{ x \})|$, which concludes this proof.
\end{proof}

We now have a disc-linked carving-decompositions of $2$-vertex-connected components, and the $1$-vertex-connected one in-between them that we can merge using Lemma~\ref{lem_disc_decomp_merge}.

\begin{proposition}\label{prop_disc_decomp_connected}
Let $G$ be a connected plane graph. Then, there exists a disc-linked carving-decomposition of $G$ of width $\cwid (G)$.
\end{proposition}

\begin{proof}[Proof of Proposition~\ref{prop_disc_decomp_connected}]
Let $D_1,\ldots, D_k$ be the $2$-vertex-connected components of $G$ (they may overlap on cut-vertices), and $D_{k+1}, \ldots, D_{n}$ be the remaining $1$-vertex-connected non-trivial components of $G$ not containing $2$-vertex-connected components of $G$. By definition, $D_1,\ldots,D_n$ is a family of connected components of $G$, covering $G$, where $D_i \cap D_j$ is either empty or a single cut-vertex if $i \neq j$. By Theorem~\ref{th_bond_link_decomp} and Lemma~\ref{lem_tree_linked}, there exists a family of disc-linked carving-decompositions $(T_1, \phi_1), \cdots, (T_n, \phi_n)$ where $(T_i, \phi_i)$ is a disc-linked carving-decomposition of $D_i$ of width $\cwid (D_i)$.

We now iteratively merge these decompositions using Lemma~\ref{lem_disc_decomp_merge}: as long as our family of decompositions has two or more elements, there exists a pair of components $D_i, D_j$ intersecting on a cut-vertex $v$, since $G$ is connected. As $G$ is embedded, $D_j \smallsetminus v$ lies entirely on a face of $D_i$. By pushing slightly the boundary of this face within itself, except at vertex $v$, we obtain a Jordan curve satisfying the condition of Lemma~\ref{lem_disc_decomp_merge}. We apply this lemma to obtain a new disc-linked carving-decomposition $(T',\phi')$ of the merged component $D = D_i \cup D_j$ by which we replace $D_i,D_j$ and $(T_i,\phi_i),(T_j,\phi_j)$ in their respective families. We iterate this process until the remaining component is $G$ and we have a disc-linked carving-decomposition $(T,\phi)$ of $G$.

A direct induction on this process yields that $\wid (T,\phi) = \max \{ \max_{v \in G} |\cut_G(\{v\})|, \max_{1 \leq i \leq n} \cwid (D_i) \}$. However, any carving-decomposition $(T', \phi')$ of $G$ can be turned into a carving-decomposition of any $D_i$ by unlabelling leaves of $T'$: we define $\psi$ such that $\psi (x) = \phi' (x)$ if $\phi' (x) \in V(D_i)$ and undefined otherwise. Every vertex removed from the labelling this way only decreases the width of the edges by definition. Hence, $\wid (T', \phi') \geq \wid (T',\psi) \geq \cwid (D_i) $. It follows that $\cwid (G) \geq \max_{1 \leq i \leq n} \cwid (D_i)$. Furthermore, $\cwid (G) \geq \max_{x \in V(G)} |\cut_{G}(\{ x \})|$. Thus $\wid (T,\phi) \leq \cwid (G)$, which implies by definition of carving-width that $\wid (T,\phi) = \cwid (G)$, which concludes the proof.
\end{proof}

\subsection{Handling disjoint connected components.}

Similarly to Proposition~\ref{prop_disc_decomp_connected}, we consider disc-linked carving-decompositions for each connected components and merge them, whence proving the existence of disc-linked carving-decomposition for general plane graphs (Proposition~\ref{prop_disc_decomp}).

\begin{proof}[Proof of Proposition~\ref{prop_disc_decomp}]
Let $G_1,\ldots,G_k$ be the connected components of $G$. By Proposition~\ref{prop_disc_decomp_connected}, there exists a family $\mathcal{T}$ of disc-linked carving-decomposition $\mathcal{T} = (T_i,\phi_i)_{1 \leq i \leq k}$ such that $(T_i, G_i)$ is a disc-linked carving-decomposition of $G_i$. 

We now merge the carving-decompositions of $\mathcal{T}$. Let us illustrate the merging of two carving-decompositions $(T_1,\phi_1)$ and $(T_2, \phi_2)$ where the graphs $G_1$ and $G_2$ share a face of $G$. First choose a vertex $v_i$ in each $V(G_i)$ such that the vertices $v_i$ are incident to the face shared by the two graphs $G_1$ and $G_2$, and let $e_i$ be the edge of $G_i$ incident to the leaf of $T_i$ whose label is $v_i$. We subdivide each $e_i$ and connect the vertices of degree $2$ in each $T_i$ by a new edge $e'$ (see Figure~\ref{pic_merge_disco}) and define $\phi$ on $T$ by $\phi(v)$ = $\phi_i(v)$ if $v$ is a leaf of $T_i$. This yields a carving-decomposition $(T,\phi)$ of $G_1 \cup G_2$ and replace $(T_1,\phi_1)$ and $(T_2, \phi_2)$ by $(T, \phi)$ in $\mathcal{T}$. Note that the width of $e'$ is $0$ and that the width of other edges is the same as the one they stem from in each $T_i$. They are separated by a Jordan curve $\gamma$ that we use to modify curves $\gamma_e$ associated to edges of $T$ as in Proposition~\ref{lem_disc_decomp_merge}.

\begin{figure}[ht]
\begin{center}
\begin{tikzpicture}
\def\e{0.05} \def\se{0.5} \def\ses{0.33}
\foreach \i in {0, ...,7}{\coordinate (a\i) at (45*\i:1.75);};
\coordinate (a5) at (-160:1.75);
\coordinate (a8) at (-40:0.8);
\coordinate (a9) at (140:1.3);

\coordinate (i0) at (0:0.7);
\coordinate (i1) at (100:0.6);
\coordinate (i2) at (190:0.3);
\coordinate (i3) at (270:0.6);
\coordinate (i4) at (0.2,0);

\begin{scope}[red!70!black]
\draw (i0) -- (i1) -- (i2) -- (i3) -- cycle;
\draw (i3) -- (i1);
\draw (i4) -- (i1);
\draw (i4) -- (i3);
\draw (i3) -- (a8);
\end{scope}
\draw (a0) -- (a1) -- (a2) -- (a3) -- (a4) -- (a5) -- (a6) -- (a7) -- cycle;

\draw [blue!70!black] ($(a8)+(-45:0.4)$) .. controls +(45:1.5) and \control{(a9)}{(45:1.75)} .. controls +(-135:1) and \control{($(a8)+(-45:0.4)$)}{(-135:1.5)}; 

\foreach \i in {0,...,7}{\fill (a\i) circle (\e cm);};
\foreach \i in {0, ...,4}{\fill [red!70!black] (i\i) circle (\e cm);};
\fill [red!70!black] (a8) circle (\e cm);

\node at ($(i0)+(80:0.2)$) [red!70!black] {\footnotesize $v_2$};
\node at (140:0.9) [red!70!black] {$G_2$};
\node at (-50:2.1) {$G_1$};
\node at ($(a1)+(80:0.2)$) {\footnotesize $v_1$};
\node at (45:1) [blue!80!black] {$\gamma$};

\begin{scope}[xshift = 3.5cm]
\coordinate (d0) at (0,0);
\coordinate (d1) at (30:\se);
\coordinate (d2) at (150:\se);
\coordinate (d3) at (-90:\se);
\coordinate (d4) at ($(d1)+(90:\se)$);
\coordinate (d5) at ($(d1)+(-30:\se)$);
\coordinate (d6) at ($(d3)+(-150:\se)$);
\coordinate (d7) at ($(d3)+(-30:\se)$);
\coordinate (d8) at ($(d7)+(-90:\se)$);
\coordinate (d9) at ($(d7)+(30:\ses)$);
\coordinate (d10) at ($(d4)+(30:\ses)$);
\coordinate (d11) at ($(d4)+(150:\ses)$);
\coordinate (d12) at ($(d8)+(-30:\ses)$);
\coordinate (d13) at ($(d8)+(-150:\ses)$);

\coordinate (a0) at (1.5,0);
\coordinate (a1) at ($(a0)+(00:\se)$);
\coordinate (a2) at ($(a1)+(-60:\se)$);
\coordinate (a3) at ($(a1)+(60:\se)$);
\coordinate (a4) at ($(a3)+(0:\se)$);
\coordinate (a5) at ($(a3)+(120:\se)$);
\coordinate (a6) at ($(a2)+(0:\se)$);
\coordinate (a7) at ($(a2)+(-120:\se)$);
\coordinate (a8) at ($(a7)+(-60:\ses)$);
\coordinate (a9) at ($(a7)+(-180:\ses)$);

\draw (d0) -- (d1); 
\draw (d0) -- (d2); 
\draw (d0) -- (d3); 
\draw (d1) -- (d4); 
\draw (d1) -- (d5); 
\draw (d3) -- (d7); 
\draw (d3) -- (d6); 
\draw (d7) -- (d8); 
\draw (d7) -- (d9); 
\draw (d4) -- (d10); 
\draw (d4) -- (d11); 
\draw (d8) -- (d12); 
\draw (d8) -- (d13); 

\begin{scope}[red!70!black]
\draw (a0) -- (a1); 
\draw (a1) -- (a2); 
\draw (a1) -- (a3); 
\draw (a3) -- (a4); 
\draw (a3) -- (a5); 
\draw (a2) -- (a6); 
\draw (a2) -- (a7); 
\draw (a7) -- (a8); 
\draw (a7) -- (a9);
\end{scope}

\node at ($(d5)+(45:0.2)$) {\footnotesize $v_1$};
\node at (a0) [below, red!70!black] {\footnotesize $v_2$};
\node at (0.2,1.45) {$T_1$};
\node at (2.2,1.45) [red!70!black] {$T_2$};

\foreach \i in {0,...,13}{\fill (d\i) circle (\e cm);};
\foreach \i in {0,...,9}{\fill [red!70!black] (a\i) circle (\e cm);};

\draw [-Stealth] (3,0.25) -- +(0.5,0);

\begin{scope}[xshift=4.25cm]
\coordinate (d0) at (0,0);
\coordinate (d1) at (30:\se);
\coordinate (d2) at (150:\se);
\coordinate (d3) at (-90:\se);
\coordinate (d4) at ($(d1)+(90:\se)$);
\coordinate (d5) at ($(d1)+(-30:\se)$);
\coordinate (d6) at ($(d3)+(-150:\se)$);
\coordinate (d7) at ($(d3)+(-30:\se)$);
\coordinate (d8) at ($(d7)+(-90:\se)$);
\coordinate (d9) at ($(d7)+(30:\ses)$);
\coordinate (d10) at ($(d4)+(30:\ses)$);
\coordinate (d11) at ($(d4)+(150:\ses)$);
\coordinate (d12) at ($(d8)+(-30:\ses)$);
\coordinate (d13) at ($(d8)+(-150:\ses)$);
\coordinate (d14) at ($(d5)+(-90:\ses)$);

\coordinate (a0) at (1.35,0.35);
\coordinate (a1) at ($(a0)+(00:\se)$);
\coordinate (a2) at ($(a1)+(-60:\se)$);
\coordinate (a3) at ($(a1)+(60:\se)$);
\coordinate (a4) at ($(a3)+(0:\se)$);
\coordinate (a5) at ($(a3)+(120:\se)$);
\coordinate (a6) at ($(a2)+(0:\se)$);
\coordinate (a7) at ($(a2)+(-120:\se)$);
\coordinate (a8) at ($(a7)+(-60:\ses)$);
\coordinate (a9) at ($(a7)+(-180:\ses)$);
\coordinate (a10) at ($(a0)+(90:\ses)$);

\draw (d0) -- (d1); 
\draw (d0) -- (d2); 
\draw (d0) -- (d3); 
\draw (d1) -- (d4); 
\draw (d1) -- (d5); 
\draw (d3) -- (d7); 
\draw (d3) -- (d6); 
\draw (d7) -- (d8); 
\draw (d7) -- (d9); 
\draw (d4) -- (d10); 
\draw (d4) -- (d11); 
\draw (d8) -- (d12); 
\draw (d8) -- (d13); 
\draw (d5) -- (d14); 

\begin{scope}[red!70!black]
\draw (a0) -- (a1); 
\draw (a1) -- (a2); 
\draw (a1) -- (a3); 
\draw (a3) -- (a4); 
\draw (a3) -- (a5); 
\draw (a2) -- (a6); 
\draw (a2) -- (a7); 
\draw (a7) -- (a8); 
\draw (a7) -- (a9);
\draw (a0) -- (a10);
\end{scope}
\draw [red!40!black] (a0) -- (d5);

\node [red!70!black] at ($(a10)+(90:0.2)$) {\footnotesize $v_2$};
\node at ($(d14)+(0:0.25)$) {\footnotesize $v_1$};
\node [red!40!black] at ($(a0)!0.5!(d5)+(90:0.2)$) {\footnotesize $e'$};
\node [red!40!black] at (1.25,1.45) {$T$};

\foreach \i in {0,...,14}{\fill (d\i) circle (\e cm);};
\foreach \i in {0,...,10}{\fill [red!70!black] (a\i) circle (\e cm);};
\fill [red!70!black] (a0) circle (\e cm);
\fill [red!70!black] (a10) circle (\e cm);
\end{scope}
\end{scope}
\end{tikzpicture}
\caption{Merging of two disc-linked carving-decompositions of a non-connected graph.}
\label{pic_merge_disco}
\end{center}
\end{figure}

We iteratively merge decompositions $(T_i, \phi_i)$ and $(T_j, \phi_j)$ whenever it is possible, \ie, when the graphs that label share a face. Note that there always exist a pair satisfying this property as long as there are two decompositions in $\mathcal{T}$.

Let $a,b$ be a pair of edges of $T$. If $a,b$ comes from the same $T_i$ then there is a natural bijection between cuts separating $A$ and $B$ in $G$ and those in $G_i$. It follows that they are linked in $T$ since they are linked in the original $T_i$. If they are not from the same $T_i$, their corresponding $A$ and $B$ are in two distinct connected components and the minimal cut separating them is empty. Furthermore, by construction, an edge merging two trees $T_i$ and $T_j$ is on the minimal path connecting them, its width is $0$. Hence, $T$ is a disc-linked carving-decomposition.
\end{proof}

\section{Planar graphs with bounded carving-width.}
\label{sec_bounded_cw}

The aim of this section is to prove that connected plane graphs of bounded degree are well-quasi-ordered by embedded immersion, \ie, Theorem~\ref{th_emb_immersion_bounded}. Since carving-width is bounded, we follow the classical approach of exploiting tree-like decomposition and leveraging an argument of Nash-Williams to exhibit a well-quasi-order. 

On a high level, the proof of Theorem~\ref{th_emb_immersion_bounded} is a proof by contradiction where we assume that a set of edges is not well-quasi-ordered, where the edges are the root of subtrees of carving-decompositions. Nash-Williams' argument is fundamentally the creation of a ``minimal'' antichain from a set which is not well-quasi-ordered. Then, by exploiting the minimality assumption, it is then possible to show that their proper subtrees, called children, are necessarily well-quasi-ordered; this argument is expressed formally by Lemma~\ref{lem_lem_trees}. To conclude the proof, there remains to merge two families of children which are ordered in order to exhibit an increasing pair in the antichain built previously.

As mentioned before, the approach of this section is classical; our contribution is to show that the approach can indeed be performed in the framework of embedded immersions. In particular, the Lemma~\ref{lem_lem_trees} is a classical lemma valid for general orders put on edges of forests. However, in order to apply it, strong assumptions need to be made on the order put on the edges considered. We show that these assumptions will be satisfied by proving Lemma~\ref{lem_ancestry_cut}, which itself exploits Lemma~\ref{lem_path_swaps} and the strong properties of disc-linked carving-decompositions established in Section~\ref{sec_disc_linked_carving_dec}. The framework that we developed for this proof and Lemmas~\ref{lem_ancestry_cut} and~\ref{lem_path_swaps} are exposed in Section~\ref{subsec_order_edge}. The final proof, as well as Lemma~\ref{lem_lem_trees}, are in Section~\ref{sec_nash_williams_argument}; a proof of Lemma~\ref{lem_lem_trees} is also included for completeness.

Finally note that disc-linked carving-decompositions are decompositions of graphs embedded in $\Sp^2$, while Theorem~\ref{th_emb_immersion_bounded} deals with graphs embedded in $\R^2$. Lemma~\ref{lem_plane_equiv} will show why this is enough to consider only on graphs embedded in $\Sp^2$ when the carving-width of graphs is bounded.

\begin{lemma}\label{lem_plane_equiv}
If embedded graphs in $\Sp^2$ with bounded carving-width are well-quasi-ordered by embedded immersion, then embedded graphs in $\R^2$ with bounded carving-width are well-quasi-ordered by embedded immersion.
\end{lemma}

\begin{proof}
An embedding in $\Sp^2$ is an embedding in $\R^2$ where a face is chosen to be the unbounded one. We show how to mark such a face in this context using a graph with high carving-width.

Let $G_n$ be a family of embedded graphs in $\R^2$ with carving-width bounded by $k \in \N$. Choose $U$, a connected embedded graph in $\R^2$ such that $\cwid (U) > k$ (it can be the tree with $k+2$ vertices made of one vertex of degree $k+1$). Set $G'_n = G_n \cup U$ where $U$ is embedded on the unbounded face of $G'_n$. Then, if $G'_i$ is an embedded immersion of $G'_j$, it follows that $G_i$ is an embedded immersion of $G_j$. Indeed the component $U$ of $G'_i$ is necessarily mapped to a component of $G'_j$ with carving-width at least $\cwid(U)$ (recall that carving-width is non-increasing with respect to immersion), by construction there is only one such component, which is the copy of $U$ in $G'_j$. Furthermore, no vertices of $G_i$ are mapped to vertices of $U$ in $G'_j$, because embedded immersion map vertices injectively. Then $G_i$ is an embedded immersion of $G_j$.

Every embedding in $\R^2$ induces an embedding in $\Sp^2$ by stereographic projection. Let us denote ${G'}_n^S$ the corresponding embedding of $G'_n$ in $\Sp^2$. By hypothesis, we assume that there exists ${G'}^S_i$ which is an embedded immersion of ${G'}^S_j$ with $i < j$. An edge $e$ of $G_i$ is incident to the unbounded face of $G_i$ if and only if the associated edge in ${G'}^S_i$ is incident to the face containing $U$ by construction. Hence, the embedded immersion of ${G'}^S_i$ to ${G'}^S_j$ directly implies that $G_i$ is an embedded immersion of $G_j$.
\end{proof}

\subsection{Order on edges.}\label{subsec_order_edge}
Let $(G_n)_{n \in \N}$ be a family of plane graphs. Thanks to Proposition~\ref{prop_disc_decomp} we can define, for all $n \in \N$, a disc-linked carving-decomposition $T_n$ of $G_n$. Furthermore, we assume that each tree $T_n$ is rooted at some unlabelled leaf called the root $r(T_n)$. If $T$ has no unlabelled leaf in the decomposition, we can add an unlabelled leaf by adding a vertex in the middle of an edge and connecting this vertex by an edge to a new leaf which will be the root ; note that adding such a root yields a disc-linked carving-decomposition. Let $a,b$ be two edges of $T \in \T$. The edge $a$ is said to be an \emphdef{ancestor} of $b$ if $a$ lies on the unique path from $r(T)$ to $b$ and if for all edges $e$ between $a$ and $b$ on this path $\wid (e) \geq \wid (a) = \wid (b)$; in such a case, $b$ is called a \emphdef{successor} of $a$. We define $\T = \bigcup_{n \in \N} T_n$ as the forest made of all the trees $T_n$.

Thanks to our assumption, for each edge $e$ of a carving-decomposition $(T, \phi)$ of a graph $G$, $T \smallsetminus e$ consists of two subtrees of $T$, only one of which, written $T^e$, does not contain the root $r(T)$. We then define $G^e$ to be the embedded subgraph of $G$ induced by the vertices labelled by the leaves of $T^e$, \ie, $G^e \colon= G[\phi^{-1}( L(T^e))]$. 

We want to keep track of how edges of $T(e)$ leave $G^e$ in the embedding. Hence, for each $e \in E(T)$, we define the \emphdef{leaving graph} $\tilde{G^e}$, which is an embedded graph summarising this information. Since $T$ satisfies the disc property, there exists a Jordan curve $\gamma_e$ such that $G^e$ lies on one side of $\gamma_e$ and intersect $G$ exactly on $T(e)$. Up to isotopy, we can assume that each edge $x$ of $T(e)$ intersects $\gamma_e$ exactly once. The side containing $G^e$ is a disc that we call $D_e$. We subdivide each edge $f$ of $T(e)$ by introducing a new vertex $s_x$ at the point of intersection with $\gamma_e$. The embedded graph induced by $V(G^e) \cup_{x \in T(e)} \{ s_x \}$ on $D_e$ is $\tilde{G}^e$ (see Figure~\ref{pic_ex_leaving_graph}). This way, $\tilde{G}^e$ consists of $G^e \cap D_e$ and $(s_x)_{x \in T(e)}$, embedded on $D_e$ and $\partial D_e$ respectively.

\begin{figure}[ht]
\begin{center}
\begin{tikzpicture}
\clip (-0.55,-1.2) rectangle (13.35,3);
\begin{scope}
\def\e{0.075}
\coordinate (l1) at (0,0);
\coordinate (l2) at ($(l1)+(-30:1)$);
\coordinate (l3) at ($(l1)+(30:1)$);

\coordinate (r1) at (2,-0.5);
\coordinate (r2) at ($(r1)+(0.75,0)$);
\coordinate (r3) at ($(r1)+(0.75,0.75)$);
\coordinate (r4) at ($(r1)+(0,0.75)$);

\coordinate (t1) at (1,1.5);

\coordinate (t2) at ($(t1)+(0.9,0)$);
\coordinate (t3) at ($(t2)+(0.75,0)$);
\coordinate (t4) at ($(t3)+(0.75,0.75)$);
\coordinate (t5) at ($(t1)+(60:1.5)$);
\coordinate (t6) at ($(t5)+(0.75,0)$);
\coordinate (t7) at ($(t1)+(0.45,0)$);

\coordinate (c) at (1.5,2);

\begin{scope}[red!50!purple, opacity = 0.7]
\draw  (l1) --(l2) -- (l3) --cycle;
\draw (l2) -- (r1);
\draw (l3) -- (r4);
\draw (l2) -- (r4);
\draw (r1) --(r2) -- (r3) -- (r4) --cycle; 
\end{scope}
\begin{scope}[opacity = 0.5]
\draw (l1) -- (t1);
\draw (l3) -- (t1);
\draw (l3) -- (t7);
\draw (r4) -- (t2);
\draw (r3) -- (t3);
\draw (r3) -- (t4);

\end{scope}
\begin{scope}[thick]
\draw (t5) -- (t6);
\draw (t2) -- (t3);
\draw (t5) -- (t1);
\draw (t5) -- (t2);
\draw (t5) -- (t3);
\draw (t6) -- (t4);
\draw (t3) -- (t4);
\end{scope}

\draw [green!40!black] (-0.25,0) .. controls +(90:1.5) and \control{(3.25,-0.2)}{(90:1.5)} .. controls + (-90:1.15) and \control{(-0.25,0)}{(-90:1.75)};

\node at (3.1,1.15) [above right, opacity = 0.5] {$T(e)$};
\foreach \i in {1,2,3}{\fill [red!50!purple] (l\i) circle (\e cm);};
\foreach \i in {1,2,3,4}{\fill [red!50!purple] (r\i) circle (\e cm);};
\foreach \i in {1,2,3,4,5,6,7}{\fill (t\i) circle (\e cm);};

\node [green!40!black, right] at (3.25,-0.2) {$\gamma_e$};
\node [red!50!purple] at (0.35,-0.5) {$G^e$};
\node at (0.25,1.75) {\Large $G$};

\draw [-Stealth] (3.65,0.5) -- +(1,0);
\end{scope}

\begin{scope}[xshift = 5.25cm]
\def\e{0.075}
\coordinate (l1) at (0,0);
\coordinate (l2) at ($(l1)+(-30:1)$);
\coordinate (l3) at ($(l1)+(30:1)$);

\coordinate (r1) at (2,-0.5);
\coordinate (r2) at ($(r1)+(0.75,0)$);
\coordinate (r3) at ($(r1)+(0.75,0.75)$);
\coordinate (r4) at ($(r1)+(0,0.75)$);

\coordinate (t1) at (1,1.5);

\coordinate (t2) at ($(t1)+(0.9,0)$);
\coordinate (t3) at ($(t2)+(0.75,0)$);
\coordinate (t4) at ($(t3)+(0.75,0.75)$);
\coordinate (t5) at ($(t1)+(60:1.5)$);
\coordinate (t6) at ($(t5)+(0.75,0)$);
\coordinate (t7) at ($(t1)+(0.45,0)$);

\coordinate (c) at (1.5,2);

\begin{scope}[red!50!purple, opacity = 0.7]
\draw  (l1) --(l2) -- (l3) --cycle;
\draw (l2) -- (r1);
\draw (l3) -- (r4);
\draw (l2) -- (r4);
\draw (r1) --(r2) -- (r3) -- (r4) --cycle; 
\end{scope}
\draw [name path =e1, opacity = 0.5] (l1) -- (t1);
\draw [name path =e2, opacity = 0.5] (l3) -- (t1);
\draw [name path =e3, opacity = 0.5] (l3) -- (t7);
\draw [name path =e4, opacity = 0.5] (r4) -- (t2);
\draw [name path =e5, opacity = 0.5] (r3) -- (t3);
\draw [name path =e6, opacity = 0.5] (r3) -- (t4);
\begin{scope}[thick]
\draw (t5) -- (t6);
\draw (t2) -- (t3);
\draw (t5) -- (t1);
\draw (t5) -- (t2);
\draw (t5) -- (t3);
\draw (t6) -- (t4);
\draw (t3) -- (t4);
\end{scope}

\draw [green!40!black, name path =blob] (-0.25,0) .. controls +(90:1.5) and \control{(3.25,-0.2)}{(90:1.5)} .. controls + (-90:1.15) and \control{(-0.25,0)}{(-90:1.75)};

\foreach \i in {1,...,6}{
\path [name intersections={of=blob and e\i,total=\tot}]
\foreach \s in {1,...,\tot}{coordinate (n\i) at (intersection-\s)}; 
\fill [green!40!black] (n\i) circle (\e cm);
};

\foreach \i in {1,2,3}{\fill [red!50!purple] (l\i) circle (\e cm);};
\foreach \i in {1,2,3,4}{\fill [red!50!purple] (r\i) circle (\e cm);};
\foreach \i in {1,2,3,4,5,6,7}{\fill (t\i) circle (\e cm);};

\draw [-Stealth] (3.45,0.5) -- +(1,0);
\end{scope}

\begin{scope}[xshift = 10cm]
\def\e{0.075}
\coordinate (l1) at (0,0);
\coordinate (l2) at ($(l1)+(-30:1)$);
\coordinate (l3) at ($(l1)+(30:1)$);

\coordinate (r1) at (2,-0.5);
\coordinate (r2) at ($(r1)+(0.75,0)$);
\coordinate (r3) at ($(r1)+(0.75,0.75)$);
\coordinate (r4) at ($(r1)+(0,0.75)$);

\coordinate (t1) at (1,1.5);

\coordinate (t2) at ($(t1)+(0.9,0)$);
\coordinate (t3) at ($(t2)+(0.75,0)$);
\coordinate (t4) at ($(t3)+(0.75,0.75)$);
\coordinate (t5) at ($(t1)+(60:1.5)$);
\coordinate (t6) at ($(t5)+(0.75,0)$);
\coordinate (t7) at ($(t1)+(0.45,0)$);

\coordinate (c) at (1.5,2);

\begin{scope}
\clip (-0.25,0) .. controls +(90:1.5) and \control{(3.25,-0.2)}{(90:1.5)} .. controls + (-90:1.15) and \control{(-0.25,0)}{(-90:1.75)};
\begin{scope}[red!50!purple, opacity = 0.7]
\draw  (l1) --(l2) -- (l3) --cycle;
\draw (l2) -- (r1);
\draw (l3) -- (r4);
\draw (l2) -- (r4);
\draw (r1) --(r2) -- (r3) -- (r4) --cycle; 
\end{scope}
\draw [name path =e1, opacity = 0.5] (l1) -- (t1);
\draw [name path =e2, opacity = 0.5] (l3) -- (t1);
\draw [name path =e3, opacity = 0.5] (l3) -- (t7);
\draw [name path =e4, opacity = 0.5] (r4) -- (t2);
\draw [name path =e5, opacity = 0.5] (r3) -- (t3);
\draw [name path =e6, opacity = 0.5] (r3) -- (t4);
\begin{scope}[thick]
\draw (t5) -- (t6);
\draw (t2) -- (t3);
\draw (t5) -- (t1);
\draw (t5) -- (t2);
\draw (t5) -- (t3);
\draw (t6) -- (t4);
\draw (t3) -- (t4);
\end{scope}
\end{scope}

\path [name path =e1, opacity = 0.5] (l1) -- (t1);
\path [name path =e2, opacity = 0.5] (l3) -- (t1);
\path [name path =e3, opacity = 0.5] (l3) -- (t7);
\path [name path =e4, opacity = 0.5] (r4) -- (t2);
\path [name path =e5, opacity = 0.5] (r3) -- (t3);
\path [name path =e6, opacity = 0.5] (r3) -- (t4);

\filldraw [fill opacity = 0.1,green!40!black, name path =blob] (-0.25,0) .. controls +(90:1.5) and \control{(3.25,-0.2)}{(90:1.5)} .. controls + (-90:1.15) and \control{(-0.25,0)}{(-90:1.75)};

\foreach \i in {1,...,6}{
\path [name intersections={of=blob and e\i,total=\tot}]
\foreach \s in {1,...,\tot}{coordinate (n\i) at (intersection-\s)}; 
\fill [green!40!black] (n\i) circle (\e cm);
};
\node [green!40!black] at (1.5,-0.8) {$D_e$};

\foreach \i in {1,2,3}{\fill [red!50!purple] (l\i) circle (\e cm);};
\foreach \i in {1,2,3,4}{\fill [red!50!purple] (r\i) circle (\e cm);};

\node at (2.25,1.5) {\Large $\tilde{G}^{e}$};
\end{scope}
\end{tikzpicture}
\caption{How to obtain $\tilde{G}^e$ from $G$, the vertices $(s_x)_{e \in T(e)}$ are in green on $\partial D_e$.}
\label{pic_ex_leaving_graph}
\end{center}
\end{figure} 

We can define a quasi-order on the set $\mathcal{E} = \cup_{T \in \T}{E(T)}$ of all edges of $\T$: for $e,f \in \mathcal{E}^2$, $e \preccurlyeq f$ if $\tilde{G}^e$ is an embedded immersion of $\tilde{G}^f$ and the image of $(s_x)_{x \in T(e)}$ in the embedded immersion is $(s_x)_{x \in T(f)}$. Lemma~\ref{lem_ancestry_cut} will then enforce that leaving graphs associated to edges from carving-decompositions are in relation if one edge is an ancestor of the other in our decompositions. To establish it, we first need to understand how to make a family of paths around a vertex tangent.

Consider a family of directed embedded paths $P = \{p_1, \ldots, p_k\}$, where $p_i$ is a sequence of edges from vertex $v_{i,0}$ to vertex $v_{i,{\ell_i}}$. If two paths $p_i,p_j$ are transverse at a vertex $v$, it is possible to swap the two subpaths of $p_i$ and $p_j$ starting at $v$, so that $p_i$ now goes from $v_{i,0}$ to $v_{j,{\ell_j}}$ and $p_i$ now goes from $v_{j,0}$ to $v_{i,{\ell_i}}$. After this operation, called \emphdef{path-swap}, $p_i$ and $p_j$ are tangent. 

It can happen that after a path-swap, a path $p$ involved in the swap ends up traversing a same vertex twice (for example, a path-swap between the red and purple paths in Figure~\ref{pic_proof_ancestry_2}). In this case, a self-loop appears in $p$, which we simply remove. After any number of path-swaps, the family $P$ joins $\{ v_{0,0}, \ldots, v_{k,0}\}$ to $\{ v_{0,{\ell_0}}, \ldots, v_{0,{\ell_k}}\}$ by $k$ paths, which may have possibly lost some edges in the process.

\begin{lemma}\label{lem_path_swaps}
Let $P$ be a family of embedded directed paths $\{p_1, \ldots, p_k\}$ and a vertex $v$. There exists a sequence of path-swaps transforming $\{p_1, \ldots, p_k\}$ to a tangent family of paths at $v$.
\end{lemma}

\begin{proof}
Since paths not containing $v$ do not matter for the result, assume that every one of these paths uses $v$. As the paths are embedded, we can consider the circular sequence $s$ of indices of paths around $v$, where each incoming edge of $p^i$ is numbered $i$, and each outgoing edge of $p^i$ is numbered $i^{-1}$ (see Figure~\ref{pic_ex_order_swap}).

\begin{figure}[ht]
\begin{center}
\begin{tikzpicture}
\foreach \i in {1,...,8}{\coordinate (c\i) at (360/8*\i:1);};

\draw [thick, -Stealth, red!50!purple] (c4) -- (0,0) -- (c1);
\node at (c4) [red!50!purple, left] {\large $1$};
\draw [thick,-Stealth, green!50!purple] (c7) -- (0,0) -- (c8);
\node at (c7) [green!50!purple, below right] {\large $2$};
\draw [thick,-Stealth, blue!80!black] (c6) -- (0,0) -- (c3);
\node at (c6) [blue!80!black, below] {\large $3$};
\draw [thick,-Stealth] (c2) -- (0,0) -- (c5);
\node at (c2) [above] {\large $4$};

\fill (0,0) circle (0.1 cm);
\node at (0,-1.75) {$1 4^{-1} 3 2 2^{-1} 1^{-1} 4 3^{-1}$};

\begin{scope}[xshift=4.5cm]
\foreach \i in {1,...,8}{\coordinate (c\i) at (360/8*\i:1);};

\draw [thick, -Stealth, red!50!purple] (c4) -- (0,0) -- (c5);
\node at (c4) [red!50!purple, left] {\large $1$};
\draw [thick,-Stealth, green!50!purple] (c7) -- (0,0) -- (c8);
\node at (c7) [green!50!purple, below right] {\large $2$};
\draw [thick,-Stealth, blue!80!black] (c6) -- (0,0) -- (c1);
\node at (c6) [blue!80!black, below] {\large $3$};
\draw [thick,-Stealth] (c2) -- (0,0) -- (c3);
\node at (c2) [above] {\large $4$};

\fill (0,0) circle (0.1 cm);
\node at (0,-1.75) {$1 1^{-1} 3 2 2^{-1} 3^{-1} 4 4^{-1}$};
\end{scope}
\end{tikzpicture}
\caption{Order around $v$ and the result of the swaps.}
\label{pic_ex_order_swap}
\end{center}
\end{figure} 

In $s$, two letters are interlaced (\ie, $i^{\pm 1}j^{\pm 1}i^{\mp 1}j^{\mp 1}$, where $x^{\pm}$ and $x^{\mp}$ have opposite exponents) if and only if their pair of associated paths are transverse at $v$. Two opposite and adjacent letters cancel out like in standard words: indeed, the path associated to this letters is tangent with no other one; it will not be swapped with another path at $v$. Finally, the path-swap operation corresponds to exchanging the place in $s$ of two $\cdot^{-1}$ letters. To prove our lemma, we need to prove that we can bring $s$ to an empty word by iterating path-swaps.

Let us consider any $i$ in $s$ that is followed by some $j^{-1}$ where $i \neq j$. Such a letter necessarily exists as long as $s$ is non-empty. Indeed, $s$ is equivalent up to circular rotations, and each letter in $s$ also has its opposite in $s$. Simply swap $j^{-1}$ and $i^{-1}$, which cancels $ii^{-1}$, and repeat. 
\end{proof}

\begin{lemma}\label{lem_ancestry_cut}
Let $G$ be a plane graph, $(T,\phi)$ a disc-linked carving-decomposition of $G$, and $a,b \in E(T)^2$ such that $a$ is an ancestor of $b$. Then $b \preccurlyeq a$.
\end{lemma}

\begin{proof}
The $\preccurlyeq$ relation between $b$ and $a$ can be defined since the carving-decomposition satisfies the disc property. We denote $k = \wid (a) = \wid (b)$. Let assume that $G^a_1 = G^a$ so that $A = V(G^a_2)$ and $B = V(G^b)$ with respect to the linked definition of $A$ and $B$. Let us call $P$ the minimal path of $T$ containing both $a$ and $b$. We define $H = G \smallsetminus E(G^a_2) \smallsetminus E(G^b)$, in other words, the edges of $G$ between $T(a)$ and $T(b)$ (where $T(a)$ and $T(b)$ are included). By construction, there are two disjoint families $\mathcal{F}_a = V(H) \cap V(A)$ and $\mathcal{F}_b = V(H) \cap V(B)$, where $\mathcal{F}_a$ consists of the endpoints of $T(a)$ in $G^a_2$ and $\mathcal{F}_b$ consists of the endpoints of $T(b)$ in $G^b$. 

\begin{figure}[ht]
\begin{center}
\begin{tikzpicture}
\def\e{0.06} \def\el{0.1}  \def\ex{0.2}
\begin{scope}
\coordinate (a0) at (0,0);
\coordinate (a1) at (0.5,-0.75);
\coordinate (a2) at (1.5,-0.25);
\coordinate (a3) at (1,1);
\coordinate (a4) at (1.2,-1);
\coordinate (a5) at (2,-0.8);

\coordinate (b0) at (2,1.75);
\coordinate (b1) at (2.1,0.9);
\coordinate (b2) at (2.4,0.5);
\coordinate (b3) at (3,0);
\coordinate (b4) at ($(b0)+(-110:1)$);
\coordinate (b5) at ($(b0)+(15:1.6)$);
\coordinate (b6) at ($(b0)+(-10:1.2)$);
\coordinate (b7) at ($(b1)+(25:0.5)$);
\coordinate (b8) at ($(b2)+(45:0.85)$);
\coordinate (b9) at ($(b3)+(80:0.7)$);
\coordinate (b10) at ($(b3)+(40:1.5)$);

\coordinate (c0) at ($(b5)+(40:1)$);
\coordinate (c1) at ($(b5)+(0:0.75)$);
\coordinate (c2) at ($(b10)+(40:1)$);
\coordinate (c3) at ($(b10)+(60:0.8)$);

\begin{scope}
\draw (a0) -- (a1) -- (a2) -- (a3) -- cycle;
\draw (a1) -- (a4) -- (a2);
\draw (a2) -- (a5) -- (a4);
\end{scope}
\begin{scope}
\draw (a3) -- (b0);
\draw (a2) -- (b1);
\draw (a2) -- (b2) -- (b8);
\draw (a5) -- (b3) -- (a2);
\draw (b8) -- (b10);
\draw (b0) -- (b1) -- (b2);
\draw (b0) -- (b4) -- (b1);
\draw (b5) -- (b6);
\draw (b0) -- (b5);
\draw (b0) -- (b6) -- (b8);
\draw (b1) -- (b7) -- (b8) -- (b9) -- (b10);
\draw (b2) -- (b3);
\draw (b2) -- (b7);
\draw (b7) -- (b6);
\draw (b2) -- (b9) -- (b3) -- (b10) -- (b5);
\end{scope}
\begin{scope}
\draw (c1) -- (c0) -- (c2) -- (c3) -- (c1);
\end{scope}
\begin{scope}
\draw (b10) -- (c1);
\draw (b5) --(c0);
\draw (b5) --(c1);
\draw (b10) --(c2);
\draw (b10) --(c3);
\end{scope}

\foreach \i in {0,...,5}{\fill (a\i) circle (\e cm);};
\foreach \i in {0,...,10}{\fill (b\i) circle (\e cm);};
\foreach \i in {0,...,3}{\fill (c\i) circle (\e cm);};

\draw [red!50!purple, opacity = 0.6] (1,2) .. controls +(-80:1) and \control{(3,-1)}{(150:1)} node [right, opacity = 1] {$T(b)$};
\draw [blue!80!black, opacity = 0.6] (3,2.5) .. controls +(-70:0.75) and \control{(4.4,0.45)}{(160:0.75)} node [right, opacity = 1] {$T(a)$};
\node at ($(b0)+(110:1)$) {$G$};
\end{scope}

\draw [-Stealth] (5.2,0) -- +(1,0);

\begin{scope}[xshift = 7cm]
\coordinate (a0) at (0,0);
\coordinate (a1) at (0.5,-0.75);
\coordinate (a2) at (1.5,-0.25);
\coordinate (a3) at (1,1);
\coordinate (a4) at (1.2,-1);
\coordinate (a5) at (2,-0.8);

\coordinate (b0) at (2,1.75);
\coordinate (b1) at (2.1,0.9);
\coordinate (b2) at (2.4,0.5);
\coordinate (b3) at (3,0);
\coordinate (b4) at ($(b0)+(-110:1)$);
\coordinate (b5) at ($(b0)+(15:1.6)$);
\coordinate (b6) at ($(b0)+(-10:1.2)$);
\coordinate (b7) at ($(b1)+(25:0.5)$);
\coordinate (b8) at ($(b2)+(45:0.85)$);
\coordinate (b9) at ($(b3)+(80:0.7)$);
\coordinate (b10) at ($(b3)+(40:1.5)$);

\coordinate (c0) at ($(b5)+(40:1)$);
\coordinate (c1) at ($(b5)+(0:0.75)$);
\coordinate (c2) at ($(b10)+(40:1)$);
\coordinate (c3) at ($(b10)+(60:0.8)$);

\def\op{0.4}
\draw [blue, smooth cycle, tension = 0.75, opacity = \op] plot coordinates {($(b5)+(-180:\el)$) ($(c0)+(80:\ex)$) ($(c2)+(-20:1.5*\el)$) ($(b10)+(-100:\el)$)};
\draw [red!50!purple, smooth cycle, tension = 0.75, opacity = \op] plot coordinates {($(a0)+(150:\el)$) ($(a3)+(90:\el)$) ($(a2)+(30:\el)$) ($(a5)+(-30:2.5*\el)$) ($(a4)+(-90:\el)$) ($(a1)+(180:2*\el)$)};
\draw [smooth cycle, tension = 0.5, opacity = \op] plot coordinates {($(a3)+(150:\ex)$) ($(a2)+(-160:\ex)$) ($(a5)+(-90:\el)$) ($(b10)+(-45:\ex)$) ($(b5)+(60:\ex)$) ($(b0)+(110:\el)$)};

\begin{scope}[thick]
\begin{scope}[red!50!purple]
\draw (a0) -- (a1) -- (a2) -- (a3) -- cycle;
\draw (a1) -- (a4) -- (a2);
\draw (a2) -- (a5) -- (a4);
\end{scope}
\begin{scope}
\draw (a3) -- (b0);
\draw (a2) -- (b1);
\draw (a2) -- (b2) -- (b8);
\draw (a5) -- (b3) -- (a2);
\draw (b0) -- (b1) -- (b2);
\draw (b0) -- (b4) -- (b1);
\draw (b0) -- (b6) -- (b8);
\draw (b1) -- (b7) -- (b8);
\draw (b2) -- (b3);
\draw (b2) -- (b7);
\draw (b7) -- (b6);
\draw (b2) -- (b9);
\end{scope}
\begin{scope}[blue!80!black]
\draw (c1) -- (c0) -- (c2) -- (c3) -- (c1);
\draw (b10) -- (c1);
\draw (b5) --(c0);
\draw (b5) --(c1);
\draw (b10) --(c2);
\draw (b10) --(c3);
\draw (b10) -- (b5);
\end{scope}

\begin{scope}
\draw (b9) -- (b10);
\draw (b3) -- (b10);
\draw (b5) -- (b6);
\draw (b0) -- (b5);
\draw (b8) -- (b10);
\end{scope}
\end{scope}

\foreach \i in {0,...,5}{\fill (a\i) circle (\e cm);};
\foreach \i in {0,...,10}{\fill (b\i) circle (\e cm);};
\foreach \i in {0,...,3}{\fill (c\i) circle (\e cm);};

\node at ($(a0)+(120:0.5)$) [red!50!purple] {$B$};
\node at ($(c2)+(-50:0.35)$) [blue!80!black] {$A$};
\node at ($(b0)+(90:0.5)$) {$H$};
\end{scope}
\end{tikzpicture}
\caption{Definition of $A$, $B$, and $H$ from $a$ and $b$.}
\label{pic_proof_ancestry}
\end{center}
\end{figure} 

By construction, every cut of $H$ separating $\mathcal{F}_a$ and $\mathcal{F}_b$ is a cut separating $A$ and $B$ so that $\mcut (A,B) \leq \mcut (\mathcal{F}_a, \mathcal{F}_b)$. However, since $(T,\phi)$ is linked, $\mcut (A,B) = \min_{e \in P} \wid (e)$. As $a$ is an ancestor of $b$, $\wid (b) = \wid (a) = k \leq \min_{e \in P} \wid (e)$. Hence, $k \leq \mcut (\mathcal{F}_a, \mathcal{F}_b)$ and by Menger's theorem ~\cite{Menger_paths} applied to the associated abstract graphs, there exists a family $P$ of $k$ edge-disjoint paths from $\mathcal{F}_a$ to $\mathcal{F}_b$. As $H$ is embedded, the paths are also embedded.

The family $P$ may contain transverse paths, since all edges of $T(a)$ and $T(b)$ are used by these paths, the only transverse vertices are the ones for which edges around them are interlaces. Let $m$ be the number of transverse vertices of $\mathcal{F}$. A path-swap at $v$ does not affect the transversality of paths at any other vertex of $H$, unless it involves the suppression of edges from some path, in which case the number of transversal paths can only decrease at any other vertex. Hence, we can apply the sequence of path-swaps from Lemma~\ref{lem_path_swaps} to reduce $m$ by one (see Figure~\ref{pic_proof_ancestry_2} for a visual description of the process).

\begin{figure}[ht]
\begin{center}
\begin{tikzpicture}
\def\e{0.06} 
\begin{scope}
\coordinate (a2) at (1.5,-0.25);
\coordinate (a3) at (1,1);
\coordinate (a5) at (2,-0.8);
\coordinate (b0) at (2,1.75);
\coordinate (b1) at (2.1,0.9);
\coordinate (b2) at (2.4,0.5);
\coordinate (b3) at (3,0);
\coordinate (b4) at ($(b0)+(-110:1)$);
\coordinate (b5) at ($(b0)+(15:1.6)$);
\coordinate (b6) at ($(b0)+(-10:1.2)$);
\coordinate (b7) at ($(b1)+(25:0.5)$);
\coordinate (b8) at ($(b2)+(45:0.85)$);
\coordinate (b9) at ($(b3)+(80:0.7)$);
\coordinate (b10) at ($(b3)+(40:1.5)$);

\begin{scope}[opacity = 0.4]
\draw (b6) -- (b8);
\draw (b0) -- (b6);
\draw (b8) -- (b9);
\draw (b7) -- (b8);
\draw (b9) -- (b3);
\draw (b1) -- (b7);
\end{scope}
\begin{scope}[thick]
\begin{scope}
\draw (a5) -- (b3);
\draw (b3) -- (b10);
\end{scope}
\begin{scope}[red!50!purple]
\draw (b0) -- (b1);
\draw (a2) -- (b1);
\draw (b0) -- (b5);
\end{scope}
\begin{scope}[blue!80!black]
\draw (a2) -- (b2);
\draw (b8) -- (b10);
\draw (b2) -- (b8);
\end{scope}
\begin{scope}[green!70!purple]
\draw (b3) -- (a2);
\draw (b2) -- (b3);
\draw (b2) -- (b7);
\draw (b7) -- (b6);
\draw (b5) -- (b6);
\end{scope}
\begin{scope}[purple!60!blue]
\draw (a3) -- (b0);
\draw (b0) -- (b4);
\draw (b1) -- (b2);
\draw (b2) -- (b9); 
\draw (b9) -- (b10);
\draw (b4) -- (b1);
\end{scope}
\end{scope}

\foreach \i in {2,3,5}{\fill (a\i) circle (\e cm);};
\foreach \i in {0,...,10}{\fill (b\i) circle (\e cm);};

\draw [smooth cycle, tension = 0.75, opacity = 0.2] plot coordinates {($(a3)+(60:0.2)$) ($(a3)+(180:0.2)$) ($(a2)+(180:0.2)$) ($(a5)+(-120:0.2)$) ($(a5)+(0:0.2)$) ($(a2)+(0:0.2)$)};
\draw [smooth cycle, tension = 0.75, opacity = 0.2] plot coordinates {($(b5)+(60:0.2)$) ($(b5)+(180:0.2)$) ($(b10)+(180:0.2)$) ($(b10)+(0:0.2)$)};

\node at (0.8,0) [opacity = 0.8] {$\mathcal{F}_b$};
\node at (4.25,2.25) [opacity = 0.8] {$\mathcal{F}_a$};
\node at ($(b0)+(90:0.5)$) {$H$};
\node at (3,-0.8) {$m=2$};

\draw [-Stealth] (4.25,0) -- +(1,0);
\end{scope}

\begin{scope}[xshift =4.5cm]
\coordinate (a2) at (1.5,-0.25);
\coordinate (a3) at (1,1);
\coordinate (a5) at (2,-0.8);
\coordinate (b0) at (2,1.75);
\coordinate (b1) at (2.1,0.9);
\coordinate (b2) at (2.4,0.5);
\coordinate (b3) at (3,0);
\coordinate (b4) at ($(b0)+(-110:1)$);
\coordinate (b5) at ($(b0)+(15:1.6)$);
\coordinate (b6) at ($(b0)+(-10:1.2)$);
\coordinate (b7) at ($(b1)+(25:0.5)$);
\coordinate (b8) at ($(b2)+(45:0.85)$);
\coordinate (b9) at ($(b3)+(80:0.7)$);
\coordinate (b10) at ($(b3)+(40:1.5)$); 

\begin{scope}[opacity = 0.4]
\draw (b6) -- (b8);
\draw (b0) -- (b6);
\draw (b8) -- (b9);
\draw (b7) -- (b8);
\draw (b9) -- (b3);
\draw (b1) -- (b7);
\end{scope}
\begin{scope}[thick]
\begin{scope}
\draw (a5) -- (b3);
\draw (b3) -- (b10);
\end{scope}
\begin{scope}[red!50!purple]
\draw (b0) -- (b1);
\draw (a2) -- (b1);
\draw (b0) -- (b5);
\end{scope}
\begin{scope}[blue!80!black]
\draw (a2) -- (b2);
\draw (b8) -- (b10);
\draw (b2) -- (b8);
\end{scope}
\begin{scope}[green!70!purple]
\draw (b3) -- (a2);
\draw (b2) -- (b3);
\draw (b2) -- (b9); 
\draw (b9) -- (b10);
\end{scope}
\begin{scope}[purple!60!blue]
\draw (a3) -- (b0);
\draw (b0) -- (b4);
\draw (b1) -- (b2);
\draw (b4) -- (b1);
\draw (b2) -- (b7);
\draw (b7) -- (b6);
\draw (b5) -- (b6);
\end{scope}
\end{scope}

\foreach \i in {2,3,5}{\fill (a\i) circle (\e cm);};
\foreach \i in {0,...,10}{\fill (b\i) circle (\e cm);};

\node at (3,-0.8) {$m=1$};
\node at ($(b0)+(90:0.5)$) {$H$};
\draw [-Stealth] (4.25,0) -- +(1,0);
\end{scope}

\begin{scope}[xshift = 9cm]
\coordinate (a2) at (1.5,-0.25);
\coordinate (a3) at (1,1);
\coordinate (a5) at (2,-0.8);
\coordinate (b0) at (2,1.75);
\coordinate (b1) at (2.1,0.9);
\coordinate (b2) at (2.4,0.5);
\coordinate (b3) at (3,0);
\coordinate (b4) at ($(b0)+(-110:1)$);
\coordinate (b5) at ($(b0)+(15:1.6)$);
\coordinate (b6) at ($(b0)+(-10:1.2)$);
\coordinate (b7) at ($(b1)+(25:0.5)$);
\coordinate (b8) at ($(b2)+(45:0.85)$);
\coordinate (b9) at ($(b3)+(80:0.7)$);
\coordinate (b10) at ($(b3)+(40:1.5)$); 

\begin{scope}[opacity = 0.4]
\draw (b6) -- (b8);
\draw (b0) -- (b6);
\draw (b8) -- (b9);
\draw (b7) -- (b8);
\draw (b9) -- (b3);
\draw (b1) -- (b7);
\draw (b0) -- (b1);
\draw (b4) -- (b1);
\draw (b0) -- (b4);
\end{scope}
\begin{scope}[thick]
\begin{scope}
\draw (a5) -- (b3);
\draw (b3) -- (b10);
\end{scope}
\begin{scope}[red!50!purple]
\draw (a2) -- (b1);
\draw (b1) -- (b2);
\draw (b2) -- (b7);
\draw (b7) -- (b6);
\draw (b5) -- (b6);
\end{scope}
\begin{scope}[blue!80!black]
\draw (a2) -- (b2);
\draw (b8) -- (b10);
\draw (b2) -- (b8);
\end{scope}
\begin{scope}[green!70!purple]
\draw (b3) -- (a2);
\draw (b2) -- (b3);
\draw (b2) -- (b9); 
\draw (b9) -- (b10);
\end{scope}
\begin{scope}[purple!60!blue]
\draw (a3) -- (b0);
\draw (b0) -- (b5);
\end{scope}
\end{scope}

\foreach \i in {2,3,5}{\fill (a\i) circle (\e cm);};
\foreach \i in {0,...,10}{\fill (b\i) circle (\e cm);};

\node at (3,-0.8) {$m=0$};
\node at ($(b0)+(90:0.5)$) {$H$};
\end{scope}
\end{tikzpicture}
\caption{An example of a sequence of path-swapping on a graph.}
\label{pic_proof_ancestry_2}
\end{center}
\end{figure} 

After this process, $P$ is a family of edge-disjoint embedded paths connecting vertices from $\mathcal{F}_b$ to vertices of $\mathcal{F}_a$, where the first edge of each path belongs to $T(b)$ and the last one to $T(a)$, which have same size. Hence, each path starts on $G_b$, goes through the edge $x \in T(b)$ whose subdivision yields $c_x$ in $\tilde{G^b}$, $H$, then the edge $y \in T(a)$ whose subdivision yields $c_y$ in $\tilde{G^a}$ and ends on $\mathcal{F}_a$. This way, we associate to each vertex $c_x$ induced by $T(b)$, a vertex $c_y$ induced by $T(a)$, and the subpath $P_x$ of $P$ starting with $x$ and ending at the vertex $c_y$ on the last one of its edge. The path $P_x$ is the image of the edge incident to $c_x$ in $\tilde{G^b}$.

By the identity, $G^b$ is obviously an embedded immersion minor of $G^a$. By construction the family $(P_x)_{x \in T(b)}$ is tangent and respects the order at their starting point in $G^b$ so that, by Lemma~\ref{lem_equiv_plan_imm}, $\tilde{G}^b$ is an embedded immersion of $\tilde{G}^a$, \ie, $b \preccurlyeq a$.
\end{proof}

\subsection{Nash-Williams' argument.}
\label{sec_nash_williams_argument}

The Lemma~\ref{lem_lem_trees} is a general lemma on trees (not restricted to cubic ones) where a quasi-order on the set of edges is considered. The concepts behind the proof of the following proposition stem from~\cite{nash-williams_kruskal} and, for the sake of completeness, we provide a proof slightly adapted from~\cite{Geelen_wqo_branchwidth}.

Recall that a \emphdef{forest} is a family or union of rooted trees (the root induces an orientation on edges which will be used for the definition of ancestor in this context). We denote by $\mathcal{E}(\mathcal{T}) = \bigcup_{T \in \mathcal{T}} E(T)$ the set of edges in this forest. A \emphdef{subforest} $\mathcal{F}'$ of $\mathcal{F}$ is a forest where each tree $T'$ is a subtree of $T$ which is rooted at the only incoming edge. A subtree induced by an edge $e \in E(T)$ is the maximal subtree of $T$ whose root is $e$. For each edge $e \in E(T)$, the set of children of $e$ is the set $\text{ch}_{\mathcal{T}}(e)$ of edges of $T$ whose tail is the head of $e$, \ie, the edges which follow $e$ for the order induced by the root.

\begin{lemma}{Lemma on trees~\cite[(3.1)]{Geelen_wqo_branchwidth}}\label{lem_lem_trees}
Let $\mathcal{T}$ be a family of trees, $n \in \N$ an integer, $w$ a map $\mathcal{E} (\mathcal{T}) \rightarrow \N$, and $\preccurlyeq$ a quasi-order on $\mathcal{E}(\mathcal{T})$ with no infinite descending sequence such that for all pairs of edges $a,b$, if $a$ is an ancestor of $b$ then $b \preccurlyeq a$. Then:

\begin{center}
\begin{tabular}{ c c}
\begin{tabular}{c}
If $\mathcal{E}(\mathcal{T})$ is not well-quasi-ordered by $\preccurlyeq$ then, there exists an antichain $(e_n)_{n \in \N}$ such that\\
the set $\cup_{n \in \N} ~\text{ch}_{\mathcal{T}} (e_n)$ is well-quasi-ordered by $\preccurlyeq$.
\end{tabular}
& (\hypertarget{prop_lem_trees}{S})
\end{tabular}
\end{center}

\end{lemma}

\begin{proof}
Let us assume by contraction that this lemma is false. Let $\mathcal{T}$ be a forest not satisfying this lemma, such that $w$ has minimal range and $\mathcal{E}(T)$ is not well-quasi-ordered by $\preccurlyeq$. In particular, for each value $k$ in its range there is an infinite number of edges with image $k$ (otherwise a finite number of trees could be deleted). 

Let $k$ be the minimum value of $w$ in its range. Then the set $K = w^{-1}(k)$ of edges labelled by $k$ is not well-quasi-ordered by $\preccurlyeq$. Indeed, if this was the case, let us consider $\mathcal{T}' = \mathcal{T} \smallsetminus K$, the subforest of $\mathcal{T}$ where edges labelled $k$ have been removed. By minimality of the range of $w$, $\mathcal{T}$ must satisfy property~(\hyperlink{prop_lem_trees}{S}). Either it is well-quasi-ordered by $\preccurlyeq$, which is absurd because by adding back $K$, $\mathcal{T}$ would still be well-quasi-ordered. Or, there exists an antichain $(e_n)_{n \in \N} \in \mathcal{E}(\mathcal{T}')^\N$ such that the set $\cup_{n \in \N} ~\text{ch}_{\mathcal{T}'} (e_n)$ is well-quasi-ordered by $\preccurlyeq$. It follows that $\cup_{n \in \N} ~\text{ch}_{\mathcal{T}'} (e_n) \cup K$ is well-quasi-ordered. This last set contains $\cup_{n \in \N} ~\text{ch}_{\mathcal{T}} (e_n)$, which is hence well-quasi-ordered as well; this contradicts property~(\hyperlink{prop_lem_trees}{S}), hence the set of edges labelled by $k$ is not well-quasi-ordered by $\preccurlyeq$.

We now construct a minimal antichain $(a_n)_{n \in \N} \in K^{\N}$ with respect to the ancestry relation. Note that, in our framework, since $k$ is the minimal value taken by $w$ on $\mathcal{E}(\mathcal{T})$, if $\wid (e)=k$ then, every edge $f$ in the subtree induced by $e$ such that $\wid (f) = k$ is a successor of $e$. We proceed by induction and first choose $a_1$ such that $a_1$ is the first element of an antichain, but none of its successor is. Then if the first $k-1$ elements are fixed, we choose $a_k$ so that $a_i \not \preccurlyeq a_k$ for $i < k$, $(a_i)_{1 \leq i \leq k}$ can be extended to an infinite antichain, and none of its successors can. Such elements exist since $K$ contains antichains, and if an element does not satisfy the successor condition, replace it by one of its successors which are in finite number since there is no infinite descending sequence.

By property~(\hyperlink{prop_lem_trees}{S}), $\cup_{n \in \N} ~\text{ch}(a_n)$ is not well-quasi-ordered. Let $\mathcal{A}$ be the subforest of $\T$ induced by $(\text{ch}_{\mathcal{A}} (a_n))_{n \in \N}$. Since any antichain of $\mathcal{A}$ is an antichain of $\T$, and $\text{ch}_{\T} |_{\mathcal{E}(\mathcal{A})} = \text{ch}_{\mathcal{A}}$, it follows that $\mathcal{A}$ also contradicts property (\hyperlink{prop_lem_trees}{S}). Hence, $K \cap \mathcal{E}(\mathcal{A})$ is not well-quasi-ordered and we can find among it another antichain $(s_n)_{n \in \N} \in (K \cap \mathcal{E}(\mathcal{A}))^\N$. 

By construction, each edge $s_n$ is the successor of exactly one edge $a_{p(n)}$. Let $N$ be such that $p(N)$ is minimal and $n > N$. Since $a_{p(n)}$ is an ancestor of $s_n$, we have $s_n \preccurlyeq a_{p(n)}$. It follows that no $a_i$ satisfies $a_i \preccurlyeq s_n$ for $i < p(N)$, by transitivity of $\preccurlyeq$. Thus, $a_1,\ldots, a_{p(N)-1}, s_{N}, s_{N+1}, \ldots$ is an antichain, which contradicts the construction of $a_{p(N)}$ and concludes the proof.
\end{proof}

We can now prove Theorem~\ref{th_emb_immersion_bounded} which states that graphs embedded in $\R^2$ with bounded carving-width are well-quasi-ordered by embedded immersion. Usually the proof of analogues in the non-embedded of Theorem~\ref{th_emb_immersion_bounded} exploits an antichain yielded by Lemma~\ref{lem_lem_trees} to construct an edge $\preccurlyeq$ to another in the antichain, which leads to a contradiction. That last construction is, in the context of abstract graph, the crux whose resolution relies on the finiteness of the number of ways to merge two ``abstract leaving graphs''. Here, our leaving graphs are much more constrained by the embedding, which modifies this part of the proof: it is harder to assemble the graphs but easier to find the ones to assemble.

\begin{proof}[Proof of Theorem~\ref{th_emb_immersion_bounded}]
Recall that, by Lemma~\ref{lem_plane_equiv}, it is enough to consider plane graphs and let $(G_n)_{n \in \N}$ be a sequence of plane graphs such that the carving-width of each graph $G_n$ is bounded by $k \in \N$. For each $n \in \N$, thanks to Proposition~\ref{prop_disc_decomp}, we choose a disc-linked carving-decomposition $T_n$ of $G_n$ of carving-width $\cwid (G_n)$ and set the forest $\mathcal{T} = \cup_{n \in \N} T_n$. Let us assume by contradiction that $(G_n)_{n \in \N}$ is not well-quasi-ordered by embedded immersion.

It follows that the set of roots of $\mathcal{T}$ is an antichain for $\preccurlyeq$ and $\mathcal{E} (\mathcal{T})$ is not well-quasi-ordered by $\preccurlyeq$. By Lemma~\ref{lem_ancestry_cut}, if two edges $a,b$ of $\mathcal{T}$ are such that $a$ is an ancestor of $b$, then $b \preccurlyeq a$. Since the number of graphs on a fixed vertex is finite and there is a finite number of embeddings of a planar graph up to isotopy, there cannot be an infinite descending sequence of graphs for $\preccurlyeq$. Hence, as $w$ is bounded by $k$ on $\mathcal{E}(\mathcal{T})$, we can apply Lemma~\ref{lem_lem_trees} and exhibit an antichain $(a_n)_{n \in \N}$ such that $\cup_{n \in \N} ~\text{ch}_{\mathcal{T}} (a_n)$ is well-quasi-ordered by $\preccurlyeq$.

Since our trees are cubic, each set $\text{ch} (a_n)$ either consists of two edges, a left child $\ell_n$ and a right child $r_n$, or is empty. However, the edges incident to leaves are well-quasi-ordered since the number of leaving graphs for such edges are in finite number (they have at most $k$ edges). Hence, there is at most a finite number of edges incident to leaves in $(a_n)_{n_\N}$ and each sequence of $(\ell_n)_{n \in \N}, (r_n)_{n \in \N}$ is infinite and well-quasi-ordered by $\preccurlyeq$. 

Up to extracting subsequences of $(a_i)_{n \in \N}$ and hereby $(\ell_n)_{n \in \N}, (r_n)_{n \in \N}$, we can assume that $\forall i \in \N, \ell_i \preccurlyeq \ell_{i+1}$, $\forall i \in \N, r_i \preccurlyeq r_{i+1}$, and $\forall i \in \N, \wid (\ell_i) = \wid (r_{i+1})$, $\wid (\ell_i) = \wid (r_{i+1})$, and $\wid (a_i) = \wid (a_{i+1})$ (each one of these properties may require an extraction). In other words, the sequences $(\ell_n)_{n \in \N}$ and $(r_n)_{n \in \N}$ are ascending chains, and each sequence has a unique value by $w$. For simplicity of notations, we consider a subset of the sequence $(G_n)_{n \in \N}$ and index it such that we can design by $G_{i}$ the graph of $(G_n)_{n \in \N}$ such that $a_i$ is an edge of a carving-decomposition of $G_{i}$.

We now exploit our embedded structure. Let us set $i \in \N$. By the disc property, there exists a Jordan curve $\gamma_{a_i}$ in the plane such that $\gamma_{a_i}$ only meets $G_{i}$ once on the interior of each edge of $T_{i}(a_i)$. Furthermore, up to switching orientation, $\gamma_{a_i}$ intersects $T_{i}(a_i)$. Similarly, we define $\gamma_{\ell_i}$ and $\gamma_{r_i}$, which intersect $T_{i}(\ell_i)$ and $T_{i}(r_i)$. Since our decompositions satisfy the disc property, vertices of $G_{i}^{\ell_i}$ and $G_{i}^{r_i}$ both lie on the same side of $\gamma_{a_i}$. Furthermore, it is possible to isotope $\gamma_{\ell_i}$, $\gamma_{r_i}$, and $\gamma_{a_i}$ so that each point of any $\gamma$ curve belongs to another one. Their union forms a $\theta$-curve as depicted in Figure~\ref{pic_def_theta_curve}. By comparing the endpoint in $G_{i}^{\ell_i}$ and the order around it, it is possible to establish a bijection between edges met along $\gamma_{\ell_i}$ with respect to $G^{a_i}$ and edges met along $\gamma_{\ell_i}$ with respect to $G^{a_i}$. This bijection can then be extended to an injection to edges met by $\theta_i$. The same arguments apply for $\theta_i$ and $\gamma_{r_i}$. In the following we will use these bijections and map to $\theta$-curves without emphasising it.

\begin{figure}[ht]
\begin{center}
\begin{tikzpicture}
\clip (-0.5,-1.7) rectangle (10.5,2.3);
\def\e{0.075}
\coordinate (l1) at (0,0);
\coordinate (l2) at ($(l1)+(-30:1)$);
\coordinate (l3) at ($(l1)+(30:1)$);

\coordinate (r1) at (2,-0.5);
\coordinate (r2) at ($(r1)+(0.75,0)$);
\coordinate (r3) at ($(r1)+(0.75,0.75)$);
\coordinate (r4) at ($(r1)+(0,0.75)$);

\coordinate (c) at (1.5,2);

\draw [red!50!purple, opacity = 0.7] (l1) --(l2) -- (l3) --cycle;
\draw [blue!80!black, opacity = 0.7] (r1) --(r2) -- (r3) -- (r4) --cycle;

\begin{scope}
\clip (1.45,-0.15) circle (1.9 cm and 1.2 cm);
\draw [opacity = 0.5, , name path=e1] (l1) -- ++(80:1) .. controls +(80:0.5) and \control{(c)}{(-150:0.5)};
\draw [opacity = 0.5, , name path=e2] (l3) -- ++(80:0.75) .. controls +(80:0.25) and \control{(c)}{(-120:0.25)};
\draw [opacity = 0.5, , name path=e3] (r3) -- ++(110:1) .. controls +(110:0.5) and \control{(c)}{(-40:0.25)};
\draw [opacity = 0.5, , name path=e4] (r3) -- ++(130:1) .. controls +(130:0.75) and \control{(c)}{(-60:0.5)};
\draw [opacity = 0.5, , name path=e5] (r4) -- ++(110:1) .. controls +(110:0.25) and \control{(c)}{(-80:0.25)};
\draw [opacity = 0.5, , name path=e6] (l2) -- (r1);
\draw [opacity = 0.5, , name path=e7] (l3) -- (r4);
\draw  [opacity = 0.5, , name path=e8](l2) -- (r4);
\end{scope}
\path [opacity = 0.5, , name path=e1] (l1) -- ++(80:1) .. controls +(80:0.5) and \control{(c)}{(-150:0.5)};
\path [opacity = 0.5, , name path=e2] (l3) -- ++(80:0.75) .. controls +(80:0.25) and \control{(c)}{(-120:0.25)};
\path [opacity = 0.5, , name path=e3] (r3) -- ++(110:1) .. controls +(110:0.5) and \control{(c)}{(-40:0.25)};
\path [opacity = 0.5, , name path=e4] (r3) -- ++(130:1) .. controls +(130:0.75) and \control{(c)}{(-60:0.5)};
\path [opacity = 0.5, , name path=e5] (r4) -- ++(110:1) .. controls +(110:0.25) and \control{(c)}{(-80:0.25)};
\path [opacity = 0.5, , name path=e6] (l2) -- (r1);
\path [opacity = 0.5, , name path=e7] (l3) -- (r4);
\path  [opacity = 0.5, , name path=e8](l2) -- (r4);

\foreach \i in {1,2,3}{\fill [red!50!purple] (l\i) circle (\e cm);};
\foreach \i in {1,2,3,4}{\fill [blue!80!black] (r\i) circle (\e cm);};

\draw [thick, name path=blob] (1.45,-0.15) circle (1.9 cm and 1.2 cm);
\draw [thick, name path=vert] (1.45,1.05) -- (1.45,-1.35);

\foreach \i in {1,...,5}{
\path [name intersections={of=blob and e\i,total=\tot}]
\foreach \s in {1,...,\tot}{coordinate (n\i) at (intersection-\s)}; 
\fill [black!50] (n\i) circle (\e cm);
};

\foreach \i in {6,7,8}{
\path [name intersections={of=vert and e\i,total=\tot}]
\foreach \s in {1,...,\tot}{coordinate (n\i) at (intersection-\s)}; 
\fill [black!50] (n\i) circle (\e cm);
};

\node at (0.9,-0.9) [red!50!purple] {$G_{i}^{\ell_i}$};
\node at (2.1,-0.9) [blue!80!black] {$G_{i}^{r_i}$};
\node at (3.3,-1) {$\theta_i$};
\node at (1.5,1.5) {\Large $\tilde{G}_{i}^{a_i}$};

\begin{scope}[xshift = 8cm, yshift = -0.15cm]
\coordinate (c) at (0,2.25);

\filldraw [red!50!purple, fill opacity = 0.1] (-0.8,0.63) arc (90:270:1 and 0.63) -- cycle;
\filldraw [blue!80!black, fill opacity = 0.1] (0.8,0.63) arc (90:-90:1 and 0.63) -- cycle;

\begin{scope}[opacity = 0.5]
\clip (0,0) circle (2.4 cm and 1.39cm);
\draw (-0.25,0.4) -- (0.25,0.6);
\draw (-0.25,0.4) -- (0.25,0.1);
\draw (-0.25,-0.3) -- (0.25,-0.3);
\draw (-0.8,0.63) -- (-0.6,0.5);
\draw [dotted] (-0.6,0.5) -- (-0.25,0.4);
\draw (-0.8,-0.2) -- (-0.6,-0.2);
\draw [dotted] (-0.6,-0.2) -- (-0.25,-0.3);
\draw (-0.8,-0.2) -- (-0.6,0);
\draw [dotted] (-0.6,0) -- (-0.25,0.4);

\draw (0.8,0.4) -- (0.6,0.5);
\draw [dotted] (0.6,0.5) -- (0.25,0.6);
\draw (0.8,0.4) -- (0.6,0.2);
\draw [dotted] (0.6,0.2) -- (0.25,0.1);
\draw (0.8,-0.3) -- (0.6,-0.3);
\draw [dotted] (0.6,-0.3) -- (0.25,-0.3);

\draw (-0.8,0.63) -- (-0.75,0.83);
\draw [dotted] (0-0.75,0.83) -- (-0.7,1.09);
\draw (-0.7,1.09) -- +(80:0.5) .. controls +(80:0.4) and \control{(c)}{(-150:0.4)};

\path (-0.8,0.63) arc (90:120:1 and 0.63) coordinate (d1);
\draw (d1) -- +(80:0.2) coordinate (d2);
\draw [dotted] (d2) -- (-1.2,1) coordinate (d3);
\draw (d3) -- +(80:0.5) .. controls +(80:0.4) and \control{(c)}{(-170:0.4)};

\path (0.8,0.63) arc (90:80:1 and 0.63) coordinate (d1);
\draw (d1) -- +(100:0.2) coordinate (d2);
\draw [dotted] (d2) -- +(110:0.25) coordinate (d3);
\draw (d3) -- +(110:0.4) .. controls +(110:0.3) and \control{(c)}{(-40:0.4)};

\path (0.8,0.63) arc (90:55:1 and 0.63) coordinate (d1);
\draw (d1) -- +(110:0.25) coordinate (d2);
\draw [dotted] (d2) -- +(110:0.25) coordinate (d3);
\draw (d3) -- +(110:0.4) .. controls +(110:0.3) and \control{(c)}{(-30:0.4)};

\path (0.8,0.63) arc (90:55:1 and 0.63) coordinate (d1);
\draw (d1) -- +(70:0.25) coordinate (d2);
\draw [dotted] (d2) -- +(100:0.2) coordinate (d3);
\draw (d3) -- +(100:0.4) .. controls +(100:0.4) and \control{(c)}{(-10:0.4)};
\end{scope}

\path [opacity = 0.5, , name path=e1] (-0.7,1.09) -- +(80:0.5) .. controls +(80:0.4) and \control{(c)}{(-150:0.4)};

\path (-0.8,0.63) arc (90:120:1 and 0.63) coordinate (d1);
\path (d1) -- +(80:0.2) coordinate (d2);
\path [dotted] (d2) -- (-1.2,1) coordinate (d3);
\path [opacity = 0.5, , name path=e2] (d3) -- +(80:0.5) .. controls +(80:0.4) and \control{(c)}{(-170:0.4)};

\path (0.8,0.63) arc (90:80:1 and 0.63) coordinate (d1);
\path (d1) -- +(100:0.2) coordinate (d2);
\path [dotted] (d2) -- +(110:0.25) coordinate (d3);
\path [opacity = 0.5, , name path=e3] (d3) -- +(110:0.4) .. controls +(110:0.3) and \control{(c)}{(-40:0.4)};

\path (0.8,0.63) arc (90:55:1 and 0.63) coordinate (d1);
\path (d1) -- +(110:0.25) coordinate (d2);
\path [dotted] (d2) -- +(110:0.25) coordinate (d3);
\path [opacity = 0.5, , name path=e4] (d3) -- +(110:0.4) .. controls +(110:0.3) and \control{(c)}{(-30:0.4)};

\path (0.8,0.63) arc (90:55:1 and 0.63) coordinate (d1);
\path (d1) -- +(70:0.25) coordinate (d2);
\path[dotted] (d2) -- +(100:0.2) coordinate (d3);
\path [opacity = 0.5, , name path=e5] (d3) -- +(100:0.4) .. controls +(100:0.4) and \control{(c)}{(-10:0.4)};
\path [opacity = 0.5, , name path=e6] (-0.25,0.4) -- (0.25,0.6);
\path [opacity = 0.5, , name path=e7] (-0.25,0.4) -- (0.25,0.1);
\path  [opacity = 0.5, , name path=e8] (-0.25,-0.3) -- (0.25,-0.3);

\foreach \i in {1,2,3}{\fill [red!50!purple] (l\i) circle (\e cm);};
\foreach \i in {1,2,3,4}{\fill [blue!80!black] (r\i) circle (\e cm);};

\draw [thick, name path=blob] (0,0) circle (2.4 cm and 1.39cm);
\draw [thick, name path=vert] (0,1.39) -- (0,-1.39);

\foreach \i in {1,...,5}{
\path [name intersections={of=blob and e\i,total=\tot}]
\foreach \s in {1,...,\tot}{coordinate (n\i) at (intersection-\s)}; 
\fill [black!50] (n\i) circle (\e cm);
};

\foreach \i in {6,7,8}{
\path [name intersections={of=vert and e\i,total=\tot}]
\foreach \s in {1,...,\tot}{coordinate (n\i) at (intersection-\s)}; 
\fill [black!50] (n\i) circle (\e cm);
};

\node at (-1.25,0) [red!50!purple] {$G_{i}^{\ell_i}$};
\node at (1.25,0) [blue!80!black] {$G_{i}^{r_i}$};
\node at (-0.5,-1) [red!50!purple] {$G_{j}^{\ell_j}$};
\node at (0.5,-1) [blue!80!black] {$G_{j}^{r_j}$};
\node at (2.3,-1) {$\theta_j$};
\node at (0,1.9) {\Large $\tilde{G}_{j}^{a_j}$};
\end{scope}
\end{tikzpicture}
\caption{Left: Definition of $\theta_i$ from $\tilde{G}_{i}^{a_i}$. Right: Visualisation of $\tilde{G}_{i}^{a_i}$ as an embedded immersion of $\tilde{G}_{j}^{a_j}$.}
\label{pic_def_theta_curve}
\end{center}
\end{figure} 
 
Let us set $j > i$, then we have two embedded immersions from $\tilde{G}_{i}^{\ell_i}$ to $\tilde{G}_{j}^{\ell_j}$ and from $\tilde{G}_{i}^{r_i}$ to $\tilde{G}_{j}^{r_j}$. The sizes of cuts $|T_{i}(r_i)|$, $|T_{i}(\ell_i)|$, and $|T_{i}(a_i)|$ determine the sizes of $|T_{i}(\ell_i) \cap T_{i}(r_i)|$, $|T_{i}(a_i) \cap T_{i}(r_i)|$, and $|T_{i}(\ell_i) \cap T_{i}(a_i)|$. Since $|T_{i}(x_i)| = |T_{j}(x_j)|$ by construction for $x \in \{ \ell, r, a\}$, we conclude that the sizes of intersections of pairs of cuts also match between $\tilde{G}_{i}^i$ and $\tilde{G}_{j}^j$, \ie, $|T_{i}(x_i) \cap T_{i}(y_i)| = |T_{j}(x_j) \cap T_{j}(y_j)|$ for $x,y \in \{ \ell, r, a \}^2$. 

By Lemma~\ref{lem_equiv_plan_imm}, we have ordered immersions $\phi_\ell, \phi_r$ where the images of edges are families of tangent paths $P_{\ell,j}$ and $P_{r,j}$. For each edge $e$ in $T_{j}(\ell_j) \cap T_{j}(r_j)$, that is the edges met on the vertical part of $\theta_j$, there exists a path $p_{\ell,j}$ of the families yielded by Lemma~\ref{lem_equiv_plan_imm} within $G^{\ell_j}_{j}$ ending with $e$. Similarly, we have a path $p_{r,j}$ within $G^{r_j}_{j}$ reaching the same edge so that we can merge these paths to get a path $p_e$. These paths are tangent since their inner vertices are also inner vertices in $P_{\ell,j}$ and $P_{r,j}$, they are also tangent paths in $G_{j}^j$, this proves that $G_{i}^i$ is an embedded immersion of $G_{j}^j$.

Then, we identify the non-merged vertices $s_x$ of $\tilde{G}_{j}^{\ell_j}$ with their equivalent among the remaining non-merged vertices $s_y$ of $\tilde{G}_{j}^{a_j}$ on the proper arc of $\theta_j$. We do similarly with the ones remaining in $\tilde{G}_{j}^{r_j}$. We then have edge-disjoint paths joining $G^{a_j}_j$ to vertices $s_x$ on $\partial D_{a_j}$ which respect the order around the starting vertices of $G^{a_j}_j$. By Lemma~\ref{lem_equiv_plan_imm} $\tilde{G}_{i}^{a_i}$ is an embedded immersion of $\tilde{G}_{j}^{a_j}$. Hence, $a_i \preccurlyeq a_j$, which is a contradiction and concludes our proof.
\end{proof}

\section{Embedded graph minors.}
\label{sec_em_gb}

The aim of this sections is to prove that plane graphs are well-quasi-ordered by embedded minor, \ie, Theorem~\ref{th_embedded_gm}. 

A way to prove this theorem for the abstract setting is to prove that planar graphs with bounded width (either branchwidth or treewidth) are well-quasi-ordered by exploiting branch-decompositions, which is similar to what we did in Section~\ref{sec_bounded_cw}. For graphs with high treewidth, it can be done by proving two properties. The first is that any plane graph can be obtained as a minor of a big enough grid~\cite{Graph_Minors_VII, Robertson_excluding}. The second is that we can find grids of arbitrary size in a family of planar graphs with unbounded treewidth {\cite{Graph_Minors_XI}}.

Adapting our techniques from Section~\ref{sec_bounded_cw} to prove directly that plane graphs with bounded branch-width (respectively treewidth) are well-quasi-ordered by the embedded minor relation fails due to the fact that a cut induced by a branch-decomposition (respectively tree-decomposition) is a vertex-cut, \ie, a set of vertices whose removal renders the graph disconnected. In consequence, although it is possible to prove the existence of bond-linked branch-decomposition, the cut-cycle duality of cuts produced by carving-decomposition does not readily imply the disc property. We circumvent this issue by considering an oriented version of the embedded immersion applied to the medial digraph in Section~\ref{subsec_im_gm} and apply the methods that we developed earlier.

The case of unbounded width plane graphs follows from adaptations of existing theorem in the literature. First of all, it is quite convincing to see that plane graphs can be obtained as embedded minors of a grid: consider an embedding of a planar graph $G$, split vertices into trees of degree less than $4$, and push each edge on a thin enough grid. Contracting back the trees and deleting the unused part of the grid is the embedded minor. Furthermore, it is already known that walls can be found as subgraph of a graph with high width, and they have few distinct embeddings in $\R^2$. Hence, finding embedded grids as embedded minor require little extra work.

\subsection{Bounded branch-width.}
\label{subsec_im_gm}

The goal of this section is to prove that plane graphs with bounded branch-width are well-quasi-ordered by embedded minors, \ie, Proposition~\ref{prop_gm_bounded_bw}. 

\begin{proposition}\label{prop_gm_bounded_bw}
Plane graphs with bounded branch-width are well-quasi-ordered by embedded minors.
\end{proposition}

Similarly to last Section, with Lemma~\ref{lem_plane_equiv}, focusing on graphs embedded in $\Sp^2$ is enough.

\paragraph*{Embedded minor and embedded directed immersions.}
To prove Proposition~\ref{prop_gm_bounded_bw}, we first connect embedded minor with embedded immersions via medial graphs. This idea is not new, it was already introduced in~\cite{Medina_wqolink}, but Proposition 2.4 used there was in fact not proven for the embedded setting. Our new setting, which considers embeddings carefully allows us to prove this proposition. The carving-width was introduced by Seymour and Thomas in 1994~\cite{Seymour_Ratcatcher} and its interplay with medial graphs plays a major role in the theory developed there. Indeed, the famous ratcatcher algorithm introduced in this paper, which computes the branch-width of planar graphs in polynomial time, boils down to computing the carving-width of the associated medial graph:

\begin{proposition}{\cite[(7.2)]{Seymour_Ratcatcher}}\label{prop_bw_cw}
Let $G$ be a plane graph such that $|E(G)| \geq 2$, and $M$ be its medial graph. Then, $2 \bwid(G) = \cwid (G)$.
\end{proposition}

The \emphdef{medial graph} $M(G)$ of a cellularly embedded graph $G$ is an embedded $4$-regular graph on vertex set $E(G)$, where two vertex $e,e' \in V(M(G))$ are adjacent if $e,e'$ are consecutive on the boundary of some face $F$. A vertex $v_e$ of $M(G)$ can be placed on the relative interior of the embedded edge $e$, this vertex can be connected by a an arc to its neighbours in each face incident to $e$. An edge forming a self loop is incident to itself on the boundary of faces bounded by the self-loop. In addition, the two edges $e,e'$ of a face $F$ bounded by exactly two edges are considered consecutive twice, \ie, there are two edges between $v_e$ and $v_{e'}$ on $F$. Choosing the arcs connecting the vertices of $M(G)$ disjoint on each face yields an embedding of $M(G)$ associated to $G$ which is unique, see Figure~\ref{pic_def_medial} for a directed version of the medial graph.

\begin{figure}
\begin{center}
\begin{tikzpicture}
\clip (-1.45,-1.85) rectangle (9.5,3);
\def\e{0.05}

\coordinate (c1) at (0,0);
\coordinate (c2) at ($(c1)+(-20:1)$);
\coordinate (c3) at ($(c2)+(-40:1)$);
\coordinate (c4) at ($(c1)+(70:1)$);
\coordinate (c5) at ($(c4)+(-20:1)$);
\coordinate (c6) at ($(c5)+(20:1)$);
\coordinate (c7) at ($(c4)+(100:1)$);
\coordinate (c8) at ($(c7)+(0:1)$);

\draw (c1) -- (c2);
\draw (c2) -- (c3);
\draw (c1) -- (c4);
\draw (c4) -- (c5);
\draw (c2) -- (c5);
\draw (c5) -- (c6);
\draw (c4) -- (c7);
\draw (c5) -- (c8);
\draw (c4) -- (c8);
\draw (c7) -- (c8);

\draw (c1) .. controls +(-110:1.25) and \control{(c3)}{(-130:1.25)};
\draw (c1) .. controls +(160:2) and \control{(c7)}{(180:2)};
\draw (c4) .. controls +(135:1.25) and \control{(c4)}{(-135:1.25)};
\draw (c3) .. controls +(50:2.5) and \control{(c8)}{(0:3)};
\draw (c7) .. controls +(90:1) and \control{(c8)}{(90:1)};

\foreach \i in {1,...,8}{\fill (c\i) circle (\e cm);};
\node at (-0.4,-0.3) {$G$};

\draw [-Stealth] (3,0)  -- +(1,0);

\begin{scope}[xshift = 6cm]
\def\op{0.4}
\coordinate (c1) at (0,0);
\coordinate (c2) at ($(c1)+(-20:1)$);
\coordinate (c3) at ($(c2)+(-60:1)$);
\coordinate (c4) at ($(c1)+(70:1)$);
\coordinate (c5) at ($(c4)+(-20:1)$);
\coordinate (c6) at ($(c5)+(20:1)$);
\coordinate (c7) at ($(c4)+(100:1)$);
\coordinate (c8) at ($(c7)+(0:1)$);

\draw [opacity = \op] (c1) -- (c2) coordinate [pos = 0.5] (m1);
\draw [opacity = \op] (c2) -- (c3) coordinate [pos = 0.5] (m2);
\draw [opacity = \op] (c1) -- (c4) coordinate [pos = 0.5] (m3);
\draw [opacity = \op] (c4) -- (c5) coordinate [pos = 0.5] (m4);
\draw [opacity = \op] (c2) -- (c5) coordinate [pos = 0.5] (m5);
\draw [opacity = \op] (c5) -- (c6) coordinate [pos = 0.5] (m6);
\draw [opacity = \op] (c4) -- (c7) coordinate [pos = 0.5] (m7);
\draw [opacity = \op] (c5) -- (c8) coordinate [pos = 0.5] (m8);
\draw [opacity = \op] (c4) -- (c8) coordinate [pos = 0.5] (m9);
\draw [opacity = \op] (c7) -- (c8) coordinate [pos = 0.5] (m10);
\draw [opacity = \op] (c1) .. controls +(-110:1.25) and \control{(c3)}{(-130:1.25)} coordinate [pos = 0.5] (m11);
\draw [opacity = \op] (c1) .. controls +(160:2) and \control{(c7)}{(180:2)} coordinate [pos = 0.5] (m12);
\draw [opacity = \op] (c4) .. controls +(135:1.25) and \control{(c4)}{(-135:1.25)} coordinate [pos = 0.5] (c9) coordinate [pos = 0.5] (m13);
\draw [opacity = \op] (c3) .. controls +(50:2.5) and \control{(c8)}{(0:3)} coordinate [pos = 0.5] (m14);
\draw [opacity = \op] (c7) .. controls +(90:1) and \control{(c8)}{(90:1)}  coordinate [pos = 0.5] (m15);

\foreach \i in {1,...,8}{\fill [black!50] (c\i) circle (\e cm);};

\begin{scope}[blue!40!black]
\draw (m1) -- (m3) node [midway, sloped] {$\arrowIn$};
\draw (m3) -- (m4) node [midway, sloped] {$\arrowOut$};
\draw (m4) -- (m5) node [midway, sloped] {$\arrowOut$};
\draw (m5) -- (m1) node [midway, sloped] {$\arrowIn$};
\draw (m1) -- (m2) node [midway, sloped] {$\arrowOut$};
\draw (m2) -- (m11) node [midway, sloped] {$\arrowIn$};
\draw (m11) -- (m1) node [midway, sloped] {$\arrowOut$};
\draw (m7) -- (m9) node [midway, sloped] {$\arrowIn$};
\draw (m9) -- (m10) node [midway, sloped] {$\arrowOut$};
\draw (m10) -- (m7) node [midway, sloped] {$\arrowOuts{1.25}$};
\draw (m4) -- (m8) node [midway, sloped] {$\arrowIns{1.25}$};
\draw (m8) -- (m9) node [midway, sloped] {$\arrowOut$};
\draw (m9) -- (m4) node [midway, sloped] {$\arrowIn$};
\draw (m13) .. controls +(30:0.5) and \control{(m13)}{(-30:0.5)} node [pos = 0.1, sloped] {$\arrowOut$};
\draw (m13) .. controls +(170:0.5) and \control{(m7)}{(-170:0.3)} node [midway, sloped] {$\arrowIns{1.25}$};
\draw (m7) .. controls +(150:0.75) and \control{(m12)}{(40:0.75)} node [midway, sloped] {$\arrowOut$};
\draw (m12).. controls + (-50:0.75) and \control{(m3)}{(-150:0.75)} node [midway, sloped] {$\arrowIn$};
\draw (m3) .. controls +(170:0.3) and \control{(m13)}{(-170:0.5)} node [midway, sloped] {$\arrowOuts{1.25}$};
\draw (m10) .. controls +(135:0.35) and \control{(m15)}{(-135:0.35)} node [midway, sloped] {$\arrowOut$};
\draw (m15) .. controls +(-45:0.35) and \control{(m10)}{(45:0.35)} node [midway, sloped] {$\arrowOut$};
\draw (m2) -- (m5) node [midway, sloped] {$\arrowIns{1.25}$};
\draw (m5) -- (m6) node [midway, sloped] {$\arrowOuts{1.5}$};
\draw (m6) .. controls +(-10:1) and \control{(m6)}{(50:1)} node [pos = 0.8, sloped] {$\arrowIn$};
\draw (m6) -- (m8) node [midway, sloped] {$\arrowIns{1.5}$};
\draw (m8) .. controls + (60:0.75) and \control{(m14)}{(130:0.75)} node [midway, sloped] {$\arrowOut$};
\draw (m14) .. controls +(-120:1) and \control{(m2)}{(40:0.5)} node [midway, sloped] {$\arrowIns{1.25}$};
\draw (m12) .. controls +(-135:1) and \control{(m11)}{(-180:1)} node [midway, sloped] {$\arrowOuts{1.25}$};
\draw (m11) .. controls +(-60:1.5) and \control{(m14)}{(-70:1.5)} node [midway, sloped] {$\arrowOuts{1.25}$};
\draw (m14) .. controls +(45:1) and \control{(m15)}{(70:1)} node [midway, sloped] {$\arrowIns{1.25}$};
\draw (m15) .. controls +(110:1) and \control{(m12)}{(135:1)} node [midway, sloped] {$\arrowIns{1.25}$};
\end{scope}

\begin{scope}[blue!40!black, opacity = 0.125]
\fill (m1) -- (m2) -- (m5) -- cycle;
\fill (m4) -- (m5) -- (m6) -- (m8) -- cycle;
\fill (m6) .. controls +(-10:1) and \control{(m6)}{(50:1)};
\fill [even odd rule] (m3) .. controls +(170:0.3) and \control{(m13)}{(-170:0.5)} .. controls +(30:0.5) and \control{(m13)}{(-30:0.5)} .. controls +(170:0.5) and \control{(m7)}{(-170:0.3)} -- (m9) -- (m4) -- cycle;
\fill (m14) .. controls +(45:1) and \control{(m15)}{(70:1)} .. controls +(-45:0.35) and \control{(m10)}{(45:0.35)} -- (m9) -- (m8) .. controls + (60:0.75) and \control{(m14)}{(130:0.75)};
\fill (m15) .. controls +(110:1) and \control{(m12)}{(135:1)} .. controls + (40:0.75) and \control{(m7)}{(150:0.75)} -- (m10) .. controls + (135:0.35) and \control{(m15)}{(-135:0.35)};
\fill (m12) .. controls +(-135:1) and \control{(m11)}{(-180:1)} -- (m1) -- (m3) .. controls +(-150:0.75) and \control{(m12)}{(-50:0.75)};
\fill (m11) .. controls +(-60:1.5) and \control{(m14)}{(-70:1.5)} .. controls +(-120:1) and \control{(m2)}{(40:0.5)} -- (m11);
\end{scope}

\foreach \i in {1,...,15}{\fill [black] (m\i) circle (\e*1.5 cm); \fill [blue!80!black] (m\i) circle (\e cm);};

\node [black!50] at (-0.4,-0.3) {$G$};
\node [blue!40!black] at (-1.5,-1.2) {$DM(G)$};
\end{scope}
\end{tikzpicture}
\caption{The medial digraph $\dirmed(G)$ of a plane graph $G$, and a checkerboard colouring of its faces.}
\label{pic_def_medial}
\end{center}
\end{figure}

Note that since $G$ is plane, $M(G)$ is plane as well. Furthermore, overlapping the two embeddings spotlights that faces of $M(G)$ are bipartitioned as in a checkerboard colouring: either the face is disjoint from $G$, or it contains exactly a vertex from $G$. 

To establish a link between embedded minor and embedded immersion via graph minor, one needs to further restrain the operations allowed. Indeed, deleting an edge of the medial graph does not translate as an embedded operation on the associated plane graph. However, performing two embedded lifts on a vertex does. The additional restriction comes from~\cite{Medina_wqolink}: we orient the medial graphs and consider embedded immersion for digraph. 

The \emphdef{medial digraph} of a graph $G$ is $\dirmed(G)$, the plane medial graph of $G$ where edges along the boundary of faces containing vertices of $G$ are oriented clockwise. Since these faces form a checkerboard colouring of faces of $M(G)$, the edges along boundaries of faces of $M(G)$ not containing vertices of $G$ are oriented counter-clockwise as illustrated in Figure~\ref{pic_def_medial}. The medial digraph $\dirmed(G)$ is then a $2$-regular digraph: each vertex is incident to two incoming edges and two outgoing edges. In addition, they alternate by construction.

\begin{lemma}\label{lem_equiv_emb_minors}
Let $H$ and $G$ be two plane graphs. Then $H$ is an embedded minor of $G$ if and only if $\dirmed(H)$ is an embedded directed immersion of $\dirmed(G)$
\end{lemma}

\begin{proof}
A \emphdef{bilift} is performing two directed embedded lifts on the same vertex of $(G)$. First not that a bilift is equivalent to deleting the matching edge of $G$ (if the edges merged follow the faces containing $V(G)$), or contracting the matching edge of $G$ (if the edges merged follow the faces not containing $V(G)$), as pictured in Figure~\ref{pic_proof_equiv_emb_minor}. Deleting a non isolated vertex of $G$ can be done by first deleting all incident edges but one, that we contract afterwards. Deleting an entire connected component of one the graph is equivalent to deleting the entire matching component.

\begin{figure}[ht]
\begin{center}
\begin{tikzpicture}
\def\e{0.045} \def\eb{0.075}
\draw [black!50] (0,0) -- (1,0);
\fill [black!50] (0,0) circle (\e cm);
\fill [black!50] (1,0) circle (\e cm);

\begin{scope}[blue!40!black]
\draw (0.5,0) -- ($(0.5,0)+(150:0.75)$) node [midway, sloped] {$\arrowIns{1.25}$};
\draw (0.5,0) -- ($(0.5,0)+(-150:0.75)$) node [midway, sloped] {$\arrowOuts{1.25}$};
\draw (0.5,0) -- ($(0.5,0)+(30:0.75)$) node [midway, sloped] {$\arrowIns{1.25}$};
\draw (0.5,0) -- ($(0.5,0)+(-30:0.75)$) node [midway, sloped] {$\arrowOuts{1.25}$};
\draw [dotted] ($(0.5,0)+(150:0.75)$) -- ++(150:0.25);
\draw [dotted] ($(0.5,0)+(-150:0.75)$) -- ++(-150:0.25);
\draw [dotted] ($(0.5,0)+(30:0.75)$) -- ++(30:0.25);
\draw [dotted] ($(0.5,0)+(-30:0.75)$) -- ++(-30:0.25);
\fill (0.5,0) circle (\eb cm);
\fill [opacity=0.125] ($(0.5,0)+(150:0.75)$) -- (0.5,0) -- ($(0.5,0)+(-150:0.75)$) -- cycle;
\fill [opacity=0.125] ($(0.5,0)+(30:0.75)$) -- (0.5,0) -- ($(0.5,0)+(-30:0.75)$) -- cycle;
\end{scope}

\draw [Stealth-Stealth] (1.5,0) -- ++(0.75,0);

\begin{scope}[xshift = 3cm]
\fill [black!50] (0,0) circle (\e cm);
\fill [black!50] (1,0) circle (\e cm);

\begin{scope}[blue!40!black]
\draw ($(0.5,0)+(150:0.75)$) .. controls +(-30:0.75) and \control{($(0.5,0)+(-150:0.75)$)}{(30:0.75)} node [pos = 0.15, sloped] {$\arrowIns{1.25}$};
\draw [dotted] ($(0.5,0)+(150:0.75)$) -- ++(150:0.25);
\draw [dotted] ($(0.5,0)+(-150:0.75)$) -- ++(-150:0.25);
\draw [dotted] ($(0.5,0)+(30:0.75)$) -- ++(30:0.25);
\draw [dotted] ($(0.5,0)+(-30:0.75)$) -- ++(-30:0.25);
\draw ($(0.5,0)+(30:0.75)$) .. controls +(-150:0.75) and \control{($(0.5,0)+(-30:0.75)$)}{(150:0.75)} node [pos = 0.85, sloped] {$\arrowOuts{1.25}$};

\fill [opacity=0.125] ($(0.5,0)+(150:0.75)$) .. controls +(-30:0.75) and \control{($(0.5,0)+(-150:0.75)$)}{(30:0.75)};
\fill [opacity=0.125] ($(0.5,0)+(30:0.75)$) .. controls +(-150:0.75) and \control{($(0.5,0)+(-30:0.75)$)}{(150:0.75)};
\end{scope}
\end{scope}

\begin{scope}[xshift = 7cm]
\draw [black!50] (0,0) -- (1,0);
\fill [black!50] (0,0) circle (\e cm);
\fill [black!50] (1,0) circle (\e cm);

\begin{scope}[blue!40!black]
\draw (0.5,0) -- ($(0.5,0)+(150:0.75)$) node [pos = 0.5, sloped] {$\arrowIns{1.25}$};
\draw (0.5,0) -- ($(0.5,0)+(-150:0.75)$) node [pos = 0.5, sloped] {$\arrowOuts{1.25}$};
\draw (0.5,0) -- ($(0.5,0)+(30:0.75)$) node [pos = 0.5, sloped] {$\arrowIns{1.25}$};
\draw (0.5,0) -- ($(0.5,0)+(-30:0.75)$) node [pos = 0.5, sloped] {$\arrowOuts{1.25}$};
\draw [dotted] ($(0.5,0)+(150:0.75)$) -- ++(150:0.25);
\draw [dotted] ($(0.5,0)+(-150:0.75)$) -- ++(-150:0.25);
\draw [dotted] ($(0.5,0)+(30:0.75)$) -- ++(30:0.25);
\draw [dotted] ($(0.5,0)+(-30:0.75)$) -- ++(-30:0.25);
\fill (0.5,0) circle (\eb cm);
\fill [opacity=0.125] ($(0.5,0)+(150:0.75)$) -- (0.5,0) -- ($(0.5,0)+(-150:0.75)$) -- cycle;
\fill [opacity=0.125] ($(0.5,0)+(30:0.75)$) -- (0.5,0) -- ($(0.5,0)+(-30:0.75)$) -- cycle;

\begin{scope}[black!50]
\clip ($(0.5,0)+(-150:0.75)$) -- (0.5,0) -- ($(0.5,0)+(150:0.75)$);
\draw (0,0) -- ++(160:1);
\draw (0,0) -- ++(-170:1);
\end{scope}
\begin{scope}[even odd rule]
\clip ($(0.5,0)+(-150:1)$) rectangle ($(0.5,0)+(30:1)$) ($(0.5,0)+(-150:0.75)$) -- (0.5,0) -- ($(0.5,0)+(150:0.75)$);
\draw [dotted] (0,0) -- ++(160:1);
\draw [dotted] (0,0) -- ++(-170:1);
\end{scope}

\begin{scope}[black!50]
\clip ($(0.5,0)+(-30:0.75)$) -- (0.5,0) -- ($(0.5,0)+(30:0.75)$);
\draw (1,0) -- ++(0:1);
\end{scope}
\begin{scope}[even odd rule]
\clip ($(0.5,0)+(-150:1)$) rectangle ($(0.5,0)+(30:1)$) ($(0.5,0)+(-30:0.75)$) -- (0.5,0) -- ($(0.5,0)+(30:0.75)$);
\draw [dotted] (1,0) -- ++(0:1);
\end{scope}
\end{scope}

\draw [Stealth-Stealth] (1.5,0) -- ++(0.75,0);

\begin{scope}[xshift = 3cm]
\begin{scope}[blue!40!black]
\draw ($(0.5,0)+(150:0.75)$) .. controls +(-30:0.75) and \control{($(0.5,0)+(30:0.75)$)}{(-150:0.75)} node [pos = 0.15, sloped] {$\arrowIns{1.25}$};
\draw ($(0.5,0)+(-150:0.75)$) .. controls +(30:0.75) and \control{($(0.5,0)+(-30:0.75)$)}{(150:0.75)} node [pos = 0.85, sloped] {$\arrowOuts{1.25}$};
\draw [dotted] ($(0.5,0)+(150:0.75)$) -- ++(150:0.25);
\draw [dotted] ($(0.5,0)+(-150:0.75)$) -- ++(-150:0.25);
\draw [dotted] ($(0.5,0)+(30:0.75)$) -- ++(30:0.25);
\draw [dotted] ($(0.5,0)+(-30:0.75)$) -- ++(-30:0.25);

\fill [opacity = 0.125] ($(0.5,0)+(150:0.75)$) .. controls +(-30:0.75) and \control{($(0.5,0)+(30:0.75)$)}{(-150:0.75)} -- ($(0.5,0)+(-30:0.75)$) .. controls +(150:0.75) and \control{($(0.5,0)+(-150:0.75)$)}{(30:0.75)} -- cycle;
\end{scope}

\begin{scope}[black!50]
\clip ($(0.5,0)+(-150:0.75)$) -- (0.5,0) -- ($(0.5,0)+(150:0.75)$);
\draw (0.5,0) -- ++(160:2);
\draw (0.5,0) -- ++(-170:2);
\end{scope}
\begin{scope}[even odd rule]
\clip ($(0.5,0)+(-150:1)$) rectangle ($(0.5,0)+(30:1)$) ($(0.5,0)+(-150:0.75)$) -- (0.5,0) -- ($(0.5,0)+(150:0.75)$);
\draw [dotted] (0.5,0) -- ++(160:2);
\draw [dotted] (0.5,0) -- ++(-170:2);
\end{scope}

\begin{scope}[black!50]
\clip ($(0.5,0)+(-30:0.75)$) -- (0.5,0) -- ($(0.5,0)+(30:0.75)$);
\draw (0.5,0) -- ++(0:2);
\end{scope}
\begin{scope}[even odd rule]
\clip ($(0.5,0)+(-150:1)$) rectangle ($(0.5,0)+(30:1)$) ($(0.5,0)+(-30:0.75)$) -- (0.5,0) -- ($(0.5,0)+(30:0.75)$);
\draw [dotted] (0.5,0) -- ++(0:2);
\end{scope}

\fill [black!50] (0.5,0) circle (\e cm);
\end{scope}
\end{scope}
\end{tikzpicture}
\caption{Deleting or contracting a, edge of $G$ is equivalent to $2$ embedded directed lift on $DM(G)$.}
\label{pic_proof_equiv_emb_minor}
\end{center}
\end{figure}

Hence, there remains to prove that embedded directed immersions on medial digraphs can always be realised by bilifts only and deleting connected components. This is in fact mentioned in~\cite[Proposition 1.2]{Hannie_PhD}. Let us consider a connected medial digraph $M_1$ which is a directed embedded immersion of a connected medial digraph $M_2$. 

Since both graphs are directed $2$-regular graphs in which edges alternate, vertices of $M_2$ that are images of vertices of $M_1$ are exactly the vertices not taking part in any bilifts and for which no incident edges are deleted. The others take part in one of a bilift, a directed lift and two edge-deletions, or $4$ edge-deletions. Hence, by following edges to be deleted in $M_2$, we identify Eulerian components (components where vertices have as much incoming edges as outgoing ones) $C_1,\ldots,C_k$ which form the subgraph to delete in $M_2$.

\begin{figure}[ht]
\begin{center}
\begin{tikzpicture}[scale = 0.9]
\def\e{0.05} \def\l{0.4} \def\lb{1}
\coordinate (i1) at (0,0);
\coordinate (i2) at (1,0);
\coordinate (i3) at (1,1);
\coordinate (i4) at (0,1);
\coordinate (c1) at (-0.75,0.5);
\coordinate (c2) at (0.5,-0.75);
\coordinate (c3) at (1.75,0.5);
\coordinate (c4) at (0.5,1.75);

\foreach \i in {1,2,3,4}{\node [circle, inner sep = \e cm, fill] (ni\i) at (i\i) {};};
\foreach \i in {1,2,3,4}{\node [circle, inner sep = \e cm, rotate = 45, fill] (no\i) at (c\i) {};};

\draw (c1) .. controls +(-45:\l) and \control{(i1)}{(180:\l)} node [midway, sloped] {$\arrowIns{1.5}$};
\draw (i1) -- (i2) node [pos = 0.5, sloped] {$\arrowOuts{1}$};
\draw (i2) .. controls +(0:\l) and \control{(c3)}{(-135:\l)} node [midway, sloped] {$\arrowIns{1.5}$};
\draw (c3) .. controls +(135:\l) and \control{(i3)}{(0:\l)} node [midway, sloped] {$\arrowOuts{1.5}$};
\draw (i3) -- (i4) node [pos = 0.5, sloped] {$\arrowIns{1}$};
\draw (i4) .. controls +(180:\l) and \control{(c1)}{(45:\l)} node [midway, sloped] {$\arrowOuts{1.5}$};

\draw (c2) .. controls +(45:\l) and \control{(i2)}{(-90:\l)} node [midway, sloped] {$\arrowIns{1.5}$};
\draw (i2) -- (i3) node [pos = 0.5, sloped] {$\arrowOuts{1}$};
\draw (i3) .. controls +(90:\l) and \control{(c4)}{(-45:\l)} node [midway, sloped] {$\arrowOuts{1.5}$};
\draw (c4) .. controls +(-135:\l) and \control{(i4)}{(90:\l)} node [midway, sloped] {$\arrowOuts{1.5}$};
\draw (i4) -- (i1) node [pos = 0.5, sloped] {$\arrowOuts{1}$};
\draw (i1) .. controls +(-90:\l) and \control{(c2)}{(135:\l)} node [midway, sloped] {$\arrowIns{1.5}$};

\draw (c1) .. controls +(-135:\lb) and \control{(c1)}{(135:\lb)} node [pos = 0.9, sloped] {$\arrowOuts{1.25}$};
\draw (c2) .. controls +(-45:\lb) and \control{(c2)}{(-135:\lb)} node [pos = 0.9, sloped] {$\arrowOuts{1.25}$};
\draw (c3) .. controls +(45:\lb) and \control{(c3)}{(-45:\lb)} node [pos = 0.9, sloped] {$\arrowIns{1.25}$};
\draw (c4) .. controls +(135:\lb) and \control{(c4)}{(45:\lb)} node [pos = 0.9, sloped] {$\arrowIns{1.25}$};

\draw [black!60, -Stealth] (2.75,0.5) -- ++(1,0);

\begin{scope}[xshift = 5.5cm]
\def\e{0.065} \def\l{0.3} \def\lb{1} \def\ls{0.06}
\coordinate (i1) at (0,0);
\coordinate (i2) at (1,0);
\coordinate (i3) at (1,1);
\coordinate (i4) at (0,1);
\coordinate (c1) at (-0.75,0.5);
\coordinate (c2) at (0.5,-0.75);
\coordinate (c3) at (1.75,0.5);
\coordinate (c4) at (0.5,1.75);

\foreach \i in {1,2,3,4}{\node [circle, inner sep = \e cm] (ni\i) at (i\i) {};};
\foreach \i in {1,2,3,4}{\node [circle, inner sep = \e cm, rotate = 45] (nc\i) at (c\i) {};};

\draw (c1) .. controls +(-45:\l) and \control{(ni1.west)}{(180:\l)} node [midway, sloped] {$\arrowIns{1.5}$};
\draw (ni1.east) -- (ni2.west) node [pos = 0.5, sloped] {$\arrowOuts{1}$};
\draw (ni2.east) .. controls +(0:\l) and \control{(c3)}{(-135:\l)} node [midway, sloped] {$\arrowIns{1.5}$};
\draw (c3) .. controls +(135:\l) and \control{(ni3.east)}{(0:\l)} node [midway, sloped] {$\arrowOuts{1.5}$};
\draw (ni3.west) -- (ni4.east) node [pos = 0.5, sloped] {$\arrowIns{1}$};
\draw (ni4.west) .. controls +(180:\l) and \control{(c1)}{(45:\l)} node [midway, sloped] {$\arrowOuts{1.5}$};

\draw (nc2.east) .. controls +(45:\l) and \control{(ni2.south)}{(-90:\l)} node [midway, sloped] {$\arrowIns{1.5}$};
\draw (ni2.north) -- (ni3.south) node [pos = 0.5, sloped] {$\arrowOuts{1}$};
\draw (ni3.north) .. controls +(90:\l) and \control{(nc4.south)}{(-45:\l)} node [midway, sloped] {$\arrowOuts{1.5}$};
\draw (nc4.west) .. controls +(-135:\l) and \control{(ni4.north)}{(90:\l)} node [midway, sloped] {$\arrowOuts{1.5}$};
\draw (ni4.south) -- (ni1.north) node [pos = 0.5, sloped] {$\arrowOuts{1}$};
\draw (ni1.south) .. controls +(-90:\l) and \control{(nc2.north)}{(135:\l)} node [midway, sloped] {$\arrowIns{1.5}$};

\draw (c1) .. controls +(-135:\lb) and \control{(c1)}{(135:\lb)} node [pos = 0.9, sloped] {$\arrowOuts{1.25}$};
\draw (nc2.south) .. controls +(-45:\lb) and \control{(nc2.west)}{(-135:\lb)} node [pos = 0.9, sloped] {$\arrowOuts{1.25}$};
\draw (c3) .. controls +(45:\lb) and \control{(c3)}{(-45:\lb)} node [pos = 0.9, sloped] {$\arrowIns{1.25}$};
\draw (nc4.north) .. controls +(135:\lb) and \control{(nc4.east)}{(45:\lb)} node [pos = 0.9, sloped] {$\arrowIns{1.25}$};

\draw (ni1.west) .. controls +(0:\ls) and \control{(ni1.south)}{(90:\ls)};
\draw (ni1.east) .. controls +(180:\ls) and \control{(ni1.north)}{(-90:\ls)};
\draw (ni2.east) .. controls +(180:\ls) and \control{(ni2.south)}{(90:\ls)};
\draw (ni2.west) .. controls +(0:\ls) and \control{(ni2.north)}{(-90:\ls)};
\draw (ni3.east) .. controls +(180:\ls) and \control{(ni3.south)}{(90:\ls)};
\draw (ni3.west) .. controls +(0:\ls) and \control{(ni3.north)}{(-90:\ls)};
\draw (ni4.east) .. controls +(180:\ls) and \control{(ni4.south)}{(90:\ls)};
\draw (ni4.west) .. controls +(0:\ls) and \control{(ni4.north)}{(-90:\ls)};
\draw (nc2.east) .. controls +(-135:\ls) and \control{(nc2.south)}{(135:\ls)};
\draw (nc2.west) .. controls +(45:\ls) and \control{(nc2.north)}{(-45:\ls)};
\draw (nc4.east) .. controls +(-135:\ls) and \control{(nc4.south)}{(135:\ls)};
\draw (nc4.west) .. controls +(45:\ls) and \control{(nc4.north)}{(-45:\ls)};

\fill (c1) circle (\e cm);
\fill (c3) circle (\e cm);

\draw [black!60, -Stealth] (2.75,0.5) -- ++(1,0);
\end{scope}

\begin{scope}[xshift = 11cm]
\def\e{0.065}
\coordinate (c1) at (-0.75,0.5);
\coordinate (c3) at (1.75,0.5);

\draw (c1) .. controls +(-45:1) and \control{(c3)}{(-135:1)} node [pos = 0.5, sloped] {$\arrowIns{1.25}$};
\draw (c1) .. controls +(45:1) and \control{(c3)}{(135:1)} node [pos = 0.5, sloped] {$\arrowOuts{1.25}$};

\draw (c1) .. controls +(-135:\lb) and \control{(c1)}{(135:\lb)} node [pos = 0.9, sloped] {$\arrowOuts{1.25}$};
\draw (c3) .. controls +(45:\lb) and \control{(c3)}{(-45:\lb)} node [pos = 0.9, sloped] {$\arrowIns{1.25}$};
\fill (c1) circle (\e cm);
\fill (c3) circle (\e cm);
\end{scope}
\end{tikzpicture}
\caption{Embedded directed immersions in medial graphs can be realised by bilifts.}
\label{pic_med_bilift}
\end{center}
\end{figure}

First perform all bilifts on vertices for which all incident edges are part of paths that are images of edges of $M_1$, this does not affect the $C_i$'s. Then, let us consider a component $C_i$, it is a directed cycle by construction. Since $M_2$ is connected, $C_i$ necessarily has a vertex $v$ with exactly one incoming edge and one outgoing one. Follow this cycle starting at $v$ and perform the appropriate embedded directed lift. After this, $C_i$ is a self-loop at $v$. Performing the proper bilift on $v$ merges it with the two other edges incident to $v$: $C_i$ was ``deleted'' with embedded directed lifts only. After all components were removed via this process, no vertices of edge to delete remain, and all vertices of $M_2$ are of degree $4$, $M_2$ is now equivalent to $M_1$ (see Figure~\ref{pic_med_bilift} for example). All vertices which are not images of vertices of $M_1$ where part of two embedded direct lifts.
\end{proof}

Thus, by Lemma~\ref{lem_equiv_emb_minors} and Proposition~\ref{prop_bw_cw}, proving Proposition~\ref{prop_gm_bounded_bw} boils down to proving:

\begin{proposition}\label{prop_imm_emb_dirmed}
Plane medial digraphs are well-quasi ordered by embedded directed immersion.
\end{proposition}

To prove this we exploit disc-linked decompositions of Section~\ref{sec_bounded_cw}, and follow the same proof method, \ie, a Nash-Williams' argument on the order put on edges of disc-linked carving-decompositions. In the following we essentially redefine this order and reprove Lemma~\ref{lem_ancestry_cut} and Theorem~\ref{th_emb_immersion_bounded} in this context.

\paragraph*{Order on edges.}
The size of a cut in the directed and undirected context is the same, as well as satisfying the disc property. A disc-linked carving-decomposition for a graph $M(G)$ is also a disc-linked carving-decomposition of its directed version $\dirmed(G)$. 

Let $(G_n)_{n \in \N}$ be a family of plane medial digraphs. Thanks to Proposition~\ref{prop_disc_decomp}, for all $n \in \N$, there exists $T_n$ a disc-linked carving-decomposition of $G_n$ that we can assume rooted at some unlabelled edge $r(T_n)$. We define $\T = \bigcup_{n \in \N} T_n$ as the forest made of all the trees $T_n$. 

Let $G$ be a plane digraph. For each $e \in E(T)$, the \emphdef{leaving digraph} $\tilde{G^e}$, is the directed version of the leaving graph defined in Section~\ref{sec_bounded_cw}: by the disc property there exist a disc $D_e$ intersecting $G$ on $G^e \cup T(e)$ which intersects each edge of $T(e)$ once on $\partial D_e$. We subdivide each edge $f$ of $T(e)$ at the place of intersection with $\partial D_e$ with a vertex $s_x$. The embedded digraph induced by $V(G^e) \cup_{x \in T(e)} \{ s_x \}$ on $D_e$ is $\tilde{G}^e$, see Figure~\ref{pic_ex_leaving_digraph} for an example. It is a digraph embedded on $D_e$ such that vertices on the inside of $D_e$ are $2$-regular and edges around them alternate, and vertices on the boundary have only one incident edge.

\begin{figure}[ht]
\begin{center}
\begin{tikzpicture}
\clip (-1.5,-2.35) rectangle (9.75,3);
\def\e{0.05}
\coordinate (c1) at (0,0);
\coordinate (c2) at ($(c1)+(-20:1)$);
\coordinate (c3) at ($(c2)+(-60:1)$);
\coordinate (c4) at ($(c1)+(70:1)$);
\coordinate (c5) at ($(c4)+(-20:1)$);
\coordinate (c6) at ($(c5)+(20:1)$);
\coordinate (c7) at ($(c4)+(100:1)$);
\coordinate (c8) at ($(c7)+(0:1)$);

\path (c1) -- (c2) coordinate [pos = 0.5] (m1);
\path (c2) -- (c3) coordinate [pos = 0.5] (m2);
\path (c1) -- (c4) coordinate [pos = 0.5] (m3);
\path (c4) -- (c5) coordinate [pos = 0.5] (m4);
\path (c2) -- (c5) coordinate [pos = 0.5] (m5);
\path (c4) -- (c7) coordinate [pos = 0.5] (m7);
\path (c5) -- (c8) coordinate [pos = 0.5] (m8);
\path (c4) -- (c8) coordinate [pos = 0.5] (m9);
\path (c7) -- (c8) coordinate [pos = 0.5] (m10);
\path (c1) .. controls +(-110:1.25) and \control{(c3)}{(-130:1.25)} coordinate [pos = 0.5] (m11);
\path (c1) .. controls +(160:1.5) and \control{(c7)}{(180:1.5)} coordinate [pos = 0.5] (m12);
\path (c3) .. controls +(50:1.5) and \control{(c8)}{(0:2.5)} coordinate [pos = 0.5] (m14);
\path (c7) .. controls +(90:1) and \control{(c8)}{(90:1)}  coordinate [pos = 0.5] (m15);

\begin{scope}[blue!40!black]
\draw (m1) -- (m3) node [midway, sloped] {$\arrowIns{1.25}$};
\draw (m3) -- (m4) node [midway, sloped] {$\arrowOut$};
\draw (m4) -- (m5) node [pos = 0.6, sloped] {$\arrowOuts{1.25}$};
\draw (m5) -- (m1) node [midway, sloped] {$\arrowIns{1.25}$};
\draw (m1) -- (m2) node [midway, sloped] {$\arrowOuts{1.5}$};
\draw (m2) -- (m11) node [midway, sloped] {$\arrowIns{1.4}$};
\draw (m11) -- (m1) node [midway, sloped] {$\arrowOut$};
\draw (m7) -- (m9) node [midway, sloped] {$\arrowIn$};
\draw (m9) -- (m10) node [midway, sloped] {$\arrowOut$};
\draw (m10) -- (m7) node [midway, sloped] {$\arrowOuts{1.5}$};
\draw (m4) -- (m8) node [midway, sloped] {$\arrowIns{1.25}$};
\draw (m8) -- (m9) node [midway, sloped] {$\arrowOuts{1.25}$};
\draw (m9) -- (m4) node [midway, sloped] {$\arrowIn$};
\draw (m7) .. controls +(150:0.75) and \control{(m12)}{(40:0.75)} node [midway, sloped] {$\arrowOuts{1.25}$};
\draw (m12).. controls + (-50:0.75) and \control{(m3)}{(-150:0.75)} node [midway, sloped] {$\arrowIns{1.5}$};
\draw (m3) .. controls +(170:0.3) and \control{(m7)}{(-170:0.5)} node [midway, sloped] {$\arrowOuts{1.25}$};
\draw (m10) .. controls +(135:0.35) and \control{(m15)}{(-135:0.35)} node [midway, sloped] {$\arrowOut$};
\draw (m15) .. controls +(-45:0.35) and \control{(m10)}{(45:0.35)} node [midway, sloped] {$\arrowOut$};
\draw (m2) -- (m5) node [midway, sloped] {$\arrowIns{1.25}$};
\draw (m5) .. controls +(40:0.65) and \control{(m8)}{(-40:0.65)} node [midway, sloped] {$\arrowOuts{1.5}$};
\draw (m8) .. controls + (60:0.75) and \control{(m14)}{(130:0.75)} node [midway, sloped] {$\arrowOuts{1.25}$};
\draw (m14) .. controls +(-120:1) and \control{(m2)}{(40:0.5)} node [midway, sloped] {$\arrowIns{1.5}$};
\draw (m12) .. controls +(-135:1) and \control{(m11)}{(-180:1)} node [pos = 0.6, sloped] {$\arrowOuts{1.5}$};
\draw (m11) .. controls +(-60:1.5) and \control{(m14)}{(-70:1.5)} node [midway, sloped] {$\arrowOuts{1.35}$};
\draw (m14) .. controls +(45:1) and \control{(m15)}{(70:1)} node [midway, sloped] {$\arrowIns{1.5}$};
\draw (m15) .. controls +(110:1) and \control{(m12)}{(135:1)} node [midway, sloped] {$\arrowIns{1.5}$};
\end{scope}

\draw [green!40!black] (-1.25,-1.25) .. controls +(90:1.5) and \control{(1.25,0.75)}{(-150:1.5)} .. controls +(30:1) and \control{(2.675,1.5)}{(180:0.75)} .. controls + (0:0.75) and \control{(3,-1)}{(90:1)} .. controls +(-90:1.5) and \control{(-1.25,-1.25)}{(-90:1.5)};

\foreach \i in {1,...,5}{\fill [black] (m\i) circle (\e*1.5 cm); \fill [blue!80!black] (m\i) circle (\e cm);};
\foreach \i in {7,...,12}{\fill [black] (m\i) circle (\e*1.5 cm); \fill [blue!80!black] (m\i) circle (\e cm);};
\foreach \i in {14,15}{\fill [black] (m\i) circle (\e*1.5 cm); \fill [blue!80!black] (m\i) circle (\e cm);};

\node [green!30!black] at (3.3,-1) {$\gamma_e$};
\node [blue!40!black] at (-0.25,2.25) {$G$};

\draw [-Stealth] (3.75,0) -- +(1,0);

\begin{scope}[xshift=6.5cm, yshift = 0.5cm]
\def\e{0.05}
\coordinate (c1) at (0,0);
\coordinate (c2) at ($(c1)+(-20:1)$);
\coordinate (c3) at ($(c2)+(-60:1)$);
\coordinate (c4) at ($(c1)+(70:1)$);
\coordinate (c5) at ($(c4)+(-20:1)$);
\coordinate (c6) at ($(c5)+(20:1)$);
\coordinate (c7) at ($(c4)+(100:1)$);
\coordinate (c8) at ($(c7)+(0:1)$);

\path (c1) -- (c2) coordinate [pos = 0.5] (m1);
\path (c2) -- (c3) coordinate [pos = 0.5] (m2);
\path (c1) -- (c4) coordinate [pos = 0.5] (m3);
\path (c4) -- (c5) coordinate [pos = 0.5] (m4);
\path (c2) -- (c5) coordinate [pos = 0.5] (m5);
\path (c4) -- (c7) coordinate [pos = 0.5] (m7);
\path (c5) -- (c8) coordinate [pos = 0.5] (m8);
\path (c4) -- (c8) coordinate [pos = 0.5] (m9);
\path (c7) -- (c8) coordinate [pos = 0.5] (m10);
\path (c1) .. controls +(-110:1.25) and \control{(c3)}{(-130:1.25)} coordinate [pos = 0.5] (m11);
\path (c1) .. controls +(160:1.5) and \control{(c7)}{(180:1.5)} coordinate [pos = 0.5] (m12);
\path (c3) .. controls +(50:1.5) and \control{(c8)}{(0:2.5)} coordinate [pos = 0.5] (m14);
\path (c7) .. controls +(90:1) and \control{(c8)}{(90:1)}  coordinate [pos = 0.5] (m15);

\begin{scope}[blue!40!black]
\clip (-1.25,-1.25) .. controls +(90:1.5) and \control{(1.25,0.75)}{(-150:1.5)} .. controls +(30:1) and \control{(2.675,1.5)}{(180:0.75)} .. controls + (0:0.75) and \control{(3,-1)}{(90:1)} .. controls +(-90:1.5) and \control{(-1.25,-1.25)}{(-90:1.5)};
\draw (m1) -- (m3) node [pos =0.4, sloped] {$\arrowIns{1.25}$};
\draw (m3) -- (m4) node [midway, sloped] {$\arrowOut$};
\draw (m4) -- (m5) node [pos = 0.6, sloped] {$\arrowOuts{1.25}$};
\draw (m5) -- (m1) node [midway, sloped] {$\arrowIns{1.25}$};
\draw (m1) -- (m2) node [midway, sloped] {$\arrowOuts{1.5}$};
\draw (m2) -- (m11) node [midway, sloped] {$\arrowIns{1.4}$};
\draw (m11) -- (m1) node [midway, sloped] {$\arrowOut$};
\draw (m7) -- (m9) node [midway, sloped] {$\arrowIn$};
\draw (m9) -- (m10) node [midway, sloped] {$\arrowOut$};
\draw (m10) -- (m7) node [midway, sloped] {$\arrowOuts{1.25}$};
\draw (m4) -- (m8) node [midway, sloped] {$\arrowIns{1.25}$};
\draw (m8) -- (m9) node [midway, sloped] {$\arrowOut$};
\draw (m9) -- (m4) node [midway, sloped] {$\arrowIn$};
\draw (m7) .. controls +(150:0.75) and \control{(m12)}{(40:0.75)} node [midway, sloped] {$\arrowOut$};
\draw (m12).. controls + (-50:0.75) and \control{(m3)}{(-150:0.75)} node [midway, sloped] {$\arrowIn$};
\draw (m3) .. controls +(170:0.3) and \control{(m7)}{(-170:0.5)} node [midway, sloped] {$\arrowOuts{1.25}$};
\draw (m10) .. controls +(135:0.35) and \control{(m15)}{(-135:0.35)} node [midway, sloped] {$\arrowOut$};
\draw (m15) .. controls +(-45:0.35) and \control{(m10)}{(45:0.35)} node [midway, sloped] {$\arrowOut$};
\draw (m2) -- (m5) node [midway, sloped] {$\arrowIns{1.25}$};
\draw (m5) .. controls +(40:0.65) and \control{(m8)}{(-40:0.65)} node [pos=0.35, sloped] {$\arrowOuts{1.5}$};
\draw (m8) .. controls + (60:0.75) and \control{(m14)}{(130:0.75)} node [pos = 0.8, sloped] {$\arrowOuts{1.5}$};
\draw (m14) .. controls +(-120:1) and \control{(m2)}{(40:0.5)} node [midway, sloped] {$\arrowIns{1.5}$};
\draw (m12) .. controls +(-135:1) and \control{(m11)}{(-180:1)} node [pos = 0.6, sloped] {$\arrowOuts{1.5}$};
\draw (m11) .. controls +(-60:1.5) and \control{(m14)}{(-70:1.5)} node [midway, sloped] {$\arrowOuts{1.35}$};
\draw (m14) .. controls +(45:1) and \control{(m15)}{(70:1)} node [pos = 0.15, sloped] {$\arrowIns{1.15}$};
\draw (m15) .. controls +(110:1) and \control{(m12)}{(135:1)} node [midway, sloped] {$\arrowIns{1.25}$};
\end{scope}

\path [name path =e1] (m4) -- (m5);
\path [name path =e2] (m1) -- (m3);
\path [name path =e3] (m12) .. controls +(-135:1) and \control{(m11)}{(-180:1)};
\path [name path =e4] (m14) .. controls +(45:1) and \control{(m15)}{(70:1)};
\path [name path =e5] (m5) .. controls +(40:0.65) and \control{(m8)}{(-40:0.65)};
\path [name path =e6] (m8) .. controls + (60:0.75) and \control{(m14)}{(130:0.75)};

\draw [green!40!black, name path= blob] (-1.25,-1.25) .. controls +(90:1.5) and \control{(1.25,0.75)}{(-150:1.5)} .. controls +(30:1) and \control{(2.675,1.5)}{(180:0.75)} .. controls + (0:0.75) and \control{(3,-1)}{(90:1)} .. controls +(-90:1.5) and \control{(-1.25,-1.25)}{(-90:1.5)};

\foreach \i in {1,...,6}{
\path [name intersections={of=blob and e\i,total=\tot}]
\foreach \s in {1,...,\tot}{coordinate (n\i) at (intersection-\s)}; 
\fill [green!30!black] (n\i) circle (\e*1.35 cm);
};

\fill [opacity =0.1, green] (-1.25,-1.25) .. controls +(90:1.5) and \control{(1.25,0.75)}{(-150:1.5)} .. controls +(30:1) and \control{(2.675,1.5)}{(180:0.75)} .. controls + (0:0.75) and \control{(3,-1)}{(90:1)} .. controls +(-90:1.5) and \control{(-1.25,-1.25)}{(-90:1.5)};

\foreach \i in {1,2,5,11,14}{\fill [black] (m\i) circle (\e*1.5 cm); \fill [blue!80!black] (m\i) circle (\e cm);};

\node [green!30!black] at (-0.7,-1.5) {$D_e$};
\node [blue!40!black] at (-0.25,0.5) {$\tilde{G^e}$};
\end{scope}
\end{tikzpicture}
\caption{How to obtain $\tilde{G}^e$ from $G$, the vertices $(s_x)_{e \in T(e)}$ are in green on $\partial D_e$.}
\label{pic_ex_leaving_digraph}
\end{center}
\end{figure} 

We define a quasi-order on the set $\mathcal{E} = \cup_{T \in \T}{E(T)}$ of all edges of $\T$: for $e,f \in \mathcal{E}^2$, $e \preccurlyeq f$ if $\tilde{G}^e$ is an embedded directed immersion of $\tilde{G}^f$ and the image of $(s_x)_{x \in T(e)}$ in the embedded immersion is $(s_x)_{x \in T(f)}$. 

We now prove the equivalent of Lemma~\ref{lem_ancestry_cut} for medial digraphs, this relies essentially on a version of Menger's theorem for digraphs.

\begin{lemma}\label{lem_ancestry_cut_dig}
Let $G$ be a plane medial digraph, $(T,\phi)$ a disc-linked carving-decomposition of $G$, and $a,b \in E(T)^2$ such that $a$ is an ancestor of $b$. Then $b \preccurlyeq a$.
\end{lemma}

\begin{proof}
Let $D_e$ be the disc associated to an edge $e$ of $T$, \ie, the component delimited by $\gamma_e$ containing $G^e$. The curve $\gamma_ e$ intersects incoming edges, \ie, edges of $T(e)$ with head in $D_e$ and outgoing one when their tail is there instead. In between two consecutive edges, $\gamma$ lies on the face bounded by these two edges. However, since they belong to the boundary of the same face, which are oriented consistently in that regard by definition of plane medial digraph, one of them is incoming and the other one is outgoing. Hence, edges met by $\gamma$ alternate (see $\gamma_e$ on Figure~\ref{pic_ex_leaving_digraph} for example). It follows that cuts of $G$ are of even size, and contain as many incoming edges with respect to $T(e)$ as outgoing ones. For the sake of clarity, we denote in this proof $\cut(A \rightarrow B)$ the set of oriented edges of $G$ with tail in $A$ and head in $B$. With this notation, for every subset $X$ of the directed graph, $|\cut(X \rightarrow X^c)| = |\cut(X^c \rightarrow X)|$

The $\preccurlyeq$ relation between $b$ and $a$ can be defined since the carving-decomposition satisfies the disc property. We denote $2k = \wid (a) = \wid (b)$. Let assume that $G^a_1 = G^a$ so that $A = V(G^a_2)$ and $B = V(G^b)$ with respect to the linked definition of $A$ and $B$. Let us call $P$ the minimal path of $T$ containing both $a$ and $b$. We define $H = G \smallsetminus E(G^a_2) \smallsetminus E(G^b)$, in other words, the edges of $G$ between $T(a)$ and $T(b)$ (where $T(a)$ and $T(b)$ are included). By construction, there are two disjoint families $\mathcal{F}_a = V(H) \cap V(A)$ and $\mathcal{F}_b = V(H) \cap V(B)$, where $\mathcal{F}_a$ consists of the endpoints of $T(a)$ in $G^a_2$ and $\mathcal{F}_b$ consists of the endpoints of $T(b)$ in $G^b$. 

\begin{figure}[ht]
\begin{center}
\begin{tikzpicture}
\clip (-1.5,-2.35) rectangle (9.75,3.5);
\def\e{0.05}
\coordinate (c1) at (0,0);
\coordinate (c2) at ($(c1)+(-20:1)$);
\coordinate (c3) at ($(c2)+(-60:1)$);
\coordinate (c4) at ($(c1)+(70:1)$);
\coordinate (c5) at ($(c4)+(-20:1)$);
\coordinate (c6) at ($(c5)+(20:1)$);
\coordinate (c7) at ($(c4)+(100:1)$);
\coordinate (c8) at ($(c7)+(0:1)$);

\path (c1) -- (c2) coordinate [pos = 0.5] (m1);
\path (c2) -- (c3) coordinate [pos = 0.5] (m2);
\path (c1) -- (c4) coordinate [pos = 0.5] (m3);
\path (c4) -- (c5) coordinate [pos = 0.5] (m4);
\path (c2) -- (c5) coordinate [pos = 0.5] (m5);
\path (c4) -- (c7) coordinate [pos = 0.5] (m7);
\path (c5) -- (c8) coordinate [pos = 0.5] (m8);
\path (c4) -- (c8) coordinate [pos = 0.5] (m9);
\path (c7) -- (c8) coordinate [pos = 0.5] (m10);
\path (c1) .. controls +(-110:1.25) and \control{(c3)}{(-130:1.25)} coordinate [pos = 0.5] (m11);
\path (c1) .. controls +(160:1.5) and \control{(c7)}{(180:1.5)} coordinate [pos = 0.5] (m12);
\path (c3) .. controls +(50:1.5) and \control{(c8)}{(0:2.5)} coordinate [pos = 0.5] (m14);
\path (c7) .. controls +(90:1) and \control{(c8)}{(90:1)}  coordinate [pos = 0.5] (m15);

\begin{scope}[blue!40!black]
\draw (m1) -- (m3) node [midway, sloped] {$\arrowIns{1.25}$};
\draw (m3) -- (m4) node [midway, sloped] {$\arrowOut$};
\draw (m4) -- (m5) node [pos = 0.6, sloped] {$\arrowOuts{1.25}$};
\draw (m5) -- (m1) node [midway, sloped] {$\arrowIns{1.25}$};
\draw (m1) -- (m2) node [midway, sloped] {$\arrowOuts{1.5}$};
\draw (m2) -- (m11) node [midway, sloped] {$\arrowIns{1.4}$};
\draw (m11) -- (m1) node [midway, sloped] {$\arrowOut$};
\draw (m7) -- (m9) node [midway, sloped] {$\arrowIn$};
\draw (m9) -- (m10) node [midway, sloped] {$\arrowOut$};
\draw (m10) -- (m7) node [midway, sloped] {$\arrowOuts{1.5}$};
\draw (m4) -- (m8) node [midway, sloped] {$\arrowIns{1.25}$};
\draw (m8) -- (m9) node [midway, sloped] {$\arrowOuts{1.25}$};
\draw (m9) -- (m4) node [midway, sloped] {$\arrowIn$};
\draw (m7) .. controls +(150:0.75) and \control{(m12)}{(40:0.75)} node [midway, sloped] {$\arrowOuts{1.25}$};
\draw (m12).. controls + (-50:0.75) and \control{(m3)}{(-150:0.75)} node [midway, sloped] {$\arrowIns{1.5}$};
\draw (m3) .. controls +(170:0.3) and \control{(m7)}{(-170:0.5)} node [midway, sloped] {$\arrowOuts{1.25}$};
\draw (m10) .. controls +(135:0.35) and \control{(m15)}{(-135:0.35)} node [midway, sloped] {$\arrowOut$};
\draw (m15) .. controls +(-45:0.35) and \control{(m10)}{(45:0.35)} node [midway, sloped] {$\arrowOut$};
\draw (m2) -- (m5) node [midway, sloped] {$\arrowIns{1.25}$};
\draw (m5) .. controls +(40:0.65) and \control{(m8)}{(-40:0.65)} node [midway, sloped] {$\arrowOuts{1.5}$};
\draw (m8) .. controls + (60:0.75) and \control{(m14)}{(130:0.75)} node [midway, sloped] {$\arrowOuts{1.25}$};
\draw (m14) .. controls +(-120:1) and \control{(m2)}{(40:0.5)} node [midway, sloped] {$\arrowIns{1.5}$};
\draw (m12) .. controls +(-135:1) and \control{(m11)}{(-180:1)} node [pos = 0.6, sloped] {$\arrowOuts{1.5}$};
\draw (m11) .. controls +(-60:1.5) and \control{(m14)}{(-70:1.5)} node [midway, sloped] {$\arrowOuts{1.35}$};
\draw (m14) .. controls +(45:1) and \control{(m15)}{(70:1)} node [midway, sloped] {$\arrowIns{1.5}$};
\draw (m15) .. controls +(110:1) and \control{(m12)}{(135:1)} node [midway, sloped] {$\arrowIns{1.5}$};
\end{scope}

\draw [green!40!black] (-1.25,-1.25) .. controls +(90:1.5) and \control{(1.25,0.75)}{(-150:1.5)} .. controls +(30:1) and \control{(2.675,1.5)}{(180:0.75)} .. controls + (0:0.75) and \control{(3,-1)}{(90:1)} .. controls +(-90:1.5) and \control{(-1.25,-1.25)}{(-90:1.5)};

\draw [red!50!purple] (-1.25,1.25) .. controls +(-90:0.5) and \control{(-0.75,1.25)}{(-90:0.5)} .. controls +(90:1) and \control{(1,2.1)}{(180:1)} .. controls +(0:0.5) and \control{(1,3.25)}{(0:0.5)} .. controls +(180:1) and \control{(-1.25,1.25)}{(90:1.5)};

\foreach \i in {1,...,5}{\fill [black] (m\i) circle (\e*1.5 cm); \fill [blue!80!black] (m\i) circle (\e cm);};
\foreach \i in {7,...,12}{\fill [black] (m\i) circle (\e*1.5 cm); \fill [blue!80!black] (m\i) circle (\e cm);};
\foreach \i in {14,15}{\fill [black] (m\i) circle (\e*1.5 cm); \fill [blue!80!black] (m\i) circle (\e cm);};

\node [red!50!purple] at (-1.2,2.5) {$\gamma_a$};
\node [green!30!black] at (3.3,-1) {$\gamma_b$};
\node [blue!40!black] at (2.5,3) {$G$};

\draw [-Stealth] (3.75,0) -- +(1,0);

\begin{scope}[xshift=6.5cm, yshift = 0.25cm]
\coordinate (c1) at (0,0);
\coordinate (c2) at ($(c1)+(-20:1)$);
\coordinate (c3) at ($(c2)+(-60:1)$);
\coordinate (c4) at ($(c1)+(70:1)$);
\coordinate (c5) at ($(c4)+(-20:1)$);
\coordinate (c6) at ($(c5)+(20:1)$);
\coordinate (c7) at ($(c4)+(100:1)$);
\coordinate (c8) at ($(c7)+(0:1)$);

\path (c1) -- (c2) coordinate [pos = 0.5] (m1);
\path (c2) -- (c3) coordinate [pos = 0.5] (m2);
\path (c1) -- (c4) coordinate [pos = 0.5] (m3);
\path (c4) -- (c5) coordinate [pos = 0.5] (m4);
\path (c2) -- (c5) coordinate [pos = 0.5] (m5);
\path (c4) -- (c7) coordinate [pos = 0.5] (m7);
\path (c5) -- (c8) coordinate [pos = 0.5] (m8);
\path (c4) -- (c8) coordinate [pos = 0.5] (m9);
\path (c7) -- (c8) coordinate [pos = 0.5] (m10);
\path (c1) .. controls +(-110:1.25) and \control{(c3)}{(-130:1.25)} coordinate [pos = 0.5] (m11);
\path (c1) .. controls +(160:1.5) and \control{(c7)}{(180:1.5)} coordinate [pos = 0.5] (m12);
\path (c3) .. controls +(50:1.5) and \control{(c8)}{(0:2.5)} coordinate [pos = 0.5] (m14);
\path (c7) .. controls +(90:1) and \control{(c8)}{(90:1)}  coordinate [pos = 0.5] (m15);

\begin{scope}[blue!40!black]
\draw (m1) -- (m3) node [midway, sloped] {$\arrowIns{1.25}$};
\draw (m3) -- (m4) node [midway, sloped] {$\arrowOut$};
\draw (m4) -- (m5) node [pos = 0.6, sloped] {$\arrowOuts{1.25}$};
\draw (m7) -- (m9) node [midway, sloped] {$\arrowIn$};
\draw (m9) -- (m10) node [midway, sloped] {$\arrowOut$};
\draw (m10) -- (m7) node [midway, sloped] {$\arrowOuts{1.5}$};
\draw (m4) -- (m8) node [midway, sloped] {$\arrowIns{1.25}$};
\draw (m8) -- (m9) node [midway, sloped] {$\arrowOuts{1.25}$};
\draw (m9) -- (m4) node [midway, sloped] {$\arrowIn$};
\draw (m7) .. controls +(150:0.75) and \control{(m12)}{(40:0.75)} node [midway, sloped] {$\arrowOuts{1.25}$};
\draw (m12).. controls + (-50:0.75) and \control{(m3)}{(-150:0.75)} node [midway, sloped] {$\arrowIns{1.5}$};
\draw (m3) .. controls +(170:0.3) and \control{(m7)}{(-170:0.5)} node [midway, sloped] {$\arrowOuts{1.25}$};
\draw (m10) .. controls +(135:0.35) and \control{(m15)}{(-135:0.35)} node [midway, sloped] {$\arrowOut$};
\draw (m15) .. controls +(-45:0.35) and \control{(m10)}{(45:0.35)} node [midway, sloped] {$\arrowOut$};
\draw (m5) .. controls +(40:0.65) and \control{(m8)}{(-40:0.65)} node [midway, sloped] {$\arrowOuts{1.5}$};
\draw (m8) .. controls + (60:0.75) and \control{(m14)}{(130:0.75)} node [midway, sloped] {$\arrowOuts{1.25}$};
\draw (m12) .. controls +(-135:1) and \control{(m11)}{(-180:1)} node [pos = 0.6, sloped] {$\arrowOuts{1.5}$};
\draw (m14) .. controls +(45:1) and \control{(m15)}{(70:1)} node [midway, sloped] {$\arrowIns{1.5}$};
\end{scope}

\draw [green!40!black, smooth cycle, tension = 0.75] plot coordinates {($(m11)+(135:0.2)$) ($(m1)+(150:0.2)$) ($(m5)+(90:0.15)$) ($(m14)+(60:0.2)$) ($(m14)+(-60:0.2)$) ($(m5)+(-90:0.2)$) ($(m1)+(-80:0.4)$) ($(m11)+(-45:0.2)$)};

\draw [red!50!purple] (-1.25,1.25) .. controls +(-90:0.5) and \control{(-0.75,1.25)}{(-90:0.5)} .. controls +(90:1) and \control{(0.7,2.5)}{(180:1)} .. controls +(0:0.5) and \control{(0.65,3)}{(0:0.5)} .. controls +(180:1) and \control{(-1.25,1.25)}{(90:1.5)};

\foreach \i in {1,3,4,5}{\fill [black] (m\i) circle (\e*1.5 cm); \fill [blue!80!black] (m\i) circle (\e cm);};
\foreach \i in {7,...,12}{\fill [black] (m\i) circle (\e*1.5 cm); \fill [blue!80!black] (m\i) circle (\e cm);};
\foreach \i in {14,15}{\fill [black] (m\i) circle (\e*1.5 cm); \fill [blue!80!black] (m\i) circle (\e cm);};

\node [red!50!purple] at (-1.2,2.65) {$\mathcal{F}_a$};
\node [green!30!black] at (1,-1) {$\mathcal{F}_b$};
\node [blue!40!black] at (2.5,2.5) {$H$};
\end{scope}
\end{tikzpicture}
\caption{Definition of $\mathcal{F}_a$, $\mathcal{F}_b$, and $H$ from $a$ and $b$.}
\label{pic_proof_ancestry_dir}
\end{center}
\end{figure}

Applying Menger’s theorem on digraphs from $\mathcal{F}_a$ to $\mathcal{F}_b$ yields a family $P$ of $k = \mcut(\mathcal{F}_a \rightarrow \mathcal{F}_b; H)$ of edge-disjoint paths. In $H’ = H \smallsetminus P$ consider $X \subset V(H)$ realising the minimal cut $\mcut(\mathcal{F}_b \rightarrow \mathcal{F}_a; H’)$ such that $\mathcal{F}_a \subset X$. Then notice that each path $p$ of $P$ must satisfy $|E(p) \cap \cut_H (X \rightarrow X^c)| = |E(p) \cap \cut_H (X^c \rightarrow X)|+ 1$, otherwise its endpoint cannot be in $\mathcal{F}_b \subset X^c$.

Assume by contradiction that $|\cut_{H’} (X^c \rightarrow X)| < k$. By construction, this is equivalent to $|\cut_H (X^c \rightarrow X)| - |E(P) \cap \cut_H (X^c \rightarrow X)| < k$. By the remark above and because paths of $P$ are edge-disjoint, this is equivalent to $|\cut_H (X^c \rightarrow X)| - |E(P) \cap \cut_H (X \rightarrow X^c)| + |P| < k$. Hence, $|\cut_H (X \rightarrow X^c)| - |E(P) \cap \cut_H (X, X^c)| < k - |P|$ because of the property of directed medial graphs. The left side of the inequality is non-negative, thus $|P| < k$ which is absurd. It follows that $\mcut(\mathcal{F}_b \rightarrow \mathcal{F}_a ; H’) \geq k$ and a second application of Menger’s theorem in $H’$ completes the family of paths that we consider: we have $k$ paths from $\mathcal{F}_a$ to $\mathcal{F}_b$ and $k$ paths from $\mathcal{F}_b$ to $\mathcal{F}_a$.

Since $G$ is embedded, $P$ is a family of edge-disjoint direct embedded paths connecting vertices from $\mathcal{F}_b$ and $\mathcal{F}_a$ where the first edge of each path belongs to one of $T(b)$ and $T(a)$ and the last one belongs to the other cut. Hence, each path starts in $\mathcal{F}_i$, goes through the edge $x \in T(i)$ whose subdivision yields $c_x$ in $\tilde{G^i}$, $H$, then the edge $y \in T(j)$ whose subdivision yields $c_y$ in $\tilde{G^j}$ and ends on $\mathcal{F}_j$ for $\{ i,j \} = \{ a,b \}$. This way, we associate to each vertex $c_x$ incident to an incoming edge, a vertex $c_y$ induced by $T(a)$, and the subpath $P_x$ of $P$ starting with $x$ and ending at the vertex $c_y$ on the last one of its edge. We associate to each vertex $c_x$ incident to an outgoing edge, a vertex $c_y$ induced by $T(a)$, and the subpath $P_x$ of $P$ starting with $c_y$ and ending with the edge $x$.

By the identity, $G^b$ is obviously an embedded immersion minor of $G^a$. By definition of medial digraph the family of directed path $(P_x)_{x \in T(b)}$ is tangent and by construction, they respect the order at their starting point in $G^b$ so that, by Lemma~\ref{lem_equiv_plan_imm_dir} $\tilde{G}^b$ is an embedded directed immersion of $\tilde{G}^a$, \ie, $b \preccurlyeq a$.
\end{proof}

\paragraph*{Nash-Williams' argument}
The proof of Proposition~\ref{prop_imm_emb_dirmed} is the very same as the proof of Theorem~\ref{th_emb_immersion_bounded} with directed paths which are handled through Lemma~\ref{lem_ancestry_cut_dig}. Again, gluing leaving digraphs is easier than in the abstract case since the embedding is much more constraint.

\begin{proof}[Proof of Proposition~\ref{prop_imm_emb_dirmed}]
Let $(G_n)_{n \in \N}$ be a sequence of medial digraphs such that their carving-width is bounded by $k \in \N$. For each $n \in \N$, thanks to Proposition~\ref{prop_disc_decomp}, we choose a disc-linked carving-decomposition $T_n$ of $G_n$ of carving-width $\cwid (G_n)$ and set the forest $\mathcal{T} = \cup_{n \in \N} T_n$. Let us assume by contradiction that $(G_n)_{n \in \N}$ is not well-quasi-ordered by embedded directed immersion.

Hence, the set of roots of $\mathcal{T}$ is an antichain for $\preccurlyeq$ and $\mathcal{E} (\mathcal{T})$ is not well-quasi-ordered by $\preccurlyeq$. By Lemma~\ref{lem_ancestry_cut_dig}, for a pair of edges of $\mathcal{T}$ if an edge $a$ is an ancestor of $b$, then $b \preccurlyeq a$. Since the number of graphs on a fixed vertex is finite and there is a finite number of embeddings of a planar graph up to isotopy, there cannot be an infinite descending sequence of graphs for $\preccurlyeq$. Hence, as $w$ is bounded by $k$ on $\mathcal{E}(\mathcal{T})$, we can apply Lemma~\ref{lem_lem_trees} and exhibit an antichain $(a_n)_{n \in \N}$ such that $\cup_{n \in \N} ~\text{ch}_{\mathcal{T}} (a_n)$ is well-quasi-ordered by $\preccurlyeq$.

Since our trees are cubic, each set $\text{ch} (a_n)$ either consists of two edges, a left child $\ell_n$ and a right child $r_n$, or is empty. However, the edges incident to leaves are well-quasi-ordered since the number of leaving graphs for such edges are in finite number (they have at most $k$ edges). Hence, there is at most a finite number of edges incident to leaves in $(a_n)_{n_\N}$ and each sequence of $(\ell_n)_{n \in \N}, (r_n)_{n \in \N}$ is infinite and well-quasi-ordered by $\preccurlyeq$. 

Up to extracting subsequences of $(a_i)_{n \in \N}$ and hereby $(\ell_n)_{n \in \N}, (r_n)_{n \in \N}$, we can assume that $\forall i \in \N, \ell_i \preccurlyeq \ell_{i+1}$, $\forall i \in \N, r_i \preccurlyeq r_{i+1}$, and $\forall i \in \N, \wid (\ell_i) = \wid (r_{i+1})$, $\wid (\ell_i) = \wid (r_{i+1})$, and $\wid (a_i) = \wid (a_{i+1})$ (each one of these properties may require an extraction). In other words, the sequences $(\ell_n)_{n \in \N}$ and $(r_n)_{n \in \N}$ are ascending chains, and each sequence has a unique value by $w$. For simplicity of notations, we consider a subset of the sequence $(G_n)_{n \in \N}$ and index it such that we can design by $G_{i}$ the graph of $(G_n)_{n \in \N}$ such that $a_i$ is an edge of a carving-decomposition of $G_{i}$.

We now exploit our embedded structure. Let us set $i \in \N$. By the disc property, there exists a Jordan curve $\gamma_{a_i}$ in the plane such that $\gamma_{a_i}$ only meets $G_{i}$ once on the interior of each edge of $T_{i}(a_i)$. Furthermore, up to switching orientation, $\gamma_{a_i}$ intersects $T_{i}(a_i)$. Similarly, we define $\gamma_{\ell_i}$ and $\gamma_{r_i}$, which intersect $T_{i}(\ell_i)$ and $T_{i}(r_i)$. Since our decompositions satisfy the disc property, vertices of $G_{i}^{\ell_i}$ and $G_{i}^{r_i}$ both lie on the same side of $\gamma_{a_i}$. Furthermore, it is possible to isotope $\gamma_{\ell_i}$, $\gamma_{r_i}$, and $\gamma_{a_i}$ so that each point of any $\gamma$ curve belongs to another one. Their union forms a $\theta$-curve as depicted in Figure~\ref{pic_def_theta_curve_dir}. By comparing the endpoint in $G_{i}^{\ell_i}$ and the order around it, it is possible to establish a bijection between edges met along $\gamma_{\ell_i}$ with respect to $G^{a_i}$ and edges met along $\gamma_{\ell_i}$ with respect to $G^{a_i}$. This bijection can then be extended to an injection to edges met by $\theta_i$. The same arguments apply for $\theta_i$ and $\gamma_{r_i}$. In the following we will use these bijection and map to $\theta$-curves without emphasising it.

\begin{figure}[ht]
\begin{center}
\begin{tikzpicture}
\clip (-0.5,-1.7) rectangle (10.5,2.3);
\def\e{0.075}

\coordinate (cl) at (0.5,-0.15);
\coordinate (cr) at (2.4,-0.15);
\coordinate (l2) at ($(cl)+(-45:0.5)$);
\coordinate (l3) at ($(cl)+(70:0.5)$);
\coordinate (r4) at ($(cr)+(80:0.5)$);
\coordinate (r3) at ($(cr)+(130:0.5)$);
\coordinate (r2) at ($(cr)+(170:0.5)$);
\coordinate (r1) at ($(cr)+(-130:0.5)$);
\coordinate (c) at (1.5,2);

\filldraw [red!50!purple, fill opacity = 0.1] (cl) circle (0.5 cm);
\filldraw [blue!80!black, fill opacity = 0.1] (cr) circle (0.5 cm);

\begin{scope}
\clip (1.45,-0.15) circle (1.9 cm and 1.2 cm);
\draw [black!50] (l3) -- ++(80:0.75) node [pos = 0.4, sloped] {$\arrowIns{1.25}$};
\draw [black!50] (r3) -- ++(90:1) node [pos = 0.4, sloped] {$\arrowIns{1.25}$};
\draw [black!50] (r3) -- ++(120:1) node [pos = 0.4, sloped] {$\arrowIns{1.5}$};
\draw [black!50] (r4) -- ++(80:1) node [pos = 0.2, sloped] {$\arrowOuts{1.25}$};
\draw [black!50] (l2) -- (r1) node [pos = 0.3, sloped] {$\arrowOuts{1.25}$};
\draw [black!50] (l3) -- (r2) node [pos = 0.3, sloped] {$\arrowOuts{1.25}$};
\draw [black!50] (l2) -- (r2) node [pos = 0.75, sloped] {$\arrowIns{1.25}$};
\end{scope}
\path [name path=e2] (l3) -- ++(80:0.75);
\path [name path=e3] (r3) -- ++(90:1);
\path [name path=e4] (r3) -- ++(120:1);
\path [name path=e5] (r4) -- ++(80:1);
\path [name path=e6] (l2) -- (r1);
\path [name path=e7] (l3) -- (r2);
\path [name path=e8] (l2) -- (r2);

\foreach \i in {2,3}{\fill [red!50!purple] (l\i) circle (\e cm);};
\foreach \i in {1,2,3,4}{\fill [blue!80!black] (r\i) circle (\e cm);};

\draw [thick, green!40!black, name path=blob] (1.45,-0.15) circle (1.9 cm and 1.2 cm);
\draw [thick, green!40!black, name path=vert] (1.45,1.05) -- (1.45,-1.35);

\foreach \i in {2,3,4,5}{
\path [name intersections={of=blob and e\i,total=\tot}]
\foreach \s in {1,...,\tot}{coordinate (n\i) at (intersection-\s)}; 
\fill [green!40!black] (n\i) circle (\e*0.75 cm);
};

\foreach \i in {6,7,8}{
\path [name intersections={of=vert and e\i,total=\tot}]
\foreach \s in {1,...,\tot}{coordinate (n\i) at (intersection-\s)}; 
\fill [green!40!black] (n\i) circle (\e*0.75 cm);
};

\node at (cl) [red!50!purple] {$G_{i}^{\ell_i}$};
\node at (cr) [blue!80!black] {$G_{i}^{r_i}$};
\node [green!40!black] at (3.3,-1) {$\theta_i$};
\node at (1.5,1.5) {\Large $\tilde{G}_{i}^{a_i}$};

\begin{scope}[xshift = 8cm, yshift = -0.15cm]
\filldraw [red!50!purple, fill opacity = 0.1] (-0.8,0.63) arc (90:270:1 and 0.63) -- cycle;
\filldraw [blue!80!black, fill opacity = 0.1] (0.8,0.63) arc (90:-90:1 and 0.63) -- cycle;

\begin{scope}[black!50]
\clip (0,0) circle (2.4 cm and 1.39cm);
\draw (-0.25,0.4) -- (0.25,0.6) node [pos = 0.75, sloped] {$\arrowOuts{0.65}$};
\draw (-0.25,0.4) -- (0.25,0.1) node [pos = 0.25, sloped] {$\arrowIns{0.65}$};
\draw (-0.25,-0.3) -- (0.25,-0.3) node [pos = 0.75, sloped] {$\arrowOuts{0.65}$};
\draw (-0.8,0.63) -- (-0.6,0.5) node [pos = 0.5, sloped] {$\arrowOuts{0.65}$};
\draw [dotted, thick] (-0.6,0.5) -- (-0.25,0.4);
\draw (-0.8,-0.2) -- (-0.6,-0.2) node [pos = 0.5, sloped] {$\arrowOuts{0.65}$};
\draw [dotted, thick] (-0.6,-0.2) -- (-0.25,-0.3);
\draw (-0.8,-0.2) -- (-0.6,0) node [pos = 0.65, sloped] {$\arrowIns{0.65}$};
\draw [dotted, thick] (-0.6,0) -- (-0.25,0.4);

\draw (0.8,0.4) -- (0.6,0.5) node [pos = 0.5, sloped] {$\arrowOuts{0.65}$};
\draw [dotted, thick] (0.6,0.5) -- (0.25,0.6);
\draw (0.8,0.4) -- (0.6,0.2) node [pos = 0.65, sloped] {$\arrowIns{0.65}$};
\draw [dotted, thick] (0.6,0.2) -- (0.25,0.1);
\draw (0.8,-0.3) -- (0.6,-0.3) node [pos = 0.5, sloped] {$\arrowOuts{0.65}$};
\draw [dotted, thick] (0.6,-0.3) -- (0.25,-0.3);

\draw (-0.8,0.63) -- (-0.75,0.83) node [pos = 0.65, sloped] {$\arrowIns{0.65}$};
\draw [dotted, thick] (0-0.75,0.83) -- (-0.7,1.09);
\draw (-0.7,1.09) -- +(80:0.5) node [pos = 0.25, sloped] {$\arrowIns{0.65}$};

\path (0.8,0.63) arc (90:80:1 and 0.63) coordinate (d1);
\draw (d1) -- +(100:0.2) coordinate (d2) node [pos = 0.5, sloped] {$\arrowIns{0.65}$};
\draw [dotted, thick] (d2) -- +(110:0.25) coordinate (d3);
\draw (d3) -- +(110:0.4) node [pos = 0.2, sloped] {$\arrowIns{0.65}$};

\path (0.8,0.63) arc (90:80:1 and 0.63) coordinate (d1);
\draw (d1) -- +(40:0.25) coordinate (d2) node [pos = 0.75, sloped] {$\arrowIns{0.65}$};
\draw [dotted, thick] (d2) -- +(90:0.25) coordinate (d3);
\draw (d3) -- +(110:0.4)  node [pos = 0.25, sloped] {$\arrowOuts{0.65}$};

\path (0.8,0.63) arc (90:55:1 and 0.63) coordinate (d1);
\draw (d1) -- +(70:0.25) coordinate (d2) node [pos = 0.5, sloped] {$\arrowOuts{0.65}$};
\draw [dotted, thick] (d2) -- +(100:0.2) coordinate (d3);
\draw (d3) -- +(100:0.4) node [pos = 0.15, sloped] {$\arrowIns{0.65}$};
\end{scope}

\path [name path=e1] (-0.7,1.09) -- +(80:0.5);

\path (0.8,0.63) arc (90:80:1 and 0.63) coordinate (d1);
\path (d1) -- +(100:0.2) coordinate (d2);
\path (d2) -- +(110:0.25) coordinate (d3);
\path [name path=e3] (d3) -- +(110:0.4);

\path (0.8,0.63) arc (90:80:1 and 0.63) coordinate (d1);
\path (d1) -- +(40:0.25) coordinate (d2);
\path(d2) -- +(90:0.25) coordinate (d3);
\path [name path=e4] (d3) -- +(110:0.4);

\path (0.8,0.63) arc (90:55:1 and 0.63) coordinate (d1);
\path (d1) -- +(70:0.25) coordinate (d2);
\path (d2) -- +(100:0.2) coordinate (d3);
\path [name path=e5] (d3) -- +(100:0.4);
\path [name path=e6] (-0.25,0.4) -- (0.25,0.6);
\path [name path=e7] (-0.25,0.4) -- (0.25,0.1);
\path [name path=e8] (-0.25,-0.3) -- (0.25,-0.3);

\foreach \i in {1,2,3,4}{\fill [blue!80!black] (r\i) circle (\e cm);};

\draw [thick, green!40!black, name path=blob] (0,0) circle (2.4 cm and 1.39cm);
\draw [thick, green!40!black, name path=vert] (0,1.39) -- (0,-1.39);

\foreach \i in {1,3,4,5}{
\path [name intersections={of=blob and e\i,total=\tot}]
\foreach \s in {1,...,\tot}{coordinate (n\i) at (intersection-\s)}; 
\fill [green!40!black] (n\i) circle (\e*0.75 cm);
};

\foreach \i in {6,7,8}{
\path [name intersections={of=vert and e\i,total=\tot}]
\foreach \s in {1,...,\tot}{coordinate (n\i) at (intersection-\s)}; 
\fill [green!40!black] (n\i) circle (\e*0.75 cm);
};

\node at (-1.25,0) [red!50!purple] {$G_{i}^{\ell_i}$};
\node at (1.25,0) [blue!80!black] {$G_{i}^{r_i}$};
\node at (-0.5,-1) [red!50!purple] {$G_{j}^{\ell_j}$};
\node at (0.5,-1) [blue!80!black] {$G_{j}^{r_j}$};
\node [green!40!black] at (2.3,-1) {$\theta_j$};
\node at (0,1.9) {\Large $\tilde{G}_{j}^{a_j}$};
\end{scope}
\end{tikzpicture}
\caption{Left: Definition of $\theta_i$ from $\tilde{G}_{i}^{a_i}$. Right: Visualisation of $\tilde{G}_{i}^{a_i}$ as an embedded immersion of $\tilde{G}_{j}^{a_j}$.}
\label{pic_def_theta_curve_dir}
\end{center}
\end{figure} 
 
Let us set $j > i$, then we have two embedded immersions from $\tilde{G}_{i}^{\ell_i}$ to $\tilde{G}_{j}^{\ell_j}$ and from $\tilde{G}_{i}^{r_i}$ to $\tilde{G}_{j}^{r_j}$. The sizes of cuts $|T_{i}(r_i)|$, $|T_{i}(\ell_i)|$, and $|T_{i}(a_i)|$ determine the sizes of $|T_{i}(\ell_i) \cap T_{i}(r_i)|$, $|T_{i}(a_i) \cap T_{i}(r_i)|$, and $|T_{i}(\ell_i) \cap T_{i}(a_i)|$. Since $|T_{i}(x_i)| = |T_{j}(x_j)|$ by construction for $x \in \{ \ell, r, a\}$, we conclude that the sizes of intersections of pairs of cuts also match between $\tilde{G}_{i}^i$ and $\tilde{G}_{j}^j$: $|T_{i}(x_i) \cap T_{i}(y_i)| = |T_{j}(x_j) \cap T_{j}(y_j)|$ for $x,y \in \{ \ell, r, a \}^2$. 

By Lemma~\ref{lem_equiv_plan_imm_dir}, we have ordered immersions $\phi_\ell, \phi_r$ where the images of edges are families of tangent directed paths $P_{\ell,j}$ and $P_{r,j}$. For each edge $e$ in $T_{j}(\ell_j) \cap T_{j}(r_j)$, that is the edges met on the vertical part of $\theta_j$, there exists a path $p_{\ell,j}$ of the families yielded by Lemma~\ref{lem_equiv_plan_imm_dir} within $G^{\ell_j}_{j}$ ending with $e$. Similarly, we have a path $p_{r,j}$ within $G^{r_j}_{j}$ reaching the same edge so that we can merge these paths to get a path $p_e$. Indeed, a path leaving $D_{\ell_j}$ at edge $p$ is a path incoming $D_{r_j}$ and vice versa, so we can merge them.

These directed paths are tangent since their inner vertices are also inner vertices in $P_{\ell,j}$ and $P_{r,j}$, they are also tangent paths in $G_{j}^j$, this proves that $G_{i}^i$ is an embedded immersion of $G_{j}^j$. Then, we identify the non-merged vertices $s_x$ of $\tilde{G}_{j}^{\ell_j}$ with their equivalent among the remaining non-merged vertices $s_y$ of $\tilde{G}_{j}^{a_j}$ on the proper arc of $\theta_j$. We do similarly with the ones remaining in $\tilde{G}_{j}^{r_j}$. We then have edge-disjoint directed paths joining $G^{a_j}_j$ to vertices $s_x$ on $\partial D_{a_j}$ which respect the order around the starting vertices of $G^{a_j}_j$. By Lemma~\ref{lem_equiv_plan_imm_dir} $\tilde{G}_{i}^{a_i}$ is an embedded directed immersion of $\tilde{G}_{j}^{a_j}$. Hence, $a_i \preccurlyeq a_j$, which is a contradiction and concludes our proof.
\end{proof}

\subsection{Unbounded branch-width.}
\label{subsec_high_bw}

The aim of this sections is to prove that plane graphs with unbounded branch-width are well-quasi-ordered by embedded minor. We do so by proving first that any plane graph can be obtained as embedded minor of a big enough embedded grid: Proposition~\ref{prop_gm_grid}. Then, we prove that if the graphs have unbounded branch-width, we can always find an embedded grid as embedded minor of a graph of the family via Proposition~\ref{prop_find_grid}. 

So for a family of plane graphs, either the branch-width of graphs of this family is bounded and handled by Proposition~\ref{prop_gm_bounded_bw}, or it is not bounded and arbitrarily large grids can be found as embedded minors of members of the family by Proposition~\ref{prop_find_grid}. Hence, any graph of the family is in relation with a graph of the family with high enough branch-width by Proposition~\ref{prop_gm_grid}. 

\begin{proposition}\label{prop_gm_grid}
Let $G$ be a plane graph. Then, there exists $n \in \N$ such that $G$ is an embedded minor of the $n \times n$ plane grid.
\end{proposition}

\begin{proposition}\label{prop_find_grid}
For all $n \in \N$, there exists $k \in \N$ such that if $G$ is a plane graph with branch-width at least $k$, then the $n \times n$ plane grid $G_{n,n}$ is an embedded minor of $G$.
\end{proposition}

As said before, everything needed here is fairly intuitive and already proven elsewhere, we need only check that the minor defined there can be realised as embedded minor.

\paragraph*{Finding a grid as embedded minor.}
The \emphdef{abstract grid $n \times m$} is the graph with $nm$ vertices $v_{i,j}$ where $(v_{i,j},v_{i',j'})$ is an edge of this grid if and only if $|i' - i| + |j - j'| = 1$. The abstract grid has a unique embedding in $\Sp^2$:

\begin{lemma}\label{lem_uniq_grid}
For all $n \in \N$, there is a unique embedding of $n \times n$ grid in $\Sp^2$.
\end{lemma}

\begin{proof}
The proof is trivial for $n \leq 1$ and is a consequence of the Jordan–Schoenflies theorem for $n=2$. Let $G$ be the abstract grid of size $n \times n$. If $n \geq 2$, for $1 \leq i,j < n$, let $C_{i,j}$ be the cycle $G[\{v_{i,j}, v_{i+1,j}, v_{i+1,j+1} , v_{i,j+1}\}]$ and $C$ be the cycle $G[\{v_{i,j} | \{i,j\} \cap \{1,n\} \neq \varnothing \}]$. 

Let us assume by contradiction that there exists an embedding $\bar{G}$ of $G$ such that $C_{i,j}$ is not the boundary of a face of $\bar{G}$. Since $C_{i,j}$ is a cycle which is not the boundary of a face, there should be at least a vertex on both sides of its embedding (it cannot be only an edge). Hence, $G \smallsetminus C_{i,j}$ is disconnected since there is no path connecting the vertices on different sides of $C_{i,j}$. This is absurd since $G \smallsetminus C_{i,j}$ is connected for $n > 2$.

Similarly, the graph $G \smallsetminus C$ is an $(n-2) \times (n-2)$ grid, which is connected so that any embedding of $G$ has $C$ as the boundary of a face. 

Each edge is present twice in $C \cup_{1 \leq i,j < n} C_{i,j}$ so that all faces of every embedding of $G$ are the same: $G$ has a unique embedding on $\Sp^2$.
\end{proof}

A \emphdef{wall} is similar to a grid except that vertical edges are placed on staggered rows instead of being aligned, it is a $3$-regular graph. Formally, the wall of \emphdef{height} $n$ for $n$ even is defined as the graph on vertices $\{ v_{i,j} |1 \leq i \leq n+1, 1 \leq j \leq 2n+2, (i,j) \not \in \{(1,2n+2),(n+1,1)\} \}$ where vertices $v_{i,j}$ and $v_{i',j'}$ are adjacent if $i = i'$ and $|j - j'| = 1$ or if $j = j'$, $|i - i'| = 1$, and $\min (i,i') + j$ is even. For $i,j$ such that $i+j$ is even and $i < n, j \leq 2n$, the \emphdef{brick-cycle} at $v_{i,j}$ is the cycle $v_{i,j} \rightarrow v_{i,j+1} \rightarrow v_{i,j+2} \rightarrow v_{i+1,j+2} \rightarrow v_{i+1,j+1} \rightarrow v_{i+1,j} \rightarrow v_{i,j}$. The same proof as Lemma~\ref{lem_uniq_grid} shows that brick-cycles are necessarily boundaries of faces of any embedding of a wall in $\Sp^2$, as well as a remaining cycle bounding a face of high degree.

The \emphdef{plane grid} $G_{n,n}$ is the embedding of the abstract $n \times n$ grid where the high-degree face of the embedding is the unbounded face of the embedding. Similarly, the \emphdef{plane wall} $W_{n}$ of height $n$ is the embedding of the wall of height $n$ where the high-degree face of the embedding is the unbounded face of the embedding. See for example the plane grid and the plane wall on Figure~\ref{pic_ex_wall_graph}. A \emphdef{subdivided wall} is a wall on which any number of edge subdivisions has been done.

\begin{figure}[ht]
\begin{center}
\begin{tikzpicture}[scale =0.75]
\def\p{6} \def\xst{0.4} \def\yst{0.5} \def\e{0.1}
\modulo{\p+1}{\p+2}{\pp}
\modulo{2*\p+1}{2*\p+3}{\q}
\modulo{2*\p+2}{2*\p+3}{\qq}

\foreach \i in {1,...,\pp}{
\foreach \j in {1,...,\qq}{
\coordinate (c-\i-\j) at (\j*\xst,\i*\yst);
};};

\begin{scope}
\draw (c-1-1) -- (c-1-\q);
\foreach \i in {2,...,\p}{\draw (c-\i-1) -- (c-\i-\qq);};
\draw (c-\pp-2) -- (c-\pp-\qq);
\foreach \i in {1,3,...,\p}{
\foreach \j in {1,3,...,\q}{
\modulo{\i+1}{\q}{\ii}
\draw (c-\i-\j) -- (c-\ii-\j);};
};
\foreach \i in {2,4,...,\p}{
\foreach \j in {2,4,...,\qq}{
\modulo{\i+1}{\q}{\ii}
\draw (c-\i-\j) -- (c-\ii-\j);};
};
\end{scope}

\draw [scale =0.5] (-10,1) grid +(6,6); 
\end{tikzpicture}
\caption{Left: the plane grid $G_{7,7}$. Right: the plane wall $W_6$.}
\label{pic_ex_wall_graph}
\end{center}
\end{figure}

It is quite well-known that branch-width is equivalent to treewidth~\cite[(5.1)]{Graph_Minors_X}, in fact for any graph $G$, $\bwid(G) < \text{tw} (G)$. Hence we will be able to use the following proposition which results from~\cite{Graph_Minors_V} and its formulated as such in~\cite{Chudnovsky_immersion_grid}:

\begin{proposition}{\cite[(1.3)]{Chudnovsky_immersion_grid}}\label{prop_find_wall}
For all $n \geq 2$ there exists $k$ such that every graph with treewidth at least $k$ contains a subdivided wall of height $h$.
\end{proposition}

It remains to explain of to get a plane grid out of such subdivided wall obtained as a subgraph.

\begin{proof}[Proof of Proposition~\ref{prop_find_grid}]
By Proposition~\ref{prop_find_wall} and equivalence between branchwidth and treewidth, there exists $k$ such that if $G$ has branch-width $k$ or more, then $G$ contains a subdivision $W$ of the Wall of height $2n$ as a subgraph. We delete every edge of $G \smallsetminus W$ so that we end up with an embedding of $W$. Then, we perform embedded contraction to have an embedding of the wall instead of one of the subdivided wall. Then, we perform embedded contractions on each edge $(v_{i,2j},v_{i,2j+1})$ for $1 \leq i \leq n+1$ and $1 \leq j \leq n+1$, so that $G$ is now an embedded grid of size $2n \times 2n$.

By Lemma~\ref{lem_uniq_grid}, the embedding of $G$ is the unique embedding of $G$ on $\Sp^2$ where one face $F$ has been chosen as the unbounded face. If $F$ is the high degree face of $G$, the proof is done. Otherwise, if $F$ is incident to $v_{i,j}$ and $v_{i+1, j+1}$, then the $4$ quarters $G[\{ v_{i',j'} | 1 \leq i' \leq i, 1 \leq j' \leq j\} ]$, $G[\{ v_{i',j'} | i < i' \leq 2n, 1 \leq j' \leq j\} ]$, $G[\{ v_{i',j'} | 1 \leq i' \leq i, j < j' \leq 2n\} ]$, and $G[\{ v_{i',j'} | i < i' \leq 2n, j < j' \leq 2n \} ]$   are plane grids as illustrated in Figure~\ref{pic_grid_emb_grid}.

\begin{figure}[ht]
\begin{center}
\begin{tikzpicture}[scale = 0.65]
\draw (0,0) grid (5,5);
\fill [opacity = 0.2] (3,0) rectangle +(1,1);
\fill [opacity = 0.2] (2,1) rectangle +(1,1);
\fill [opacity = 0.2] (4,1) rectangle +(1,1);
\fill [opacity = 0.2] (3,2) rectangle +(1,1);
\fill [opacity = 0.2, red!50!purple] (1,0) rectangle +(1,1);
\fill [opacity = 0.2, red!50!purple] (1,2) rectangle +(1,1);
\fill [opacity = 0.2, red!50!purple] (2,3) rectangle +(1,1);
\fill [opacity = 0.2, red!50!purple] (4,3) rectangle +(1,1);
\fill [opacity = 0.2, green!70!purple] (0,1) rectangle +(1,1);
\fill [opacity = 0.2, green!70!purple] (0,3) rectangle +(1,1);
\fill [opacity = 0.2, green!70!purple] (1,4) rectangle +(1,1);
\fill [opacity = 0.2, green!70!purple] (3,4) rectangle +(1,1);

\fill [blue!80!black] (3.5,1.5) circle (0.05cm);
\node at (2.5,-1) {$\Sp^2$};

\draw [-Stealth, thin] (5.5,2.5) -- +(1,0);

\begin{scope}[xshift = 7cm]
\coordinate (c31) at (0,0);
\coordinate (c41) at (5,0);
\coordinate (c32) at (0,5);
\coordinate (c42) at (5,5);
\coordinate (c20) at ($(0,0)+(60:1)+(30:1)$);
\coordinate (c50) at ($(c41)+(120:1)+(150:1)$);
\coordinate (c51) at ($(c41)+(120:1)$);
\coordinate (c40) at ($(c41)+(150:1)$);
\coordinate (c21) at ($(0,0)+(60:1)$);
\coordinate (c30) at ($(0,0)+(30:1)$);
\coordinate (c11) at ($(c21)+(60:0.75)$);
\coordinate (c10) at ($(c11)+(30:0.75)$);
\coordinate (c23) at ($(c32)+(-60:1)+(-30:1)$);
\coordinate (c22) at ($(c32)+(-60:1)$);
\coordinate (c33) at ($(c32)+(-30:1)$);
\coordinate (c12) at ($(c22)+(-60:0.75)$);
\coordinate (c13) at ($(c12)+(-30:0.75)$);
\coordinate (c34) at ($(c33)+(-30:0.75)$);
\coordinate (c24) at ($(c34)+(-60:0.75)$);
\coordinate (c52) at ($(c42)+(-120:1)$);
\coordinate (c43) at ($(c42)+(-150:1)$);
\coordinate (c53) at ($(c43)+(-120:1)$);
\coordinate (c44) at ($(c43)+(-150:0.75)$);
\coordinate (c54) at ($(c43)+(-120:0.75)+(-150:0.75)$);
\coordinate (c14) at ($(c23)+(-30:0.5)+(-60:0.5)$);
\coordinate (c35) at ($(c34)+(-30:0.5)$);
\coordinate (c45) at ($(c44)+(-150:0.5)$);
\coordinate (c55) at ($(c44)+(-150:0.5)+(-120:0.5)$);
\coordinate (c05) at ($(c14)+(-60:0.375)+(-30:0.375)$);
\coordinate (c04) at ($(c14)+(-60:0.5)$);
\coordinate (c01) at ($(c11)+(60:0.5)$);
\coordinate (c00) at ($(c01)+(30:0.5)$);
\coordinate (c02) at ($(c12)+(-60:0.5)$);
\coordinate (c03) at ($(c12)+(-60:0.5)+(-30:0.5)$);
\coordinate (c15) at ($(c14)+(-30:0.5)$);
\coordinate (c25) at ($(c34)+(-30:0.5)+(-60:0.5)$);

\fill [opacity = 0.2] (0,0) -- ($(0,0)+(60:1)$) -- ($(0,5)+(-60:1)$) -- (0,5) -- (0,0);
\fill [opacity = 0.2] (0,0) -- ($(0,0)+(30:1)$) -- ($(5,0)+(150:1)$) -- (5,0) -- (0,0);
\fill [opacity = 0.2] (5,0) -- ($(5,0)+(120:1)$) -- ($(5,5)+(-120:1)$) -- (5,5) -- (5,0);
\fill [opacity = 0.2] (0,5) -- ($(0,5)+(-30:1)$) -- ($(5,5)+(-150:1)$) -- (5,5) -- (0,5);
\fill [opacity = 0.2, red!50!purple] (c10) -- (c20) -- (c21) -- (c11) -- (c10);
\fill [opacity = 0.2, red!50!purple] (c12) -- (c22) -- (c23) -- (c13) -- (c12);
\fill [opacity = 0.2, red!50!purple] (c33) -- (c34) -- (c24) -- (c23) -- (c33);
\fill [opacity = 0.2, red!50!purple] (c43) -- (c53) -- (c54) -- (c44) -- (c43);
\fill [opacity = 0.2, green!70!purple] (c01) -- (c11) -- (c12) -- (c02) -- (c01);
\fill [opacity = 0.2, green!70!purple] (c34) -- (c44) -- (c45) -- (c35) -- (c34);
\fill [opacity = 0.2, green!70!purple] (c03) -- (c13) -- (c14) -- (c04) -- (c03);
\fill [opacity = 0.2, green!70!purple] (c14) -- (c24) -- (c25) -- (c15) -- (c14);

\foreach \i in {0,...,4}{
\foreach \j in {0,...,4}{
\modulo{\i+1}{6}{\ii}
\modulo{\j+1}{6}{\jj}
\draw (c\i\j) -- (c\ii\j);
\draw (c\i\j) -- (c\i\jj);
};};
\draw (c50) -- (c51) -- (c52)  -- (c53)  -- (c54) -- (c55) -- (c45) -- (c35) -- (c25) -- (c15) -- (c05);

\draw [-Stealth, thin] (6,2.5) -- +(1,0);

\node at (2.5,-1) {$\R^2$};
\fill [blue] (5.5, 3.5) circle (0.075);
\end{scope}

\begin{scope}[xshift = 13.5cm]
\draw (1,1) grid (4,4);

\fill [opacity = 0.2, green!70!purple] (3,2) rectangle +(1,1);
\fill [opacity = 0.2, green!70!purple] (2,1) rectangle +(1,1);
\fill [opacity = 0.2, red!50!purple] (2,3) rectangle +(1,1);
\fill [opacity = 0.2, red!50!purple] (1,2) rectangle +(1,1);

\fill [blue!80!black] (4.5,2.5) circle (0.075cm);
\node at (2.5,-1) {$\R^2$};
\end{scope}
\end{tikzpicture}
\caption{The $5 \times 5$ grid embedded in $\Sp^2$, and its embedding in $\R^2$ where the chosen unbounded face is the face of $\Sp^2$ containing the blue dot.}
\label{pic_grid_emb_grid}
\end{center}
\end{figure}

Then, the quarter containing the more vertices contains the plane grid $G_{n,n}$, removing everything outside of is the final step to obtain $G_{n,n}$ as embedded minor of $G$.
\end{proof}

\paragraph*{Embedded minor in a plane grid.}
We detail here a construction of~\cite{Robertson_excluding} and remark that it is well adapted for embedded minors we claim no novelty for neither the methods nor the constructions.

\begin{lemma}{\cite[(1.3)]{Robertson_excluding}}\label{lem_hamil_embed}
Let $H$ be a simple Hamiltonian planar graph embedded in $\Sp^2$ with $n$ vertices, then $H$ is an embedded minor of the unique embedding of $G_{n,n}$ in the sphere.
\end{lemma}

\begin{proof}
We detail the construction of~\cite{Robertson_excluding}. One can first add edges to $H$ so that $H$ is triangulated, \ie, each face has degree $3$. Such edges can be deleted in the end after all the contraction operations. Let us first chose an Hamiltonian cycle $C$ of $H$, \ie, a cycle $C$ passing exactly through each vertex of $H$. Let us number these vertices $\{v_1, \ldots, v_n\}$ by picking a first one and following the order along $C$. Since $V(C) = V(H)$, and $C$ is embedded, every edge of $H$ is either on the bounded face $A$ of $C$, the unbounded face $B$ of $C$, or is an edge of $C$. Let us define $H_A$ and $H_B$, the subgraphs induced by $A \cup C$ and $B \cup C$ respectively, \ie, $H_A = H \cap (C \cup A)$ and $H_B = H \cap (C \cup B)$.

Set $G = G_{n,n}$, the $n \times n$ plane grid. Let us define subgraphs of $G$ thanks to the quantities:
\begin{center}
$
\left\lbrace
\begin{array}{rl}
\text{For } 1 < k < n, & i(k) = \min \{1 \leq i \leq n | (v_k, v_i) \in E(H_A) \}\\
\text{For } 1 < k < n, & j(k) = \min \{1 \leq j \leq n | (v_k, v_j) \in E(H_B) \}\\
\text{For } 1 < k < n, & j'(k) = \max \{1 \leq j \leq n | (v_k, v_j) \in E(H_A) \}\\
\text{For } 1 < k < n, & i'(k) = \max \{1 \leq i \leq n | (v_k, v_j) \in E(H_B) \}\\
\end{array}
\right.
$
\end{center}

Then, for $1<k<n$, $X_k = G[ \{v_{i,k} | i(k) < i < i'(k) \} \cup \{v_{k,j} | j(k) < j < j'(k) \} ]$, $X_1 = G[ \{v_{i,1} | 1 \leq i < n \} \cup \{v_{1,j} | 1 \leq i < n \} ]$, and $X_n = G[ \{v_{i,n} | 1 \leq i \leq n \} \cup \{v_{n,j} | 1 \leq i \leq n \} ]$. See each black component centred at $v_{k,k}$ in Figure~\ref{pic_ex_constr_emb_grid} for illustration.

\begin{figure}[ht]
\begin{center}
\begin{tikzpicture}[scale=1]
\def\e{0.05}
\coordinate (c1) at (4,3.2);
\coordinate (c2) at ($(c1)+(-45:1)$);
\coordinate (c3) at ($(c2)+(-110:1.5)$);
\coordinate (c4) at ($(c3)+(-160:1)$);
\coordinate (c5) at ($(c4)+(170:1.5)$);
\coordinate (c6) at ($(c5)+(130:1.2)$);
\coordinate (c7) at ($(c6)+(40:1.7)$);
\coordinate (c8) at ($(c7)+(-30:1)$);

\foreach \i in {1,...,7}{
\modulo{\i + 1}{9}{\ii} 
\draw [blue!80!black] (c\i) -- (c\ii);
};
\draw [blue!80!black] (c1) -- (c8);

\begin{scope}[red!80!black]
\draw (c1) -- (c4);
\draw (c2) -- (c4);
\draw (c4) -- (c8);
\draw (c5) -- (c8);
\draw (c6) -- (c8);

\draw (c1) -- (c7);
\draw (c1) .. controls +(120:1.5) and \control{(c6)}{(85:1.25)}; 
\draw (c2) .. controls +(80:2.5) and \control{(c6)}{(100:2.5)}; 
\draw (c2) .. controls +(-45:1.5) and \control{(c4)}{(-30:1.5)}; 
\draw (c4) .. controls +(-130:1.5) and \control{(c6)}{(-100:1.5)}; 
\end{scope}

\fill [opacity = 0.125] (c1) -- (c2) -- (c3) -- (c4) -- (c5) -- (c6) -- (c7) -- (c8) -- cycle;

\node at ($(c1)+(30:0.25)$) {$8$};
\node at ($(c2)+(30:0.25)$) {$1$};
\node at ($(c3)+(-30:0.25)$) {$2$};
\node at ($(c4)+(-90:0.25)$) {$3$};
\node at ($(c5)+(-150:0.25)$) {$4$};
\node at ($(c6)+(180:0.25)$) {$5$};
\node at ($(c7)+(70:0.25)$) {$6$};
\node at ($(c8)+(90:0.25)$) {$7$};

\node [blue!80!black] at ($(c3)+(70:0.5)+(0:0.25)$) {$C$};
\node at ($(c4)+(120:1)$) {$A$};
\node at ($(c6)+(120:0.75)$) {$B$};

\foreach \i in {1,...,8}{\fill (c\i) circle (\e cm);};

\draw [-Stealth, black!50] (5.5,2) -- +(1,0);

\begin{scope}[xshift = 7cm, scale = 0.5]
\def\e{0.15}
\draw [opacity=0.3] (1,1) grid (8,8);

\begin{scope}[very thick]
\draw (1,1) -- (1,7);
\draw (1,1) -- (7,1);
\draw (2,3) -- (7,3);
\draw (4,4) -- (6,4);
\draw (5,5) -- (5,7);
\draw (1,1) -- (1,7);
\draw (3,2) -- (3,4);
\draw (2,5) -- (6,5);
\draw (6,6) -- (6,7);
\draw (7,4) -- (7,7);
\draw (1,8) -- (8,8);
\draw (8,1) -- (8,8);
\end{scope}

\begin{scope}[very thick, red!70!black]
\draw (1,3) -- +(1,0);
\draw (1,5) -- +(1,0);
\draw (7,3) -- +(1,0);
\draw (6,4) -- +(1,0);
\draw (6,5) -- +(1,0);
\draw (3,1) -- +(0,1);
\draw (3,4) -- +(0,1);
\draw (7,3) -- +(0,1);
\draw (5,7) -- +(0,1);
\draw (6,7) -- +(0,1);
\end{scope}

\begin{scope}[very thick, blue!80!black]
\draw (7,1) -- +(1,0);
\draw (6,6) -- +(1,0);
\draw (2,1) -- +(0,1);
\draw (2,2) -- +(0,1);
\draw (4,3) -- +(0,1);
\draw (5,4) -- +(0,1);
\draw (6,5) -- +(0,1);
\draw (7,7) -- +(0,1);
\end{scope}

\fill [opacity = 0.125] (2,1) -- (2,3) -- (4,3) -- (4,4) -- (5,4) -- (5,5) -- (6,5) -- (6,6) -- (7,6) -- (7,8) -- (8,8) -- (8,1) -- cycle;

\foreach \i in {1,...,8}{\fill (\i,\i) circle (\e cm);};
\foreach \i in {1,...,8}{\node at (\i,0.5) {$\i$};};
\foreach \i in {1,...,8}{\node at (0.5,\i) {$\i$};};
\end{scope}
\end{tikzpicture}
\caption{Construction of the embedding of a plane Hamiltonian graph.}
\label{pic_ex_constr_emb_grid}
\end{center}
\end{figure}

It follows directly from the definitions that each $X_i$ contains $v_{i,i}$. Furthermore, $X_k$ crosses $X_\ell$ with $k < \ell$ if and only if $j'(k) > i(\ell)$ or $i'(k) > j(\ell)$. However, this would imply the existence of edges crossing themselves in $A$ and $B$ respectively, which is absurd since the graph is embedded. Hence the $X_i$'s are disjoint (they are disjoint from $X_1$ and $X_n$ by definition). 

If an edge $(k,\ell)$ with $k < \ell$ exists in $A$, then $k \leq j'(\ell)$ and $i(\ell) \leq k$. Since $H$ is triangulated, one of these inequalities must be strict. Indeed, the faces incident to $(k, \ell)$ are bounded by another edge with either $k$ or $\ell$ as an endpoint and one of these edges will imply that $i(\ell) < k$ or $k < j'(\ell)$. It follows that one of $(v_{k,\ell-1}, v_{k,\ell})$ or $(v_{k,\ell}, v_{k+1,\ell})$ is incident to both $X_k$ and $X_\ell$. Furthermore, $H$ is simple so that these edges cannot appear between $k$ and $k+1$. Hence such edges are to the right of the boundary of the diagonal faces $v_{i,i}, v_{i+1,i}, v_{i+1,i+1}, v_{i,i+1},$ of $G_{n,n}$. Similarly, if an edge $(k,\ell)$ with $k < \ell$ exists in $B$, then one of $(v_{\ell,k+1}, v_{k,\ell})$ or $(v_{\ell,k-1}, v_{\ell, k})$ is incident to both $X_k$ and $X_\ell$ on the left of the diagonal faces of $G_{n,n}$ (see Figure~\ref{pic_ex_constr_emb_grid} for example).

Similar arguments prove that we can always connect $X_i$ and $X_{i+1}$  for $i < n$ on the boundary of diagonal faces. By joining $X_n$ and $X_1$ with $(v_{1,n-1}, v_{1,n})$, we enclose all edges stemming from $A$ by the cycle stemming from $H$ and the sets $X_i$. In addition, edges of $B$ lie outside of this cycle. Hence, by deleting all unused edges of $G_{n,n}$ and contracting all sets $X_i$ to vertices, we obtain an embedding of $H$ with all edges on the proper side of $C$, and in the proper order by construction (they appear in the proper order along each $X_i$). Hence, $H$ is an embedded minor of $G_{n,n}$.
\end{proof}

\begin{proof}[Proof of Proposition~\ref{prop_gm_grid}]
The construction of Lemma~\ref{lem_hamil_embed} can be translated to the context of embedding in $\R^2$ but then, the unbounded face of the embedded minor is the one incident to the edge $(1,n)$. Hence, if $H$ is a plane Hamiltonian graph with an edge of a Hamiltonian cycle incident to the unbounded face, we know that it is an embedded minor of a plane grid. In fact, this is easy to reduce any plane Hamiltonian graph to this case by replacing a vertex incident to the unbounded face by an edge taking part in the Hamiltonian cycle. And preserve the property when the graph is triangulated.

Hence, there remains to show that a plane graph $G$ can be obtained as an embedded minor of a Hamiltonian plane graph. Again, this is explained in~\cite{Robertson_excluding}. Whitney's theorem~\cite{Whitney_hamiltonian} ensures that if a simple triangulated graph has no separating triangle, then it is Hamiltonian. There is a finite number of them; subdividing an edge of such triangle and connecting the newly created vertex to the two non-neighbouring vertices of incident faces (see Figure~\ref{pic_proof_trig_removal}) reduces the number of separating triangles by at least one. A finite number of such operation allows to remove all separating triangles and thus obtain a triangulated Hamiltonian graph. These operation can be reversed by edge deletion and contraction so that $G$ is an embedded minor of this graph.

\begin{figure}[ht]
\begin{center}
\begin{tikzpicture}
\def\e{0.05}
\coordinate (c0) at (-30:1);
\coordinate (c1) at (90:1);
\coordinate (c2) at (-150:1);
\coordinate (c3) at (0,0);
\coordinate (c4) at ($(c0)+(90:1)$);

\draw [red!50!purple] (c1) -- (c2);
\draw (c1) -- (c3);
\draw (c1) -- (c4);
\draw [red!50!purple] (c1) -- (c0);
\draw (c3) -- (c0);
\draw (c0) -- (c4);
\draw [red!50!purple] (c0) -- (c2);
\node [red!50!purple] at (30:0.65) {$e$};

\foreach \i in {0,...,4}{\fill (c\i) circle (\e cm);}

\draw [black!50,-Stealth] (1.5,0) -- +(0.75,0);

\begin{scope}[xshift = 3.5cm]
\coordinate (c0) at (-30:1);
\coordinate (c1) at (90:1);
\coordinate (c2) at (-150:1);
\coordinate (c3) at (0,0);
\coordinate (c4) at ($(c0)+(90:1)$);
\coordinate (c5) at ($(c1)!0.5!(c0)$);

\draw (c1) -- (c2);
\draw (c1) -- (c3);
\draw (c1) -- (c4);
\draw (c1) -- (c0);
\draw (c3) -- (c0);
\draw (c0) -- (c4);
\draw (c0) -- (c2);
\draw (c3) -- (c5) --(c4);

\foreach \i in {0,...,4}{\fill (c\i) circle (\e cm);}
\fill [green!50!black] (c5) circle (\e cm);
\end{scope}
\end{tikzpicture}
\caption{Removing a separating triangle.}
\label{pic_proof_trig_removal}
\end{center}
\end{figure}

Finally, if $G$ is non-simple, subdivide it enough so that it becomes simple before applying the constructions above. This concludes this proof.
\end{proof}

\subparagraph*{Acknowledgements.}
This project was supported by the ANR project AlgoKnot: ANR-20-CE48-0007. Corentin Lunel was partially supported by the GA\v{C}R grant 25-16847S. The authors want to thank Pierre Dehornoy and Arnaud de Mesmay for introducing them to this problem. In addition, the authors want to thank the anonymous referees for their helpful comments and remarks.

\bibliographystyle{plainurl}
\bibliography{biblio}

\begin{thebibliography}{10}

\bibitem{baader_positive}
Sebastian Baader.
\newblock Positive braids of maximal signature.
\newblock {\em L’Enseignement Math{\'e}matique}, 59(3):351--358, 2014.

\bibitem{baader_divide}
Sebastian Baader and Pierre Dehornoy.
\newblock Minor theory for surfaces and divides of maximal signature.
\newblock 2012.
\newblock URL: \url{https://arxiv.org/abs/1211.7348}.

\bibitem{baader_minor_quasipositive}
Sebastian Baader, Pierre Dehornoy, and Livio Liechti.
\newblock Minor theory for quasipositive surfaces.
\newblock In Athanase Papadopoulos, editor, {\em Essays in geometry, dedicated
  to Norbert A'Campo}, pages 351--358. EMS Pres, 2023.

\bibitem{Burton_2018}
Benjamin~A. Burton, Cl{\'{e}}ment Maria, and Jonathan Spreer.
\newblock Algorithms and complexity for {T}uraev{\textendash}{V}iro invariants.
\newblock {\em Journal of Applied and Computational Topology}, 2(1-2):33--53,
  2018.

\bibitem{Kawarabayashi_immer_embedable}
Dario Cavallaro, Ken ichi Kawarabayashi, and Stephan Kreutzer.
\newblock Well-quasi-ordering {E}ulerian digraphs embeddable in surfaces by
  strong immersion, 2025.
\newblock \href {https://arxiv.org/abs/2509.26260} {\path{arXiv:2509.26260}}.

\bibitem{Cavallaro_wqo_bd_cw}
Dario Cavallaro, Ken ichi Kawarabayashi, and Stephan Kreutzer.
\newblock Well-quasi-ordering eulerian digraphs: Bounded carving width, 2026.
\newblock URL: \url{https://arxiv.org/abs/2605.07468}, \href
  {https://arxiv.org/abs/2605.07468} {\path{arXiv:2605.07468}}.

\bibitem{Chudnovsky_immersion_grid}
M.~Chudnovsky, Zdenek Dvor{\'a}k, Tereza Klimo{\v {s}}ov{\'a}, and Paul~D.
  Seymour.
\newblock Immersion in four-edge-connected graphs.
\newblock {\em J. Comb. Theory B}, 116:208--218, 2013.

\bibitem{Conder_min_gen}
Marston D.~E. Conder and Klara Stokes.
\newblock New methods for finding minimum genus embeddings of graphs on
  orientable and non-orientable surfaces.
\newblock {\em Ars Math. Contemp.}, 17:1--35, 2019.

\bibitem{Davey_book_order}
B.A. Davey and H.A. Priestley.
\newblock {\em Introduction to Lattices and Order}.
\newblock Cambridge Univ. Press, 2002.

\bibitem{Lunel_defect}
Pierre Dehornoy, Corentin Lunel, and Arnaud de~Mesmay.
\newblock Hopf arborescent links, minor theory, and decidability of the genus
  defect.
\newblock {\em Discrete \& Computational Geometry}, 75:294 -- 319, 2023.

\bibitem{Devos_min_degree_immersion}
Matt Devos, Zden\v{e}k Dvo\v{r}\'{a}k, Jacob Fox, Jessica McDonald, Bojan
  Mohar, and Diego Scheide.
\newblock A minimum degree condition forcing complete graph immersion.
\newblock {\em Combinatorica}, 34(3):279–298, 2014.

\bibitem{DeVos_forbid_clique}
Matt DeVos, Jessica McDonald, Bojan Mohar, and Diego Scheide.
\newblock A note on forbidding clique immersions.
\newblock {\em Electronic Journal of Combinatorics}, 20, 2012.

\bibitem{Diestel_Graph_Theory}
Reinhard Diestel.
\newblock {\em Graph theory}.
\newblock Number 173 in Graduate texts in mathematics. Springer, New York, 5th
  edition, 2016.

\bibitem{Diestel_Walls_strong}
Reinhard Diestel, Raphael~W. Jacobs, Paul Knappe, and Paul Wollan.
\newblock A grid theorem for strong immersions of walls.
\newblock {\em Journal of Graph Theory}, 110, 2023.

\bibitem{Bodlaender_planar_efficient}
Frederic Dorn, Eelko Penninkx, Hans~L Bodlaender, and Fedor~V Fomin.
\newblock Efficient exact algorithms on planar graphs: Exploiting sphere cut
  branch decompositions.
\newblock In {\em European Symposium on Algorithms}, pages 95--106. Springer,
  2005.

\bibitem{Dvorak_structure_th}
Zden\v{e}k Dvo\v{r}\'{a}k and Paul Wollan.
\newblock A structure theorem for strong immersions.
\newblock {\em Journal of Graph Theory}, 83(2):152–163, 2015.

\bibitem{Dvorak_minor_book}
Zdeněk Dvoř{\'a}k.
\newblock Graph minors: Theory and applications.
\newblock {\em Graph Minors}, 2025.

\bibitem{Geelen_wqo_branchwidth}
James~F. Geelen, A.M.H. Gerards, and Geoff Whittle.
\newblock Branch-width and well-quasi-ordering in matroids and graphs.
\newblock {\em J. Comb. Theory Ser. B}, 84(2):270–290, 2002.

\bibitem{Hannie_PhD}
Stefan Hannie.
\newblock {\em 2-Regular Digraphs on Surfaces}.
\newblock PhD thesis, Simon Fraser University, 2018.

\bibitem{Tarjan_planar}
John Hopcroft and Robert Tarjan.
\newblock Efficient planarity testing.
\newblock {\em Journal of the ACM}, 21(4):549–568, 1974.

\bibitem{Huszar_JSJ}
Krist{\'{o}}f Husz{\'{a}}r and Jonathan Spreer.
\newblock On the width of complicated {JSJ} decompositions.
\newblock In {\em Proceedings of the 39th International Symposium on
  Computational Geometry (SoCG 2023)}, volume 258 of {\em LIPIcs}, pages
  42:1--42:18. Schloss Dagstuhl - Leibniz-Zentrum f{\"{u}}r Informatik, 2023.

\bibitem{Huszar_treewidth}
Krist{\'o}f Husz{\'a}r, Jonathan Spreer, and Uli Wagner.
\newblock On the treewidth of triangulated 3-manifolds.
\newblock In {\em Proceedings of the 33th International Symposium on
  Computational Geometry (SoCG 2017)}, volume~99, pages 46:1--46:15. Schloss
  Dagstuhl - Leibniz-Zentrum f{\"{u}}r Informatik, 2018.

\bibitem{Jordan_lecture}
Camille Jordan.
\newblock {\em Cours d'analyse de l'{É}cole {P}olytechnique}.
\newblock 1887.

\bibitem{Kauffman_Jones_poly}
Louis~H. Kauffman.
\newblock The {J}ones polynomial, knots, diagrams, and categories.
\newblock {\em Bulletin of the American Mathematical Society}, 2023.

\bibitem{Kruskal_tree_minor}
Joseph~B. Kruskal.
\newblock Well-quasi-ordering, the tree theorem, and {V}azsonyi’s conjecture.
\newblock {\em Transactions of the American Mathematical Society}, 95:210--225,
  1960.

\bibitem{Tarjan_separator}
Richard~J. Lipton and Robert~Endre Tarjan.
\newblock Applications of a planar separator theorem.
\newblock {\em SIAM Journal on Computing}, 9(3):615--627, 1980.

\bibitem{Liu_wqo_strong_digraph}
Chun-Hung Liu and Irene Muzi.
\newblock Well-quasi-ordering digraphs with no long alternating paths by the
  strong immersion relation.
\newblock {\em J. Comb. Theory B}, 158:210--251, 2020.

\bibitem{Lunel_spherewidth}
Corentin Lunel and Arnaud de~Mesmay.
\newblock A structural approach to tree decompositions of knots and spatial
  graphs.
\newblock {\em Algebraic \& Geometric Topology}, 26(6):2015--2047, 2026.

\bibitem{Lovasz_minor_survey}
Lovász László.
\newblock Graph minor theory.
\newblock {\em Bulletin of The American Mathematical Society - BULL AMER MATH
  SOC}, 43:75--87, 2005.

\bibitem{Malowsky_Jones}
J.A. Makowsky and J.P. Mariño.
\newblock The parametrized complexity of knot polynomials.
\newblock {\em Journal of Computer and System Sciences}, 67(4):742--756,
  December 2003.

\bibitem{Manolescu_floer_diag}
Ciprian Manolescu, Peter~S. Ozsv{\'a}th, and Sucharit Sarkar.
\newblock A combinatorial description of knot floer homology.
\newblock {\em Annals of Mathematics}, 169:633--660, 2006.

\bibitem{maria_FPT_quantum}
Cl\'{e}ment Maria.
\newblock Parameterized complexity of quantum knot invariants.
\newblock In {\em Proceedings of the 37th International Symposium on
  Computational Geometry (SoCG 2021)}. Schloss Dagstuhl-Leibniz-Zentrum fuer
  Informatik, 2021.

\bibitem{mariapurcell}
Cl{\'e}ment Maria and Jessica Purcell.
\newblock Treewidth, crushing and hyperbolic volume.
\newblock {\em Algebraic \& Geometric Topology}, 19(5):2625--2652, 2019.

\bibitem{Medina_wqolink}
Carolina Medina, Bojan Mohar, and Gelasio Salazar.
\newblock Well-quasi-order of plane minors and an application to link diagrams,
  2019.
\newblock \href {https://arxiv.org/abs/1905.01830} {\path{arXiv:1905.01830}}.

\bibitem{Menger_paths}
Karl Menger.
\newblock Zur allgemeinen kurventheorie.
\newblock {\em Fundamenta Mathematicae}, 10:96--115.

\bibitem{nash-williams_kruskal}
C.~St. J.~A. Nash-Williams.
\newblock On well-quasi-ordering finite trees.
\newblock {\em Mathematical Proceedings of the Cambridge Philosophical
  Society}, 59(4):833–835, 1963.

\bibitem{Graph_Minors_V}
Neil Robertson and Paul~D. Seymour.
\newblock Graph minors. {V}. {Excluding} a planar graph.
\newblock {\em Journal of Combinatorial Theory, Series B}, 41(1):92--114,
  August 1986.

\bibitem{Graph_Minors_VII}
Neil Robertson and Paul~D. Seymour.
\newblock Graph minors. {VII}. {D}isjoint paths on a surface.
\newblock {\em Journal of Combinatorial Theory, Series B}, 45(2):212--254,
  1988.

\bibitem{Graph_Minors_IV}
Neil Robertson and {Paul D.} Seymour.
\newblock Graph minors. {IV}. {T}ree-width and well-quasi-ordering.
\newblock {\em Journal of Combinatorial Theory, Series B}, 48(2):227--254,
  1990.

\bibitem{Graph_Minors_X}
Neil Robertson and Paul~D. Seymour.
\newblock Graph minors. {X}. {Obstructions} to tree-decomposition.
\newblock {\em Journal of Combinatorial Theory, Series B}, 52(2):153--190, July
  1991.

\bibitem{Graph_Minors_XI}
Neil Robertson and {Paul D.} Seymour.
\newblock Graph minors. {XI}. {C}ircuits on a {S}urface.
\newblock {\em Journal of Combinatorial Theory. Series B}, 60(1):72--106, 1994.

\bibitem{Graph_Minors_XVII}
Neil Robertson and Paul~D. Seymour.
\newblock Graph minors: {XVII}. {T}aming a vortex.
\newblock {\em Journal of Combinatorial Theory, Series B}, 77:162--210, 1999.

\bibitem{Graph_Minor_XVI}
Neil Robertson and Paul~D. Seymour.
\newblock Graph minors. {XVI}. {E}xcluding a non-planar graph.
\newblock {\em Journal of Combinatorial Theory, Series B}, 89:43--76, 2003.

\bibitem{Graph_Minors_XX}
Neil Robertson and Paul~D. Seymour.
\newblock Graph minors. {XX}. {W}agner's conjecture.
\newblock {\em Journal of Combinatorial Theory, Series B}, 92(2):325--357,
  2004.

\bibitem{Graph_Minors_XXIII}
Neil Robertson and Paul~D. Seymour.
\newblock Graph minors. {XXIII}. {N}ash-{W}illiams' immersion conjecture.
\newblock {\em Journal of Combinatorial Theory, Series B}, 100(2):181--205,
  2010.

\bibitem{Robertson_excluding}
Neil Robertson, Paul~D. Seymour, and Robin Thomas.
\newblock Quickly excluding a planar graph.
\newblock {\em Journal of Combinatorial Theory, Series B}, 62(2):323--348,
  1994.

\bibitem{Rolfsen_Knots}
Dale Rolfsen.
\newblock {\em Knots and links}.
\newblock AMS Chelsea Pub, Providence, R.I, 2003.
\newblock OCLC: ocm52901393.

\bibitem{Seymour_Ratcatcher}
Paul~D. Seymour and Robin Thomas.
\newblock Call routing and the ratcatcher.
\newblock {\em Combinatorica}, 14(2):217--241, 1994.

\bibitem{Thomassen_min_gen_hard}
Carsten Thomassen.
\newblock The graph genus problem is {NP}-complete.
\newblock {\em Journal of Algorithms}, 10:568–576, 1989.

\bibitem{Question_mathover_2}
Lukas~Lewark (\url{https://mathoverflow.net/users/23455/lukas-lewark}).
\newblock Is the ``surface-minor'' ordering of plane graphs a
  well-quasi-ordering?
\newblock MathOverflow, 2017.
\newblock URL: \url{https://mathoverflow.net/q/275282}.

\bibitem{Question_mathover_1}
Pierre~Dehornoy (\url{https://mathoverflow.net/users/9248/pierre-dehornoy}).
\newblock Planar minor graphs.
\newblock MathOverflow, 2012.
\newblock URL: \url{https://mathoverflow.net/q/92421}.

\bibitem{Reidemeister_reide_th}
Kurt van Reidemeister.
\newblock Elementare {B}egr{\"u}ndung der {K}notentheorie (in german).
\newblock {\em Abhandlungen aus dem Mathematischen Seminar der Universit{\"a}t
  Hamburg}, 5:24--32, 1927.

\bibitem{Whitney_hamiltonian}
Hassler Whitney.
\newblock A theorem on graphs.
\newblock {\em The Annals of Mathematics}, 1931.

\bibitem{Whitney_unique}
Hassler Whitney.
\newblock 2-isomorphic graphs.
\newblock {\em American Journal of Mathematics}, 55:245--254, 1933.

\end{thebibliography}

\end{document}